\documentclass{harvard-thesis}
\usepackage{graphicx}
\usepackage{amssymb}
\usepackage{minitoc}
\usepackage{hyperref}
\usepackage{color}
\usepackage{multirow}
\usepackage{amsfonts,amsmath}
\usepackage{booktabs}
\usepackage{enumitem}
\usepackage{amsthm}
\usepackage{algpseudocode}

\usepackage{tcolorbox}
\usepackage{lmodern}
\usepackage{framed}
\usepackage[linesnumbered,ruled,vlined]{algorithm2e}

\usepackage{caption}
\usepackage{subcaption}
 \usepackage{times}
\usepackage{float}

\usepackage{xcolor}

\hypersetup{
    colorlinks=true,
    linkcolor=red!80!black,
    urlcolor=red!80!black,
    citecolor=blue
}

\usepackage{cleveref}

\newtheorem{ex}{Example}

\newtheorem{observation}{Observation}[section]

\newtheorem{theorem}{Theorem}[section]
\newtheorem{conjecture}{Conjecture}[section]
\newtheorem{open-question}{Open Question}[section]

\newtheorem{meta-question}{Meta Question}[section]

\newtheorem{proposition}{Proposition}[section]
\newtheorem{lemma}{Lemma}[section]

\newtheorem{definition}{Definition}[section]
\newtheorem{corollary}{Corollary}

\usepackage{tikz}

\usepackage{float}
\usepackage{makeidx}
\usepackage{alltt}

\definecolor{gra}{RGB}{105,105,105}

\usepackage{etoolbox}

\DeclareMathOperator{\rad}{rad}

\DeclareMathOperator{\diam}{diam}

\DeclareMathOperator{\MP}{mp}

\DeclareMathOperator{\MMP}{MP}

\usepackage{url}

\usepackage{bm}

\usepackage{cleveref}

 \usepackage{todonotes}

 \usepackage{thmtools}
\usepackage{thm-restate}

\newenvironment{proofsketch}{
  \proof}{\endproof}

  \newenvironment{claimproof}{
  \proof}{\hfill$\triangleleft$}

\makeatletter
\patchcmd{\ttlh@hang}{\parindent\z@}{\parindent\z@\leavevmode}{}{}
\patchcmd{\ttlh@hang}{\noindent}{}{}{}
\makeatother

\makeindex
\begin{document}


\title{\LARGE Multipacking on graphs and Euclidean metric space}

\author{Sk Samim Islam}
\advisor{Prof. Sandip Das}

\degree{Doctor of Philosophy}
\field{Algorithms on graphs and geometry}
\degreeyear{2025}
\degreemonth{July}

\department{Advanced Computing and Microelectronics Unit}
\university{Indian Statistical Institute}
\universitycity{Kolkata, India}


\frontpage
\abstract
\dominitoc
\tableofcontents
\listoffigures
\dedicationpage
\acknowledgments

%

\chapter{Introduction}\label{chapter:intro}\hypertarget{chapter:introhref}{}
\minitoc 


\section{Graph Theory Preliminaries} \label{defn:graph}

Graphs are the primary objects of study in this thesis. Therefore, we begin by introducing the necessary definitions and preliminaries from graph theory. For a more detailed exposition, the books by West~\cite{west} and Diestel~\cite{diestel} are excellent resources.

A \textit{graph} $G$ is a triple consisting of a \textit{vertex set} $V(G)$, an \textit{edge set} $E(G)$, and a relation that associates each edge with two (not necessarily distinct) vertices called its \textit{endpoints}. An edge whose endpoints are the same is called a \textit{loop}, and \textit{multiple edges} are distinct edges sharing the same pair of endpoints.

A \textit{simple graph} is a graph that contains neither loops nor multiple edges. We denote a simple graph $G$ as $G = (V, E)$, where $V = V(G)$ is the vertex set and $E = E(G)$ is the edge set. Here, $E$ is a set of unordered pairs of distinct vertices, and we denote an edge $e$ with endpoints $u$ and $v$ by $e = uv$ (or $e = vu$). Unless stated otherwise, all graphs considered in this thesis are assumed to be simple. Furthermore, when the context is clear, we refer to the vertex and edge sets simply as $V$ and $E$, respectively. The \textit{order} of a graph refers to the number of vertices and is denoted by $n = |V|$.

If $u$ and $v$ are endpoints of an edge $e$, we say that $u$ and $v$ are \textit{neighbours} or are \textit{adjacent}. For a vertex $u$, the \textit{open neighbourhood}, denoted by $N(u)$, is defined as $N(u) = \{ v \in V \mid uv \in E \}$. The \textit{closed neighbourhood} of $u$, denoted by $N[u]$, is defined as $N[u] = N(u) \cup \{u\}$. A vertex $v$ is said to be \textit{dominating} or \textit{universal} if $N[v] = V$.

The \textit{degree} of a vertex $v$, denoted $d(v)$, is the number of edges incident to $v$. The \textit{minimum degree} of a graph $G$ is denoted by $\delta(G) = \min_{v \in V} d(v)$, and the \textit{maximum degree} is denoted by $\Delta(G) = \max_{v \in V} d(v)$. A graph is called \textit{$k$-regular} if every vertex has degree $k$.

A \textit{homomorphism} from a graph $G$ to a graph $H$ is a function $f : V(G) \rightarrow V(H)$ (denoted $f: G \rightarrow H$) that preserves adjacency; that is, if $uv \in E(G)$, then $f(u)f(v) \in E(H)$. An \textit{isomorphism} from $G$ to $H$ is a bijective function $f : V(G) \rightarrow V(H)$ such that $uv \in E(G)$ if and only if $f(u)f(v) \in E(H)$. Two graphs $G$ and $H$ are said to be \textit{isomorphic}, denoted by $G \cong H$, if there exists an isomorphism between them.

A sequence of vertices $u_1, \ldots, u_k$ is called a \textit{path} if $u_iu_{i+1} \in E$ for all $i < k$. The vertices $u_1$ and $u_k$ are referred to as the \textit{endpoints} of the path. A path $P$ is called a \textit{$u,v$-path} if $u$ and $v$ are its endpoints. A path on $n$ vertices is denoted by $P_n$. A graph is said to be \textit{connected} if there exists a path between every pair of vertices.

A \textit{cycle} on $n$ vertices, denoted by $C_n$, is a graph isomorphic to one with vertex set $V = \{v_1, v_2, \ldots, v_n\}$ and edge set $E = \{ v_iv_{(i+1) \bmod n} \mid v_i \in V \}$. A graph is \textit{acyclic} (also called a \textit{forest}) if it contains no cycles. A \textit{tree} is a connected and acyclic graph. The \textit{girth} of a graph is the length of its shortest cycle. A \textit{chord} is an edge joining two non-consecutive vertices in a path or a cycle.

 For a graph $G=(V,E)$, $d_G(u,v)$ is the length of a shortest path joining two vertices $u$ and $v$ in  $G$, we simply write $d(u,v)$ when there is no confusion.  $\diam(G):=\max\{d(u,v):u,v\in V(G)\}$. Diameter is a path of $G$  of the length  $\diam(G)$. 
 A ball of radius $r$ around $u$, $N_r[u]:=\{v\in V:d(u,v)\leq r\}$ where $u\in V$. The \textit{eccentricity} $e(w)$  of a vertex $w$ is $\min \{r:N_r[w]=V\}$. The \textit{radius} of the graph $G$ is $\min\{e(w):w\in V\}$, denoted by $\rad(G)$.  The \textit{center} $C(G)$ of the graph $G$  is the set of all vertices of minimum eccentricity, i.e., $C(G):=\{v\in V:e(v)=\rad(G)\}$. Each vertex in the set $C(G)$ is called a \textit{central vertex} of the graph $G$.  If $P$ is a path in $G$, then we say $V(P)$ is the vertex set of the path $P$, $E(P)$ is the edge set of the path $P$, and  $l(P)$ is the length of the path $P$, i.e., $l(P)=|E(P)|$. 
 

For a graph $G$ and a subset $S \subseteq V(G)$, the \textit{induced subgraph} on $S$, denoted by $G[S]$, is the subgraph of $G$ whose vertex set is $S$ and whose edge set consists of all edges of $G$ with both endpoints in $S$. A \textit{complete graph} on $n$ vertices, denoted by $K_n$, is a graph in which every pair of distinct vertices is adjacent. A subgraph $H$ of $G$ is called a \textit{clique} if $H$ induces a complete graph. A set of vertices $S$ is called \textit{independent} if $G[S]$ contains no edges.

A graph is \textit{bipartite} if its vertex set can be partitioned into two independent sets. More generally, a graph is $k$-\textit{partite} if its vertex set can be partitioned into $k$ independent sets. A \textit{complete bipartite graph}, denoted by $K_{m,n}$, is a bipartite graph whose vertex set is partitioned into sets of sizes $m$ and $n$, with every vertex in one part adjacent to all vertices in the other part.

The distance between two vertices $u$ and $v$, denoted $d(u,v)$, is the length of a shortest $u,v$-path. If $u = v$, then $d(u,v) = 0$; if there is no $u,v$-path, then $d(u,v) = \infty$. A path is called \textit{isometric} if it is a shortest path between its endpoints.

A set of vertices $S$ in a graph $G$ is a \textit{dominating set} if every vertex not in $S$ has a neighbor in $S$. A dominating set of minimum cardinality is called a \textit{minimum dominating set}, and its size, denoted by $|S|$, is called the \textit{domination number} of $G$.

\section{Multipacking and Broadcast Domination on graphs}\label{defn:broadcast_multipacking}

Covering and packing problems are fundamental in graph theory and algorithms (See the book~\cite{cornuejols2001combinatorial}). Here, we study two dual covering and packing problems called \emph{broadcast domination} and \emph{multipacking} (See the book~\cite{haynes2021structures}). The broadcast domination problem has a natural motivation in telecommunication networks: imagine a network with radio emission towers, where each tower can broadcast information at any radius $r$  for a cost of $r$. The goal is to cover the whole network by minimizing the total cost. The multipacking problem is its natural packing counterpart and generalizes various other standard packing problems. Unlike many standard packing and covering problems, these two problems involve arbitrary distances in graphs, which makes them challenging. 




The concept of broadcast domination or dominating broadcast was initially proposed by Erwin in his Ph.D.~thesis~\cite{erwin2004dominating,erwin2001cost} in 2001 as a framework for modeling communication networks. For a graph $ G = (V, E) $ with the diameter $\mathrm{diam}(G)$, a function
	$ f : V \rightarrow \{0, 1, 2, . . . , \mathrm{diam}(G)\} $ is called a \textit{broadcast} on $ G $. Suppose $G$ is a graph with a broadcast $f$. Let $d(u,v)$ be  the length of a shortest path joining the vertices $u$ and $v$ in $G$.  We say $v\in V$ is a \textit{tower} of $G$ if $f(v)>0$.	
Suppose $u, v  \in  V$ (possibly, $ u = v $) such that $ f (v) > 0 $ and
	$ d(u, v) \leq f (v) $, then we say $v$ \textit{broadcasts} (or  \textit{dominates}) $u$, and $u$ \textit{hears} a broadcast from $v$.

For each
	vertex $ u \in V  $, if  there exists a vertex $ v $ in $ G $ (possibly, $ u = v $) such that $ f (v) > 0 $ and
	$ d(u, v) \leq f (v) $, then $ f $ is called a \textit{dominating broadcast} on $ G $.
The \textit{cost} of the broadcast $f$ is the quantity $ \sigma(f)  $, which is the sum
	of the weights of the broadcasts over all vertices in $ G $. So, $\sigma(f)=  \sum_{v\in V}f(v)$. The minimum cost of a dominating broadcast in $G$ (taken over all dominating broadcasts)  is the \textit{broadcast domination number} of $G$, denoted by $ \gamma_{b}(G) $.  So, $ \gamma_{b}(G) = \displaystyle\min_{f\in D(G)} \sigma(f)= \displaystyle\min_{f\in D(G)} \sum_{v\in V}f(v)$, where $D(G)=$ set of all dominating broadcasts on $G$. 
	
	Suppose $f$ is a dominating broadcast with $f(v)\in \{0,1\}$, $\forall v\in V(G)$, then $\{v\in V(G):f(v)=1\}$ is a \textit{dominating set} on $G$. The minimum cardinality of a dominating set is the \textit{domination number} which is denoted by $ \gamma(G) $.

An \textit{optimal broadcast} (also known as \textit{optimal dominating broadcast}, or \textit{minimum dominating broadcast}) on a graph $G$ is a dominating broadcast with a cost equal to $ \gamma_{b}(G) $.	
A dominating broadcast is \textit{efficient} if no vertex
hears broadcasts from two different towers. So, no tower can hear a broadcast from another tower in an efficient broadcast. There is a theorem that says,	for every graph there is an optimal dominating broadcast that is also efficient \cite{dunbar2006broadcasts}. 

If we take a dominating set of a graph as a set of towers where each tower has weight $1$, then this  forms a dominating broadcast of the graph. Therefore, the broadcast domination number is upper bounded by the domination number. Moreover, if we take a central vertex of a graph as a tower of strength of radius of the graph, then it forms a dominating broadcast of the graph. Therefore, the broadcast domination number can be at most the radius of the graph. We state this formally as follows.

\begin{observation}[\cite{erwin2001cost,erwin2004dominating}] \label{obs:erwin_ineq}
    If $ G $ is a connected graph of order at least 2  having radius $ \mathrm{rad}(G) $, broadcast domination number $ \gamma_{b}(G) $ and domination number $ \gamma (G) $, then \[ \gamma_{b}(G) \leq min\{\gamma (G),\mathrm{rad}(G)\}. \] 
\end{observation}


Herke and Mynhardt~\cite{herke2009radial} determined an upper bound for the broadcast domination number of a tree in terms of its number of vertices. 

\begin{theorem}[\cite{herke2009radial}]\label{thm:treebroadcastbound}
    If $T$ is a tree of order $n$, then $\gamma_{b}(T)\leq  \lceil \frac{n}{3}\rceil$.
\end{theorem} 

Observe that, if $f$ is a dominating broadcast of a spanning tree of a graph $G$, then $f$ is also a dominating broadcast of $G$. Therefore, if $G$ is a connected graph of order $n$, then $\gamma_{b}(G)\leq  \lceil \frac{n}{3}\rceil$ by Theorem \ref{thm:treebroadcastbound}.

Finding the domination number (size of the smallest dominating set) of a graph is known to be \textsc{NP-hard}~\cite{garey1979computers}. Since broadcast domination is one kind of general version of the domination problem, we might expect that determining the broadcast domination number $\gamma_b(G)$ of a graph is also \textsc{NP-hard}. 
 Surprisingly, it is not \textsc{NP-hard}! In 2006, Heggernes and Lokshtanov~\cite{heggernes2006optimal} have shown that an optimal dominating broadcast of any graph can be found in polynomial-time. 

\begin{theorem}[\cite{heggernes2006optimal}]
The broadcast domination number of any graph with $n$ vertices can be determined in $O(n^6)$ time.
\end{theorem}

Heggernes and S{\ae}ther~\cite{heggernes2012broadcast} later conjectured that the broadcast domination number can be determined in $O(n^5)$ time for general graphs. They have shown that the problem can be solved in $O(n^4)$ time for chordal graphs. The same problem can be solved in linear time for trees~\cite{brewster2019broadcast}.  Blair, Heggernes, Horton, and Manne have shown that the broadcast domination number can be computed in $O(n^3)$ time for interval graphs and in $O(nr^4)$ time
for series-parallel graphs where $r$ is the radius of the graph.
Chang and Peng~\cite{chang2010linear} later improved the result for interval graphs. They have shown that the broadcast domination number of an interval graph of vertex set size $n$ and edge set size $m$ can be computed in $O(n + m)$ time. Also, the problem can be solved in $O(n^3)$ time for strongly chordal graphs~\cite{brewster2019broadcast,yang2019limited}.

Multipacking was  introduced in Teshima’s Master’s Thesis \cite{teshima2012broadcasts} in 2012 (also see \cite{beaudou2019multipacking,haynes2021structures,cornuejols2001combinatorial,dunbar2006broadcasts,meir1975relations}). Multipacking can be seen as a refinement of the classic independent set problems. 	A \textit{multipacking} is a set $ M \subseteq V  $ in a
	graph $ G = (V, E) $ such that   $|N_r[v]\cap M|\leq r$ for each vertex $ v \in V $ and for every integer $ r \geq 1 $.  It is worth noting that in this definition, it suffices to consider values of $r$ in the set $\{1,2,\dots, \mathrm{diam}(G)\}$, where $\mathrm{diam}(G)$ denotes the diameter of the graph. Moreover, if $c$ is a central vertex of the graph $G$, then $N_r[c]=V(G)$ when $r=\mathrm{rad}(G)$. Therefore, size of any multipacking in $G$ is upper-bounded by $\mathrm{rad}(G)$. Therefore, we can further restrict $r$ (in the definition of multipacking) to the set $\{1,2,\dots, \mathrm{rad}(G)\}$, where $\mathrm{rad}(G)$ is the radius of the graph. The \textit{multipacking number} of $ G $ is the maximum cardinality of a multipacking of $ G $ and it
	is denoted by $ \MP(G) $. A \textit{maximum multipacking} is a multipacking $M$  of a graph $ G  $ such that	$|M|=\MP(G)$.

Multipacking was further studied by Hartnell and Mynhardt~\cite{hartnell2014difference} in 2014. They have shown the following.
\begin{proposition}[\cite{hartnell2014difference}]\label{prop:hartnell_ineq}
    If $ G $ is a connected graph of order at least 2  having diameter $ \diam(G) $, multipacking number  $ \MP(G) $, broadcast domination number $ \gamma_{b}(G) $, then \[\Big\lceil{\dfrac{\diam(G)+1}{3}\Big\rceil}\leq \MP(G)\leq \gamma_{b}(G). \] 
\end{proposition}

In 2013, Brewster, Mynhardt, and Teshima~\cite{brewster2013new} formulated the broadcast domination problem as an Integer Linear Program (ILP). The dual of this ILP led to the nice combinatorial problem \emph{multipacking}. Suppose $V(G)=\{v_1,v_2,v_3,\dots,v_n\}$. Let $c$ and $x$ be the vectors indexed by $(i,k)$ where $v_i\in V(G)$ and $1\leq k\leq diam(G)$, with the  entries $c_{i,k}=k$ and $x_{i,k}=1$ when $f(v_i)=k$ and $x_{i,k}=0$ when $f(v_i)\neq k$. Let $A=[a_{j,(i,k)}]$ be a matrix with the entries 
\begin{equation*}	a_{j,(i,k)}=
    \begin{cases}
        1 & \text{if }  v_j\in N_k[v_i]\\
        0 & \text{otherwise. }
    \end{cases} 
 \end{equation*}

\noindent	Hence, the broadcast domination number can be expressed as an integer linear program:  $$\gamma_b(G)=\min \{c.x :  Ax\geq \mathbf{1}, x_{i,k}\in \{0,1\}\}.$$	
	
\noindent  The \textit{maximum multipacking problem} is the dual integer program of the above problem. If $M$ is a multipacking, we define   a vector $y$ with the entries $y_j=1$ when $v_j\in M$ and $y_j=0$ when $v_j\notin M$.  So, $$\MP(G)=\max \{y.\mathbf{1} :  yA\leq c, y_{j}\in \{0,1\}\}.$$ 

We know that $\MP(G)\leq \gamma_b(G)$ from Proposition \ref{prop:hartnell_ineq}. Moreover, we have $3\MP(G)-1\geq \diam(G)\geq \rad(G)\geq \gamma_b(G)$, from the Proposition \ref{prop:hartnell_ineq} and Observation \ref{obs:erwin_ineq}. Hartnell and Mynhardt~\cite{hartnell2014difference} studied the relation between $\MP(G)$ and $\gamma_b(G)$. They have shown that $ \gamma_b(G) \leq 3\MP(G)-2$ for the connected graphs with $\MP(G)\geq 2$. This inequality was improved by Beaudou, Brewster, and Foucaud~\cite{beaudou2019broadcast}.

\begin{theorem}[\cite{beaudou2019broadcast}]\label{thm:gammaleq2mp3}
    For any connected graph $G$,  $\gamma_b(G) \leq 2\MP(G)+3.$
\end{theorem}

In 2012, Teshima~\cite{teshima2012broadcasts} showed that the graph $G$ depicted in Figure~\ref{fig:ratio2} satisfies $\gamma_b(G) = 4$ and $\MP(G) = 2$.

\begin{figure}[h]
    \centering
   \includegraphics[height=3.3cm]{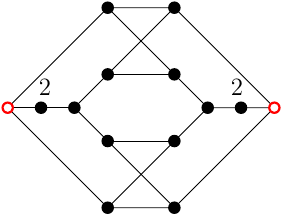}
    \caption{A connected $G$ graph with $\gamma_b(G)=4$ and $\MP(G)=2$.}
    \label{fig:ratio2}
\end{figure}

 Surprisingly, to date, no graph G has been found such that $\frac{\gamma_b(G)}{\MP(G)}>2$.  Beaudou, Brewster, and Foucaud~\cite{beaudou2019broadcast} conjectured the following.

\begin{conjecture}[\cite{beaudou2019broadcast}]
    For any connected graph $G$,  $\gamma_b(G) \leq 2\MP(G).$
\end{conjecture}

Hartnell and Mynhardt~\cite{hartnell2014difference} constructed a graph $H_k$ (depicted in Figure \ref{fig:Hk1}) as follows:  (i) The vertex set of $H_k$ is  $V(H_k)=\{w_i,y_i,x_i,u_i,v_i,z_i:1\leq i\leq 3k\}$. (ii) For each $1\leq i\leq 3k$, the vertex set $\{w_i,y_i,x_i,u_i,v_i,z_i\}$ forms a $K_{2,4}$ (complete bipartite graph) with partite sets $\{w_i,y_i\}$ and $ \{x_i,u_i,v_i,z_i\}$. (iii) For each $1\leq i\leq 3k-1$, $z_i$ and $x_{i+1}$ are adjacent. They proved that $\gamma_b(H_k)=4k$ and $\MP(H_k)=3k$. 

\begin{figure}[ht]
    \centering
    \includegraphics[width=\textwidth]{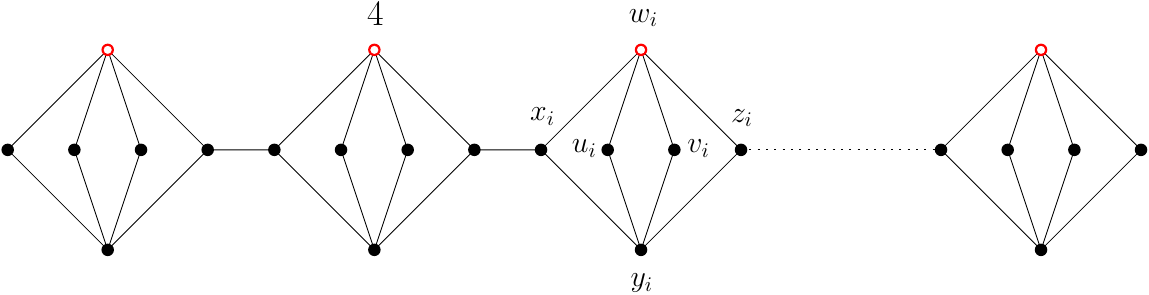}
    \caption{The $H_k$ graph with $\gamma_b(H_k)=4k$ and $\MP(H_k)=3k$. The set $\{w_i:1\leq i\leq 3k\}$ is a maximum multipacking of $H_k$. This family of graphs was constructed by Hartnell and Mynhardt~\cite{hartnell2014difference}.} 
    \label{fig:Hk1}
\end{figure}

This yields the following result.

\begin{theorem}[\cite{hartnell2014difference}]\label{thm:gamma_mp_ratio_4_3}
For every integer \( k \geq 1 \), there exists a connected graph \( H_k \) satisfying $\gamma_b(H_k) = 4k \text{ and } \MP(H_k) = 3k.$  Thus, the following hold for the graph \( H_k \).

\begin{enumerate}
    \item[1)] $  \gamma_b(H_k) - \MP(H_k) = k.$
    
    \item[2)] $\gamma_b(H_k)/\MP(H_k) = \frac{4}{3}.$
    
\end{enumerate}
\end{theorem}

Theorem \ref{thm:gammaleq2mp3} and Theorem \ref{thm:gamma_mp_ratio_4_3} yields the following result.

\begin{corollary}[\cite{beaudou2019broadcast,hartnell2014difference}]\label{cor:gammabG/mpG_general} For connected graphs, 
$$\frac{4}{3}\leq\lim_{\MP(G)\to \infty}\sup\Bigg\{\frac{\gamma_{b}(G)}{\MP(G)}\Bigg\}\leq 2.$$
\end{corollary}

In 2025, Rajendraprasad et al.~\cite{rajendraprasad2025multipacking} have determined the exact value of the above limit, which is $2$. Their proof is based on the $n$-dimensional hypercube\footnote{An $n$-dimensional hypercube $Q_n$ is a graph with vertex set $\{0,1\}^n$, where two vertices are adjacent if they differ in exactly one coordinate.} $Q_n$, for which they established $\left\lfloor \tfrac{n}{2} \right\rfloor \leq \MP(Q_n) \leq \tfrac{n}{2} + 6\sqrt{2n}.$ Earlier, in 2019, Bre{\v{s}}ar and {\v{S}}pacapan~\cite{brevsar2009broadcast} showed that $\gamma_b(Q_n)=n-1$ for $n\geq 3$, and $\gamma_b(Q_n)=n$ for $n\in\{1,2\}$. Combining these two results yields the following.


\begin{theorem}[\cite{rajendraprasad2025multipacking}]\label{value_of_lim_sum_is_2}
    For connected graphs, 
$$\lim_{\MP(G)\to \infty}\sup\Bigg\{\frac{\gamma_{b}(G)}{\MP(G)}\Bigg\}= 2.$$
\end{theorem}


The value of the expression $\lim_{\MP(G)\to \infty}\sup\{\gamma_{b}(G)/\MP(G)\}$ is also known for certain subclasses of graphs. In particular, Brewster et al.~\cite{brewster2019broadcast} proved that $\gamma_b(G)=\MP(G)$ whenever $G$ is strongly chordal, implying that the value of the expression is $1$ for strongly chordal graphs. Since trees form a subclass of strongly chordal graphs, the equality $\gamma_b(G)=\MP(G)$ also holds for trees.

 From algorithmic point of view, the \textsc{Multipacking} problem is the following.
    
\medskip
\medskip
\noindent
\fbox{%
  \begin{minipage}{\dimexpr\linewidth-2\fboxsep-2\fboxrule}
  \textsc{ Multipacking} problem
  \begin{itemize}[leftmargin=1.5em]
    \item \textbf{Input:} An undirected graph $G = (V, E)$, an integer $k \in \mathbb{N}$.
    \item \textbf{Question:} Does there exist a multipacking $M \subseteq V$ of $G$ of size at least $k$?
  \end{itemize}
  \end{minipage}%
}
\medskip


It remained an open problem for over a decade whether the \textsc{Multipacking} problem  is \textsc{NP-complete} or solvable in polynomial time. This question was repeatedly addressed by numerous authors (\cite{brewster2019broadcast,foucaud2021complexity,teshima2012broadcasts,yang2015new,yang2019limited}),  yet the problem remained as an unsolved challenge. In 2025, we (Das, Islam and Lokshtanov~\cite{das2025hardnessMultipacking}) answered this question by proving that the \textsc{Multipacking} problem  is \textsc{NP-complete} which is one of the parts of this thesis (the detailed discussion is in Chapter \ref{chapter:NPcomplete}). In 2019, Beaudou, Brewster, and Foucaud~\cite{beaudou2019broadcast} presented an approximation algorithm for computing multipacking of general graphs.

\begin{theorem}[\cite{beaudou2019broadcast}]
    There is a $(2 + o(1))$-approximation algorithm for computing multipacking in general graphs.
\end{theorem}

 However, polynomial-time algorithms for computing maximum multipacking are known for trees and more generally, strongly chordal graphs~\cite{brewster2019broadcast}. Even for trees, the algorithm for finding a maximum multipacking   is very non-trivial.


  A \textit{$d$-packing} is a set of vertices in $G$, such that the shortest path between each pair of the vertices from the set has at least $d +1$ edges. When $d =1$, a $1$-packing is nothing but an independent set.
 The usual $2$-packing is a $1$-multipacking of a graph. In 2013, Brewster, Mynhardt, and Teshima \cite{brewster2013new} obtained a
generalization of 2-packings called $r$-multipackings.

The notion of $r$-multipacking was first introduced in Teshima’s Master’s thesis~\cite{teshima2012broadcasts} in 2012, where it was originally referred to as "$k$-multipacking". In this work, we adopt the notation $r$ instead of $k$, as $k$ is more commonly used to denote the solution size or query parameter in decision versions of optimization problems.

Later, in 2018, Henning, MacGillivray, and Yang~\cite{henning2018k} used the term "$k$-multipacking" to refer to a different problem. Throughout this thesis, we follow the original definition of "$k$-multipacking" as introduced by Teshima~\cite{teshima2012broadcasts}, and we refer to it as "$r$-multipacking" for consistency with standard notation.


An \textit{$r$-multipacking} is a set $ M \subseteq V  $ in a
	graph $ G = (V, E) $ such that   $|N_s[v]\cap M|\leq s$ for each vertex $ v \in V $ and for every integer $s$ where $ 1\leq s \leq r $. The \textit{$r$-multipacking number} of $ G $ is the maximum cardinality of an $r$-multipacking of $ G $ and it is denoted by $ \MP_r(G) $. The \textsc{$r$-Multipacking} problem is as follows:

 \medskip   
\medskip
\noindent
\fbox{%
  \begin{minipage}{\dimexpr\linewidth-2\fboxsep-2\fboxrule}
  \textsc{ $r$-Multipacking} problem
  \begin{itemize}[leftmargin=1.5em]
    \item \textbf{Input:} An undirected graph $G = (V, E)$, an integer $k \in \mathbb{N}$.
    \item \textbf{Question:} Does there exist an $r$-multipacking $M \subseteq V$ of $G$ of size at least $k$?
  \end{itemize}
  \end{minipage}%
}
\medskip

    If $M$ is an $r$-multipacking, where $r=\diam(G)$, the diameter of $G$, then $M$ is called a \textit{multipacking} of $G$, we have seen this earlier. It is obvious that $\MP_r(G)\geq \MP(G)$ for any $r$, since a multipacking is always an $r$-multipacking of $G$ by their definitions.

For a given integer $d\geq 2$, a \textit{distance-$d$ independent set} of graph $G$ is the same as a $(d-1)$-packing of $G$. This is a natural generalization of the concept of "independent set". Given an undirected graph $G$ and a positive integer $k$, the \textsc{Distance-$d$ Independent Set (D$d$IS)} problem asks if $G$ has a distance-$d$ independent set of size at least $k$. Therefore, D$2$IS problem is nothing but the \textsc{Independent Set} problem.  Eto et al.~\cite{eto2014distance} showed that,  for $d\geq 3$, the D$d$IS problem is \textsc{NP-complete} for planar bipartite graphs of maximum degree $3$. Moreover, they have shown that,  D$d$IS is \textsc{NP-complete} for chordal graphs when $d$ is odd and $d\geq 3$. Since D$3$IS is nothing but $2$-packing or $1$-multipacking, therefore all of these hardness results hold for $1$-multipacking. The complexity of \textsc{$r$-Multipacking} problem, for $r\geq 2$, was unknown. In 2025, we (Das et al.~\cite{das2025rMultipacking}) answered this question by proving that the \textsc{$r$-Multipacking} problem, for $r\geq 2$,  is \textsc{NP-complete} which is one of the parts of this thesis (the detailed discussion is in Chapter \ref{chapter:k_multi}). 

Foucaud, Gras, Perez, and Sikora~\cite{foucaud2021complexity} studied the complexity of \textsc{Broadcast Domination} and \textsc{Multipacking} problems on digraph (or directed graph). A \emph{dominating broadcast} of a digraph \( D \) is a function \( f \colon V(D) \to \mathbb{N} \) such that for every vertex \( v \in V(D) \), there exists a vertex \( t \in V(D) \) with \( f(t) > 0 \) and a directed path from \( t \) to \( v \) of length at most \( f(t) \). The \emph{cost} of \( f \) is the sum of \( f(v) \) over all vertices \( v \in V(D) \). Whereas, a \emph{multipacking} in a digraph \( D \) is a set \( M \subseteq V(D) \) such that for every vertex \( v \in V(D) \) and for every integer \( r \geq 1 \), there are at most \( r \) vertices from \( M \) within directed distance at most \( r \) from \( v \). Let \textsc{Broadcast Domination} denote the decision problem of whether a given digraph \( D \) admits a dominating broadcast of cost at most \( k \), and let \textsc{Multipacking} denote the problem of determining whether \( D \) has a multipacking of size at least \( k \).

In 2021, Foucaud, Gras, Perez, and Sikora~\cite{foucaud2021complexity} have shown that both \textsc{Broadcast Domination} and \textsc{Multipacking} problems are \textsc{NP-complete} for digraphs, even for planar layered acyclic digraphs of bounded maximum degree. Furthermore, they proved that when parameterized by the cost or size of the solution, the problems are \textsc{W[2]-hard} and \textsc{W[1]-hard} respectively.

\section{Multipacking and Broadcast Domination in geometric domain}\label{defn:multipacking_geometry}
The problem of locating facilities, both desirable and undesirable, is of utmost importance due to its huge number of applications in the real world~\cite{wolf2011facility}. Allocation of restricted resources within an arena has been extensively studied across various models and methodologies since the 1980s. Drezner and Wesolowsky did a "maximin" problem where a certain facility must serve a given set of users but must be placed in such a manner that its minimum distance from the users is maximized (see~\cite{drezner1980maximin,drezner1983location}). Rodriguez et al.~\cite{rodriguez2006general} gave a generic model for the obnoxious facility location problem inside a given polygonal region. Rakas et al. \cite{rakas2004multi} devised a multi-objective model to determine the optimal locations for undesirable facilities.  The book by Daskin \cite{daskin1997network} contains a nice collection of problems, models, and algorithms regarding obnoxious facility location problems in the domain of supply chain networks. Here, we study the \emph{Multipacking} \cite{teshima2012broadcasts} problem in Euclidean space, which resembles the obnoxious facility location problem in computational geometry. 
 


Given a point set $P \subseteq \mathbb{R}^d$ of size $n$, the \emph{neighborhood of size $r$} of a point $v \in P$ is the subset of $P$ that consists of the $r$ nearest points of $v$ and $v$ itself. This subset is denoted by $N_r[v]$. The distance between points are measured using the Euclidean distance metric. A natural and realistic assumption here is the distances between all pairs of points are distinct. Thus the $r$-th neighbor of every point is unique. This assumption is easy to achieve through some form of controlled perturbation (see \cite{mehlhorn2006reliable,seidel1998nature}). A \textit{multipacking} is a subset $ M $ of a point set $ P  $, such that $|N_r[v]\cap M|\leq (r+1)/2$ for each point $ v \in P $ and for every integer $  r\geq 1  $, i.e. $M$ contains at most half of the number of elements of $N_r[v]$. It is worth noting that in this definition, it suffices to consider values of $r$ in the set $\{1,2,\dots, n-1\}$. The \emph{multipacking number} of $ P $ is the maximum cardinality of a multipacking of $ P $ and is denoted by $ \MP(P) $.   The \textsc{Multipacking} problem in $\mathbb{R}^d$ is the following.
    
\medskip
\medskip
\noindent
\fbox{%
  \begin{minipage}{\dimexpr\linewidth-2\fboxsep-2\fboxrule}
  \textsc{ Multipacking} problem (in geometry)
  \begin{itemize}[leftmargin=1.5em]
    \item \textbf{Input:} A point set $P$ in $\mathbb{R}^d$.
    \item \textbf{Question:} Does there exist a multipacking $M$ of $P$ of size at least $k$?
  \end{itemize}
  \end{minipage}%
}
\medskip

\noindent So far, the computational complexity of the \textsc{Multipacking} problem in geometric settings remains unknown.

For a natural number $r \leq n-1$, an \textit{$r$-multipacking} is a subset $ M $ of a point set $ P  $, such that $|N_s[v]\cap M|\leq (s+1)/2$ for each point $ v \in P $ and for every integer $ 1\leq s \leq r  $, i.e. $M$ contains at most half of the number of elements of $N_s[v]$.
 The \emph{$r$-multipacking number} of $ P $ is the maximum cardinality of an $r$-multipacking of $ P $ and is denoted by $ \MP_r(P) $. Note that, when $r=n-1$, an $(n-1)$-multipacking of $P$ is nothing but a \emph{multipacking} of $P$.  It is obvious that $\MP_r(P)\geq \MP(P)$ for any $r$, since a multipacking is always an $r$-multipacking of $P$ by their definitions.  Computing maximum $r$-multipacking  involves selecting the maximum number of points from the given point set satisfying the constraints of $r$-multipacking.  
This variant of facility location problem is interesting due to their wide application in operation research~\cite{rodriguez2006general}, resource allocation~\cite{melachrinoudis1999bicriteria}, VLSI design  \cite{abravaya2010maximizing}, network planning~\cite{cappanera2003discrete,rakas2004multi}, and clustering  \cite{henzinger2020dynamic}.  The \textsc{$r$-Multipacking} problem in $\mathbb{R}^d$ is the following.
    
\medskip
\medskip
\noindent
\fbox{%
  \begin{minipage}{\dimexpr\linewidth-2\fboxsep-2\fboxrule}
  \textsc{ $r$-Multipacking} problem (in geometry)
  \begin{itemize}[leftmargin=1.5em]
    \item \textbf{Input:} A point set $P$ in $\mathbb{R}^d$.
    \item \textbf{Question:} Does there exist an $r$-multipacking $M$ of $P$ of size at least $k$?
  \end{itemize}
  \end{minipage}%
}
\medskip

\noindent We (Das et al.~\cite{das2025geometrymultipackingCALDAM}) have shown that, for a point set in $\mathbb{R}^2$, a maximum $1$-multipacking can be computed in polynomial time, whereas computing a maximum $2$-multipacking is \textsc{NP-hard}. These results are included in this thesis (the detailed discussion is in Chapter \ref{chapter:geo_multi}).


Next, we discuss broadcast domination in geometric domain. The conventional approach to broadcasting requires strategic placement of transmitters to ensure complete signal coverage for a given set of receivers. Efficiency in this context primarily refers to minimizing the number of transmitters required, reducing the deployment costs. However, in many practical scenarios, transmitters possess heterogeneous transmission ranges, necessitating additional considerations and adjustments based on their individual capacities. Considering both transmitters and receivers are discretely distributed over the arena, the problem can be modeled as a \emph{dominating set} problem on graphs, where the vertices represent the locations for transmitters and receivers. Unlike the classical dominating set, each vertex in the dominating set is equipped with a range, such that the number of vertices dominated by it is bounded by that range. The problem is formally defined as the \emph{Broadcast domination}~\cite{henning2021broadcast,teshima2012broadcasts} in graphs.  The continual synergy between the discrete graph metric and the Euclidean metric motivates the investigation of the complexity of the problem for a point set in $\mathbb{R}^d$.


Consider a point $v \in P$. Here also a basic assumption is that the distances between all pairs of points are distinct. Thus the $r$-th neighbor of every point is unique. Recall that we denote the subset of $P$ that has only the $r$ nearest points of $v$ and the point $v$ itself by $N_r[v]$. Formally, $N_r[v] = \{ p \in P | \text{ $p$ is the $i$-th neighbor of $v$}, 1 \leq i \leq r \} \cup \{v\}$. Hence, $|N_r[v]|=r+1$ for $1\leq r\leq n-1$.

A function 	$ f : P \rightarrow \{0, 1, 2, \dots , n-1\} $ is called a \emph{broadcast} on $P$ and $v\in P$ is a \textit{tower} of $P$ if $f(v)>0$. Let $f$ be a broadcast on $P$. For each point $ u \in P  $, if  there exists a point $ v $ in $ P $ (possibly, $ u = v $) such that $ f (v) > 0 $ and $ u\in N_{f(v)}[v] $, then $ f $ is called a \emph{dominating broadcast}  or \emph{broadcast domination} on $ P$.
 The \textit{cost} of the broadcast $f$ is the quantity $ \sigma(f)$, which is the sum of the weights of the broadcasts over all the points in $ P $. Therefore, $\sigma(f)=  \sum_{v\in P}f(v)$. The minimum cost of a dominating broadcast in $P$ (taken over all dominating broadcasts)  is the \textit{broadcast domination number} of $P$, denoted by $ \gamma_{b}(P) $.  So, $ \gamma_{b}(P) = \min_{f\in D(P)} \sigma(f)= \min_{f\in D(P)} \sum_{v\in P}f(v)$, where $D(P)$ is the set of all dominating broadcasts on $P$.
A \emph{minimum dominating broadcast} (or \emph{optimal broadcast}, or \emph{minimum broadcast}) on $ P$ is a dominating broadcast on $ P$ with cost equal to $ \gamma_{b}(P)$.

We (Das et al.~\cite{das2025broadcastgeometry}) have shown that a minimum dominating broadcast in $\mathbb{R}^d$ can be computed in polynomial time. This result is included in this thesis (the detailed discussion is in Chapter \ref{chapter:geo_broad}).

Let $f$ be a minimum dominating broadcast and $M$ be a maximum multipacking on $P$. Let $T$ be the set of towers, i.e. $T=\{v\in P:f (v) > 0\}$. If $v\in T$, then  $ |N_{f(v)}[v]\cap M| \leq \frac{1}{2}(f(v)+1)$. Therefore,  $\MP(P)\leq \sum_{v:f(v)>0}|N_{f(v)}[v]\cap M| \leq \sum_{v:f(v)>0}\frac{1}{2}(f(v)+1)= \frac{1}{2}(\gamma_b(P)+|T|)$. Observe that $|T| \leq \gamma_b(P)$. This implies that $\MP(P)\leq \gamma_b(P)$. We state this in the following proposition.

\begin{proposition}\label{prop:mp_gammab_geo_ineq} Let $P $ be a point set in $\mathbb{R}^d$ having multipacking number  $ \MP(P) $ and broadcast domination number $ \gamma_{b}(P) $,  then  $$\MP(P)\leq \gamma_b(P).$$
\end{proposition}

\noindent Note that the same inequality holds between multipacking number and broadcast domination number in graphs (See Proposition \ref{prop:hartnell_ineq}).

\section{Contributions and thesis overview} \label{chapter1:results}

\subsection{Chapter 2: Complexity of Multipacking}\label{subsec:contribution_NP_complete}


The question of the algorithmic complexity of the \textsc{Multipacking} problem (for undirected graphs) has been repeatedly addressed by numerous authors (\cite{brewster2019broadcast,foucaud2021complexity,teshima2012broadcasts,yang2015new,yang2019limited}), yet it has persisted as an unsolved challenge for the past decade. Foucaud, Gras, Perez, and Sikora~\cite{foucaud2021complexity} [\textit{Algorithmica} 2021] made a step towards solving the open question by showing that the \textsc{Multipacking} problem  is \textsc{NP-complete} for directed graphs and it is \textsc{W[1]-hard} when parameterized by the solution size.  

In Chapter \ref{chapter:NPcomplete}, we answer the open question. First, we present a simple but elegant reduction from the well-known \textsc{NP-complete} (also \textsc{W[2]-complete} when parametrized by solution size) problem \textsc{Hitting Set}~\cite{cygan2015parameterized,garey1979computers} to the \textsc{Multipacking} problem to establish the following.

\begin{restatable}{theorem}{NPCchordal}\label{thm:NPC_chordal}
  \textsc{Multipacking} problem is \textsc{NP-complete}. Moreover, \textsc{Multipacking} problem  is \textsc{W[2]-hard} when parameterized by the solution size.
\end{restatable}

The \textit{Exponential Time Hypothesis (ETH)} states that the \textsc{3-SAT} problem cannot be solved in time $2^{o(n)}$, where $n$ is the number of variables of the input CNF (Conjunctive Normal Form) formula. The reduction we use in the proof of Theorem \ref{thm:NPC_chordal} yields the following.

\begin{restatable}{theorem}{ETHmultipackingSubExp}\label{thm:ETH_multipacking_fkno(k)}
    Unless \textsc{ETH} fails, there is no $f(k)n^{o(k)}$-time algorithm for the \textsc{Multipacking} problem, where $k$ is the solution size and $n$ is the number of vertices of the graph.
\end{restatable}


We also present hardness results on $\frac{1}{2}$-hyperbolic graphs. Gromov~\cite{gromov1987hyperbolic} introduced the $\delta$-hyperbolicity measures to explore the geometric group theory through Cayley graphs. Let $d(\cdot,\cdot)$ denote the distance between two vertices in a graph $G$. The graph $G$ is called a \emph{$\delta$-hyperbolic graph} if for any four
vertices $u, v, w, x \in V(G)$, the two larger of the three sums $d(u, v) + d(w, x)$, $d(u, w) + d(v, x)$, $d(u, x) +
d(v, w)$ differ by at most $2\delta$. A graph class $\mathcal{G}$ is said to be hyperbolic if there exists a constant $\delta$ such that every graph $G \in \mathcal{G}$ is $\delta$-hyperbolic. A graph is $0$-hyperbolic iff it is a block graph\footnote{A graph is a \textit{block graph} if every block (maximal 2-connected component) is a clique.}~\cite{howorka1979metric}. Trees are $0$-hyperbolic graphs, since trees are block graphs. 

It is known that the \textsc{Multipacking} problem is solvable in $O(n^3)$-time for strongly chordal\footnote{A \emph{strongly chordal graph} is  an undirected graph $G$ which is a chordal graph and every cycle of even length ($\geq 6$) in $G$ has an \emph{odd chord}, i.e., an edge that connects two vertices that are an odd distance ($>1$) apart from each other in the cycle.} graphs and linear time for trees~\cite{brewster2019broadcast}. All strongly chordal graphs are $\frac{1}{2}$-hyperbolic~\cite{wu2011hyperbolicity}. Whereas, all chordal graphs are $1$-hyperbolic \cite{brinkmann2001hyperbolicity}. Here, we prove that the \textsc{Multipacking} problem  is \textsc{NP-complete} for chordal $\cap$ $\frac{1}{2}$-hyperbolic graphs, which is a superclass of strongly chordal graphs.


\begin{restatable}{theorem}{NPChalfhyperbolicchordal}\label{thm:NPC_half_hyperbolic_chordal}
  \textsc{Multipacking} problem  is \textsc{NP-complete} for chordal $\cap$ $\frac{1}{2}$-hyperbolic graphs. Moreover, \textsc{Multipacking} problem  is \textsc{W[2]-hard} for chordal $\cap$ $\frac{1}{2}$-hyperbolic graphs when parameterized by the solution size.
\end{restatable}

Further, we use two different reductions from the \textsc{NP-complete} problem \textsc{Hitting Set}~\cite{garey1979computers} to prove the following theorems on bipartite and on claw-free\footnote{\textit{Claw-free graph} is a graph that does not contain a claw (or $K_{1,3}$) as an induced subgraph.} graphs.

\begin{restatable}{theorem}{NPCbipartite}\label{thm:NPC_bipartite}
  \textsc{Multipacking} problem  is \textsc{NP-complete} for bipartite graphs. Moreover, \textsc{Multipacking} problem  is \textsc{W[2]-hard} for bipartite graphs when parameterized by the solution size.
\end{restatable}

\begin{restatable}{theorem}{NPCclawfree}\label{thm:NPC_clawfree}
  \textsc{Multipacking} problem  is \textsc{NP-complete} for claw-free graphs. Moreover, \textsc{Multipacking} problem  is \textsc{W[2]-hard} for claw-free graphs when parameterized by the solution size.
\end{restatable}

 It is known that \textsc{Total Dominating Set} problem is \textsc{NP-complete} for cubic (3-regular) graphs~\cite{garey1979computers}. We reduce this problem to our problem to show the following.

\begin{restatable}{theorem}{multipackingregularNPc}\label{thm:multipacking_regular_NPc}
    \textsc{Multipacking} problem  is \textsc{NP-complete} for regular graphs. 
\end{restatable}


Figure~\ref{fig:graph_class_map} illustrates the complexity landscape of the \textsc{Multipacking} problem across the graph classes discussed in Chapter~\ref{chapter:NPcomplete}, along with some related classes.

\begin{figure}[H]
    \centering
   \includegraphics[width=\textwidth]{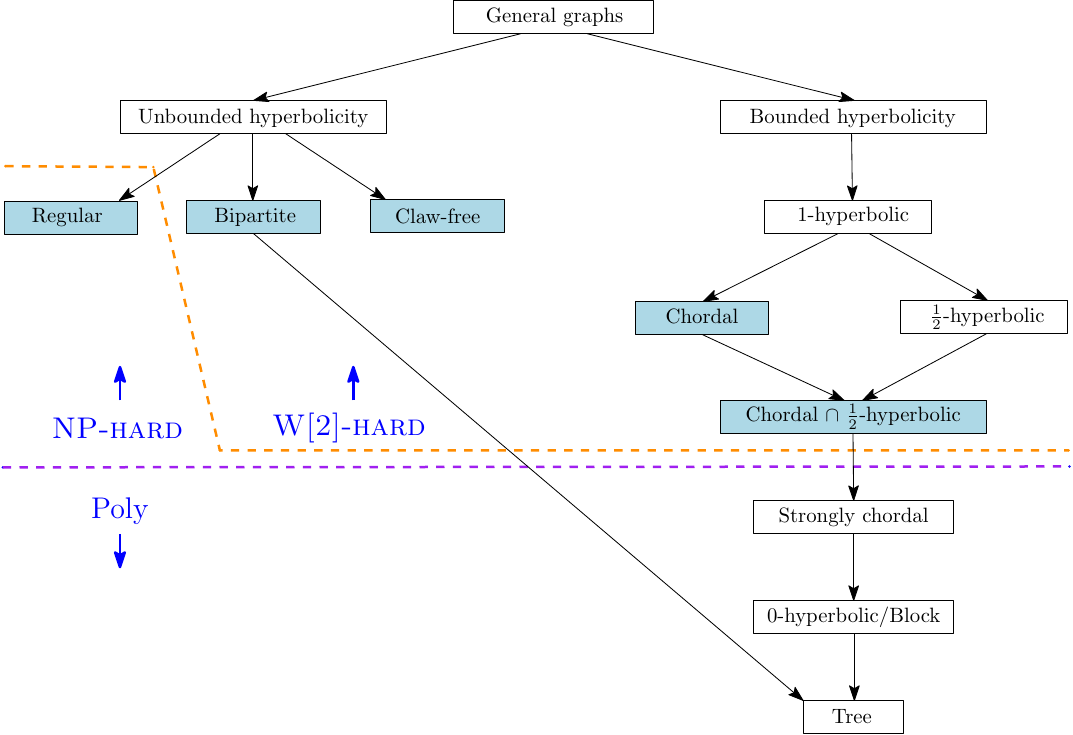}
    \caption{Inclusion diagram for graph classes mentioned in Chapter \ref{chapter:NPcomplete} (and related ones). If a class $A$ has a downward path to class $B$, then $B$ is a subclass of $A$. The \textsc{Multipacking} problem  is \textsc{NP-hard} for the graph classes above the dashed (purple) straight-line and it is polynomial-time solvable for the graph classes below the dashed (purple) straight-line. Moreover, the \textsc{Multipacking} problem  is \textsc{W[2]-hard} for the graph classes above the dashed (darkorange) curve. In Chapter \ref{chapter:NPcomplete}, we have discussed the hardness of the \textsc{Multipacking} problem for the graph classes in the colored (lightblue) block.}
    \label{fig:graph_class_map}
\end{figure}

Further, we studied the computational complexity of the \textsc{Multipacking} problem on a geometric intersection graph-class: CONV graphs. A graph $G$ is a \textit{CONV graph} or \textit{Convex Intersection Graph} if and only if there exists a family of convex sets on a plane such that the graph has a vertex for each convex set and an edge for each intersecting pair of convex sets. It is known that 
 \textsc{Total Dominating Set} problem is \textsc{NP-complete} for planar graphs~\cite{garey1979computers}.  We reduce this problem to our problem to show the following.

\begin{restatable}{theorem}{multipackingCONVNPc}\label{thm:multipacking_CONV_NPc}
    \textsc{Multipacking} problem  is \textsc{NP-complete} for CONV graphs. 
\end{restatable}

\subsection{Chapter 3: Multipacking on a bounded hyperbolic graph-class: chordal graphs}

A \textit{chordal graph} is an undirected simple graph in which each cycle on four or more vertices has a chord, which is an edge that is not part of the cycle but connects two vertices of the cycle.  Chordal graph is a superclass of interval graph. In Chapter \ref{chapter:chordal}, we study the multipacking problem on chordal graphs.


We also generalize this study for $\delta$-hyperbolic graphs (the definition of $\delta$-hyperbolic graph is given in Subsection \ref{subsec:contribution_NP_complete}). Trees are $0$-hyperbolic graphs~\cite{buneman1974note}, and chordal graphs are 1-hyperbolic~\cite{brinkmann2001hyperbolicity}. In general, hyperbolicity is a measurement of the deviation of the distance function of a graph from a tree metric. Many other interesting graph
classes including 
co-comparability graphs \cite{corneil2013ldfs}, asteroidal-triple free graphs \cite{corneil1997asteroidal}, permutation graphs \cite{golumbic2004algorithmic}, 
graphs with bounded chordality or treelength are hyperbolic \cite{chepoi2008diameters,keil2017algorithm}. See Figure~\ref{fig:diagram} for a diagram representing the inclusion hierarchy between these classes. Moreover, hyperbolicity is a measure that captures properties of real-world graphs such as
the Internet graph \cite{shavitt2004curvature} or database relation graphs \cite{walter2002interactive}. 
Here, we study the multipacking problem of hyperbolic graphs.

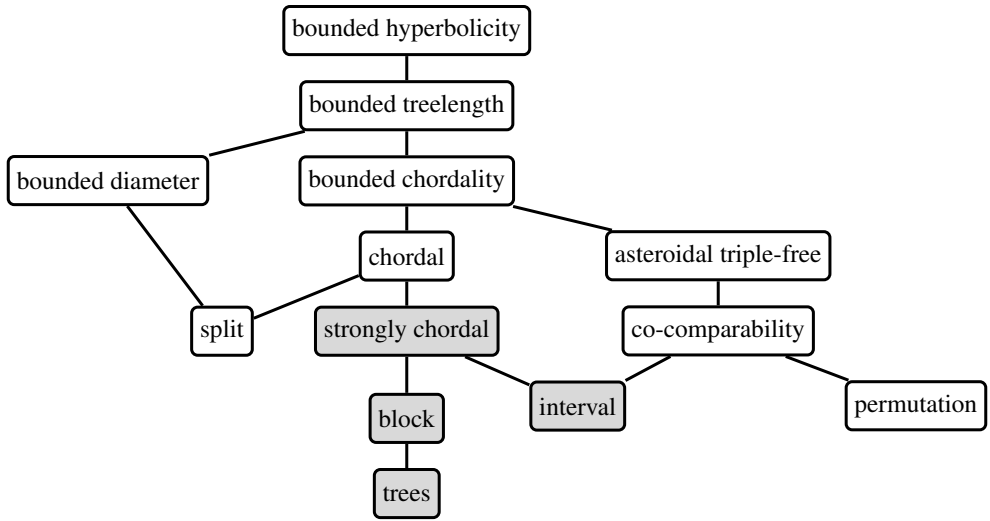
\begin{figure}[t]
\centering
\scalebox{0.8}{\begin{tikzpicture}[node distance=7mm]

\tikzstyle{mybox}=[fill=white,line width=0.5mm,rectangle, minimum height=.8cm,fill=white!70,rounded corners=1mm,draw];
\tikzstyle{myedge}=[line width=0.5mm]
\newcommand{\tworows}[2]{\begin{tabular}{c}{#1}\\{#2}\end{tabular}}

 \node[mybox] (hyper)  {bounded hyperbolicity};
    \node[mybox] (treelength)  [below=of hyper,yshift=3mm]  {bounded treelength} edge[myedge] (hyper);
    \node[mybox] (chordality) [below=of treelength,yshift=3mm] {bounded chordality} edge[myedge] (treelength);
      \node[mybox] (diam) [below left=of treelength,xshift=-10mm,yshift=1mm] {bounded diameter} edge[myedge] (treelength);
    \node[mybox] (ATfree) [below right=of chordality,xshift=10mm,yshift=1mm] {asteroidal triple-free} edge[myedge, fill=red!30] (chordality);
    \node[mybox] (cocomp) [below =of ATfree,yshift=3mm] {co-comparability} edge[myedge] (ATfree);
    \node[mybox] (chordal) [below =of chordality,yshift=3mm] {chordal} edge[myedge] (chordality);
    \node[mybox, fill=gray!30] (strongly) [below =of chordal,yshift=3mm] {strongly chordal} edge[myedge] (chordal);
    \node[mybox] (split) [left =of strongly,xshift=-3mm] {split} edge[myedge] (chordal) edge[myedge] (diam);
    \node[mybox, fill=gray!30] (interval) [below right=of strongly,yshift=1mm] {interval} edge[myedge] (strongly) edge[myedge] (cocomp);
    \node[mybox] (perm) [below right=of cocomp,yshift=1mm] {permutation} edge[myedge] (cocomp);
    \node[mybox, fill=gray!30] (block) [below =of strongly,yshift=1mm]
    {block} edge[myedge] (strongly);
    \node[mybox, fill=gray!30] (tree) [below=of block,yshift=3mm] {trees} edge[myedge] (block);
  \end{tikzpicture}}

\caption{Inclusion diagram for graph classes mentioned in Chapter \ref{chapter:chordal} (and related ones). If a class $A$ has an upward path to class $B$, then $A$ is included in $B$. For the graphs in the gray classes, the broadcast domination number is equal to the multipacking number, but this is not true for the white classes. 
}
\label{fig:diagram}
\end{figure}

We start by bounding the multipacking number of a chordal graph: 

\begin{restatable}{proposition}{multipackingbroadcastrelation}
\label{prop:gammabGleq3/2mpG}
 If $G$ is a connected chordal graph, then $\gamma_{b}(G)\leq \big\lceil{\frac{3}{2} \MP(G)\big\rceil} $. 
\end{restatable}

We have shown that, the \textsc{Multipacking} problem  is \textsc{NP-complete} for chordal $\cap$ $\frac{1}{2}$-hyperbolic graphs and it is \textsc{W[2]-hard} for chordal $\cap$ $\frac{1}{2}$-hyperbolic graphs when parameterized by the solution size (See Theorem \ref{thm:NPC_half_hyperbolic_chordal} of Chapter \ref{chapter:NPcomplete}).
However, polynomial-time algorithms are known for trees and more generally, strongly chordal graphs~\cite{brewster2019broadcast}. Even for trees, the algorithm for finding a maximum multipacking is very non-trivial. Chordal graphs form a superclass of the class of strongly chordal graphs. Here, we provide a $(\frac{3}{2}+o(1))$-approximation algorithm to find a multipacking on chordal graphs.   

\begin{restatable}{proposition}{approximationalgorithm}\label{prop:3/2mpGapprox} If $G$ is a connected chordal graph, there is a polynomial-time algorithm to construct a multipacking of $G$ of size at least $  \big\lceil{\frac{2\MP(G)-1}{3} \big\rceil}$.
\end{restatable}

This approximation algorithm is based on finding a diametral path which is almost double the radius for chordal graphs. The set of every third vertex on this path yields a multipacking.

Hartnell and Mynhardt~\cite{hartnell2014difference} constructed a family of connected graphs $H_k$ such that $\MP(H_k)=3k$ and $\gamma_b(H_k)=4k$. There is no such construction for trees and more generally, strongly chordal graphs, since multipacking number and broadcast domination number are the same for these graph classes~\cite{brewster2019broadcast}. Here, we construct a family of connected chordal graphs $H_k$ such that $\MP(H_k)=9k$ and $\gamma_b(H_k)=10k$ (Figure~\ref{fig:Names}).

\begin{restatable}{theorem}{multipackingbroadcastgapechordal}\label{thm:9k10k}
 For each positive integer $k$, there is a connected chordal graph $H_k$ such that $\MP(H_k)=9k$ and $\gamma_b(H_k)=10k$.
\end{restatable}

Theorem \ref{thm:9k10k} directly establishes the following corollary.

\begin{restatable}{corollary}{gammabGdiffmpGchordal}\label{cor:gammabG-mpGchordal} The difference $ \gamma_{b}(G) -  \MP(G) $ can be arbitrarily large for connected chordal graphs.
\end{restatable}

We mentioned earlier that, for general connected graphs, the value of the expression $\lim_{\MP(G)\to \infty}\sup\{\gamma_{b}(G)/\MP(G)\}$ is $2$~\cite{rajendraprasad2025multipacking}. We found the range of this expression for chordal graphs. Proposition \ref{prop:gammabGleq3/2mpG} and Theorem \ref{thm:9k10k} yield the following corollary.

\begin{restatable}{corollary}{gammabGbympGchordal} \label{cor:gammabG/mpGchordal} For connected chordal graphs $G$, 
$$\displaystyle\frac{10}{9}\leq\lim_{\MP(G)\to \infty}\sup\Bigg\{\frac{\gamma_{b}(G)}{\MP(G)}\Bigg\}\leq \frac{3}{2}.$$
\end{restatable}

We also make a connection with the \emph{fractional} versions of the two concepts dominating broadcast and multipacking, as introduced in~\cite{brewster2013broadcast}.


Further, we improve the lower bound of the expression $\lim_{\MP(G)\to \infty}$ $\sup\{\gamma_{b}(G)/\MP(G)\}$ for connected chordal graphs to $4/3$ where the previous lower bound was $10/9$  \cite{das2023relation} in Corollary \ref{cor:gammabG/mpGchordal}.  To prove this, we have shown the following.

\begin{restatable}{theorem}{chordalmultipackingbroadcastgape}\label{thm:chordal_graphs}
For each positive integer $k$, there is a connected chordal graph $F_k$ such that $\MP(F_k)=3k$ and $\gamma_b(F_k)=4k$.
    
\end{restatable}

\begin{restatable}{corollary}{chordalgammabGBYmpGimproved} \label{cor:chordal_gammabG/mpG} For connected chordal graphs $G$, \text{ }
$$\displaystyle\frac{4}{3}\leq\lim_{\MP(G)\to \infty}\sup\Bigg\{\frac{\gamma_{b}(G)}{\MP(G)}\Bigg\}\leq \frac{3}{2}.$$
\end{restatable}

We establish a relation between broadcast domination and multipacking numbers of $\delta$-hyperbolic graphs using the same method  that
works for chordal graphs.  We state that in the following  proposition:

\begin{restatable}{proposition}{deltaMultipackingBroadcastRelation}\label{prop:delta_multipacking_broadcast_relation}
     If $G$ is a $\delta$-hyperbolic graph, then $\gamma_{b}(G)\leq \big\lfloor{\frac{3}{2} \MP(G)+2\delta\big\rfloor} $.

\end{restatable}

We used the same method that
works for chordal graphs to  provide an approximation algorithm to find a large multipacking of $\delta$-hyperbolic graphs.

\begin{restatable}{proposition}{approxdeltaMultipacking}\label{prop:approx_delta_multipacking}
    If $G$ is a $\delta$-hyperbolic graph, there is a polynomial-time algorithm to construct a multipacking of $G$ of size at least $  \big\lceil{\frac{2\MP(G)-4\delta}{3} \big\rceil} $.
\end{restatable}

A graph is \emph{$k$-chordal} (or, bounded chordality graph with bound $k$) if it does not contain any induced $n$-cycle for $n > k$. The class of $3$-chordal graphs is the class of usual chordal graphs. Brinkmann, Koolen and Moulton \cite{brinkmann2001hyperbolicity} proved that every chordal graph is $1$-hyperbolic. In 2011, Wu and Zhang \cite{wu2011hyperbolicity} established that a $k$-chordal graph is always a $\frac{\lfloor\frac{k}{2}\rfloor}{2}$-hyperbolic graph, for every $k\geq 4$. A graph $G$ has \emph{treelength} at most $d$ if it admits a tree-decomposition where the maximum distance between any two vertices in the same bag have distance at most $d$ in $G$~\cite{dourisboure2007tree}. Clearly, a graph of diameter $d$ has treelength at most $d$. Graphs of treelength at most $d$ have hyperbolicity at most $d$~\cite{chepoi2008diameters}. Moreover, $k$-chordal graphs have treelength at most $2k$~\cite{dourisboure2007tree}. 
Additionally, it is known that the hyperbolicity is $1$ for the graph classes: asteroidal-triple free, co-comparability and permutation graphs~\cite{wu2011hyperbolicity}. 
These results about hyperbolicity and Proposition \ref{prop:delta_multipacking_broadcast_relation} directly establish the following.

\begin{corollary}\label{cor:k-chordal} Let $G$ be a graph.

\noindent(i) If $G$ is an asteroidal-triple free, co-comparability or permutation graph, then $\gamma_{b}(G)\leq \big\lfloor{\frac{3}{2} \MP(G)+2}\big\rfloor $.

\noindent(ii) If $G$ is a $k$-chordal graph where $k\geq 4$, then $\gamma_{b}(G)\leq \big\lfloor{\frac{3}{2} \MP(G)+\lfloor\frac{k}{2}\rfloor}\big\rfloor $.

\noindent(iii) If $G$ has treelength $d$, then $\gamma_{b}(G)\leq \big\lfloor{\frac{3}{2} \MP(G)+2d}\big\rfloor $.
\end{corollary}

\subsection{Chapter 4: Multipacking on an unbounded hyperbolic graph-class: cactus graphs}

 A \textit{cactus} is a connected graph in which any two  cycles have at most one vertex in common. Equivalently, it is a connected graph in which every edge belongs to at most one  cycle. Cactus is a superclass of tree and a subclass of outerplanar graph. In Chapter \ref{chapter:cactus}, we study the multipacking problem on the cactus graph. We  establish a relation between multipacking and dominating broadcast on the same graph class. 
We start by bounding the multipacking number of a cactus:

\begin{restatable}{theorem}{multipackingbroadcastrelationcactus}
\label{thm:multipacking_broadcast_relation}
    Let $G$ be a cactus, then $\gamma_b(G)\leq \frac{3}{2}\MP(G)+\frac{11}{2}$.
\end{restatable} 

Note that, for $\delta$-hyperbolic graphs,  $\gamma_{b}(G)\leq \big\lfloor{\frac{3}{2} \MP(G)+2\delta\big\rfloor} $~\cite{das2023relationarxiv}. Chordal graphs\footnote{A \textit{chordal graph} is an undirected simple graph in which each cycle on four or more vertices has a chord, which is an edge that is not part of the cycle but connects two vertices of the cycle.} are $1$-hyperbolic \cite{brinkmann2001hyperbolicity}. Earlier we mentioned that  for any connected chordal graph $G$, $\gamma_{b}(G)\leq \big\lceil{\frac{3}{2} \MP(G)\big\rceil}$~\cite{das2023relation}. The cactus graphs are far from  being chordal or  hyperbolic since cactus graphs can have unbounded hyperbolicity and that shows the importance of Theorem \ref{thm:multipacking_broadcast_relation}.  The proof of Theorem \ref{thm:multipacking_broadcast_relation} is based on finding some paths and a cycle (if needed) in the graph whose total size is almost double the radius. The set of every third vertex on these paths and the cycle yields a multipacking under some conditions. 




There is a polynomial time algorithm to find a maximum 2-packing set in a cactus \cite{flores2018algorithm}. A multipacking is a 2-packing, but the reverse is not true. Here, we provide an approximation algorithm to find multipacking on cactus.  Our proof technique is based on the basic structure of the cactus graphs, and the technique is completely different from the existing polynomial time solution for the $2$-packing problem on the cactus.

\begin{restatable}{theorem}{MultipackingAlgorithm}\label{thm:multipacking_algorithm}
 If $G$ is a cactus graph, there is an $O(n)$-time algorithm to construct a multipacking of $G$ of size at least $ \frac{2}{3}\MP(G)-\frac{11}{3} $  where $n=|V(G)|$.
\end{restatable}

\begin{figure}[ht]
    \centering
    \includegraphics[width=\textwidth]{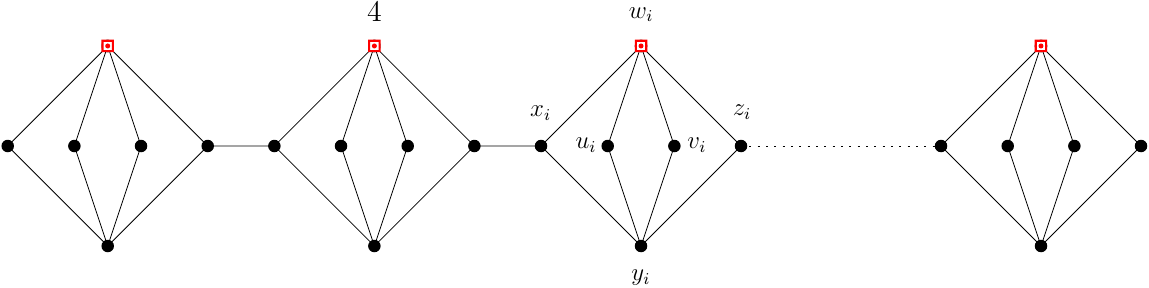}
    \caption{The $H_k$ graph with $\gamma_b(H_k)=4k$ and $\MP(H_k)=3k$. The set $\{w_i:1\leq i\leq 3k\}$ is a maximum multipacking of $H_k$. This family of graphs was constructed by Hartnell and Mynhardt~\cite{hartnell2014difference}.} 
    \label{fig:Hk}
\end{figure}

Hartnell and Mynhardt~\cite{hartnell2014difference} constructed a family of connected graphs $H_k$ such that $\MP(H_k)=3k$ and $\gamma_b(H_k)=4k$ (Fig. \ref{fig:Hk}). We provide a simpler family of connected graphs $G_k$ such that $\MP(G_k)=3k$ and $\gamma_b(G_k)=4k$ (Fig. \ref{fig:pentagon}). Not only that, the family of graphs $G_k$ covers many graph classes, including   cactus (a subclass of outerplanar graphs), AT-free\footnote{An independent set of three vertices such that each pair is joined by a path that avoids the neighborhood of the third is called an \textit{asteroidal triple}. A graph is asteroidal triple-free or AT-free if it contains
no asteroidal triples.} (a subclass of $C_n$-free graphs, for $n\geq 6$) graphs, etc. We state that in the following theorem:

\begin{restatable}{theorem}{multipackingbroadcastgapecactus}\label{thm:3k4k}
 For each positive integer $k$, there is a cactus graph (and AT-free graph) $G_k$ such that $\MP(G_k)=3k$ and $\gamma_b(G_k)=4k$.
\end{restatable}


Theorem \ref{thm:3k4k} directly establishes the following corollary.

\begin{restatable}{corollary}{gammabDIFFGmpGcactus}\label{cor:gammabG-mpGcactus} The difference $ \gamma_{b}(G) -  \MP(G) $ can be arbitrarily large for cactus graphs (and for AT-free graphs).
\end{restatable}

We also make a connection with the \emph{fractional} versions of the two concepts dominating broadcast and multipacking, as introduced in~\cite{brewster2013broadcast}.

Further, we presented a range of the expression $\lim_{\MP(G)\to \infty}$ $\sup\{\gamma_{b}(G)/\MP(G)\}$ for cactus graphs. Theorem \ref{thm:multipacking_broadcast_relation} and Theorem \ref{thm:3k4k} yield the following:

\begin{restatable}{corollary}{gammabGBYmpGcactus} \label{cor:gammabG/mpGcactus} For cactus graphs $G$, 
$$\displaystyle\frac{4}{3}\leq\lim_{\MP(G)\to \infty}\sup\Bigg\{\frac{\gamma_{b}(G)}{\MP(G)}\Bigg\}\leq \frac{3}{2}.$$
\end{restatable}  

Corollary \ref{cor:gammabG/mpGcactus} yields that, for cactus graphs, we cannot form a bound in the form $\gamma_b(G)\leq c_1\cdot \MP(G)+c_2$, for any constant $c_1<4/3$ and $c_2$.


We also answer some questions of the relation between broadcast domination number and multipacking number on the hyperbolic graph classes (the definition of $\delta$-hyperbolic graph is given in Subsection \ref{subsec:contribution_NP_complete}).  It is known that a graph is $0$-hyperbolic iff it is a block graph\footnote{A graph is a \textit{block graph} if every block (maximal 2-connected component) is a clique.}~\cite{howorka1979metric}. Trees are $0$-hyperbolic graphs, since trees are block graphs. Moreover, it is known that all strongly chordal graphs are $\frac{1}{2}$-hyperbolic~\cite{wu2011hyperbolicity}, and also $\MP(G)=\gamma_b(G)$ holds for the strongly chordal graphs~\cite{brewster2019broadcast}. Since block graphs are strongly chordal graphs, therefore $\MP(G)=\gamma_b(G)$ holds for the graphs with hyperbolicity $0$. Furthermore, we know that $ \gamma_{b}(G) -  \MP(G) $ can be arbitrarily large for connected chordal graphs which is a subclass of $1$-hyperbolic graphs.  This leads to a natural question: what is the minimum value of $\delta$, for which we can say that the difference $ \gamma_{b}(G) -  \MP(G) $ can be arbitrarily large for $\delta$-hyperbolic graphs? We answer this question. We show that the family of graphs $G_k$ (Fig. \ref{fig:pentagon}) is $\frac{1}{2}$-hyperbolic. This yields the following theorem.

\begin{restatable}{theorem}{hypermultipackingbroadcastgape}\label{thm:1/2-Hyperbolic_graphs}
For each positive integer $k$, there is a $\frac{1}{2}$-hyperbolic graph $G_k$ such that $\MP(G_k)=3k$ and $\gamma_b(G_k)=4k$.
    
\end{restatable}

Theorem \ref{thm:1/2-Hyperbolic_graphs} directly establishes the following corollary.

\begin{restatable}{corollary}{hypergammabDIFFGmpG}\label{cor:hyper_gammabG-mpG} The difference $ \gamma_{b}(G) -  \MP(G) $ can be arbitrarily large for $\frac{1}{2}$-hyperbolic graphs.
\end{restatable}

We discussed earlier that for all $0$-hyperbolic graphs, we have $ \gamma_{b}(G) =  \MP(G) $. Therefore, Corollary \ref{cor:hyper_gammabG-mpG} yields the following.

\begin{restatable}{theorem}{mindelta}\label{thm:min_delta} For $\delta$-hyperbolic graphs, the minimum value of $\delta$ is $\frac{1}{2}$ for which, the difference $ \gamma_{b}(G) -  \MP(G) $ can be arbitrarily large.
\end{restatable}

It is known that, if $G$ is a $\delta$-hyperbolic graph, then $\gamma_{b}(G)\leq \big\lfloor{\frac{3}{2} \MP(G)+2\delta\big\rfloor} $~\cite{das2023relationarxiv}. This result, together with Theorem \ref{thm:1/2-Hyperbolic_graphs}, yields the following.

\begin{restatable}{corollary}{hypergammabGBYmpG} \label{cor:hyper_gammabG/mpG} For $\frac{1}{2}$-hyperbolic graphs $G$, 
$$\displaystyle\frac{4}{3}\leq\lim_{\MP(G)\to \infty}\sup\Bigg\{\frac{\gamma_{b}(G)}{\MP(G)}\Bigg\}\leq \frac{3}{2}.$$
\end{restatable}

\subsection{Chapter 5: Complexity of $r$-Multipacking}

\textsc{$1$-Multipacking} problem is \textsc{NP-complete} for planar bipartite graphs of maximum degree $3$, and chordal graphs~\cite{eto2014distance}. In Chapter \ref{chapter:k_multi}, we study the hardness of \textsc{$r$-Multipacking} problem, for $r\geq 2$. 

It is known that the \textsc{Independent Set} problem is \textsc{NP-complete} for  planar graphs of maximum degree $3$~\cite{garey1974some}.  We reduce this problem to our problem to show the following.

\begin{restatable}{theorem}{NPCplanarkmulti}\label{thm:NPC_planar_r_multi}
  For $r\geq 2$, the \textsc{$r$-Multipacking} problem is \textsc{NP-complete} even for planar bipartite graphs with maximum degree $\max \{4,r\}$. 
\end{restatable}

Next, we use two reductions from the well-known \textsc{NP-complete} problem \textsc{Independent Set}~\cite{garey1979computers} to our problems to show the following two results.

\begin{restatable}{theorem}{NPCchordalkmulti}\label{thm:NPC_chordal_r_multi}
  For $r\geq 2$, the \textsc{$r$-Multipacking} problem is \textsc{NP-complete}  for chordal graphs with maximum radius $r+1$. 
\end{restatable}

\begin{restatable}{theorem}{NPCbipartitekmulti}\label{thm:NPC_bipartite_r_multi}
  For $r\geq 2$, the \textsc{$r$-Multipacking} problem is \textsc{NP-complete} for bipartite graphs with maximum radius $r+1$. 
\end{restatable}

Next, we show that the \textsc{$r$-Multipacking} problem does not admit a subexponential-time algorithm, unless the \textsc{ETH}\footnote{The \textit{Exponential Time Hypothesis (ETH)} states that the \textsc{3-SAT} problem cannot be solved in time $2^{o(n)}$, where $n$ is the number of variables of the input CNF (Conjunctive Normal Form) formula.} fails.

\begin{restatable}{theorem}{rmultiETH}\label{thm:ETH_r_multi}
Unless the \textsc{ETH} fails, for $r \geq 2$, the \textsc{$r$-Multipacking} problem admits no algorithm running in time $2^{o(n+m)}$, even when restricted to chordal or bipartite graphs of maximum radius $r+1$, where $n$ and $m$ denote the numbers of vertices and edges, respectively.

\end{restatable}

\subsection{Chapter 6: Multipacking in geometric domain}

 In Chapter \ref{chapter:geo_multi}, first, we study the multipackings of the point sets on $\mathbb{R}^1$. We start by bounding the multipacking number of $P\subset \mathbb{R}^1$ and we have shown the tightness of the bounds by constructing examples.

\begin{restatable}{theorem}{thmmponeDtightlowerupperbound}
\label{thm:mp 1D tight lower bound n/3 upper bound n/2}
    Let $P $ be a point set on $\mathbb{R}^1$. Then $\lfloor\frac{n}{3}\rfloor \leq \MP(P)\leq \lfloor \frac{n}{2}\rfloor$. 
    Both the bounds are tight.
\end{restatable}

We also provide an algorithm to find a maximum $r$-multipacking of a point set on a line. 


\begin{restatable}{theorem}{rkthmonedcorrectness}\label{thm:rk1D}
For any $1 \leq r \leq n-1$, a maximum $r$-multipacking for a given point set $P \subset \mathbb{R}^1$ with $|P|=n$ can be computed in $O(n^2r)$ time.
\end{restatable}

\noindent Theorem \ref{thm:rk1D} yields that, a maximum multipacking for a given point set $P \subset \mathbb{R}^1$ with $|P|=n$ can be computed in $O(n^3)$ time, since an $(n-1)$-multipacking is nothing but a multipacking.



Further, we define the function $\MMP_{r}(t)$ to be the smallest number such that any point set $P \subset \mathbb{R}^2$ of size $ \MMP_{r}(t)$ admits an $r$-multipacking of size at least $t$. We denote $\MMP_{n-1}(t)$ as $\MMP(t)$, i.e. $\MMP(t)$ is the smallest number such that any point set $P \subset \mathbb{R}^2$ of size $ \MMP(t)$ admits a \emph{multipacking} of size at least $t$. It can be followed that $\MMP(1)=1$ from the definition of multipacking. No general solution has been found for this problem yet. 


We provide the value of $\MMP(2)$ in the following theorem:

\begin{restatable}{theorem}{MPtwo}\label{thm:MP2}
    $\MMP(2)=6$.
\end{restatable}


Until now, there is no known polynomial-time algorithm to find a maximum multipacking for a given point set $P\subset \mathbb{R}^2$, and the problem is also not known to be NP-hard. Here, we make an important step towards solving the problem by studying the $1$-multipacking and $2$-multipacking in $\mathbb{R}^2$. We considered the NNG (Nearest-Neighbor-Graph) $G$ of $P\subset \mathbb{R}^2$. Using the observation, the maximum independent sets of $G$ are the maximum $1$-multipackings of $P$, we conclude that a maximum $1$-multipacking can be computed in polynomial time. However the $2$-multipacking problem in $\mathbb{R}^2$ is much more complicated. We have studied the hardness of the $2$-multipacking problem in $\mathbb{R}^2$. We reduce the well-known NP-complete problem Planar Rectilinear Monotone 3-SAT~\cite{de2012optimal} to our problem to show the following. 

\begin{restatable}{theorem}{NPC}\label{thm:NPC}
Computing a maximum $2$-multipacking of a given set of $n$ points $P$ in $\mathbb{R}^2$ is \textsc{NP-hard}.
\end{restatable}



Moreover, we provided approximation and parameterized solutions to this problem. We have shown that a maximum $2$-multipacking of a given set of points $P$ in $\mathbb{R}^2$ can be approximated in polynomial time within a ratio arbitrarily close to $4$ and $2$-multipacking of size $k$ of a given set of $n$ points in $\mathbb{R}^2$ can be computed in time $18^kn^{O(1)}$, if it exists.




\subsection{Chapter 7: Dominating Broadcast in $d$-dimensional space}

 In Chapter \ref{chapter:geo_broad}, we address the problem of computing the minimum dominating broadcast for a point set in $\mathbb{R}^d$.  We start by proving the following theorem that helps to design an algorithm for finding a minimum broadcast of a given point set $P$ in $\mathbb{R}^d$. 

\begin{restatable}{theorem}{thmbroadcastzeroone}
\label{thm:f(u)in0,1}
    Let $P $ be a point set in $\mathbb{R}^d$. There exists a minimum broadcast $f$ such that $f(u)\in \{0,1\}$ for each point $ u \in P  $.
\end{restatable}

We use Theorem \ref{thm:f(u)in0,1} and the construction of nearest neighbor graph (NNG) of $P$ to provide an algorithm to find a minimum broadcast of $P$ in $O(n\log n)$ time, where $|P| = n$.  We state that in the following theorem:

\begin{restatable}{theorem}{thmbroadcastalgortihmcomplexity}\label{thm:broadcast_algortihm_complexity_nlogn}
  A minimum broadcast on a given point set $P$ in $\mathbb{R}^d$ can be computed in $O(n\log n)$ time, where $|P| = n$. 
\end{restatable}

Kissing number \cite{levenshtein1979bounds,musin2008kissing,odlyzko1979new,pfender2004kissing} is a very well-known topic in the field of mathematics and computer science. 
The \textit{kissing number}, denoted by $\tau_d$, is the maximum number of non-overlapping (i.e., having disjoint interiors)  unit balls in the Euclidean $d$-dimensional space $\mathbb{R}^d$ that can touch a given unit ball. For a survey of kissing numbers, see \cite{boyvalenkov2015survey}. 

The values of $\tau_d$ for small dimensions are known: $\tau_1=2$, $\tau_2=6$, $\tau_3=12$ \cite{anstreicher2004thirteen,boroczky2003newton,maehara2001isoperimetric,schutte1952problem}, $\tau_4=24$ \cite{musin2008kissing}, etc. However, the exact value of $\tau_d$ for general $d$ remains unknown. Nevertheless, bounds for $\tau_d$ are available.

The best known upper bound was given by Kabatjanskii and Levenstein~\cite{kabatjansky1978bounds},  yielding
\[
\tau_d \leq 2^{0.4041...\cdot d}.
\]

For the lower bound,  Jenssen, Joos, and Perkins~\cite{jenssen2018kissing} proved that
\[
\tau_d \geq (1 + o(1)) \frac{\sqrt{3\pi}}{2\sqrt{2}} \log \frac{3}{2\sqrt{2}} \cdot d^{3/2} \cdot \left( \frac{2}{\sqrt{3}} \right)^d.
\]

Later, Fernandez et al.~\cite{fernandez2025new} improved the  constant factor of the kissing number lower bound in high dimensions. They have shown that, for $d \to \infty$, the kissing number satisfies
\[
\tau_d \geq (1 + o(1)) \frac{\sqrt{3\pi}}{4\sqrt{2}} \log \frac{3}{2} \cdot d^{3/2} \cdot \left( \frac{2}{\sqrt{3}} \right)^d.
\]

We provide an upper  bound for broadcast domination number of $P \subseteq \mathbb{R}^d$ using kissing number. 

\begin{restatable}{theorem}{thmgammablowerboundford}\label{thm:gammab lower bound n/2 for d}
    Let $P $ be a point set on $\mathbb{R}^d$ with $|P|=n$. Then $\frac{n}{2}\leq \gamma_b(P)\leq \frac{\tau_d}{\tau_d+1} n$, where $\tau_d$ is the kissing number in $\mathbb{R}^d$. The lower bound is tight.
\end{restatable}

Moreover, we provide an upper and lower bound for broadcast domination number of $P \subseteq \mathbb{R}^2$ and construct examples to show the tightness of the bounds.

\begin{restatable}{theorem}{thmgammablowerbound}\label{thm:gammab lower bound n/2}
    Let $P $ be a point set on $\mathbb{R}^2$ with $|P|=n$. Then $ \gamma_b(P)\leq \frac{5n}{6}$. The bound is tight.
\end{restatable}

\chapter{Complexity of Multipacking}\label{chapter:NPcomplete}\hypertarget{chapter:introhref}{}

\minitoc




    In this chapter, we prove that the \textsc{Multipacking} problem  is \textsc{NP-complete} for undirected graphs, which answers the open question (mentioned in Subsection \ref{subsec:contribution_NP_complete}). Moreover, it is \textsc{W[2]-hard} for undirected graphs when parameterized by the solution size. Furthermore, we have shown that the problem is \textsc{NP-complete} and \textsc{W[2]-hard} (when parameterized by the solution size) even for various subclasses: chordal, bipartite, and claw-free graphs. Whereas, it is \textsc{NP-complete} for regular, and CONV graphs. Additionally, the problem is \textsc{NP-complete} and \textsc{W[2]-hard} (when parameterized by the solution size) for chordal $\cap$ $\frac{1}{2}$-hyperbolic graphs, which is a superclass of strongly chordal graphs where the problem is polynomial-time solvable.

\section{Chapter overview}
 In Section \ref{sec:preliminaries_NPcomplete}, we recall some definitions and notations. In Section \ref{sec:multipacking_npc_w}, we prove that the \textsc{Multipacking} problem is \textsc{NP-complete} and also the problem is \textsc{W[2]-hard} when parameterized by the solution size. In Section \ref{sec:hardness_on_subclasses}, we discuss the hardness of the \textsc{Multipacking} problem on some sublcasses. We conclude this chapter in Section \ref{sec:conclusion_NPcomplete} by presenting several important questions and future research directions.

\section{Preliminaries}\label{sec:preliminaries_NPcomplete}

Let $G=(V,E)$ be a graph. For a vertex $v\in V$, the \emph{open neighborhood} of $v$ is
\(
N(v) := \{u \in V : d(u,v)=1\},
\)
and the \emph{closed $r$-neighborhood} of $v$ is
\(
N_r[v] := \{u \in V : d(u,v) \le r\},
\)
that is, the ball of radius $r$ centered at $v$.


A \textit{multipacking} is a set $ M \subseteq V  $ in a
	graph $ G = (V, E) $ such that   $|N_r[v]\cap M|\leq r$ for each vertex $ v \in V $ and for every integer $ r \geq 1 $. The \textsc{Multipacking} problem is as follows:

\medskip
\noindent
\fbox{%
  \begin{minipage}{\dimexpr\linewidth-2\fboxsep-2\fboxrule}
  \textsc{ Multipacking} problem
  \begin{itemize}[leftmargin=1.5em]
    \item \textbf{Input:} An undirected graph $G = (V, E)$, an integer $k \in \mathbb{N}$.
    \item \textbf{Question:} Does there exist a multipacking $M \subseteq V$ of $G$ of size at least $k$?
  \end{itemize}
  \end{minipage}%
}
\medskip

A complete bipartite graph $K_{1,3}$ is called a \textit{claw}, and the vertex that is incident with all the edges in a claw is called the \textit{claw-center}. \textit{Claw-free graph} is a graph that does not contain a claw as an induced subgraph.

Let $G$ be a connected graph, and let $d(\cdot,\cdot)$ denote the distance between two vertices in $G$.
 For any four
vertices $u, v, x, y \in V(G)$, we define $\delta(u, v, x, y)$ as half of the difference between the two larger of the three sums $d(u, v) + d(x,y)$, $d(u, y) + d(v, x)$, $d(u, x) +
d(v, y)$. The \textit{hyperbolicity} of a graph $G$, denoted by $\delta(G)$, is the value of $\max_{u, v, x, y\in V(G)}\delta(u, v, x, y)$. For every $\delta\geq \delta(G)$, we say that $G$ is $\delta$-\textit{hyperbolic}. A graph class $\mathcal{G}$ is said to be \textit{hyperbolic} (or bounded hyperbolicity class) if there exists a constant $\delta$ such that every graph $G \in \mathcal{G}$ is $\delta$-hyperbolic.

\medskip
\noindent\textbf{Parameterized complexity: }
A \emph{parameterized problem} is of a decision problem paired with an integer \emph{parameter} $k$ that is determined by the problem instance. Such a problem is called \emph{fixed-parameter tractable} (FPT) if there exists an algorithm that can solve any instance $I$ of size $|I|$ with parameter $k$ in time $f(k) \cdot |I|^c$, where $f$ is a computable function and $c$ is a constant. For a given parameterized problem $P$, a \emph{kernel} refers to a function that maps each instance of $P$ to another instance of $P$ that is equivalent but has its size bounded by a function $h$ of the parameter. If $h$ is a polynomial function, the kernel is called a \emph{polynomial kernel}. An \emph{FPT-reduction} between two parameterized problems $P$ and $Q$ is a function that transforms an instance $(I, k)$ of $P$ into an instance $(f(I), g(k))$ of $Q$. Here, $f$ and $g$ must be computable in FPT time with respect to $k$, and the transformation must preserve the YES/NO answer. If, in addition, $f$ can be computed in polynomial time and $g$ is polynomial in $k$, the reduction is termed a \emph{polynomial time and parameter transformation}. These reductions are useful for establishing conditional lower bounds. Specifically, if a parameterized problem $P$ is not FPT (or lacks a polynomial kernel) and there exists an FPT-reduction (or a polynomial time and parameter transformation) from $P$ to another problem $Q$, then $Q$ is also unlikely to admit an FPT algorithm (or a polynomial kernel). These conclusions depend on certain widely accepted complexity assumptions; for further details, see \cite{cygan2015parameterized}.


\section{Complexity of the \textsc{Multipacking} problem}\label{sec:multipacking_npc_w}

In this section, we present a reduction from the \textsc{Hitting Set} problem to show the first hardness result on the \textsc{Multipacking} problem.

\NPCchordal*

\begin{proof}

We reduce the well-known \textsc{NP-complete} (also \textsc{W[2]-complete} when parametrized by solution size) problem \textsc{Hitting Set}~\cite{cygan2015parameterized,garey1979computers} to the \textsc{Multipacking} problem.

\medskip
\noindent
\fbox{%
  \begin{minipage}{\dimexpr\linewidth-2\fboxsep-2\fboxrule}
  \textsc{ Hitting Set} problem
  \begin{itemize}[leftmargin=1.5em]
    \item \textbf{Input:} A finite set $U$, a collection $\mathcal{F}$ of subsets of $U$, and an integer $k \in \mathbb{N}$.
    \item \textbf{Question:} Does there exist a \emph{hitting set} $S \subseteq U$ of size at most $k$; that is, a set of at most $k$ elements from $U$ such that each set in $\mathcal{F}$ contains at least one element from $S$?
  \end{itemize}
  \end{minipage}%
}
\medskip

Let $U=\{u_1,u_2,\dots,u_n\}$ and  $\mathcal{F}=\{S_1,S_2,\dots,S_m\}$ where $S_i\subseteq U$ for each $i$. We construct a graph $G$ in the following way: (i) Include each element of $\mathcal{F}$ in the vertex set of $G$ and $\mathcal{F}$ forms a clique in $G$. (ii) Include $k-2$ length path $u_i = u_i^1 u_i^2 \dots u_i^{k-1}$  in $G$,  for each element $u_i\in U$.
(iii) Join $u_i^1$ and $S_j$ by an edge iff $u_i \notin S_j$. Fig. \ref{fig:reduction_chordal} gives an illustration.

\begin{figure}[ht]
    \centering
   \includegraphics[width=\textwidth]{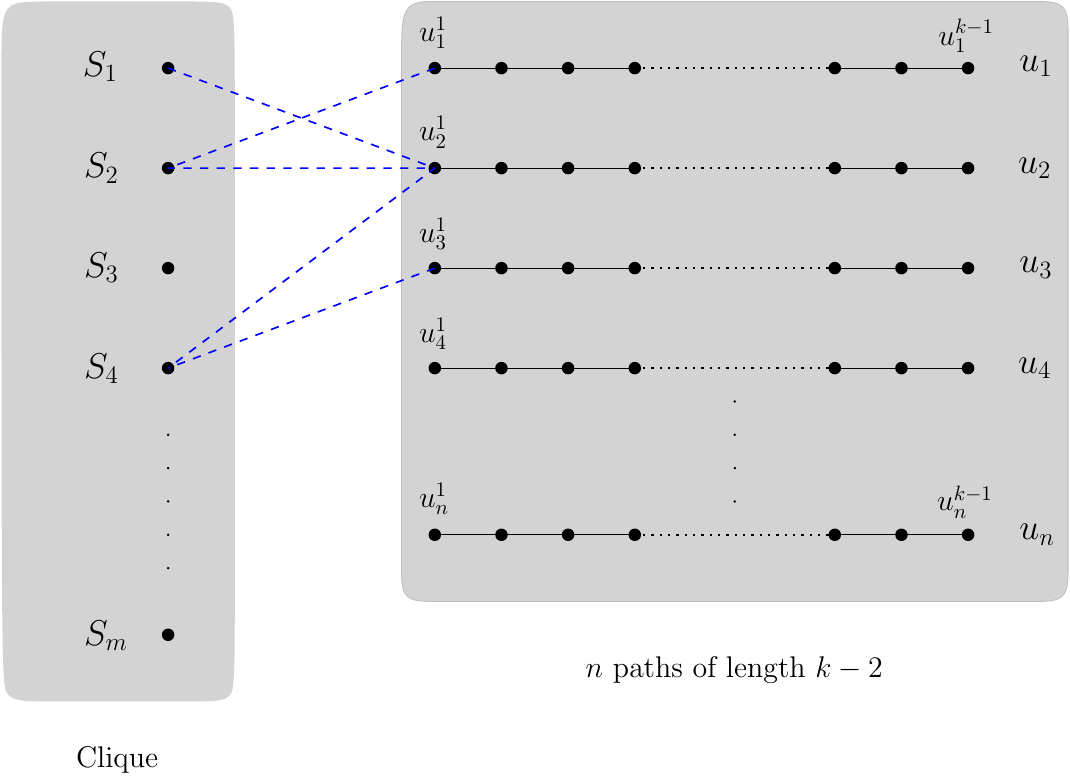}
    \caption{An illustration of the construction of the graph $G$ used in the proof of Theorem~\ref{thm:NPC_chordal}, demonstrating the hardness of the \textsc{Multipacking} problem.}
    \label{fig:reduction_chordal}
\end{figure}


\vspace{0.3cm}
\noindent
\textbf{Claim \ref{thm:NPC_chordal}.1. } $\mathcal{F}$ has a hitting set of size at most $k$ iff $G$ has a multipacking of size at least $k$, for $k\geq 2$.

\begin{claimproof} \textsf{(If part)} Suppose $\mathcal{F}$ has a hitting set $H = \{ u_1, u_2, \dots, u_k \}$. 
We want to show that $M =\{ u_1^{k-1}, u_2^{k-1}, \dots, u_k^{k-1} \}$ is a multipacking in $G$. If $r < k-1$, then 
$| N_{r}[v] \cap M |\leq 1 \leq r$, $ \forall v \in V(G)$, since $d(u_i^{k-1}, u_j^{k-1})\geq 2k-2$, $\forall i\neq j$.
If $r > k-1$, then 
$| N_{r}[v] \cap M | \leq | M | = k \leq r$, $\forall v \in V(G)$. Suppose $r = k-1$. If $M$ is not a multipacking,  there exists some $v \in V(G)$ such that 
$| N_{k-1}[v] \cap M | = k.$
Therefore, 
$M \subseteq N_{k-1}[v].$ Suppose $v=u_p^q$ for some $p$ and $q$. In this case, $d(u_p^q,u_i^{k-1})\geq k$, $\forall i\neq p$. Therefore, $v\notin \{u_i^j:1\leq i\leq n, 1\leq j\leq k-1\}$. 
This implies that $v \in \{ S_1, S_2, \dots, S_m \}$. Let $v = S_t$ for some $t$. If there exists an $i\in\{1,2,\dots,k\}$ such that $S_t$ is not adjacent to $u_i^1$, then $d(S_t,u_i^{k-1})\geq k$. Therefore, 
$M \subseteq N_{k-1}[S_t] $ implies that  $S_t$ is adjacent to each vertex of the set  $\{u_1^1, u_2^1, \dots, u_k^1\}$.
Then 
$u_1, u_2, \dots, u_k \notin S_t.$
Therefore, $H$ is not a hitting set of $\mathcal{F}$. This is a contradiction. Hence, $M$ is a multipacking of size $k$.

 \textsf{(Only-if part)}  Suppose $G$ has a multipacking $M$ of size $k$. Let 
$H = \{ u_i : u_i^j \in M, 1\leq i\leq n, 1\leq j\leq k-1 \}.$
Then 
$|H| \leq |M| = k.$
Without loss of generality, let 
$H = \{ u_1, u_2, \dots, u_{k'} \} $ where $ k' \leq k$.
Now we want to show that $H$ is a hitting set of $\mathcal{F}$. Suppose $H$ is not a hitting set. Then there exists $t$ such that 
$S_t \cap H = \emptyset.$
Therefore, $S_t$ is adjacent to each vertex of the set $\{ u_1^{1}, u_2^{1}, \dots, u_{k'}^{1} \}$.
Then $M \subseteq N_{k-1}[S_t]$. This implies that 
$| N_{k-1}[S_t] \cap M | = k.$
This is a contradiction. Therefore, $H$ is a hitting set of size at most $k$. 
\end{claimproof}

\noindent Since the above reduction is a polynomial time reduction, therefore the \textsc{Multipacking} problem  is \textsc{NP-complete}.  Moreover, the reduction is also an FPT reduction. Hence, the \textsc{Multipacking} problem  is \textsc{W[2]-hard} when parameterized by the solution size.
\end{proof}

\begin{corollary}\label{cor:NPC_chordal}
  \textsc{Multipacking} problem is \textsc{NP-complete} for chordal graphs. Moreover, \textsc{Multipacking} problem  is \textsc{W[2]-hard} for chordal graphs when parameterized by the solution size.
    
\end{corollary}

\begin{proof}
    We prove that the graph $G$ of the reduction we use in the proof of the Theorem \ref{thm:NPC_chordal} is chordal. Note that if $G$ has a cycle $C=c_0c_1c_2\dots c_{t-1}c_0$ of length at least $4$, then there exists $i$ such that $c_i,c_{i+2 \! \pmod{t}}\in \mathcal{F}$ because no two vertices of the set $\{ u_1^{1}, u_2^{1}, \dots, u_n^{1} \}$ are adjacent. Since $\mathcal{F}$ forms a clique, so $c_i$ and $c_{i+2 \! \pmod{t}}$ are endpoints of a chord. Therefore, the graph $G$ is a chordal graph.
\end{proof}


It is known that, unless \textsc{ETH}\footnote{The \textit{Exponential Time Hypothesis (ETH)} states that the \textsc{3-SAT} problem cannot be solved in time $2^{o(n)}$, where $n$ is the number of variables of the input CNF (Conjunctive Normal Form) formula.} fails, there is no $f(k)(m+n)^{o(k)}$-time algorithm for the \textsc{Hitting Set} problem, where $k$ is the solution size, $|U|=n$ and  $|\mathcal{F}|=m$ \cite{cygan2015parameterized}. We use this result to show the following.

\ETHmultipackingSubExp*

\begin{proof}
   Let $|V(G)|=n_1$. From the construction of $G$ in the proof of Theorem \ref{thm:NPC_chordal}, we have $n_1=m+n(k-1)=O(m+n^2)$. Therefore, the construction of $G$ takes $O(m+n^2)$-time. Suppose that there is an $f(k)n_1^{o(k)}$-time algorithm for the \textsc{Multipacking} problem. Therefore, we can solve the \textsc{Hitting Set} problem in $f(k)n_1^{o(k)}+O(m+n^2)=f(k)(m+n^2)^{o(k)}=f(k)(m+n)^{o(k)}$ time. This is a contradiction, since there is no $f(k)(m+n)^{o(k)}$-time algorithm for the \textsc{Hitting Set} problem, assuming \textsc{ETH}.  
\end{proof}

\section{Hardness results on some subclasses} \label{sec:hardness_on_subclasses}

In this section, we present hardness results on some subclasses.

\subsection{Chordal $\cap \text{ } \frac{1}{2}$-hyperbolic graphs}

In Corollary \ref{cor:NPC_chordal},  we have already shown that, for chordal graphs, the \textsc{Multipacking} problem is \textsc{NP-complete} and it is \textsc{W[2]-hard} when parameterized by the solution size. Here we improve the result. It is known that the \textsc{Multipacking} problem is solvable in cubic-time for strongly chordal graphs. We show the hardness results for chordal $\cap \text{ } \frac{1}{2}$-hyperbolic graphs which is a superclass of strongly chordal graphs. To show that we need the following result.

\begin{theorem}[\cite{brinkmann2001hyperbolicity}]\label{thm:chordal_one_hyperbolic}
 If $G$ is a chordal graph, then the hyperbolicity  $\delta(G)\leq 1$. Moreover, $\delta(G)=1$  if and only if  $G$ does contain one of the graphs in Fig. \ref{fig:chordal_half_hyperbolic_forbidden_graphs} as an isometric subgraph.  Equivalently, a chordal graph $G$ is $\frac{1}{2}$-hyperbolic (or, $\delta(G)\leq \frac{1}{2}$) if and only if $G$ does not contain any of the graphs in Fig. \ref{fig:chordal_half_hyperbolic_forbidden_graphs} as an isometric subgraph.
\end{theorem}

\begin{figure}[ht]
    \centering
    \includegraphics[height=4.5cm]{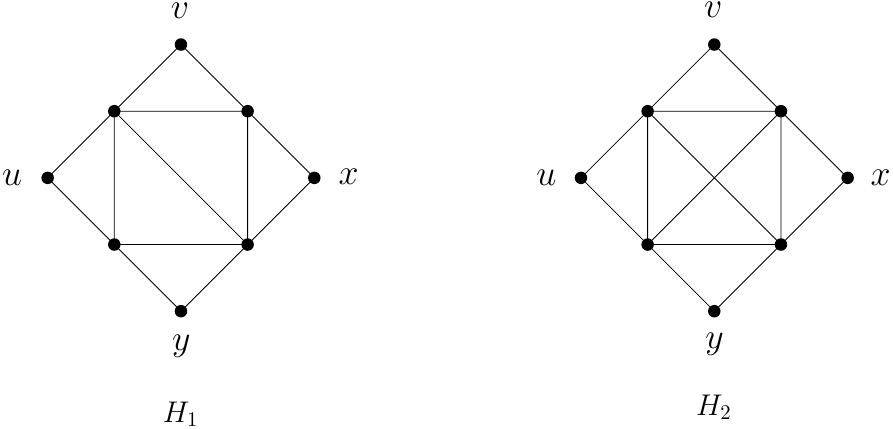}
    \caption{Forbidden isometric subgraphs for chordal $\cap$ $\frac{1}{2}$-hyperbolic graphs} 
    \label{fig:chordal_half_hyperbolic_forbidden_graphs}
\end{figure}



\NPChalfhyperbolicchordal*


\begin{proof} We reduce the \textsc{Hitting Set} problem to the \textsc{Multipacking} problem  to show that the \textsc{Multipacking} problem  is \textsc{NP-complete} for chordal $\cap$ $\frac{1}{2}$-hyperbolic graphs. 

Let $U=\{u_1,u_2,\dots,u_n\}$ and  $\mathcal{F}=\{S_1,S_2,\dots,S_m\}$ where $S_i\subseteq U$ for each $i$. We construct a graph $G$ in the following way: (i) Include each element of $\mathcal{F}$ in the vertex set of $G$. (ii) Include $k-2$ length path $u_i = u_i^1 u_i^2 \dots u_i^{k-1}$  in $G$,  for each element $u_i\in U$.
(iii) Join $u_i^1$ and $S_j$ by an edge iff $u_i \notin S_j$. (iv) Include the vertex set  $Y=\{y_{i,j}:1\leq i<j\leq n\}$ in $G$. (v) For each $i,j$ where $1\leq i<j\leq n$, $y_{i,j}$ is adjacent to $u_i^1$ and $u_j^1$. (vi) $\mathcal{F}\cup Y$ forms a clique in $G$. Fig. \ref{fig:reduction_chordal_half_hyperbolic} gives an illustration.  

\begin{figure}[p]
    \centering
   \includegraphics[width=\textwidth]{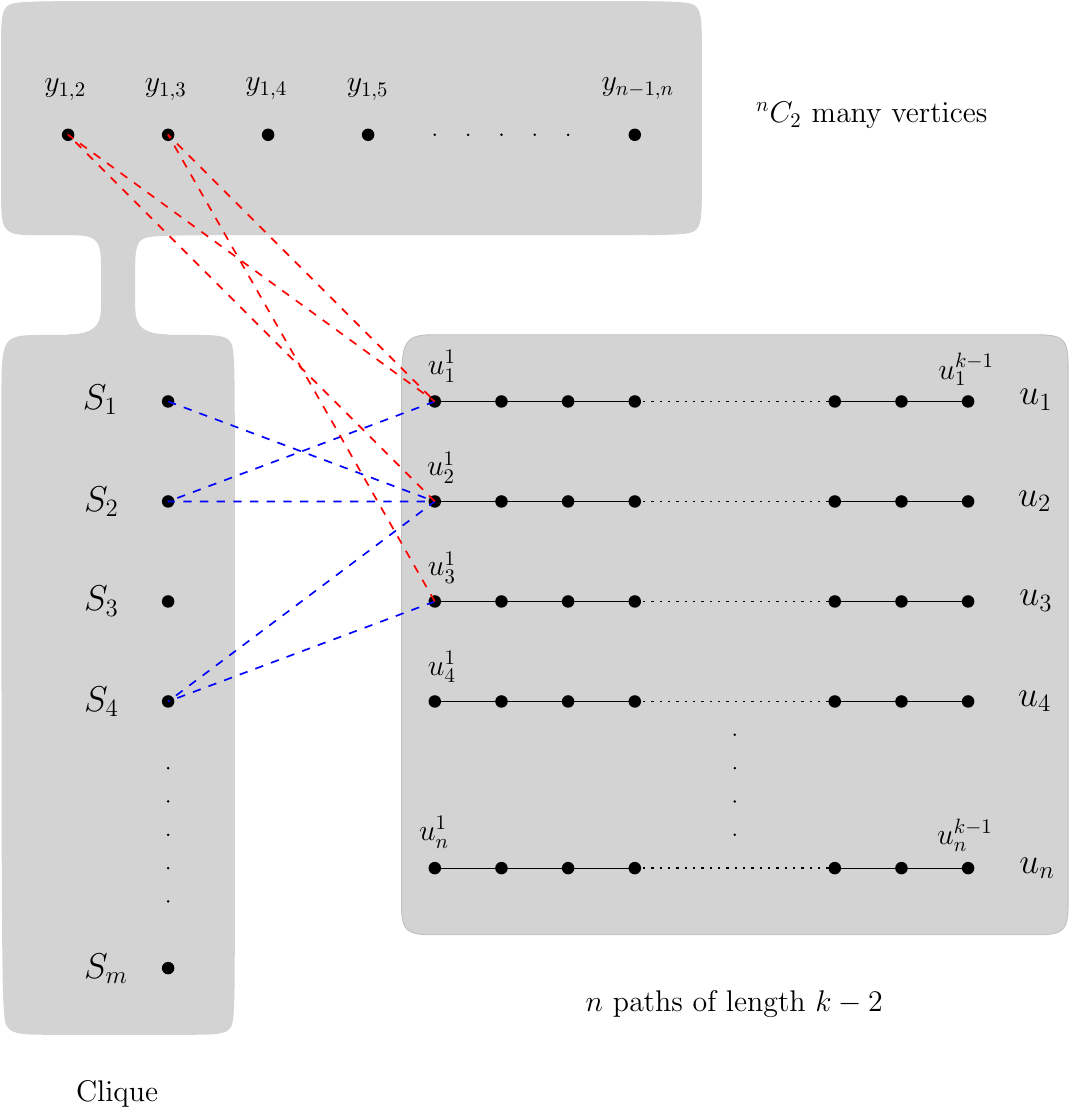}
    \caption{An illustration of the construction of the graph $G$ used in the proof of Theorem~\ref{thm:NPC_half_hyperbolic_chordal}, demonstrating the hardness of the \textsc{Multipacking} problem for chordal $\cap$ $\frac{1}{2}$-hyperbolic graphs.}
    \label{fig:reduction_chordal_half_hyperbolic}
\end{figure}

\vspace{0.3cm}
\noindent
\textbf{Claim \ref{thm:NPC_half_hyperbolic_chordal}.1. } $G$ is a chordal graph.

\begin{claimproof} If $G$ has a cycle $C=c_0c_1c_2\dots c_{t-1}c_0$ of length at least $4$, then there exists $i$ such that $c_i,c_{i+2 \! \pmod{t}}\in \mathcal{F}\cup Y$ because no two vertices of the set $\{ u_1^{1}, u_2^{1}, \dots, u_n^{1} \}$ are adjacent. Since $\mathcal{F}\cup Y$ forms a clique,  $c_i$ and $c_{i+2 \! \pmod{t}}$ are endpoints of a chord. Therefore, the graph $G$ is a chordal graph.    
\end{claimproof}

\vspace{0.3cm}
\noindent
\textbf{Claim \ref{thm:NPC_half_hyperbolic_chordal}.2. } $G$ is a  $\frac{1}{2}$-hyperbolic graph.

\begin{claimproof} Let $G'$ be the induced subgraph of $G$ on the vertex set $\mathcal{F}\cup Y\cup \{ u_1^{1}, u_2^{1}, \dots, u_n^{1} \}$. We want to show that $\mathrm{diam}(G')\leq 2$. Note that the distance between any two vertices of the set $\mathcal{F}\cup Y$ is $1$, since $\mathcal{F}\cup Y$ forms a clique in $G$. Furthermore, the distance between any two vertices of the set $\{ u_1^{1}, u_2^{1}, \dots, u_n^{1} \}$ is $2$, since they are non-adjacent and have a common neighbor in the set $Y$. Suppose $v\in \mathcal{F}\cup Y$ and $u_i^1\in \{ u_1^{1}, u_2^{1}, \dots, u_n^{1} \}$ where $v$ and $u_i^1$ are non-adjacent. In that case, $u_iy_{1,i}v$ is a $2$-length path that joins $v$ and $u_i^1$ since $\mathcal{F}\cup Y$ forms a clique in $G$. Hence, $\mathrm{diam}(G')\leq 2$. From the  Claim \ref{thm:NPC_half_hyperbolic_chordal}.1 we have that $G$ is chordal. From Theorem \ref{thm:chordal_one_hyperbolic}, we can say that the hyperbolicity of $G$ is at most $1$. Suppose $G$ has hyperbolicity exactly $1$. In that case, from Theorem \ref{thm:chordal_one_hyperbolic}, we can say that $G$  contains at least one of the graphs $H_1$ or $H_2$  (Fig. \ref{fig:chordal_half_hyperbolic_forbidden_graphs}) as 
 isometric subgraphs. Note that every vertex of $H_1$ and $H_2$ is a vertex of a cycle. Therefore, $H_1$ or $H_2$ does not have a vertex from the set $\{u_i^j:1\leq i\leq n, 2\leq j\leq k-1\}$. This implies that $H_1$ or $H_2$ are isometric subgraphs of $G'$. But both $H_1$ and $H_2$ have diameter $3$, while $\mathrm{diam}(G')\leq 2$. This is a contradiction. Hence, $G$ is a  $\frac{1}{2}$-hyperbolic graph.   
\end{claimproof}

\vspace{0.3cm}
\noindent
\textbf{Claim \ref{thm:NPC_half_hyperbolic_chordal}.3. } $\mathcal{F}$ has a hitting set of size at most $k$ iff $G$ has a multipacking of size at least $k$, for $k\geq 3$.

\begin{claimproof} \textsf{(If part)} Suppose $\mathcal{F}$ has a hitting set $H = \{ u_1, u_2, \dots, u_k \}$. 
We want to show that $M =\{ u_1^{k-1}, u_2^{k-1}, \dots, u_k^{k-1} \}$ is a multipacking in $G$. If $r < k-1$, then 
$| N_{r}[v] \cap M |\leq 1 \leq r$, $ \forall v \in V(G)$, since $d(u_i^{k-1}, u_j^{k-1})\geq 2k-2$, $\forall i\neq j$.
If $r > k-1$, then 
$| N_{r}[v] \cap M | \leq | M | = k \leq r$, $\forall v \in V(G)$. Suppose $r = k-1$. If $M$ is not a multipacking,  there exists some $v \in V(G)$ such that 
$| N_{k-1}[v] \cap M | = k.$
Therefore, 
$M \subseteq N_{k-1}[v].$ Suppose $v=u_p^q$ for some $p$ and $q$. In this case, $d(u_p^q,u_i^{k-1})\geq k$, $\forall i\neq p$. Therefore, $v\notin \{u_i^j:1\leq i\leq n, 1\leq j\leq k-1\}$. This implies that $v \in \mathcal{F}\cup Y$. Suppose $v \in Y$. Let $v=y_{p,q}$ for some $p$ and $q$. Then $d(y_{p,q},u_i^{k-1})\geq k$, $\forall i\in \{1,2,\dots,n\}\setminus\{p,q\} $. Therefore, $v\notin Y$. This implies that $v \in \mathcal{F}$. Let $v = S_t$ for some $t$. If there exists an $i\in\{1,2,\dots,k\}$ such that $S_t$ is not adjacent to $u_i^1$, then $d(S_t,u_i^{k-1})\geq k$. Therefore, 
$M \subseteq N_{k-1}[S_t] $ implies that  $S_t$ is adjacent to each vertex of the set  $\{u_1^1, u_2^1, \dots, u_k^1\}$.
Then 
$u_1, u_2, \dots, u_k \notin S_t.$
Therefore, $H$ is not a hitting set of $\mathcal{F}$. This is a contradiction. Hence, $M$ is a multipacking of size $k$.

 \textsf{(Only-if part)}  Suppose $G$ has a multipacking $M$ of size $k$. Let 
$H = \{ u_i : u_i^j \in M, 1\leq i\leq n, 1\leq j\leq k-1 \}.$
Then 
$|H| \leq |M| = k.$
Without loss of generality, let 
$H = \{ u_1, u_2, \dots, u_{k'} \} $ where $ k' \leq k$.
Now we want to show that $H$ is a hitting set of $\mathcal{F}$. Suppose $H$ is not a hitting set. Then there exists $t$ such that 
$S_t \cap H = \emptyset.$
Therefore, $S_t$ is adjacent to each vertex of the set $\{ u_1^{1}, u_2^{1}, \dots, u_{k'}^{1} \}$.
Then $M \subseteq N_{k-1}[S_t]$. This implies that 
$| N_{k-1}[S_t] \cap M | = k.$
This is a contradiction. Therefore, $H$ is a hitting set of size at most $k$.
\end{claimproof}

\noindent Since the above reduction is a polynomial time reduction, therefore the \textsc{Multipacking} problem  is \textsc{NP-complete} for chordal $\cap$ $\frac{1}{2}$-hyperbolic graphs.  Moreover, the reduction is also an FPT reduction. Hence, the \textsc{Multipacking} problem  is \textsc{W[2]-hard} for chordal $\cap$ $\frac{1}{2}$-hyperbolic graphs when parameterized by the solution size.
\end{proof}

\subsection{Bipartite graphs}

Here, we discuss the hardness of the \textsc{Multipacking} problem for bipartite graphs.

\NPCbipartite*

\begin{proof}

We reduce the \textsc{Hitting Set} problem to the \textsc{Multipacking} problem  to show that the \textsc{Multipacking} problem  is \textsc{NP-complete} for bipartite graphs. 

Let $U=\{u_1,u_2,\dots,u_n\}$ and  $\mathcal{F}=\{S_1,S_2,\dots,S_m\}$ where $S_i\subseteq U$ for each $i$. We construct a graph $G$ in the following way: (i)  Include each element of $\mathcal{F}$ in the vertex set of $G$. (ii) Include a vertex $C$ in $G$ where all the vertices of $\mathcal{F}$ are adjacent to the vertex $C$. (iii) Include $k-2$ length path $u_i = u_i^1 u_i^2 \dots u_i^{k-1}$  in $G$,  for each element $u_i\in U$.
(iv) Join $u_i^1$ and $S_j$ by an edge iff $u_i \notin S_j$.  Fig. \ref{fig:reduction_bipartite} gives an illustration. 

Note that $G$ is a bipartite graph where the partite sets are $B_1=\mathcal{F}\cup \{u_i^j:1\leq i\leq n \text{ and } j\text{ is even} \}$ and $B_2=\{C\}\cup \{u_i^j:1\leq i\leq n \text{ and } j\text{ is odd} \}$. 

\begin{figure}[ht]
    \centering
   \includegraphics[width=\textwidth]{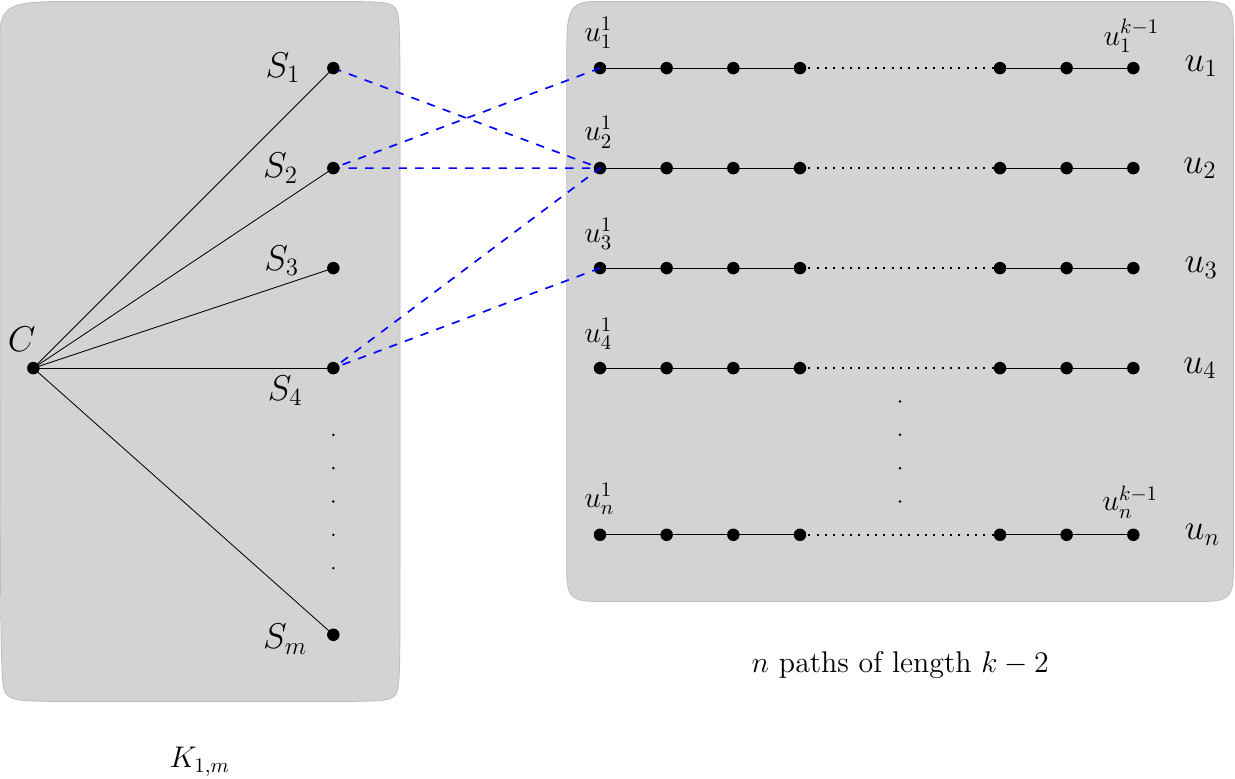}
    \caption{An illustration of the construction of the graph $G$ used in the proof of Theorem~\ref{thm:NPC_bipartite}, demonstrating the hardness of the \textsc{Multipacking} problem for bipartite graphs.}
    \label{fig:reduction_bipartite}
\end{figure}

\vspace{0.3cm}
\noindent
\textbf{Claim \ref{thm:NPC_bipartite}.1. } $\mathcal{F}$ has a hitting set of size at most $k$ iff $G$ has a multipacking of size at least $k$, for $k\geq 2$.

\begin{claimproof} \textsf{(If part)} Suppose $\mathcal{F}$ has a hitting set $H = \{ u_1, u_2, \dots, u_k \}$. 
We want to show that $M =\{ u_1^{k-1}, u_2^{k-1}, \dots, u_k^{k-1} \}$ is a multipacking in $G$. If $r < k-1$, then 
$| N_{r}[v] \cap M |\leq 1 \leq r$, $ \forall v \in V(G)$, since $d(u_i^{k-1}, u_j^{k-1})\geq 2k-2$, $\forall i\neq j$.
If $r > k-1$, then 
$| N_{r}[v] \cap M | \leq | M | = k \leq r$, $\forall v \in V(G)$. Suppose $r = k-1$. If $M$ is not a multipacking,  there exists some $v \in V(G)$ such that 
$| N_{k-1}[v] \cap M | = k.$
Therefore, 
$M \subseteq N_{k-1}[v].$ Suppose $v=u_p^q$ for some $p$ and $q$. In this case, $d(u_p^q,u_i^{k-1})\geq k$, $\forall i\neq p$. Therefore, $v\notin \{u_i^j:1\leq i\leq n, 1\leq j\leq k-1\}$. Moreover, $v\neq C$, since $d(C,u_i^{k-1})\geq k$, $\forall 1\leq i\leq n$. 
This implies that $v \in \{ S_1, S_2, \dots, S_m \}$. Let $v = S_t$ for some $t$.  If there exists an $i\in\{1,2,\dots,k\}$ such that $S_t$ is not adjacent to $u_i^1$, then $d(S_t,u_i^{k-1})\geq k$. Therefore,  
$M \subseteq N_{k-1}[S_t] $ implies that  $S_t$ is adjacent to each vertex of the set  $\{u_1^1, u_2^1, \dots, u_k^1\}$.
Then 
$u_1, u_2, \dots, u_k \notin S_t.$
Therefore, $H$ is not a hitting set of $\mathcal{F}$. This is a contradiction. Hence, $M$ is a multipacking of size $k$.

 \textsf{(Only-if part)}  Suppose $G$ has a multipacking $M$ of size $k$. Let 
$H = \{ u_i : u_i^j \in M, 1\leq i\leq n, 1\leq j\leq k-1 \}.$
Then 
$|H| \leq |M| = k.$
Without loss of generality, let 
$H = \{ u_1, u_2, \dots, u_{k'} \} $ where $ k' \leq k$.
Now we want to show that $H$ is a hitting set of $\mathcal{F}$. Suppose $H$ is not a hitting set. Then there exists $t$ such that 
$S_t \cap H = \emptyset.$
Therefore, $S_t$ is adjacent to each vertex of the set $\{ u_1^{1}, u_2^{1}, \dots, u_{k'}^{1} \}$.
Then $M \subseteq N_{k-1}[S_t]$. This implies that 
$| N_{k-1}[S_t] \cap M | = k.$
This is a contradiction. Therefore, $H$ is a hitting set of size at most $k$. 
\end{claimproof}

\noindent Since the above reduction is a polynomial time reduction, therefore the \textsc{Multipacking} problem  is \textsc{NP-complete} for bipartite graphs.  Moreover, the reduction is also an FPT reduction. Hence, the \textsc{Multipacking} problem  is \textsc{W[2]-hard} for bipartite graphs when parameterized by the solution size.
\end{proof}

\subsection{Claw-free graphs}

Next, we discuss the hardness of the \textsc{Multipacking} problem for claw-free graphs.

\NPCclawfree*

\begin{proof}
    
We reduce the \textsc{Hitting Set} problem to the \textsc{Multipacking} problem  to show that the \textsc{Multipacking} problem  is \textsc{NP-complete} for claw-free graphs. 

Let $U=\{u_1,u_2,\dots,u_n\}$ and  $\mathcal{F}=\{S_1,S_2,\dots,S_m\}$ where $S_i\subseteq U$ for each $i$. We construct a graph $G$ in the following way: (i)  Include each element of $\mathcal{F}$ in the vertex set of $G$ and $\mathcal{F}$ forms a clique in $G$. (ii) Include $k-3$ length path $u_i = u_i^1 u_i^2 \dots u_i^{k-2}$  in $G$,  for each element $u_i\in U$.
(iii) Join $S_j$ and $u_i^1$ by a $2$-length path $S_jw_{j,i}u_i^1$  iff $u_i \notin S_j$. Let $W=\{w_{j,i}:u_i \notin S_j, 1\leq i\leq n, 1\leq j\leq m\}$. (iv) Join $w_{j,i}$ and $w_{q,p}$ by an edge iff either $j=q$ or $i=p$. Fig. \ref{fig:reduction_claw_free} gives an illustration.

\begin{figure}[ht]
    \centering
   \includegraphics[width=\textwidth]{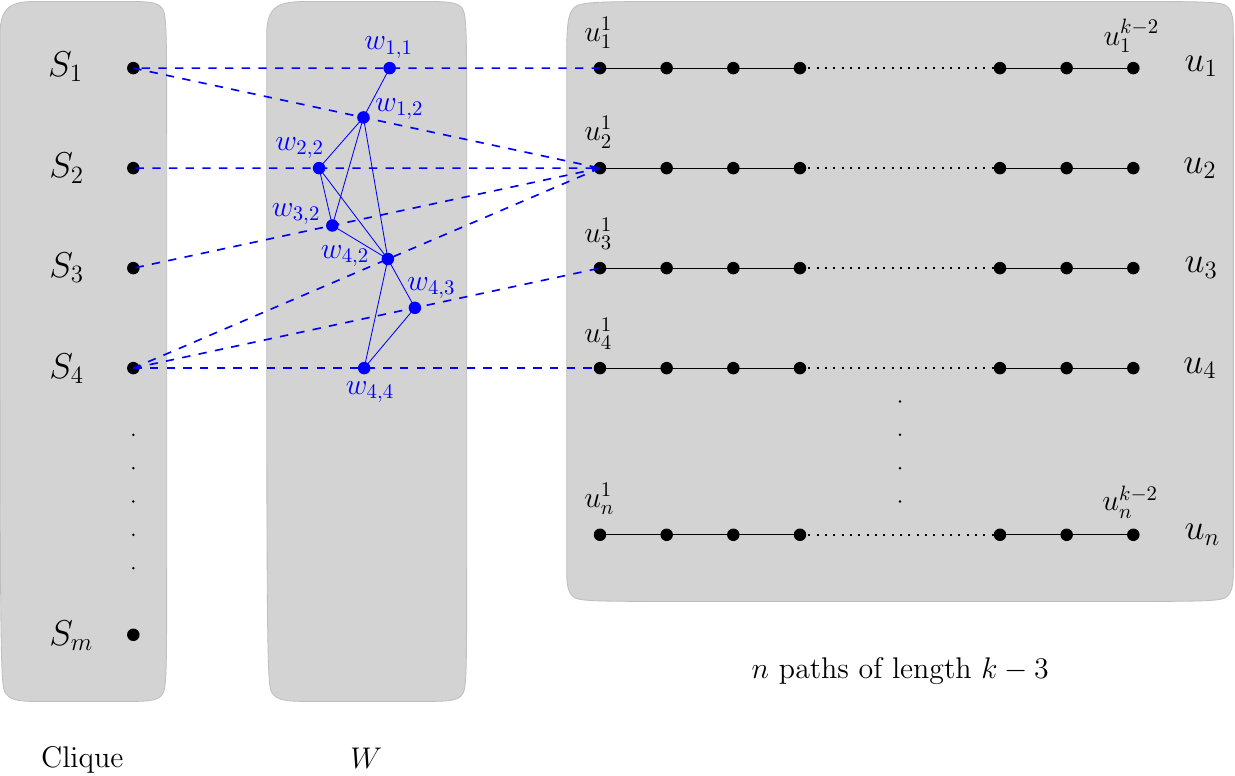}
    \caption{An illustration of the construction of the graph $G$ used in the proof of Theorem~\ref{thm:NPC_clawfree}, demonstrating the hardness of the \textsc{Multipacking} problem for claw-free graphs.}
    \label{fig:reduction_claw_free}
\end{figure}

\vspace{0.3cm}
\noindent
\textbf{Claim \ref{thm:NPC_clawfree}.1.} $G$ is a claw-free graph.

\begin{claimproof} Suppose $G$ contains an induced claw $C$ whose vertex set is $V(C)=\{c,x,y,z\}$ and the claw-center is $c$. Note that $c\notin \{u_i^j:1\leq i\leq n, 2\leq j\leq k-2\}$, since each element of this set has degree at most $2$. 

Suppose $c\in \{u_i^1:1\leq i\leq n\}$. Let $c=u_t^1$ for some $t$. Therefore, $\{x,y,z\}\subseteq N(c)=\{u_t^2\}\cup \{w_{j,t}:u_t\notin S_j,1\leq j\leq m\}$. Then at least $2$ vertices of the set $\{x,y,z\}$ belong to the set $\{w_{j,t}:u_t\notin S_j,1\leq j\leq m\}$. Note that $\{w_{j,t}:u_t\notin S_j,1\leq j\leq m\}$ forms a clique in $G$. This implies that at least $2$ vertices of the set $\{x,y,z\}$ are adjacent. This is a contradiction, since $C$ is an induced claw in $G$. Therefore,  $c\notin \{u_i^1:1\leq i\leq n\}$. 

Suppose $c\in \mathcal{F}$. Let $c=S_t$ for some $t$. Therefore, $\{x,y,z\}\subseteq N(c)=\{S_j:j\neq t, 1\leq j\leq m\}\cup \{w_{t,i}:u_i\notin S_t,1\leq i\leq n\}$. Note that both sets $\{S_j:j\neq t, 1\leq j\leq m\}$ and $ \{w_{t,i}:u_i\notin S_t,1\leq i\leq n\}$ form cliques in $G$. This implies that at least $2$ vertices of the set $\{x,y,z\}$ are adjacent. This is a contradiction, since $C$ is an induced claw in $G$. Therefore, $c\notin \mathcal{F}$.

Suppose $c\in W$. Let $c=w_{q,p}$ for some $q$ and $p$. Let $A_q=\{S_q\}\cup \{w_{q,i}:u_i\notin S_q,i\neq p,1\leq i\leq n\}$ and $B_p=\{u_p^1\}\cup \{w_{j,p}:u_p\notin S_j,j\neq q,1\leq j\leq m\}$. Note that both sets $A_q$ and $B_p$ form cliques in $G$ and $\{x,y,z\}\subseteq N(c)=A_q\cup B_p$. This implies that at least $2$ vertices of the set $\{x,y,z\}$ are adjacent. This is a contradiction, since $C$ is an induced claw in $G$. Therefore, $c\notin W$.

Therefore, $G$ is a claw-free graph. \end{claimproof}

\vspace{0.3cm}
\noindent
\textbf{Claim \ref{thm:NPC_clawfree}.2.} $\mathcal{F}$ has a hitting set of size at most $k$ iff $G$ has a multipacking of size at least $k$, for $k\geq 2$.

\begin{claimproof} \textsf{(If part)} Suppose $\mathcal{F}$ has a hitting set $H = \{ u_1, u_2, \dots, u_k \}$. 
We want to show that $M =\{ u_1^{k-2}, u_2^{k-2}, \dots, u_k^{k-2} \}$ is a multipacking in $G$. If $r < k-1$, then 
$| N_{r}[v] \cap M |\leq 1 \leq r$, $ \forall v \in V(G)$, since $d(u_i^{k-2}, u_j^{k-2})\geq 2k-3$, $\forall i\neq j$.
If $r > k-1$, then 
$| N_{r}[v] \cap M | \leq | M | = k \leq r$, $\forall v \in V(G)$. Suppose $r = k-1$. If $M$ is not a multipacking,  there exists some $v \in V(G)$ such that 
$| N_{k-1}[v] \cap M | = k.$
Therefore, 
$M \subseteq N_{k-1}[v]$.

\vspace{0.22cm}

 \noindent \textit{\textbf{Case 1: }} $v\in \{u_i^j:1\leq i\leq n, 1\leq j\leq k-2\}$.
 
Let $v=u_p^q$ for some $p$ and $q$. In this case, $d(u_p^q,u_i^{k-2})\geq k$, $\forall i\neq p$. Therefore, $v\notin \{u_i^j:1\leq i\leq n, 1\leq j\leq k-2\}$.

\vspace{0.22cm}

 \noindent \textit{\textbf{Case 2: }} $v \in W$.
 
 Let $v=w_{p,q}$ for some $p$ and $q$. Therefore, $u_q\notin S_p$. Since $M \subseteq N_{k-1}[w_{p,q}]$, we have $d(w_{p,q},u_i^{k-2})\leq k-1$, $\forall 1\leq i\leq k$. This implies that $d(w_{p,q},u_i^1)\leq 2$, $\forall 1\leq i\leq k$. Note that $d(w_{p,q},u_i^1)\neq 1$, $\forall i\neq q, 1\leq i\leq k$. Therefore, $d(w_{p,q},u_i^1)=2$, $\forall i\neq q, 1\leq i\leq k$. Suppose $t\neq q$ and $1\leq t\leq k$. Then $d(w_{p,q},u_t^1)=2$.
Therefore, $w_{p,q}$ is adjacent to $w_{h,t}$, for some $h$. This is only possible when $h=p$, since $t\neq q$. Therefore, $u_t^1$ and $S_p$ are joined by a $2$-length path $u_t^1-w_{p,t}-S_p$. This implies that $u_t\notin S_p$. Therefore, $u_1, u_2, \dots, u_k \notin S_p$.  Therefore, $H$ is not a hitting set of $\mathcal{F}$. This is a contradiction. Therefore, $v \notin W$.

\vspace{0.22cm}

 \noindent \textit{\textbf{Case 3: }} $v\in \mathcal{F}$.

 Let $v = S_t$ for some $t$. If there exists an $i\in\{1,2,\dots,k\}$ such that $u_i\in S_t$, then $d(S_t,u_i^1)\geq 3$. Therefore, $d(S_t,u_i^{k-2})\geq k$. This implies that $u_i^{k-2}\notin N_{k-1}[S_t]$. Therefore, $M \nsubseteq N_{k-1}[S_t] $. Hence, $M \subseteq N_{k-1}[S_t] $ implies that  $u_1, u_2, \dots, u_k \notin S_t$. Then $H$ is not a hitting set of $\mathcal{F}$. This is a contradiction. Therefore, $v\notin \mathcal{F}$.

Hence, $M$ is a multipacking of size $k$.

\medskip
 \textsf{(Only-if part)}  Suppose $G$ has a multipacking $M$ of size $k$. Let 
$H = \{ u_i : u_i^j \in M, 1\leq i\leq n, 1\leq j\leq k-1 \}.$
Then 
$|H| \leq |M| = k.$
Without loss of generality, let 
$H = \{ u_1, u_2, \dots, u_{k'} \} $ where $ k' \leq k$.
We want to show that $H$ is a hitting set of $\mathcal{F}$. Suppose $H$ is not a hitting set. Then there exists $t$ such that 
$S_t \cap H = \emptyset.$
Therefore,  $S_t$ and $u_i^1$ are joined by a $2$-length path $S_t-w_{t,i}-u_i^1$, for each $i\in\{1,2,\dots,k'\}$. 
So, $M \subseteq N_{k-1}[S_t]$. This implies that 
$| N_{k-1}[S_t] \cap M | = k.$
This is a contradiction. Therefore, $H$ is a hitting set of size at most $k$. 
\end{claimproof}

\noindent Since the above reduction is a polynomial time reduction, therefore the \textsc{Multipacking} problem  is \textsc{NP-complete} for claw-free graphs.  Moreover, the reduction is also an FPT reduction. Hence, the \textsc{Multipacking} problem  is \textsc{W[2]-hard} for claw-free graphs when parameterized by the solution size.
\end{proof}

\subsection{Regular graphs}

Next, we prove NP-completeness for regular graphs. We will need the following result to prove our theorem.

From Erdős-Gallai Theorem~\cite{erdos1960grafok} and Havel-Hakimi criterion~\cite{hakimi1962realizability,havel1955remark}, we can say the following.

\begin{lemma}[\cite{erdos1960grafok,hakimi1962realizability,havel1955remark}]\label{lem:ErdosHavelHakimi}
    A simple $d$-regular graph on $n$ vertices can be formed in polynomial time when  $n \cdot d$ is even and $d < n$.
\end{lemma}


\multipackingregularNPc*

\begin{proof} It is known that the \textsc{Total Dominating Set} problem is \textsc{NP-complete} for cubic (3-regular) graphs~\cite{garey1979computers}.  We reduce the \textsc{Total Dominating Set} problem of cubic graph to the \textsc{Multipacking} problem  of regular graph. 

\medskip
\noindent
\fbox{%
  \begin{minipage}{\dimexpr\linewidth-2\fboxsep-2\fboxrule}
  \textsc{ Total Dominating Set} problem
  \begin{itemize}[leftmargin=1.5em]
    \item \textbf{Input:} An undirected graph $G = (V, E)$ and an integer $k \in \mathbb{N}$.
    \item \textbf{Question:} Does there exist a \emph{total dominating set} $S \subseteq V$ of size at most $k$; that is, a set of at most $k$ vertices such that every vertex in $V$ has at least one neighbor in $S$?
  \end{itemize}
  \end{minipage}%
}
\medskip

Let $(G,k)$ be an instance of the \textsc{Total Dominating Set} problem, where $G$ is a $3$-regular graph  with the vertex set $V=\{v_1,v_2,\dots,v_n\}$ where $n\geq 6$. Therefore, $n$ is even. Let $d=n-4$. So, $d$ is also even. Now we construct a graph $G'$ in the following way (Fig. \ref{fig:reduction_regular} gives an illustration). 

\noindent(i) For every vertex $v_a\in V(G)$, we construct a graph $H_a$ in $G'$ as follows. 

\begin{itemize}
    \item Add the vertex sets $S_a^i=\{u_a^{i,j}:1\leq j\leq d\}$ for all $2\leq i\leq k-2$ in $H_a$. Each vertex of $S_a^i$ is adjacent to each vertex of $S_a^{i+1}$ for each $i$.

    \item Add a vertex $u_a^{1}$ in $H_a$ where $u_a^{1}$ is adjacent to each vertex of $S_a^2$. Moreover, $S_a^2$ forms a clique.

    \item  Add the vertex set $T_a=\{u_a^{k-1,j}:1\leq j\leq d^2\}$ in $H_a$ where the induced subgraph $H_a[T_a]$ is a $(2d-1)$-regular graph. We can construct such a graph in polynomial-time by Lemma \ref{lem:ErdosHavelHakimi} since $d$ is an even number and $2d-1 < d^2 $ (since $n\geq 6$). 

    \item Let $U_1, U_2, \ldots, U_d$ be a partition of $T_a$ into $d$ disjoint sets, each of size $d$, that is, $T_a=U_1\sqcup U_2\sqcup \dots \sqcup U_d$ and $|U_j|=d$ for each $1\leq j\leq d$. Each vertex in $U_j$ is adjacent to $u_a^{k-2,j}$ for all $1\leq j\leq d$.
    
\end{itemize}

\noindent(ii) If $v_i$ and $v_j$ are different and not adjacent in $G$, we make $u^1_i$ and $u^1_j$ adjacent in $G'$ for every $1\leq i,j\leq n, i\neq j$. \\

Note that $G'$ is a $2d$-regular graph. 

\begin{figure}[H]
    \centering
   \includegraphics[width=\textwidth]{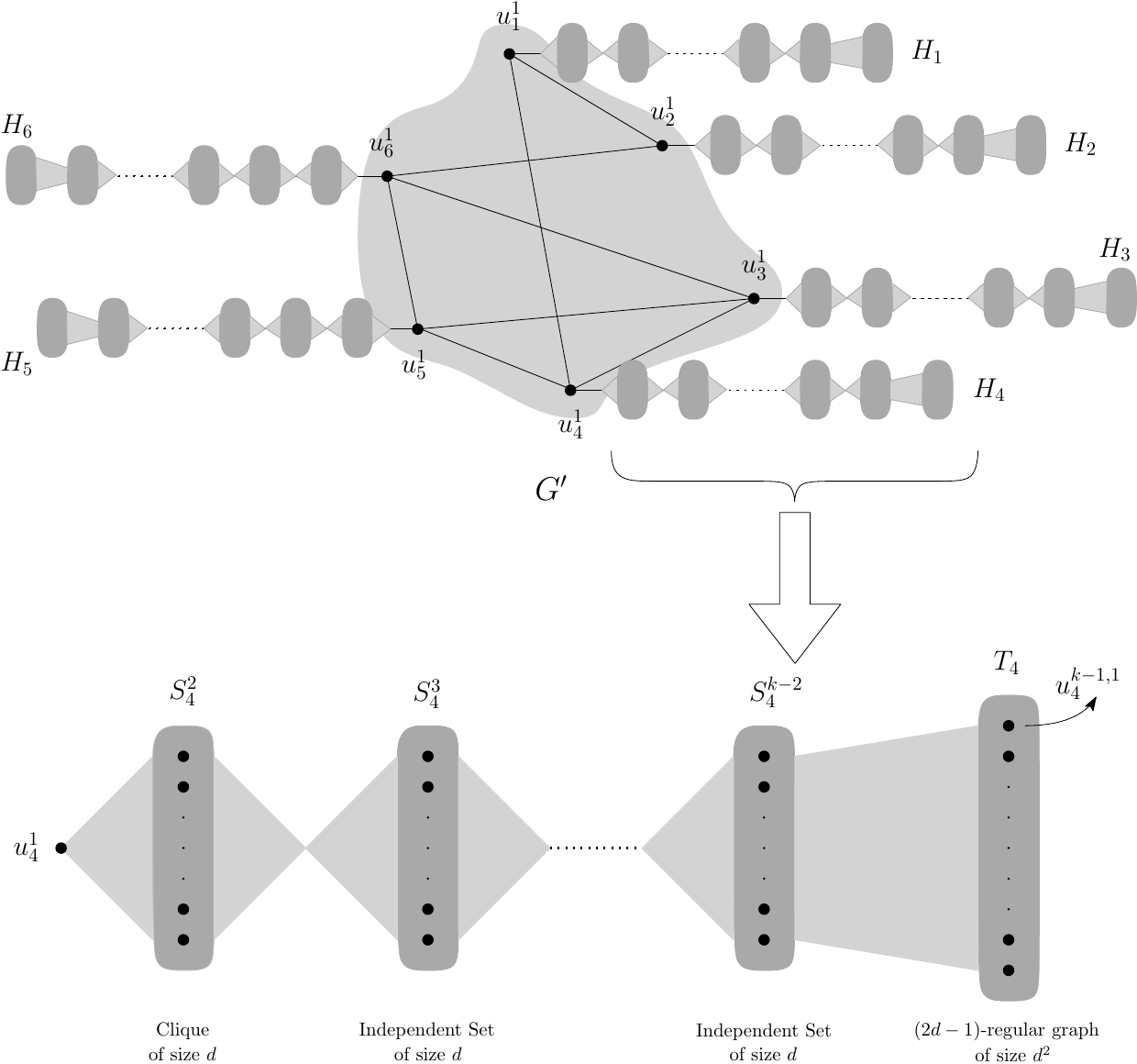}
    \caption{An illustration of the construction of the graph $G'$ used in the proof of Theorem~\ref{thm:multipacking_regular_NPc}, demonstrating the hardness of the \textsc{Multipacking} problem for regular graphs. The induced subgraph $G'[u_1^1,u_2^1,\dots,u_6^1]$ is isomorphic to $G^c$.}
    \label{fig:reduction_regular}
\end{figure}

\vspace{0.3cm}
\noindent
\textbf{Claim \ref{thm:multipacking_regular_NPc}.1. } $G$ has a total dominating set of size at most $k$ iff $G'$ has a multipacking of size at least $k$, for $k\geq 2$.

\begin{claimproof} \textsf{(If part)} 
Suppose $G$ has a total dominating set $D = \{ v_1, v_2, \dots$ $, v_k \}$. 
We want to show that $M =\{ u_1^{k-1,1}, u_2^{k-1,1}, \dots, u_k^{k-1,1} \}$ is a multipacking in $G'$.  If $r < k-1$, then 
$| N_{r}[v] \cap M |\leq 1 \leq r$, $ \forall v \in V(G')$, since $d(u_i^{k-1,1}, u_j^{k-1,1})\geq 2k-3$, $\forall i\neq j$.
If $r > k-1$, then 
$| N_{r}[v] \cap M | \leq | M | = k \leq r$, $\forall v \in V(G')$. Suppose $r = k-1$. If $M$ is not a multipacking,  there exists some $v \in V(G')$ such that 
$| N_{k-1}[v] \cap M | = k$. Therefore, 
$M \subseteq N_{k-1}[v]$. Suppose $v\in V(H_p)\setminus\{u_p^1\}$ for some $p$. In this case, $d(v,u_i^{k-1,1})\geq k$, $\forall i\neq p$. Therefore, $v\notin V(H_p)\setminus\{u_p^1\}$ for any $p$.  This implies that $v \in \{u_i^1:1\leq i\leq n\}$. Let $v = u_t^1$ for some $t$. If there exists an $i\in\{1,2,\dots,k\}$ such that neither $u_t^1=u_i^1$ nor $u_t^1$ is adjacent to $u_i^1$, then $d(u_t^1,u_i^{k-1,1})\geq k$. Therefore, 
$M \subseteq N_{k-1}[u_t^1] $ implies that either  $u_t^1$ is adjacent to each vertex of the set  $\{u_1^1, u_2^1, \dots, u_k^1\}$ when $u_t^1\notin \{u_1^1, u_2^1, \dots, u_k^1\}$ or $u_t^1$ is adjacent to each vertex of the set  $\{u_1^1, u_2^1, \dots, u_k^1\}\setminus \{u_t^1\}$ when $u_t^1\in \{u_1^1, u_2^1, \dots, u_k^1\}$. Therefore, in the graph $G$, either $v_t$ is not adjacent to any vertex of $D$ when $v_t\notin D$ or  $v_t$ is not adjacent to any vertex of $D\setminus\{v_t\}$ when $v_t\in D$. In both cases, $D$ is not a total dominating set of $G$. This is a contradiction.    Hence, $M$ is a multipacking of size $k$ in $G'$.

 \textsf{(Only-if part)}  Suppose $G'$ has a multipacking $M$ of size $k$. Let 
$D = \{ v_i : V(H_i)\cap M\neq \emptyset, 1\leq i\leq n \}.$
Then 
$|D| \leq |M| = k.$
Without loss of generality, let 
$D = \{ v_1, v_2, \dots, v_{k'} \} $ where $ k' \leq k$.
Now we want to show that $D$ is a total dominating set of $G$. Suppose $D$ is not a total dominating set.  Then there exists $t$ such that either 
$v_t\notin D$ and no vertex in $D$ is adjacent to $v_t$ or $v_t\in D$ and no vertex in $D\setminus\{v_t\}$ is adjacent to $v_t$. Therefore, in $G'$, either  $u_t^1$ is adjacent to each vertex of the set  $\{u_1^1, u_2^1, \dots, u_k^1\}$ when $u_t^1\notin \{u_1^1, u_2^1, \dots, u_k^1\}$ or $u_t^1$ is adjacent to each vertex of the set  $\{u_1^1, u_2^1, \dots, u_k^1\}\setminus \{u_t^1\}$ when $u_t^1\in \{u_1^1, u_2^1, \dots, u_k^1\}$.  Then $M \subseteq N_{k-1}[u_t^1]$. This implies that 
$| N_{k-1}[u_t^1] \cap M | = k$.
This is a contradiction. Therefore, $D$ is a total dominating set of size at most $k$ in $G$.   
\end{claimproof}

\noindent Hence, the \textsc{Multipacking} problem  is \textsc{NP-complete} for regular graphs. 
\end{proof}

\subsection{Convex Intersection graphs}

Here, we discuss the hardness of the \textsc{Multipacking} problem for a geometric intersection graph class: CONV graphs.  A graph $G$ is a \textit{CONV graph} if and only if there exists  a family of convex sets on a plane such that the graph has a vertex for each convex set and an edge for each intersecting pair of convex sets. We show that the \textsc{Multipacking} problem is \textsc{NP-complete} for CONV graphs. To prove this we will need the following result.

\begin{theorem}[\cite{kratochvil1998intersection}]\label{thm:coplanar_CONV}
    Compliment of a planar graph (co-planar graph) is a CONV graph. 
\end{theorem}

    

\multipackingCONVNPc*

\begin{proof} It is known that the \textsc{Total Dominating Set} problem is \textsc{NP-complete} for planar graphs~\cite{garey1979computers}. We reduce the \textsc{Total Dominating Set} problem of planar graph to the \textsc{Multipacking} problem  of CONV graph to show that the \textsc{Multipacking} problem  is \textsc{NP-complete} for CONV graphs.

Let $(G,k)$ be an instance of the \textsc{Total Dominating Set} problem, where $G$ is a planar graph  with the vertex set $V=\{v_1,v_2,\dots,v_n\}$. We construct a graph $G'$ in the following way (Fig. \ref{fig:reduction_CONV} gives an illustration). 

\noindent(i) For every vertex $v_i\in V(G)$, we introduce a path $u_i = u_i^1 u_i^2 \dots u_i^{k-1}$ of length $k-2$ in $G'$. 

\noindent(ii) If $v_i$ and $v_j$ are different and not adjacent in $G$, we make $u^1_i$ and $u^1_j$ adjacent in $G'$ for every $1\leq i,j\leq n, i\neq j$.

\begin{figure}[ht]
    \centering
   \includegraphics[width=\textwidth]{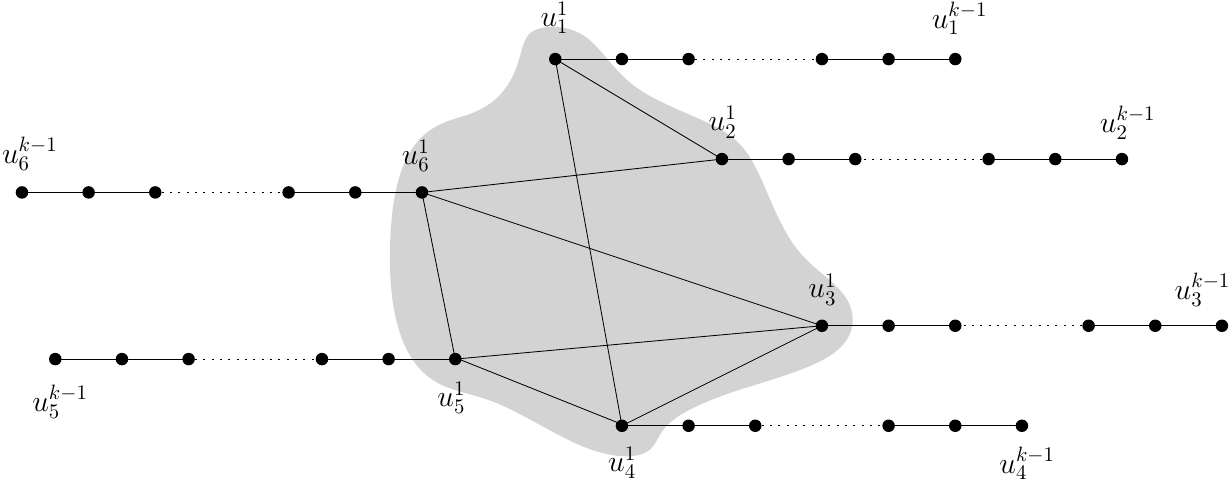}
    \caption{An illustration of the construction of the graph $G'$ used in the proof of Theorem~\ref{thm:multipacking_CONV_NPc}, demonstrating the hardness of the \textsc{Multipacking} problem for CONV graphs. The induced subgraph $G'[u_1^1,u_2^1,\dots,u_6^1]$ (the colored region) is isomorphic to $G^c$.}
    \label{fig:reduction_CONV}
\end{figure}

\vspace{0.3cm}
\noindent
\textbf{Claim \ref{thm:multipacking_CONV_NPc}.1. } $G'$ is a CONV graph.

\begin{claimproof} Note that the induced subgraph $G'[u_1^1,u_2^1,\dots,u_n^1]$ is isomorphic to $G^c$. Since $G$ is a planar graph, $G'[u_1^1,u_2^1,\dots,u_n^1]$ is a CONV graph by Theorem \ref{thm:coplanar_CONV}. Hence, $G'$ is a CONV graph.    
\end{claimproof}

\vspace{0.3cm}
\noindent
\textbf{Claim \ref{thm:multipacking_CONV_NPc}.2. } $G$ has a total dominating set of size at most $k$ iff $G'$ has a multipacking of size at least $k$, for $k\geq 2$.

\begin{claimproof} \textsf{(If part)} 
Suppose $G$ has a total dominating set $D = \{ v_1, v_2, \dots$ $, v_k \}$. 
We want to show that $M =\{ u_1^{k-1}, u_2^{k-1}, \dots, u_k^{k-1} \}$ is a multipacking in $G'$. If $r < k-1$, then 
$| N_{r}[v] \cap M |\leq 1 \leq r$, $ \forall v \in V(G')$, since $d(u_i^{k-1}, u_j^{k-1})\geq 2k-3$, $\forall i\neq j$.
If $r > k-1$, then 
$| N_{r}[v] \cap M | \leq | M | = k \leq r$, $\forall v \in V(G')$. Suppose $r = k-1$. If $M$ is not a multipacking,  there exists some $v \in V(G')$ such that 
$| N_{k-1}[v] \cap M | = k$. Therefore, 
$M \subseteq N_{k-1}[v].$ Suppose $v=u_p^q$ for some $p$ and $q$ where $1\leq p\leq n, 2\leq q\leq k-1$. In this case, $d(u_p^q,u_i^{k-1})\geq k$, $\forall i\neq p$. Therefore, $v\notin \{u_i^j:1\leq i\leq n, 2\leq j\leq k-1\}$.  This implies that $v \in \{u_i^1:1\leq i\leq n\}$.  Let $v = u_t^1$ for some $t$. If there exists an $i\in\{1,2,\dots,k\}$ such that neither $u_t^1=u_i^1$ nor $u_t^1$ is adjacent to $u_i^1$, then $d(u_t^1,u_i^{k-1})\geq k$. Therefore, 
$M \subseteq N_{k-1}[u_t^1] $ implies that either  $u_t^1$ is adjacent to each vertex of the set  $\{u_1^1, u_2^1, \dots, u_k^1\}$ when $u_t^1\notin \{u_1^1, u_2^1, \dots, u_k^1\}$ or $u_t^1$ is adjacent to each vertex of the set  $\{u_1^1, u_2^1, \dots, u_k^1\}\setminus \{u_t^1\}$ when $u_t^1\in \{u_1^1, u_2^1, \dots, u_k^1\}$. Therefore, in the graph $G$, either $v_t$ is not adjacent to any vertex of $D$ when $v_t\notin D$ or  $v_t$ is not adjacent to any vertex of $D\setminus\{v_t\}$ when $v_t\in D$. In both cases, $D$ is not a total dominating set of $G$. This is a contradiction.    Hence, $M$ is a multipacking of size $k$ in $G'$.

 \textsf{(Only-if part)}  Suppose $G'$ has a multipacking $M$ of size $k$. Let 
$D = \{ v_i : u_i^j \in M, 1\leq i\leq n, 1\leq j\leq k-1 \}.$
Then 
$|D| \leq |M| = k.$
Without loss of generality, let 
$D = \{ v_1, v_2, \dots, v_{k'} \} $ where $ k' \leq k$.
Now we want to show that $D$ is a total dominating set of $G$. Suppose $D$ is not a total dominating set.  Then there exists $t$ such that either 
$v_t\notin D$ and no vertex in $D$ is adjacent to $v_t$ or $v_t\in D$ and no vertex in $D\setminus\{v_t\}$ is adjacent to $v_t$. Therefore, in $G'$, either  $u_t^1$ is adjacent to each vertex of the set  $\{u_1^1, u_2^1, \dots, u_k^1\}$ when $u_t^1\notin \{u_1^1, u_2^1, \dots, u_k^1\}$ or $u_t^1$ is adjacent to each vertex of the set  $\{u_1^1, u_2^1, \dots, u_k^1\}\setminus \{u_t^1\}$ when $u_t^1\in \{u_1^1, u_2^1, \dots, u_k^1\}$.  Then $M \subseteq N_{k-1}[u_t^1]$. This implies that 
$| N_{k-1}[u_t^1] \cap M | = k$.
This is a contradiction. Therefore, $D$ is a total dominating set of size at most $k$ in $G$.   
\end{claimproof}

\noindent Hence, the \textsc{Multipacking} problem  is \textsc{NP-complete} for CONV graphs. 
\end{proof}

\subsubsection{SEG Graphs or Segment Graphs}

 A graph $G$ is a \textit{SEG graph} (or \textit{Segment Graph}) if and only if there exists  a family of straight-line segments in a plane such that the graph has a vertex for each straight-line segment and an edge for each intersecting pair of straight-line segments. It is known that planar graphs are SEG graphs~\cite{chalopin2009every}. 

In 1998, Kratochv{\'\i}l and Kub{\v{e}}na~\cite{kratochvil1998intersection} posed the question of whether the complement of any planar graph (i.e., a co-planar graph) is also a SEG graph. As far as we know, this is still an open question. A positive answer to this question would imply that the \textsc{Multipacking} problem is \textsc{NP-complete} for SEG graphs, since in that case the induced subgraph $G'[u_1^1, u_2^1, \dots, u_n^1]$ in the proof of Theorem~\ref{thm:multipacking_CONV_NPc} would be a SEG graph, which in turn would imply that $G'$ itself is a SEG graph.

\section{Conclusion}\label{sec:conclusion_NPcomplete}

In this work, we have established the computational complexity of the \textsc{Multipacking} problem by proving its \textsc{NP-completeness} and demonstrating its \textsc{W[2]-hardness} when parameterized by the solution size. Furthermore, we have extended these hardness results to several important graph classes, including chordal $\cap$ $\frac{1}{2}$-hyperbolic graphs, bipartite graphs, claw-free graphs, regular graphs, and CONV graphs.

Our results naturally lead to several open questions that merit further investigation:
\begin{enumerate}

    \item Does there exist an FPT algorithm for general graphs when parameterized by treewidth?
    \item Can we design a sub-exponential algorithm for this problem?
    \item What is the complexity status on planar graphs? Does an FPT algorithm for planar graphs exist when parameterized by solution size?
    \item While a $(2 + o(1))$-factor approximation algorithm is known~\cite{beaudou2019broadcast}, can we achieve a PTAS for this problem?
    \item What is the approximation hardness of the \textsc{Multipacking} problem?
\end{enumerate}

Addressing these questions would provide a more complete understanding of the computational landscape of the \textsc{Multipacking} problem and could lead to interesting algorithmic developments.

\chapter{Multipacking on a bounded hyperbolic graph-class: chordal graphs}\label{chapter:chordal}\hypertarget{chapter:introhref}{}

 \vspace{-0.9cm}
\minitoc

  
  In this chapter, we show that for any connected chordal graph $G$, $\gamma_{b}(G)\leq \big\lceil{\frac{3}{2} \MP(G)\big\rceil}$. We also show that $\gamma_b(G)-\MP(G)$ can be arbitrarily large for connected chordal graphs by constructing an infinite family of connected chordal graphs such that the ratio $\gamma_b(G)/\MP(G)=10/9$, with $\MP(G)$ arbitrarily large.   Moreover, we show that $\gamma_{b}(G)\leq \big\lfloor{\frac{3}{2} \MP(G)+2\delta\big\rfloor} $ holds for all $\delta$-hyperbolic graphs. In addition, we provide a polynomial-time algorithm to construct a multipacking of a $\delta$-hyperbolic graph $G$ of size at least $  \big\lceil{\frac{2\MP(G)-4\delta}{3} \big\rceil} $. 

  Further, we construct an infinite family of connected chordal graphs such that the ratio $\gamma_b(G)/\MP(G)=4/3$, with $\MP(G)$ arbitrarily large.  This improves the lower bound of the expression $\lim_{\MP(G)\to \infty}$ $\sup\{\gamma_{b}(G)/\MP(G)\}$ for connected chordal graphs to $4/3$, and it says, for connected chordal graphs, we cannot improve the bound $\gamma_{b}(G)\leq \big\lceil{\frac{3}{2} \MP(G)\big\rceil}$ to a bound in the form $\gamma_b(G)\leq c_1\cdot \MP(G)+c_2$, for any constant $c_1<4/3$ and $c_2$.

\section{Chapter overview}
In Section \ref{sec:Definitions and notation_chordal}, we recall some definitions and notations. In Section \ref{sec:An inequality linking Broadcast domination and Multipacking numbers}, we  establish a relation between multipacking and dominating broadcast on the chordal graphs. In the same section, we provide a $(\frac{3}{2}+o(1))$-factor approximation algorithm for finding multipacking on the same graph class. In Section \ref{sec:Unboundedness of the gap between Broadcast domination and Multipacking numbers of Chordal graphs}, we prove our main result which says that the difference $ \gamma_{b}(G) -  \MP(G) $ can be arbitrarily large for connected chordal graphs. In Section \ref{sec:chordal graphs}, we improve the lower bound of the expression $\lim_{\MP(G)\to \infty}$ $\sup\{\gamma_{b}(G)/\MP(G)\}$ for connected chordal graphs. In Section \ref{sec:A study of Broadcast domination and Multipacking numbers on Hyperbolic graphs},  we relate the broadcast domination and multipacking number of $\delta$-hyperbolic graphs and provide an approximation algorithm for the multipacking problem of the same. We conclude this chapter in Section \ref{sec:Conclusion_chordal}.

\section{Preliminaries}\label{sec:Definitions and notation_chordal}

 Let $G=(V,E)$ be a graph and $d_G(u,v)$ be the length of a shortest path joining two vertices $u$ and $v$ in  $G$, we simply write $d(u,v)$ when there is no confusion. Let $diam(G)=\max\{d(u,v):u,v\in V(G)\}$. Diameter (or diametral path) is a path of $G$  of the length  $diam(G)$. 
 $N_r[u]=\{v\in V:d(u,v)\leq r\}$ where $u\in V$. The \textit{eccentricity} $e(w)$  of a vertex $w$ is $\min \{r:N_r[w]=V\}$. The \textit{radius} of the graph $G$ is $\min\{e(w):w\in V\}$, denoted by $\rad(G)$.  The \textit{center} $C(G)$ of the graph $G$  is the set of all vertices of minimum eccentricity, i.e., $C(G)=\{v\in V:e(v)=\rad(G)\}$. Each vertex in the set $C(G)$ is called a \textit{central vertex} of the graph $G$. 

A \textit{chordal graph} is an undirected simple graph in which all cycles of four or more vertices have a chord, which is an edge that is not part of the cycle but connects two vertices of the cycle.

 Let $d$ be the shortest-path metric of a graph $G$. The graph $G$ is called a \emph{$\delta$-hyperbolic graph} if for any four
vertices $u, v, w, x \in V(G)$, the two larger of the three sums $d(u, v) + d(w, x)$, $d(u, w) + d(v, x)$, $d(u, x) +
d(v, w)$ differ by at most $2\delta$. A graph class $\mathcal{G}$ is said to be hyperbolic if there exists a constant $\delta$ such that every graph $G \in \mathcal{G}$ is $\delta$-hyperbolic. Trees are $0$-hyperbolic graphs~\cite{buneman1974note}, and chordal graphs are 1-hyperbolic~\cite{brinkmann2001hyperbolicity}.

\section{An inequality linking Broadcast domination and Multipacking numbers of Chordal Graphs}\label{sec:An inequality linking Broadcast domination and Multipacking numbers}

In this section, we use results from the literature to show that the general bound connecting multipacking number and broadcast domination number can be improved for chordal graphs.

\begin{theorem}[\cite{hartnell2014difference}]\label{d+1/3leqmpG}
If $G$ is a connected graph of order at least 2  having  diameter $ d $ and  multipacking number  $ \MP(G) $, where $P=v_0,\dots,v_d$ is a diametral  path of $G$, then the set $M=\{v_i:i\equiv 0 \text{ } (mod \text{ } 3), i=0,1,\dots,d\}$ is a multipacking of $G$ of size $\big\lceil{\frac{d+1}{3}\big\rceil}$ and  $\big\lceil{\frac{d+1}{3}\big\rceil}\leq \MP(G)$.
\end{theorem}

\begin{theorem}[\cite{erwin2001cost,teshima2012broadcasts}]  \label{mpGleqgammabG}  If $G$ is a connected graph of order at least 2  having radius $ r $, diameter $ d $, multipacking number  $\MP(G) $, broadcast domination number $ \gamma_{b}(G) $ and domination number $\gamma(G)$, then $ \MP(G)\leq \gamma_{b}(G) \leq min\{\gamma (G),r\}$.
\end{theorem} 

\begin{theorem}[\cite{laskar1983powers}] \label{2rleqd+2} If $ G $ is a connected chordal graph with radius $ r $ and diameter $ d $, then $ 2r\leq d+2 $.
\end{theorem} 

Using these results we prove the following proposition:

\multipackingbroadcastrelation*

\begin{proof}
From Theorem \ref{d+1/3leqmpG}, $\big\lceil{\frac{d+1}{3}\big\rceil}\leq \MP(G)$  which implies that $d \leq 3\MP(G)-1$. Moreover, from  Theorem \ref{mpGleqgammabG} and Theorem \ref{2rleqd+2},  $\gamma_{b}(G)\leq r \leq \big\lfloor{\frac{d+2}{2}\big\rfloor} \leq \big\lfloor{\frac{(3\MP(G)-1)+2}{2}\big\rfloor}$ $=\big\lfloor{\frac{3}{2} \MP(G)+\frac{1}{2}\big\rfloor}  $. Therefore, $\gamma_{b}(G)\leq \big\lfloor{\frac{3}{2} \MP(G)+\frac{1}{2}\big\rfloor} = \big\lceil{\frac{3}{2} \MP(G)\big\rceil}$. 
\end{proof}

The proof of Proposition~\ref{prop:gammabGleq3/2mpG} has the following algorithmic application.

\approximationalgorithm*

\begin{proof} If $P=v_0,\dots,v_d$ is a diametrical path of $G$, then the set $M=\{v_i:i\equiv 0 \text{ } (mod \text{ } 3), i=0,1,\dots,d\}$ is a multipacking of $G$ of size $\big\lceil{\frac{d+1}{3}\big\rceil}$ by Theorem \ref{d+1/3leqmpG}. We can construct $M$ in polynomial-time since we can find a diametral path of a graph $G$ in polynomial-time. Moreover, from  Theorem \ref{d+1/3leqmpG}, Theorem \ref{mpGleqgammabG} and Theorem \ref{2rleqd+2},  $\big\lceil{\frac{2\MP(G)-1}{3}\big\rceil}\leq\big\lceil{\frac{2r-1}{3}\big\rceil}\leq\big\lceil{\frac{d+1}{3}\big\rceil}\leq \MP(G)$. 
\end{proof}


\begin{figure}[h]
    \centering
   \includegraphics[height=3.1cm]{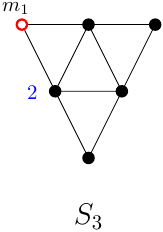}
    \caption{$S_3$ is a connected chordal graph   with $\gamma_b(S_3)=2$ and $\MP(S_3)=1$}
    \label{fig:S3}
\end{figure}

\begin{figure}[h]
    \centering
   \includegraphics[height=3.3cm]{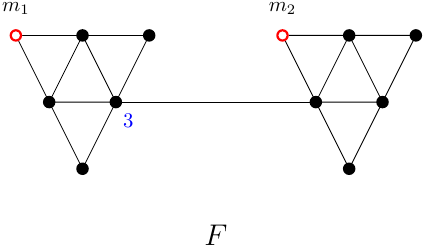}
    \caption{$F$  is a connected chordal graph   with $\gamma_b(F)=3$ and $\MP(F)=2$}
    \label{fig:F}
\end{figure}

\begin{figure}[h]
    \centering
   \includegraphics[height=3.2cm]{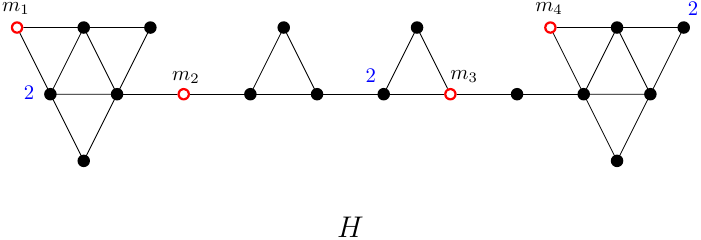}
    \caption{$H$ is a connected chordal graph  with $\gamma_b(H)=6$ and $\MP(H)=4$}
    \label{fig:H}
\end{figure}

Here we have some examples of graphs that achieve the equality of the bound in Proposition \ref{prop:gammabGleq3/2mpG}.

\begin{ex} The connected  chordal graph  $S_{3}$ (Figure~\ref{fig:S3}) has $  \MP(S_{3}) =1 $ and $\gamma_{b}(S_{3}) =2$. So, here  $ \gamma_{b}(S_{3}) = \big\lceil{\frac{3}{2} \MP(S_3)\big\rceil} $.
\end{ex}

\begin{ex} The  connected   chordal graph  $F$ (Figure~\ref{fig:F})  has $  \MP(F) =2 $ and $\gamma_{b}(F) =3$. So, here  $ \gamma_{b}(F) = \big\lceil{\frac{3}{2} \MP(F)\big\rceil} $.
\end{ex}

\begin{ex}
The  connected  chordal graph   $H$ (Figure~\ref{fig:H})  has $  \MP(H) =4 $ and $\gamma_{b}(H) =6$. So, here  $ \gamma_{b}(H) = \big\lceil{\frac{3}{2} \MP(H)\big\rceil} $.
\end{ex}

We could not find an example of  connected  chordal graph with $\MP(G)=3$ and $\gamma_{b}(G) =\big\lceil{\frac{3}{2} \MP(G)\big\rceil}=5$. 

\section{Unboundedness of the gap between Broadcast domination and Multipacking numbers of Chordal graphs}\label{sec:Unboundedness of the gap between Broadcast domination and Multipacking numbers of Chordal graphs}

In this section, our goal is to show that  the difference between broadcast domination number and multipacking number of connected chordal graphs  can be arbitrarily large. We prove this using the following theorem that we prove later.

\multipackingbroadcastgapechordal*

Theorem \ref{thm:9k10k} yields the following.

\gammabGdiffmpGchordal*

Moreover, Proposition \ref{prop:gammabGleq3/2mpG} and Theorem \ref{thm:9k10k} yield the following.

\gammabGbympGchordal*

\subsection{Proof of Theorem \ref{thm:9k10k}}

\begin{figure}[h]
    \centering
   \includegraphics[height=3.3cm]{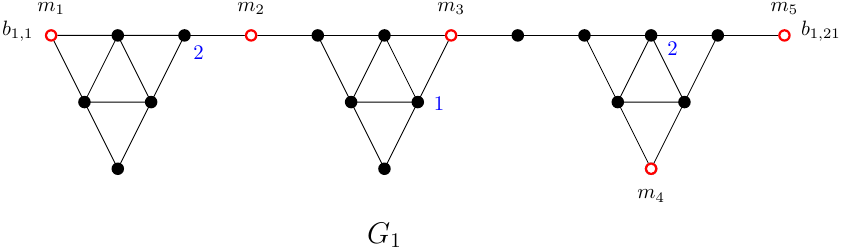}
    \caption{$G_1$ is a connected chordal graph  with $\gamma_b(G_1)=5$ and $\MP(G_1)=5$. $M_1=\{m_i:1\leq i \leq  5\}$ is a multipacking of size $5$.}
    \label{fig:G1m}
\end{figure}

Consider the graph $G_{1}$ as in Figure \ref{fig:G1m}.  
Let $B_1$ and $B_2$ be two isomorphic copies of $G_1$.  Join $b_{1,21}$ of  $B_1$ and $b_{2,1}$ of $B_2$ by an edge (Figure~\ref{fig:G2m} and \ref{fig:Names}). We denote this new graph by $G_2$ (Figure~\ref{fig:G2m}). In this way, we form $G_k$ by joining $k$ isomorphic copies of $G_1$ : $B_1,B_2,\cdots,B_k$ (Figure~\ref{fig:Names}). Here $B_i$ is joined with $B_{i+1}$ by joining $b_{i,21}$ and $b_{i+1,1}$.  We say that $B_i$ is the $i$-th block of $G_k$.  $B_i$ is an induced subgraph of $G_k$ as given by $B_i= G_k[\{b_{i,j}:1\leq j \leq 21\}]$. Similarly, for $1\leq i\leq2k-1$, we define $B_i\cup B_{i+1}$, induced subgraph of $G_{2k}$, as $B_i\cup B_{i+1}=G_{2k}[\{b_{i,j},b_{i+1,j}:1\leq j \leq 21\}]$.  We prove Theorem \ref{thm:9k10k}  by establishing  that $\gamma_b(G_{2k})=10k$ and $\MP(G_{2k})=9k$. 

\begin{figure}[h]
    \centering
   \includegraphics[height=7.1cm]{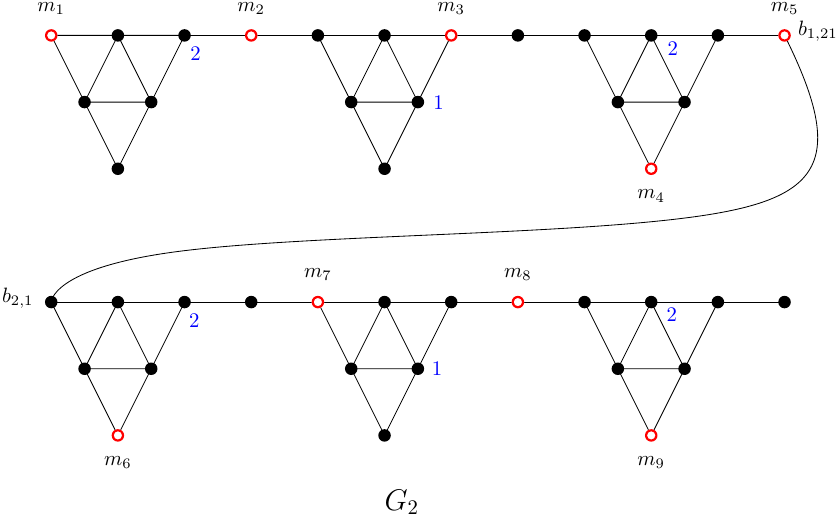}
    \caption{Graph $G_2$ with $\gamma_b(G_2)=10$ and $\MP(G_2)=9$. $M=\{m_i:1\leq i \leq  9\}$ is a multipacking of size $9$.}
    \label{fig:G2m}
\end{figure}

\begin{figure}[ht]
    \centering
   \includegraphics[height=11.3cm]{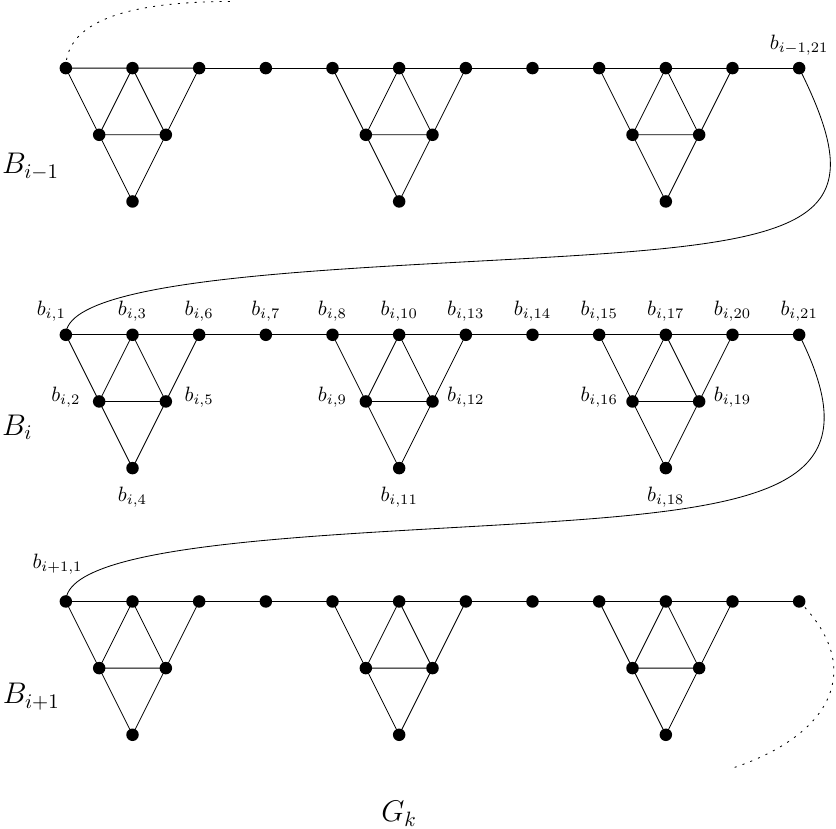}
    \caption{Graph $G_k$}
    \label{fig:Names}
\end{figure}

Our proof of Theorem \ref{thm:9k10k} is accomplished through a set of lemmas which are stated and proved below. We begin by observing a basic fact about multipacking in a graph. We formally state it in Lemma \ref{duv>=3} for ease of future reference.

\begin{lemma} \label{duv>=3} Suppose $M$ is a multipacking in a graph $G$. If $u,v \in M$ and $u \neq v$, then $d(u,v)\geq 3$.
\end{lemma}
\begin{proof}  If $d(u,v)=1$, then $u,v\in N_1[v]\cap M$, then $M$ cannot be a multipacking. So, $d(u,v)\neq 1$. If $d(u,v)=2$, then there exists a common neighbour $w$ of $u$ and $v$. So, $u,v\in N_1[w]\cap M$, then $M$ cannot be a multipacking. So, $d(u,v)\neq 2$. Therefore,  $d(u,v)>2$. 
\end{proof}

\begin{lemma}\label{mpG2k>=9k} $\MP(G_{2k})\geq 9k$, for each positive integer $k$.
\end{lemma}

\begin{proof}
Consider the set 
$M_{2k}=\{b_{2i-1,1},b_{2i-1,7},b_{2i-1,13},b_{2i-1,18},b_{2i-1,21},b_{2i,4}, $ $ b_{2i,8},b_{2i,14}, b_{2i,18} :1\leq i \leq k\}$ (Figure~\ref{fig:Names}) of size $9k$. 
 We want to show that $M_{2k}$ is a multipacking of $G_{2k}$. So, we have to prove that, $|N_r[v]\cap M_{2k}|\leq r$ for each vertex $ v \in V(G_{2k}) $ and for every integer $ r \geq 1 $. We prove this statement using induction on $r$.	It can be checked that $|N_r[v]\cap M_{2k}|\leq r$ for each vertex $ v \in V(G_{2k}) $ and for each  $ r \in \{1,2,3,4\} $. Now assume that the statement is true for $r=s$, we want to prove that, it is true for $r=s+4$. Observe that, $|(N_{s+4}[v]\setminus N_{s}[v])\cap M_{2k}|\leq 4$ for every vertex $ v\in V(G_{2k})$. Therefore,  $|N_{s+4}[v]\cap M_{2k}|\leq |N_{s}[v]\cap M_{2k}|+4\leq s+4$. So, the statement is true. Therefore, $M_{2k}$ is a multipacking of $G_{2k}$. So, $\MP(G_{2k})\geq |M_{2k}|= 9k$. 
\end{proof}

\begin{lemma} \label{mpG1=5}  $\MP(G_{1})=5$.
\end{lemma}
\begin{proof}  $V(G_1)=N_3[b_{1,7}]\cup N_2[b_{1,17}]$. Suppose $M$  is a multipacking on $G_1$ such that $|M|=\MP(G_{1})$. So, $|M\cap N_3[b_{1,7}]|\leq 3$ and $|M\cap N_2[b_{1,17}]|\leq 2$. Therefore, $|M\cap (N_3[b_{1,7}]\cup N_2[b_{1,17}])|\leq 5$. So, $|M\cap V(G)|\leq 5$, that implies $|M|\leq 5$. Let $M_1=\{b_{1,1},b_{1,7},b_{1,13},b_{1,18},b_{1,21}\}$. Since $|N_r[v]\cap M|\leq r$ for each vertex $ v \in V(G_{1}) $ and for every integer $ r \geq 1 $, so $M_1$ is a multipacking of size $5$. Then  $ 5=|M_1|\leq |M|$. So, $|M|=5$. Therefore, $\MP(G_{1})=5$. 
\end{proof}

 So, now we have $\MP(G_1)=5$. Using this fact we prove that $\MP(G_2)=9$.

\begin{lemma}\label{mpG2=9}  $\MP(G_{2})=9$.
\end{lemma}
\begin{proof} As mentioned before,  $B_i= G_k[\{b_{i,j}:1\leq j \leq 21\}]$, $1\leq i \leq 2$. So, $B_1$ and $B_2$ are two blocks in $G_2$ which are isomorphic to $G_1$. Let $M$ be a multipacking of $G_2$ with size $\MP(G_{2})$. So, $|M|\geq 9$ by Lemma \ref{mpG2k>=9k}. Since $M$ is a multipacking of $G_2$, so $M\cap V(B_1)$  and $M\cap V(B_2)$  are multipackings of $B_1$ and $B_2$,  respectively. Let $M\cap V(B_1)=M_1$  and $M\cap V(B_2)=M_2$. Since $B_1\cong G_1$ and $B_2\cong G_1$, so $\MP(B_1)=5$ and $\MP(B_2)=5$ by Lemma \ref{mpG1=5}. This implies $|M_1|\leq 5$ and $|M_2|\leq 5$. Since $V(B_1)\cup V(B_2)=V(G_2)$ and $V(B_1)\cap V(B_2)=\phi$, so $M_1\cap M_2=\phi$ and $|M|=|M_1|+|M_2|$.  Therefore, $9\leq |M|=|M_1|+|M_2|\leq  10$. So, $9\leq |M|\leq 10$. 

We establish this lemma by using contradiction on $|M|$. In the first step, we prove that if  $|M_1|= 5$, then the particular vertex $b_{1,21}\in M_1$. Using this, we can show that $|M_2|\leq 4$. In this way we show that $|M|\leq 9$.

For the purpose of contradiction, we assume that $|M|=10$. So, $|M_1|+ |M_2|=10$, and also   $|M_1|\leq 5$, $|M_2|\leq 5$. Therefore,   $|M_1|=|M_2|= 5$.

\vspace{0.2cm}
\noindent
\textbf{Claim \ref{mpG2=9}.1. } If $|M_1|= 5$, then $b_{1,21}\in M_1$.
\begin{claimproof}
Suppose $b_{1,21}\notin M$. Let $S=\{b_{1,7},b_{1,14}\}$, $S_1=\{b_{1,r}:1\leq r\leq 6\}$, $S_2=\{b_{1,r}:8\leq r\leq 13\}$, $S_3=\{b_{1,r}:15\leq r\leq 20\}$. If $u,v\in S_t$, then $d(u,v)\leq 2$, this holds  for each $ t\in\{1,2,3\}$. So, by Lemma \ref{duv>=3}, $u,v$ together cannot be in a multipacking. Therefore $|S_t\cap M_1|\leq 1$ for $t=1,2,3$  and  $|S\cap M_1|\leq |S|= 2$. Now, $5=|M_1|=|M_1\cap [V(G_1)\setminus \{b_{1,21}\}|=|M_1\cap(S\cup S_1\cup S_2\cup S_3)|=|(M_1\cap S)\cup(M_1\cap S_1)\cup(M_1\cap S_2)\cup(M_1\cap S_3)|\leq |M_1\cap S|+|M_1\cap S_1|+|M_1\cap S_2|+|M_1\cap S_3|\leq 2+1+1+1=5$. Therefore,  $|S_t\cap M_1|= 1$ for $t=1,2,3$  and  $|S\cap M_1|= 2$, so $b_{1,7},b_{1,14}\in M_1$. Since $|S_2\cap M_1|= 1$,  there exists $ w \in S_2\cap M_1$. Then $N_2[b_{1,10}]$ contains three vertices $b_{1,7},b_{1,14},w$ of $M_1$, which is not possible. So, this is a contradiction. Therefore, $b_{1,21}\in M_1$.  
\end{claimproof}

\vspace{0.1cm}
\noindent
\textbf{Claim \ref{mpG2=9}.2. } If $|M_1|=5$, then $|M_2|\leq 4$.
\begin{claimproof}

Let $S'=\{b_{2,14},b_{2,21}\}$, $S_4=\{b_{2,r}:1\leq r\leq 6\}$, $S_5=\{b_{2,r}:8\leq r\leq 13\}$, $S_6=\{b_{2,r}:15\leq r\leq 20\}$. By  Lemma \ref{duv>=3},   $|S_t\cap M_2|\leq 1$ for $t=4,5,6$  and also $|S'\cap M_2|\leq |S'|= 2$.

Observe that, if $S_4\cap M_2\neq \phi$, then  $b_{2,7} \notin M_2$ (i.e.  if $b_{2,7} \in M_2$, then $S_4\cap M_2 = \phi$). [Suppose not, then $S_4\cap M_2 \neq \phi$ and  $b_{2,7} \in M_2$, so,  there exists $ u\in S_4\cap M_2$. Then $N_2[b_{2,3}]$ contains three vertices $b_{1,21},b_{2,7},u$ of $M$, which is not possible. This is a contradiction].

Suppose $S_4\cap M_2 \neq \phi$, then  $b_{2,7} \notin M_2$. Now, $5=|M_2|=|M_2\cap [V(B_2)\setminus \{b_{2,7}\}]|=|M_2\cap(S'\cup S_4\cup S_5\cup S_6)|=|(M_2\cap S')\cup(M_2\cap S_4)\cup(M_2\cap S_5)\cup(M_2\cap S_6)|\leq |M_2\cap S'|+|M_2\cap S_4|+|M_2\cap S_5|+|M_2\cap S_6|\leq 2+1+1+1=5$. Therefore  $|S_t\cap M_2|= 1$ for $t=4,5,6$  and  $|S'\cap M_2|= 2$. Since $|M_2\cap S_6|= 1$, so there exists $ u_1\in M_2\cap S_6$. Then $N_2[b_{2,17}]$ contains three vertices $b_{2,14},b_{2,21},u_1$ of $M_2$, which is not possible. So, this is a contradiction. 

Suppose $S_4\cap M_2 = \phi$, then either  $b_{2,7} \in M_2$ or $b_{2,7} \notin M_2$. First consider  $b_{2,7} \notin M_2$, then $5=|M_2|=|M_2\cap(S'\cup S_5\cup S_6)|=|(M_2\cap S')\cup(M_2\cap S_5)\cup(M_2\cap S_6)|\leq |M_2\cap S'|+|M_2\cap S_5|+|M_2\cap S_6|\leq 2+1+1=4$. So, this is a contradiction. And if $b_{2,7} \in M_2$, then  $5=|M_2|=|M_2\cap(S'\cup S_5\cup S_6\cup \{b_{2,7}\})|=|(M_2\cap S')\cup(M_2\cap S_5)\cup(M_2\cap S_6)\cup(M_2\cap \{b_{2,7}\})|\leq |M_2\cap S'|+|M_2\cap S_5|+|M_2\cap S_6|+|M_2\cap \{b_{2,7}\}|\leq 2+1+1+1=5$. Therefore  $|S_t\cap M_2|= 1$ for $t=5,6$  and  $|S'\cap M_2|= 2$. Since $|M_2\cap S_6|= 1$, so there exists $ u_2\in M_2\cap S_6$. Then $N_2[b_{2,17}]$ contains three vertices $b_{2,14},b_{2,21},u_2$ of $M_2$, which is not possible. So, this is a contradiction. So, $|M_1|=5\implies |M_2|\leq 4$.  
\end{claimproof}

Recall that for contradiction,  we assume $|M|=10$, which implies $|M_2|=5$. In the proof of the above claim, we established  $|M_2|\leq 4$, which in turn contradicts our assumption. So, $|M|\neq 10$.   Therefore, $|M|=9$. 
\end{proof}

Notice that  graph $G_{2k}$ has $k$ copies of $G_2$. Moreover, we have $\MP(G_2)=9$. Using the Pigeonhole principle, we show that $\MP(G_{2k})=9k$.

\begin{lemma}\label{mpG2k=9k} $\MP(G_{2k})=9k$, for each positive integer $k$.
\end{lemma}

\begin{proof} For $k=1$ it is true by Lemma \ref{mpG2=9}. Moreover, we know $\MP(G_{2k})\geq 9k$  by Lemma \ref{mpG2k>=9k}.  Suppose $k>1$ and assume  $\MP(G_{2k})>9k$. Let $\hat{M}$ be a multipacking of $G_{2k}$ such that $|\Hat{M}|>9k$. Let $\hat{B}_j$ be a subgraph of $G_{2k}$ defined as $\hat{B}_j=B_{2j-1}\cup B_{2j}$ where $1\leq j \leq k$. So, $V(G_{2k})=\bigcup_{j=1}^kV(\hat{B}_j)$ and $V(\hat{B}_p)\cap V(\hat{B}_q)=\phi$ for all $p\neq q$ and  $p,q\in \{1,2,3,\dots,k\}$. Since $|\Hat{M}|>9k$, so by the Pigeonhole principle there exists a  number  $j\in \{1,2,3,\dots,k\}$ such that $|\Hat{M}\cap \hat{B}_j|>9$. Since $\Hat{M}\cap \hat{B}_j$ is a multipacking of $\hat{B}_j$, so $\MP(\hat{B}_j)>9$.  But $\hat{B}_j \cong G_2$ and $\MP(G_2)=9$ by Lemma \ref{mpG2=9}, so $\MP(\hat{B}_j)=9$, which is a contradiction. Therefore, $\MP(G_{2k})=9k$.
\end{proof}

R. C. Brewster and  L. Duchesne \cite{brewster2013broadcast} introduced fractional multipacking in 2013 (also see \cite{teshima2014multipackings}).  Suppose $G$ is a graph with $V(G)=\{v_1,v_2,v_3,\dots,v_n\}$ and $w:V(G)\rightarrow [0,\infty)$  is a function. So, $w(v)$ is a weight on a vertex $v\in V(G)$. Let $w(S)=\sum_{u\in S}w(u)$ where $S\subseteq V(G)$.  We say   $w$ is a \textit{fractional multipacking} of $G$,  if $w( N_r[v])\leq r$ for each vertex $ v \in V(G) $ and for every integer $ r \geq 1 $. The \textit{fractional multipacking number} of $ G $ is the  value $\displaystyle \max_w w(V(G)) $ where $w$ is any fractional multipacking and it
	is denoted by $ \MP_f(G) $. A \textit{maximum fractional multipacking} is a fractional multipacking $w$  of a graph $ G $ such that	$ w(V(G))=\MP_f(G)$. If $w$ is a fractional multipacking, we define   a vector $y$ with the entries $y_j=w(v_j)$.  So,  $$\MP_f(G)=\max \{y.\mathbf{1} :  yA\leq c, y_{j}\geq 0\}.$$  So, this is a  linear program which is the dual of the linear program   $\min \{c.x :  Ax\geq \mathbf{1}, x_{i,k}\geq 0\}$. Let,  $$\gamma_{b,f}(G)=\min \{c.x : Ax\geq \mathbf{1}, x_{i,k}\geq 0\}.$$ Using the strong duality theorem for linear programming, we can say that  $$\MP(G)\leq \MP_f(G)= \gamma_{b,f}(G)\leq \gamma_{b}(G).$$

\begin{figure}[ht]
    \centering
   \includegraphics[height=11.3cm]{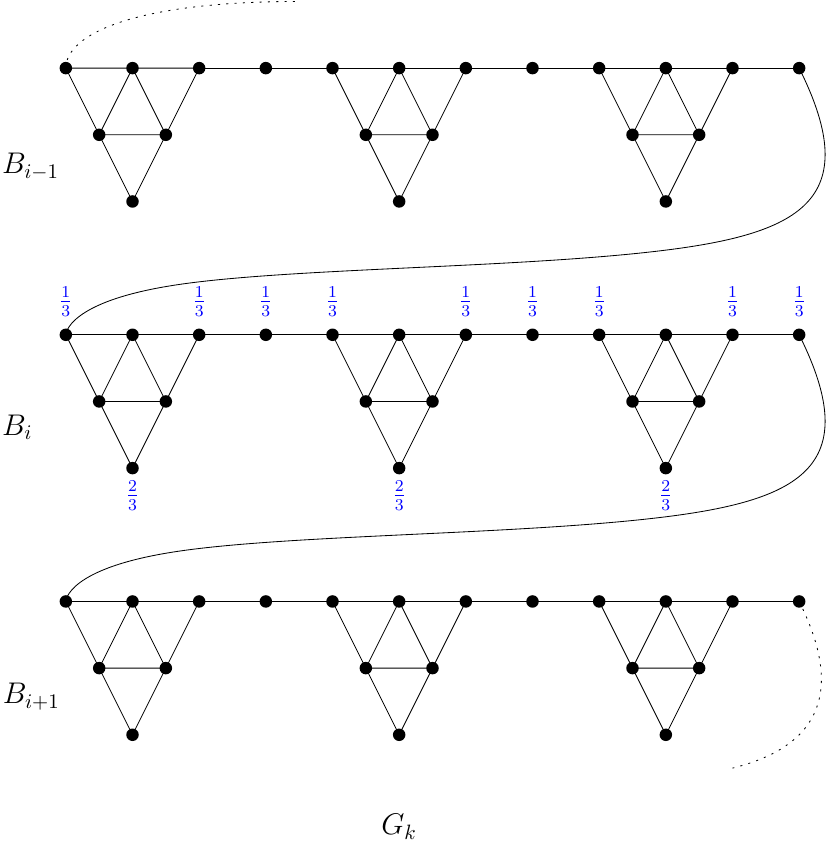}
    \caption{Graph $G_k$}
    \label{fig:Gkmpf}
\end{figure}

\begin{lemma} \label{mpfGk>=5k}
If $k$ is a positive integer, then $ \MP_f(G_{k})\geq 5k$.
\end{lemma}

\begin{proof}
We define a function $w: V(G_k)\rightarrow [0,\infty)$  where $w(b_{i,1})=w(b_{i,6})=w(b_{i,7})=w(b_{i,8})=w(b_{i,13})=w(b_{i,14})=w(b_{i,15})=w(b_{i,20})=w(b_{i,21})=\frac{1}{3}$ and $w(b_{i,4})=w(b_{i,11})=w(b_{i,18})=\frac{2}{3}$ for each $i\in \{1,2,3,\dots,k\}$ (Figure~\ref{fig:Gkmpf}).  So, $w(G_k)=5k$. 
 We want to show that $w$ is a fractional multipacking of $G_{k}$. So, we have to prove that $w(N_r[v])\leq r$ for each vertex $ v \in V(G_{k}) $ and for every integer $ r \geq 1 $. We prove this statement using induction on $r$.	It can be checked that $w(N_r[v])\leq r$ for each vertex $ v \in V(G_{k}) $ and for each  $ r \in \{1,2,3,4\} $. Now assume that the statement is true for $r=s$, we want to prove that it is true for $r=s+4$. Observe that, $w(N_{s+4}[v]\setminus N_{s}[v])\leq 4$, $\forall v\in V(G_{k})$. Therefore,  $w(N_{s+4}[v])\leq w(N_{s}[v])+4\leq s+4$. So, the statement is true. So, $w$ is a fractional multipacking of $G_{k}$. Therefore, $\MP_f(G_{k})\geq 5k$. 
\end{proof}

\begin{lemma} \label{mpfGk=gammabGk=5k}
If $k$ is a positive integer, then $\MP_f(G_{k})=\gamma_b(G_{k})= 5k$.
\end{lemma}

\begin{proof} Define a broadcast ${f}$ on $G_k$ as ${f}(b_{i,j})=
    \begin{cases}
        2 & \text{if } 1\leq i \leq k \text{ and } j=6,17\\
        1 & \text{if } 1\leq i \leq k \text{ and } j=12 \\
        0 & \text{otherwise }
    \end{cases}$.\\
Here  ${f}$ is an efficient dominating broadcast and $\sum_{v\in V(G_k)}{f}(v)=5k$. So, $\gamma_b(G_k)\leq 5k$, $\forall k\in \mathbb{N}$. So, by the strong duality theorem and Lemma \ref{mpfGk>=5k},  $5k\leq \MP_f(G_k)= \gamma_{b,f}(G_k)\leq \gamma_{b}(G_k)\leq 5k$. Therefore, $\MP_f(G_{k})=\gamma_b(G_{k})= 5k$. 
\end{proof}

So, $\gamma_b(G_{2k})=10k$ by Lemma \ref{mpfGk=gammabGk=5k} and  $\MP(G_{2k})=9k$ by Lemma \ref{mpG2k=9k}. Take $G_{2k}=H_k$. Thus we prove Theorem \ref{thm:9k10k}.


\begin{corollary} \label{mpfG-mpG}  The difference $\MP_f(G)-\MP(G)$ can be arbitrarily large for  connected  chordal graphs.
\end{corollary}

\begin{proof}
    We get $\MP_f(G_{2k})=10k$ by Lemma \ref{mpfGk=gammabGk=5k} and  $\MP(G_{2k})=9k$ by Lemma \ref{mpG2k=9k}. Therefore,   $\MP_f(G_{2k})-\MP(G_{2k})=k$  for all positive integers $k$. 
\end{proof}

\begin{corollary}  \label{gammabG2k/mpG2k}
For every integer $k \geq 1$, there is a connected chordal graph $G_{2k}$ with $\MP(G_{2k})=9k$, $\MP_f(G_{2k})/\MP(G_{2k})=10/9$ and $\gamma_b(G_{2k})/\MP(G_{2k})=10/9$.
\end{corollary}


\section{Improved bound of the ratio between broadcast domination and multipacking numbers of connected chordal graphs}\label{sec:chordal graphs}

In this section, we improve the lower bound of the expression $\lim_{\MP(G)\to \infty}$ $\sup\{\gamma_{b}(G)/\MP(G)\}$ for connected chordal graphs to $4/3$ where the previous lower bound was $10/9$  \cite{das2023relation} in Corollary \ref{cor:gammabG/mpGchordal}. To show this, we construct the graph $F_k$ as follows. Let $A_i$ be a 
 graph having the vertex set $\{a_i,b_i,c_i,d_i,e_i,g_i\}$ where  $(g_i,a_i,b_i,c_i,d_i,e_i,g_i)$ is a 6-cycle and $(a_i,c_i,e_i)$ is a 3-cycle,  for each  $i=1,2,\dots,3k$. We form $F_k $ by joining $b_i$ to $g_{i+1}$ for each $i=1,2,\dots,3k-1$ (See Fig. \ref{fig:pentagon}). This  graph is a  connected chordal graph. We show that $\MP(F_k)=3k$ and $\gamma_b(F_k)=4k$.

 \begin{figure}[ht]
    \centering

\includegraphics[width=\textwidth]{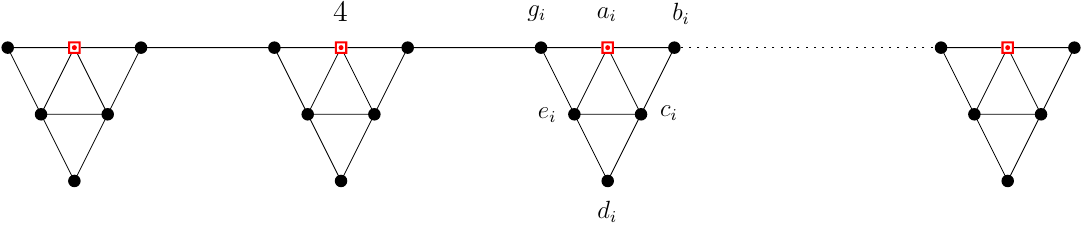}
    \caption{The $F_k$ graph with $\gamma_b(F_k)=4k$ and $\MP(F_k)=3k$. The set $\{a_i:1\leq i\leq 3k\}$ is a maximum multipacking of $F_k$.} 
    \label{fig:chordal}
\end{figure}

\begin{lemma} \label{lem:chordal_mpFk=3k} $\MP(F_{k})=3k$, for each positive integer $k$.
\end{lemma}

\begin{proof} The path $P=(g_1,a_1,b_1,g_2,a_2,b_2,\dots,g_{3k}a_{3k}b_{3k})$ is a diametral path of $F_k$ (Fig.\ref{fig:chordal}). This implies $P$ is an isometric path of $F_k$ having the length  $l(P)=3.3k-1$. By Lemma \ref{lem:isometric_path}, every third vertex on this path form a multipacking of size $\big\lceil\frac{3.3k-1}{3}\big\rceil=3k$. Therefore, $\MP(F_k)\geq 3k$. Note that, diameter of $A_i$ is $2$ for each $i$. Therefore, any multipacking of $F_k$ can contain at most one vertex of $A_i$ for each $i$. So, $\MP(F_k)\leq 3k$. Hence $\MP(F_k)= 3k$.
\end{proof}

\begin{figure}[ht]
    \centering
    \includegraphics[width=\textwidth]{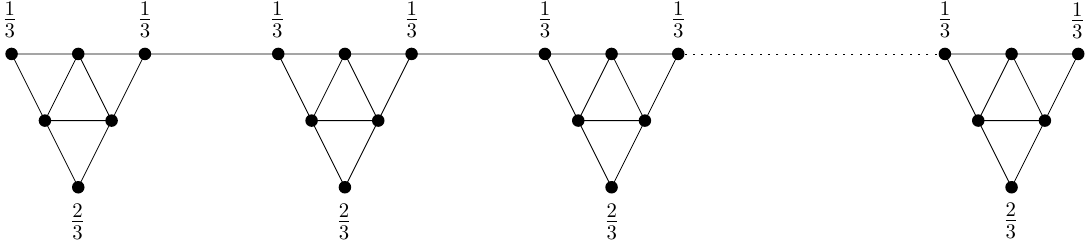}
    \caption{The $F_k$ graph with $\MP_f(F_k)=4k$.}
    \label{fig:chordal_frac}
\end{figure}

\begin{lemma}  \label{lem:chordal_mpfFk=gammabFk=4k}
If $k$ is a positive integer, then $\MP_f(F_{k})=\gamma_b(F_{k})= 4k$.
\end{lemma}

\begin{proof}  
We define a function $w: V(F_k)\rightarrow [0,\infty)$  where $w(g_{i})=w(b_{i})=\frac{1}{3}$ and $w(d_{i})=\frac{2}{3}$   for each $i=1,2,3,\dots,3k$ (Fig. \ref{fig:chordal_frac}).  So, $w(F_k)=4k$. 
 We want to show that $w$ is a fractional multipacking of $F_{k}$. So, we have to prove that $w(N_r[v])\leq r$ for each vertex $ v \in V(F_{k}) $ and for every integer $ r \geq 1 $. We prove this statement using induction on $r$.	It can be checked that $w(N_r[v])\leq r$ for each vertex $ v \in V(F_{k}) $ and for each  $ r \in \{1,2,3,4\} $. Now assume that the statement is true for $r=s$, we want to prove that it is true for $r=s+4$. Observe that, $w(N_{s+4}[v]\setminus N_{s}[v])\leq 4$, $\forall v\in V(F_{k})$. Therefore,  $w(N_{s+4}[v])\leq w(N_{s}[v])+4\leq s+4$. So, the statement is true. So, $w$ is a fractional multipacking of $F_{k}$. Therefore, $\MP_f(F_{k})\geq 4k$.

Define a broadcast ${f}$ on $F_k$ as ${f}(v)=
    \begin{cases}
        4 & \text{if } v=a_i  \text{ and } i\equiv 2 \text{ (mod $3$)}  \\
        0 & \text{otherwise }
    \end{cases}$.\\
Here  ${f}$ is an efficient dominating broadcast and $\sum_{v\in V(F_k)}{f}(v)=4k$. So, $\gamma_b(F_k)\leq 4k$, for all $ k\in \mathbb{N}$. So, by the strong duality theorem  we have  $4k\leq \MP_f(F_k)= \gamma_{b,f}(F_k)\leq \gamma_{b}(F_k)\leq 4k$. Therefore, $\MP_f(F_{k})=\gamma_b(F_{k})= 4k$.  
\end{proof}

Lemma \ref{lem:chordal_mpFk=3k}  and  Lemma \ref{lem:chordal_mpfFk=gammabFk=4k} imply the following results.

\chordalmultipackingbroadcastgape*

Therefore, the range of the expression $\lim_{\MP(G)\to \infty}$ $\sup\{\gamma_{b}(G)/\MP(G)\}$ which was previously $[10/9,3/2]$  \cite{das2023relation} is improved to $[4/3,3/2]$.

\chordalgammabGBYmpGimproved*

\section{Hyperbolic graphs}
\label{sec:A study of Broadcast domination and Multipacking numbers on Hyperbolic graphs}
In this section, we relate the broadcast domination and multipacking number of $\delta$-hyperbolic graphs. In addition, we provide an approximation algorithm for the multipacking problem of the same. 

 Chepoi et al.~\cite{chepoi2008diameters} established a relation between the length of the radius and diameter of a $\delta$-hyperbolic graph, that we state in the following theorem:

\begin{theorem}[\cite{chepoi2008diameters}]\label{delta_hyperbolic_radius_diameter}
    For any $\delta$-hyperbolic graph $G$, we have $diam(G)\geq 2 \rad(G)-4\delta-1$. 
    
\end{theorem}

Using this theorem, we can establish a relation between $\gamma_{b}(G)$ and $\MP(G)$ of a $\delta$-hyperbolic graph $G$.

\deltaMultipackingBroadcastRelation*

\begin{proof} Let $G$ be a $\delta$-hyperbolic graph with radius $r$ and diameter $d$. 
    From Theorem \ref{d+1/3leqmpG}, $\big\lceil{\frac{d+1}{3}\big\rceil}\leq \MP(G)$  which implies that $d \leq 3\MP(G)-1$. Moreover, from  Theorem \ref{mpGleqgammabG} and Theorem \ref{delta_hyperbolic_radius_diameter},  $\gamma_{b}(G)\leq r \leq \big\lfloor{\frac{d+4\delta +1}{2}\big\rfloor} \leq \big\lfloor{\frac{(3\MP(G)-1)+4\delta +1}{2}\big\rfloor}$ $=\big\lfloor{\frac{3}{2} \MP(G)+2\delta \big\rfloor}  $. Therefore, $\gamma_{b}(G)\leq \big\lfloor{\frac{3}{2} \MP(G)+2\delta\big\rfloor}$. 
\end{proof}

\approxdeltaMultipacking*

\begin{proof}
    The proof is similar to the proof of Propostition \ref{prop:3/2mpGapprox}. If $P=v_0,\dots,v_d$ is a diametrical path of $G$, then the set $M=\{v_i:i\equiv 0 \text{ } (mod \text{ } 3), i=0,1,\dots,d\}$ is a multipacking of $G$ of size $\big\lceil{\frac{d+1}{3}\big\rceil}$ by Theorem \ref{d+1/3leqmpG}. We can construct $M$ in polynomial-time since we can find a diametral path of a graph $G$ in polynomial-time.   Theorem \ref{d+1/3leqmpG}, Theorem \ref{mpGleqgammabG} and Theorem \ref{delta_hyperbolic_radius_diameter} yield that  $\big\lceil{\frac{2\MP(G)-4\delta}{3}\big\rceil}\leq\big\lceil{\frac{2r-4\delta}{3}\big\rceil}\leq\big\lceil{\frac{d+1}{3}\big\rceil}\leq \MP(G)$. 
\end{proof}

\section{Conclusion}
\label{sec:Conclusion_chordal}

We have shown that the  bound $\gamma_b(G)\leq 2\MP(G)+3$ for general graphs $G$ can be improved to $\gamma_{b}(G)\leq \big\lceil{\frac{3}{2} \MP(G)\big\rceil}$ for connected  chordal graphs. It is known that for strongly chordal graphs, $\gamma_{b}(G)=\MP(G)$, we have shown that this is not the case for connected  chordal graphs. Even more, $\gamma_b(G)-\MP(G)$ can be arbitrarily large for connected chordal graphs, as we have constructed infinitely many connected chordal graphs $G$ where $\gamma_b(G)/ \MP(G)=10/9$ and $\MP(G)$ is arbitrarily large. Further, we improve this by constructing infinitely many connected chordal graphs $G$ for which $\gamma_{b}(G) / \MP(G) = 4/3$ and where $\MP(G)$ is arbitrarily large. Hence, we determined that, for connected chordal graphs, the expression \[
\lim_{\MP(G)\to \infty} \sup \left\{ \frac{\gamma_{b}(G)}{\MP(G)} \right\}\in [4/3,3/2].
\] It remains an interesting open problem to determine the best possible value of $\displaystyle\lim_{\MP(G)\to \infty}\sup\bigg\{\displaystyle\frac{\gamma_{b}(G)}{\MP(G)}\bigg\}$  for connected chordal graphs, whereas the exact value of the expression is $2$ for general connected graphs \cite{rajendraprasad2025multipacking}. This problem could also be studied for other interesting graph classes such as the one in Figure~\ref{fig:diagram}, or others.

Moreover, we have shown that $\gamma_{b}(G)\leq \big\lfloor{\frac{3}{2} \MP(G)+2\delta\big\rfloor} $ holds for the $\delta$-hyperbolic graphs. Note that connected split graphs have radius at most~2, so their broadcast number and their multipacking number are both at most~2. The graph from Figure~\ref{fig:S3} is a split graph, showing that the multipacking number of a split graph can indeed be different from its broadcast domination number.

\chapter{Multipacking on an unbounded hyperbolic graph-class: cactus graphs}\label{chapter:cactus}\hypertarget{chapter:introhref}{}
\minitoc


  In this chapter, we show that for any cactus  graph $G$,  $\gamma_b(G)\leq \frac{3}{2}\MP(G)+\frac{11}{2}$. Although cactus graphs form a narrow graph class, we used some non-trivial techniques to provide the bound. These techniques make an important step towards generating a tighter bound for general graphs. We also show that $\gamma_b(G)-\MP(G)$ can be arbitrarily large for cactus graphs and asteroidal triple-free graphs by constructing an infinite family of cactus graphs which are also asteroidal triple-free graphs such that the ratio $\gamma_b(G)/\MP(G)=4/3$, with $\MP(G)$ arbitrarily large.  Moreover, we provide an $O(n)$-time algorithm to construct a multipacking of cactus graph $G$ of size at least $ \frac{2}{3}\MP(G)-\frac{11}{3} $,  where $n$ is the number of vertices of the graph $G$. The hyperbolicity of the cactus graph class is unbounded. For $0$-hyperbolic graphs, $\MP(G)=\gamma_b(G)$. Moreover, $\MP(G)=\gamma_b(G)$ holds for the strongly chordal graphs which is a subclass of $\frac{1}{2}$-hyperbolic graphs. Now it's a natural question: what is the minimum value of $\delta$, for which we can say that the difference $ \gamma_{b}(G) -  \MP(G) $ can be arbitrarily large for $\delta$-hyperbolic graphs? We show that the minimum value of $\delta$ is $\frac{1}{2}$ using a construction of an infinite family of cactus graphs with hyperbolicity $\frac{1}{2}$.

\section{Chapter overview} 
In Section \ref{sec:Definitions and notation_cactus}, we recall some definitions and notations. In Section \ref{sec:inequality_linking_Broadcast domination_and_Multipacking}, we prove Theorem \ref{thm:multipacking_broadcast_relation}. In Section \ref{sec:approximation_algorithm_to_find_Multipacking}, we provide a $(\frac{3}{2}+o(1))$-factor approximation algorithm for finding multipacking on the cactus graphs. In Section \ref{sec:Unboundedness_of_the_gap_between_Broadcast_domination_and_Multipacking}, we prove  that the difference $ \gamma_{b}(G) -  \MP(G) $ can be arbitrarily large for cactus graphs. In Section \ref{sec:1/2-Hyperbolic graphs},  we show that the difference $ \gamma_{b}(G) -  \MP(G) $ can be arbitrarily large for $\frac{1}{2}$-hyperbolic graphs also.  We conclude this chapter in Section \ref{sec:conclusion_cactus}.

\section{Preliminaries}\label{sec:Definitions and notation_cactus}

 Let $G=(V,E)$ be a graph and $d_G(u,v)$ be the length of a shortest path joining two vertices $u$ and $v$ in  $G$, we simply write $d(u,v)$ when there is no confusion. Let $\diam(G)=\max\{d(u,v):u,v\in V(G)\}$. A diameter of $G$ is a path of the length  $\diam(G)$. If $u\in V$, then we denote
 $N_r[u]=\{v\in V:d(u,v)\leq r\}$. The \textit{eccentricity} $e(w)$  of a vertex $w$ is $\min \{r:N_r[w]=V\}$. The \textit{radius} of the graph $G$ is $\min\{e(w):w\in V\}$, denoted by $\rad(G)$.  The \textit{center} $C(G)$ of the graph $G$  is the set of all vertices of minimum eccentricity, i.e., $C(G)=\{v\in V:e(v)=\rad(G)\}$. Each vertex in the set $C(G)$ is called a \textit{central vertex} of the graph $G$.  
 A subgraph $H$ of a graph $G$ is called an \textit{isometric subgraph} if $d_H(u, v) = d_G(u, v)$ for every pair of vertices 
$u$ and $v$ in $H$, where $d_H(u,v)$ and $d_G(u,v)$ are the distances between $u$ and $v$ in $H$ and $G$, respectively. An \textit{isometric path} is an isometric subgraph which is a path.  If $H_1$ and $H_2$ are two subgraphs of $G$, then $H_1\cup H_2$ denotes the subgraph whose vertex set is $V(H_1)\cup V(H_2)$ and edge set is $E(H_1)\cup E(H_2)$. We denote an  indicator function as  $1_{[x<y]}$  that takes the value $1$ when $x<y$, and $0$ otherwise. If $P$ is a path in $G$, then we say $V(P)$ is the vertex set of the path $P$, $E(P)$ is the edge set of the path $P$, and  $l(P)$ is the length of the path $P$, i.e., $l(P)=|E(P)|$.

 A \textit{cactus} is a connected graph in which any two  cycles have at most one vertex in common. Equivalently, it is a connected graph in which every edge belongs to at most one  cycle.

  Let $d$ be the shortest-path metric of a graph $G$. The graph $G$ is called a \emph{$\delta$-hyperbolic graph} if for any four
vertices $u, v, w, x \in V(G)$, the two larger of the three sums $d(u, v) + d(w, x)$, $d(u, w) + d(v, x)$, $d(u, x) +
d(v, w)$ differ by at most $2\delta$. A graph class $\mathcal{G}$ is said to be hyperbolic if there exists a constant $\delta$ such that every graph $G \in \mathcal{G}$ is $\delta$-hyperbolic.

\section{Proof of Theorem \ref{thm:multipacking_broadcast_relation}}\label{sec:inequality_linking_Broadcast domination_and_Multipacking}

In this section, we prove a relation between the broadcast domination number and the radius of a cactus. Using this we establish our main result that relates the  broadcast domination number and the multipacking number for cactus graphs.

We start with a lemma that is true for general graphs.


\begin{lemma} \label{lem:disjoint_radial_path} \textbf{(Disjoint radial path lemma)}
Let $G$ be a graph with radius $r$ and central vertex $c$, where $r\geq 1$. Let $P$ be an isometric path in $G$ such that $l(P)=r$ and $c$ is one endpoint of $P$. Then there exists an isometric path $Q$ in $G$ such that $V(P)\cap V(Q)=\{c\}$, $ r-1 \leq l(Q)\leq  r$ and $c$ is one endpoint of $Q$.
\end{lemma}







\begin{proof}

\textcolor{black}{Let \(P = (c, v_1, \ldots, v_r)\) be a shortest path to a vertex at distance \(r\) from \(c\). 
    Let \(w_{r'}\) be a vertex at maximum distance from \(v_1\) and \(Q = (c, w_1, \ldots, w_{r'})\) 
    be a shortest path from \(c\) to \(w_{r'}\). If \(v_i \in Q\) for some \(i\), then there is a shortest 
    path from \(c\) to \(w_{r'}\) passing through \(v_1\). This implies \(d(v_1, w_{r'}) = d(c, w_{r'}) - 1 \leq r - 1\), 
    which is a contradiction since radius of $G$ is $r$. The result follows.}    
\end{proof}

For this section, we assume that $G$ is a cactus. 
Let $c$ be a central vertex  and $r$ be the radius of $G$ where $r\geq 1$. We can find an isometric path $P$ of length $r$ whose one endpoint is $c$, since  $\rad(G)=r$. Let $P=(c,v_1,v_2,v_3,\dots ,v_r)$.  By Lemma \ref{lem:disjoint_radial_path}, we can find an isometric path $Q$ in $G$ such that $V(P)\cap V(Q)=\{c\}$, $ r-1 \leq l(Q)\leq  r$ and $c$ is one endpoint of $Q$. Let $Q=(c,w_1,w_2,w_3,\dots ,w_{r'})$ where $ r-1 \leq r'\leq  r$. Now we define some expressions to study the subgraph $P\cup Q$ in $G$ (See Fig. \ref{fig:bfs3}).  For $v_i\in V(P)\setminus \{c\}$ and $w_j\in V(Q)\setminus \{c\}$, we define  $ X_{P,Q}(v_i,w_j)=$\{$P_1$ : $P_1$ is a path in $G$ that joins $v_i$ and $w_j$ such that $V(P)\cap V(P_1)=\{v_i\}$ and $V(Q)\cap V(P_1)=\{w_j\}\}$. Note that $c \notin V(P_1)$. Let $X_{P,Q}=\{(v_i,w_j): X_{P,Q}(v_i,w_j)\neq \emptyset \}$. Since $G$ is a cactus, it implies every edge belongs to at most one cycle. Therefore, $|X_{P,Q}(v_i,w_j)|\leq 1$ and also $|X_{P,Q}|\leq 1$. Therefore, there is at most one path (say, $P_1$) that does not pass through $c$ and joins a vertex of $V(P)\setminus \{c\}$ with a vertex of $V(Q)\setminus \{c\}$ and $P_1$ intersects $P$ and $Q$ only at their joining points. So, the following observation is true. 

\begin{observation}
\label{obs:joining_path_atmost_1}
    Let $G$ be a cactus with $\rad(G)=r$ and central vertex $c$. Suppose $P$ and $Q$  are two isometric paths  in $G$ such that   $V(P)\cap V(Q)=\{c\}$,  $l(P)=r$, $r-1\leq l(Q)\leq r$  and both have  one endpoint  $c$. Then 
    
    (i) $|X_{P,Q}|\leq 1$ and  $|X_{P,Q}(v_i,w_j)|\leq 1$ for all $(v_i,w_j)$.

(ii) $X_{P,Q}= \{(v_i,w_j)\}$ iff  $|X_{P,Q}(v_i,w_j)|= 1$.
    
\end{observation}

\begin{figure}[ht]
    \centering
    \includegraphics[width=\textwidth]{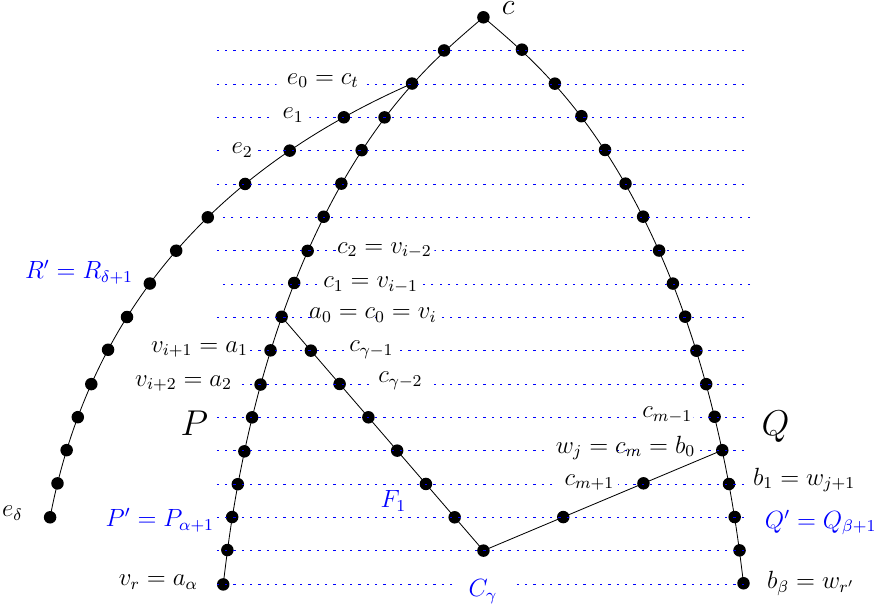}
    \caption{The subgraph $P\cup Q\cup F_1\cup R'=H_\gamma (c_0,\alpha,c_t,\delta,c_m,\beta)$}
    \label{fig:bfs3}
\end{figure}

\medskip
\noindent \textbf{Finding a special subgraph $\bm{H_\gamma(c_0,\alpha,c_t,\delta,c_m,\beta)}$  in cactus $\bm{G}$ : }
From Observation \ref{obs:joining_path_atmost_1}, there are two possible cases: either $|X_{P,Q}| =0$ or $|X_{P,Q}|=1$. First, we want to study the case when $|X_{P,Q}|=1$. In this case, suppose $ X_{P,Q}=\{(v_i,w_j)\}$ and $X_{P,Q}(v_i,w_j)=\{F_1\}$ (See Fig.  \ref{fig:bfs3}). We are interested in finding a multipacking of $G$ from the subgraph $P\cup Q\cup F_1$.   Let $F_2=(v_i,v_{i-1},v_{i-2},\dots,v_1,c,w_1,w_2,\dots,$ $w_{j-1},w_j)$, $P'=(v_i,v_{i+1},\dots,v_{r})$ and $Q'=(w_j,w_{j+1},\dots,w_{r'})$. Here $F_1\cup F_2$ is a cycle and $P'$ and $Q'$ are two isometric paths that are attached to the cycle. Suppose there is another isometric path $R'$ which is disjoint with $P'$ and $Q'$ and whose one endpoint belongs to $V(F_2)$. Then  $P\cup Q\cup F_1\cup R'=F_1\cup F_2\cup P'\cup Q' \cup R'$ is a subgraph of $G$ (See Fig.\ref{fig:bfs3}).   Let $H=P\cup Q\cup F_1\cup R'$. Note that, we can always find $P$ and $Q$ in 
$G$ due to Lemma \ref{lem:disjoint_radial_path}, but  $F_1$ or $R'$ might not be there in $G$. In this case, we can assume that $V(F_1)$ or $V(R')$ is empty in $H$. In the rest of this section, our main goal is to find a multipacking of $G$ of size at least  $\frac{2}{3}\rad(G)-\frac{11}{3}$ from the subgraph $H$, so that we can prove Theorem \ref{thm:multipacking_broadcast_relation}. 

To give $H$ a general structure, we want to rename the vertices. Suppose  $F_1\cup F_2$ is a cycle of length $\gamma$. Let $F_1\cup F_2=C_\gamma=(c_0,c_1,c_2,\dots,c_{\gamma-2},c_{\gamma-1},c_0)$. Suppose $l(P')=\alpha$, $l(Q')=\beta$ and $l(R')=\delta$. We rename the paths $P'$, $Q'$ and $R'$ as $P_{\alpha+1}$, $Q_{\beta+1}$ and $R_{\delta+1}$ respectively. 
Let $P_{\alpha+1}=(a_0,a_{1},\dots,a_{\alpha})$,  $Q_{\beta+1}=(b_0,b_{1},\dots,b_{\beta})$ and  $R_{\delta+1}=(e_0,e_{1},\dots,e_{\delta})$   such that $c_0=a_0$,  $c_t=e_0$, $c_m=b_0$  (See Fig. \ref{fig:bfs1}). Here $P_{\alpha+1}$, $Q_{\beta+1}$ and $R_{\delta+1}$ are three isometric paths in $G$.  According to the structure, we have $V(P_{\alpha+1})\cap V(Q_{\beta+1})=\emptyset$,  $V(Q_{\beta+1}) \cap V(R_{\delta+1}) =\emptyset$,
$V(R_{\delta+1})\cap V(P_{\alpha+1})=\emptyset$,
$V(C_\gamma)\cap V(P_{\alpha+1}) =\{c_0\}$, $V(C_\gamma)\cap V(R_{\delta+1})=\{c_t\}$, $V(C_\gamma)\cap V(Q_{\beta+1})=\{c_m\}$. Now we can write $H$ as a variable of $\alpha, \beta, \gamma $ and $\delta$. Let $H=H_\gamma(c_0,\alpha,c_t,\delta,c_m,\beta)$.

\begin{figure}[ht]
    \centering
    \includegraphics[width=\textwidth]{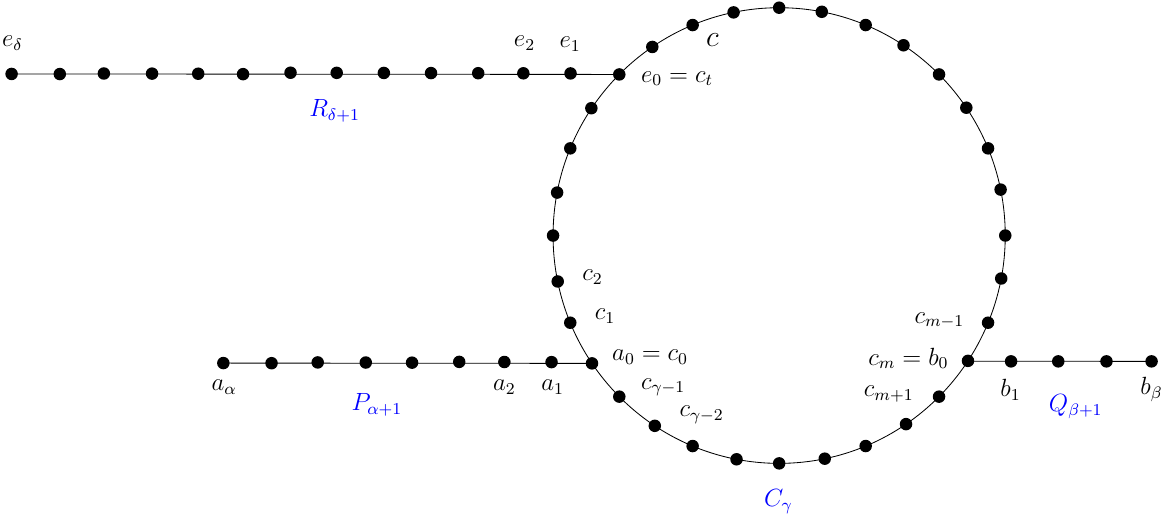}
    \caption{The subgraph $H_\gamma (c_0,\alpha,c_t,\delta,c_m,\beta)$}
    \label{fig:bfs1}
\end{figure}

In the rest of the section, we study how to find a multipacking of $G$ from the subgraph $H_\gamma(c_0,\alpha,c_t,$ $\delta,c_m,\beta)$. Using this, we prove Lemma \ref{lem:multipacking_radius_relation} that yields Theorem \ref{thm:multipacking_broadcast_relation} which is our main result in this section.

\begin{lemma} [\cite{beaudou2019broadcast}]  \label{lem:isometric_path}
    Let $G$ be a graph, $k$ be a positive integer and $P=(v_0,v_1, \dots ,$ $ v_{k-1})$ be an isometric path in $G$ with $k$ vertices. Let $M = \{v_i:0\leq i\leq k, i\equiv 0 \text{ (mod $3$)} \}$ be the set of every third vertex
on this path. Then $M$ is a multipacking in $G$ of size $\big\lceil\frac{k}{3}\big\rceil$.
\end{lemma}

Note that, if $H_\gamma(c_0,\alpha,c_t,\delta,c_m,\beta)$ is a subgraph of $G$, then $P_{\alpha+1}$, $Q_{\beta+1}$ and $R_{\delta+1}$ are three isometric paths in $G$ by the definition of $H_\gamma(c_0,\alpha,c_t,\delta,c_m,\beta)$. Moreover, since $G$ is a cactus, two cycles cannot share an edge. Therefore, the following observation is true.

\begin{observation} 
\label{obs:isometric_graph}
    If $G$ is a cactus and $H_\gamma(c_0,\alpha,c_t,\delta,c_m,\beta)$ is a subgraph of $G$ such that $\gamma\geq 3 $ and $c_0,c_t,c_m$ are distinct vertices of $C_\gamma$, then $H_\gamma(c_0,\alpha,c_t,\delta,c_m,\beta)$ is an isometric subgraph of $G$.
\end{observation}

\begin{observation} 
\label{obs:isometric_graph_path_F_1}
    Let $G$ be a cactus and $H_\gamma(c_0,\alpha,c_t,\delta,c_m,\beta)$ be a subgraph of $G$ such that $\gamma\geq 3 $ and $c_0,c_t,c_m$ are distinct vertices of $C_\gamma$. Let $F_1$ and $F_2$ be two paths such that  $F_1=(c_m,c_{m+1},\dots,$ $c_{\gamma-1},c_0)$ and $F_2=(c_0,c_1,\dots,c_m)$. Then
    
 \noindent  (i)  If $l(F_1)> l(F_2)$, then $P_{\alpha+1}\cup F_2\cup Q_{\beta +1}$ is an isometric path of $G$.

 \noindent (ii)  If $l(F_1)<l(F_2)$, then $P_{\alpha+1}\cup F_1\cup Q_{\beta +1}$ is an isometric path of $G$. 

 \noindent (iii) If $l(F_1)=l(F_2)$, then 
 both of $P_{\alpha+1}\cup F_1\cup Q_{\beta +1}$ and $P_{\alpha+1}\cup F_2\cup Q_{\beta +1}$ are isometric paths of $G$.

\end{observation}

The following lemma tells why it is sufficient to find a multipacking in $H_\gamma(c_0,\alpha,$ $c_t,\delta,c_m,\beta)$ to provide a multipacking in $G$.

\begin{lemma}   
\label{lem:multipacking_subgraph}
    Let $G$ be a cactus and $H_\gamma(c_0,\alpha,c_t,\delta,c_m,\beta)$ be a subgraph of $G$. If   $M$ is a multipacking of $H_\gamma(c_0,\alpha,c_t,\delta,c_m,\beta)$, then $M$ is a  
 multipacking of $G$.
\end{lemma}

\begin{proof} 
Let $H=H_\gamma(c_0,\alpha,c_t,\delta,c_m,\beta)$. Since $M$ is multipacking of $H$, therefore $|N_r[z]\cap M|\leq r$ for all $z\in V(H)$ and $r\geq 1$. Let $z\in V(G)\setminus V(H)$ and    $v\in V(H)$.  Define $ P_H(z,v)=$\{$P$ :  $P$ is a path in $G$ that joins $z$ and $v$ such that $V(P)\cap V(H)=\{v\}$\}. 
Note that, if $v_1,v_2\in V(H)$ and $v_1\neq v_2$, then $P_H(z,v_1)\cap P_H(z,v_2)=\emptyset$. Since $G$ is connected, therefore $|\{v\in V(H):P_H(z,v)\neq \emptyset\}|\geq 1$.

\medskip
\noindent
\textbf{Claim \ref{lem:multipacking_subgraph}.1. }  If $z\in V(G)\setminus V(H)$, then $|\{v\in V(H):P_H(z,v)\neq \emptyset \}|\leq 2$. 

\begin{claimproof}Suppose $|\{v\in V(H):P_H(z,v)\neq \emptyset\}|\geq 3$. Let $v_1,v_2,v_3\in \{v\in V(H):P_H(z,v)\neq \emptyset\}$ where $v_1,v_2,v_3$ are distinct. So, there are $3$ distinct paths $P_1,P_2,P_3$ such that $P_i\in P_H(z,v_i)$ for $i=1,2,3$. Then there are two cycles formed by the paths $P_1,P_2,P_3$ that have at least one common edge, which is a contradiction, since $G$ is a cactus.     
\end{claimproof}

\vspace{0.2cm}
\noindent
\textbf{Claim \ref{lem:multipacking_subgraph}.2. }    If $z\in V(G)\setminus V(H)$ and  $|\{v\in V(H):P_H(z,v)\neq \emptyset\}|=1 $, then  $|N_r[z]\cap M|\leq r$ for all $r\geq 1$.

\begin{claimproof} Let $ \{v\in V(H):P_H(z,v)\neq \emptyset\}=\{v_1\}$. Therefore,  if $v$ is any vertex in $V(H)$, then  any path joining $z$ and $v$ passes through $v_1$. Let $d(z,v_1)=k$ for some $k\geq 1$. We have $|N_{r}[v_1]\cap M|\leq r$ for all $r\geq 1$. If $1\leq r<k$, then  $|N_r[z]\cap M|=0<r$. If $r\geq k$, then $|N_r[z]\cap M|= |N_{r-k}[v_1]\cap M|\leq r-k\leq r$.   
\end{claimproof}

\vspace{0.2cm}
\noindent
\textbf{Claim \ref{lem:multipacking_subgraph}.3. }    If $z\in V(G)\setminus V(H)$ and $|\{v\in V(H):P_H(z,v)\neq \emptyset\}|=2 $, then  $|N_r[z]\cap M|\leq r$ for all $r\geq 1$.

\begin{claimproof} Let $\{v\in V(H):P_H(z,v)\neq \emptyset\}=\{v_1,v_2\}$. Therefore,  if $w$ is any vertex in $V(H)$, then  any path joining $z$ and $w$ passes through $v_1$ or $v_2$. Note that both of  $v_1,v_2$ belong to either $P_{\alpha+1}$, $Q_{\beta+1}$ or $R_{\delta+1}$, otherwise $G$ cannot be a cactus. W.l.o.g. assume that $v_1,v_2\in P_{\alpha+1}$.  Let $d(z,v_1)=k_1$ and $d(z,v_2)=k_2$ for some $k_1,k_2\geq 1$. Since $P_{\alpha+1}$ is an isometric path in $G$,  $d(v_1,v_2)\leq d(z,v_1)+d(z,v_2)=k_1 + k_2$. Let $r$ be a positive integer  and  $S=V(H)$. If $N_{r-k_1}[v_1]\cap N_{r-k_2}[v_2]\cap S=\emptyset$,  then  $(r-k_1)+(r-k_2)\leq d(v_1,v_2)\leq k_1+k_2$ $\implies$ $r\leq k_1+k_2$. Therefore,   $|N_r[z]\cap M|=|N_{r-k_1}[v_1]\cap M|+|N_{r-k_2}[v_2]\cap M|\leq r-k_1+r-k_2=2r-(k_1+k_2)\leq r$. Suppose  $N_{r-k_1}[v_1]\cap N_{r-k_2}[v_2]\cap S\neq \emptyset$. Let $v_1=a_i$, $v_2=a_j$, $h=\big\lfloor\frac{i+j}{2}\big\rfloor$, $v=a_h$ and $s=\big\lfloor\frac{(r-k_1)+(r-k_2)+ d(v_1,v_2)+1}{2}\big\rfloor$.  So,


\begin{align*}
N_r[z]\cap S
   &= (N_{r-k_1}[v_1]\cup N_{r-k_2}[v_2])\cap S
      \subseteq N_s[v]\cap S \\[4pt]
\implies\quad
N_r[z]\cap M
   &\subseteq N_s[v]\cap M \\[4pt]
\implies\quad
|N_r[z]\cap M|
   &\le |N_s[v]\cap M|
      \le s
      = \Big\lfloor 
          \frac{(r-k_1)+(r-k_2)+ d(v_1,v_2)+1}{2}
        \Big\rfloor \\[4pt]
&\le 
\Big\lfloor 
   \frac{(r-k_1)+(r-k_2)+ k_1+k_2+1}{2}
\Big\rfloor \\[4pt]
&\le 
\Big\lfloor \frac{2r+1}{2} \Big\rfloor
   = \Big\lfloor r+\tfrac12 \Big\rfloor
   = r .
\end{align*}

\end{claimproof}

From the above results, we can say that   $|N_{r}[z]\cap M|\leq r$ for all $z\in V(G)$ and $r\geq 1$. Therefore, $M$ is a   multipacking of $G$. 
\end{proof}

Now our goal is to find a multipacking of $H_\gamma(c_0,\alpha,c_t,\delta,c_m,\beta)$. Whichever multipacking we find for the subgraph $H_\gamma(c_0,\alpha,c_t,\delta,c_m,\beta)$ will be a multipacking for $G$ by Lemma \ref{lem:multipacking_subgraph}. There could be several ways to choose a multipacking from the subgraph $H_\gamma(c_0,\alpha,c_t,\delta,c_m,\beta)$. 
In order to prove Lemma \ref{lem:multipacking_radius_relation}, which is the main lemma to prove Theorem \ref{thm:multipacking_broadcast_relation}, we have to ensure that the size of the multipacking is at least $\frac{2}{3}\rad(G)-\frac{11}{3}$.

We discuss three choices to find  a  multipacking in $H_\gamma(c_0,\alpha,c_t,\delta,c_m,\beta)$ in the following subsections.

We henceforth write $H$ in place of $H_\gamma(c_0,\alpha,c_t,\delta,c_m,\beta)$ for simplicity.

\begin{figure}[ht]
    \centering
    \includegraphics[width=\textwidth]{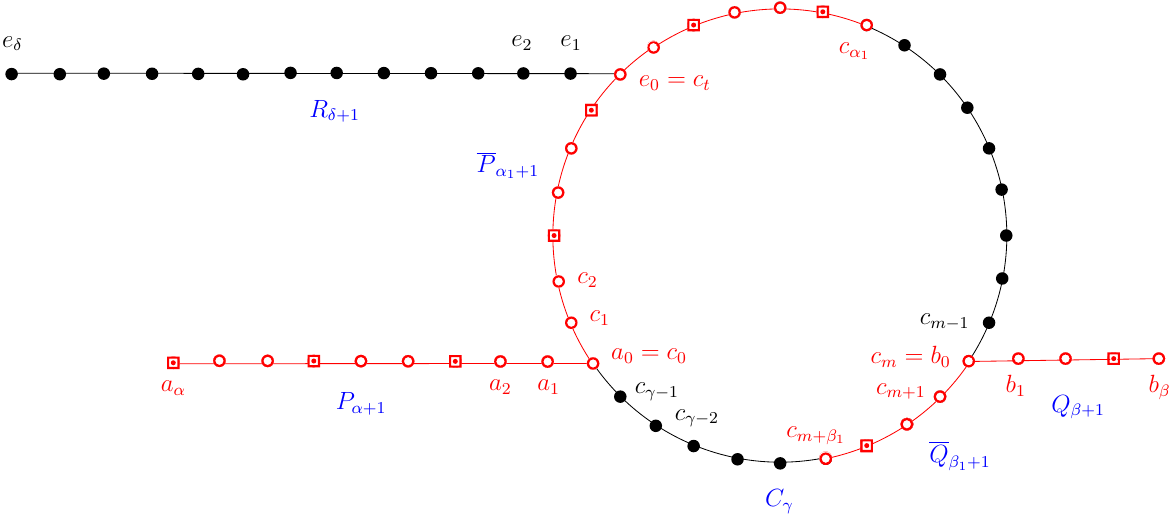}
 \caption{ The circles and squares  represent the  subgraph $P_{\alpha+1}\cup \overline{P}_{\alpha_1+1}\cup Q_{\beta+1}\cup \overline{Q}_{\beta_1+1}$ and the squares represent the set $M_\gamma(c_0,\alpha,\alpha_{1},c_m,\beta,\beta_1)$.}
    \label{fig:m_2}
\end{figure}

\medskip
\noindent\textbf{Finding a multipacking of  $\bm{H_\gamma(c_0,\alpha,c_t,\delta,c_m,\beta)} $ 
 according to Choice 1 :}\label{subec:multipacking_technique_1}
\smallskip

Here we find two paths of  $H$ so that the set of every third vertex (with some exceptions) 
on those paths provide a multipacking of $H$. Let $\overline{P}_{\alpha_1+1}=(c_0,c_1,\dots,c_{\alpha_1})$
and $\overline{Q}_{\beta_1+1}=(c_m,c_{m+1},\dots,c_{m+\beta_1})$, where $0\leq \alpha_1\leq m-1$ and $0\leq \beta_1\leq (\gamma-1)-m$. Now we choose every third vertex (with some exceptions) from the paths  $P_{\alpha+1}\cup \overline{P}_{\alpha_1+1}$ and $Q_{\beta+1}\cup \overline{Q}_{\beta_1+1}$ to construct a set, called $M_\gamma(c_0,\alpha,\alpha_1,c_m,\beta,\beta_1)$ (See Fig. \ref{fig:m_2}). Formally, we define \begin{align*}
M_\gamma(c_0,\alpha,\alpha_1,c_m,\beta,\beta_1)
   ={}&
   \{\,a_i : 0 \le i \le \alpha,\; i \equiv 0 \pmod{3}\,\} \\
   &\cup \{\,c_i : 0 \le i \le \alpha_1,\; i \equiv 0 \pmod{3}\,\} \\
   &\cup \{\,b_i : 0 \le i \le \beta,\; i \equiv 0 \pmod{3}\,\} \\
   &\cup \{\,c_i : m \le i \le m+\beta_1,\; i \equiv m \pmod{3}\,\} \setminus \{c_0, c_m\}.
\end{align*}


This set will be a multipacking of $G$ under the conditions stated in the following lemma.

\begin{lemma} 
\label{lem:multipacking_subgraph1}
Let $G$ be a cactus and $H_\gamma(c_0,\alpha,c_t,\delta,c_m,\beta)$  be a subgraph of $G$. Let  $\alpha_1,\alpha_2,\beta_1,\beta_2$ be non-negative integers such that $\alpha_2=(m-1)-\alpha_1 $, $\beta_2=(\gamma-1)-(m+\beta_1) $. If $\alpha_1\leq 3\beta_2+\alpha_2+\beta_1$ and $\beta_1\leq 3\alpha_2+\beta_2+\alpha_1$, then 
 $M_\gamma(c_0,\alpha,\alpha_1,c_m,\beta,\beta_1)$  is a multipacking of $G$ of size at least
 $ \big\lfloor\frac{\alpha+\alpha_1+1}{3}\big\rfloor+\big\lfloor\frac{\beta+\beta_1+1}{3}\big\rfloor -2$.

\end{lemma}

\begin{proof}

Let $H=H_\gamma(c_0,\alpha,c_t,\delta,c_m,\beta)$ and   $M=M_\gamma(c_0,\alpha,\alpha_1,c_m,\beta,\beta_1)$.

Let $v=c_{x_1}$ for some $x_1$ where $0\leq x_1\leq m-1$. Note that $m-1=\alpha_1+\alpha_2$ and $\gamma=\alpha_1+\alpha_2+\beta_1+\beta_2+2$, i.e., $\big\lfloor\frac{\gamma}{2}\big\rfloor=\big\lfloor\frac{\beta_{1}+\beta_{2}+\alpha_{1}+\alpha_2}{2}\big\rfloor+1$.

Recall that we denote an  indicator function as  $1_{[x<y]}$  that takes the value $1$ when $x<y$, and $0$ otherwise.

\vspace{0.22cm}

 \noindent \textit{\textbf{Case 1: }} $1\leq r\leq \max\{\alpha_{1}+\alpha_2-x_1,\beta_{2}+x_1\}$. 

 In this case,  $N_r[v]\cap V(P_{\alpha+1}\cup \overline{P}_{\alpha_1+1})\neq \emptyset$ but $N_r[v]\cap V(Q_{\beta+1}\cup \overline{Q}_{\beta_1+1})= \emptyset$. Therefore, $|N_r[v]\cap M|\leq \frac{1}{3} [\{2r-2(\alpha_{1}+\alpha_2-x_1)-1+(\alpha_1-x_1)+r+1\}\times 1_{[\alpha_{1}+\alpha_2-x_1\leq \beta_{2}+x_1]}+\{r-x_1-\beta_{2}+2r+1\}\times 1_{[\alpha_{1}+\alpha_2-x_1> \beta_{2}+x_1]}]\leq r $. Therefore $|N_r[v]\cap M|\leq r$.

\vspace{0.22cm}

 \noindent \textit{\textbf{Case 2: }} $\max\{\alpha_{1}+\alpha_2-x_1,\beta_{2}+x_1\}<r\leq \big\lfloor\frac{\beta_{1}+\beta_{2}+\alpha_{1}+\alpha_2}{2}\big\rfloor$.

 In this case,  $N_r[v]\cap V(P_{\alpha+1}\cup \overline{P}_{\alpha_1+1})\neq \emptyset$ and $N_r[v]\cap V(Q_{\beta+1}\cup \overline{Q}_{\beta_1+1})\neq \emptyset$. Therefore, $|N_r[v]\cap M|\leq \frac{1}{3} [ r+1+(\alpha_1-x_1)+2\{r-(\alpha_{1}+\alpha_2-x_1)\}-1+r-(x_1+\beta_{2})]=\frac{1}{3} [4r-\alpha_{1}-\beta_{2}-2\alpha_2]$.
We know that $\beta_{1}\leq \alpha_{1}+\beta_{2}+3\alpha_2$. Now $\beta_{1}\leq \alpha_{1}+\beta_{2}+3\alpha_2$ $\implies$ $\beta_{1}+\beta_{2}+\alpha_{1}+\alpha_2\leq 2\alpha_{1}+2\beta_{2}+4\alpha_2$ $\implies$ $\frac{\beta_{1}+\beta_{2}+\alpha_{1}+\alpha_2}{2}\leq \alpha_{1}+\beta_{2}+2\alpha_2$. Since $r\leq \big\lfloor\frac{\beta_{1}+\beta_{2}+\alpha_{1}+\alpha_2}{2}\big\rfloor$, therefore $r\leq \alpha_{1}+\beta_{2}+2\alpha_2$. Now $r\leq \alpha_{1}+\beta_{2}+2\alpha_2$  $\implies$ $4r-\alpha_{1}-\beta_{2}-2\alpha_2\leq 3r$ $\implies$ $|N_r[v]\cap M|\leq r$.

\vspace{0.22cm}

 \noindent \textit{\textbf{Case 3: }} $\big\lfloor\frac{\beta_{1}+\beta_{2}+\alpha_{1}+\alpha_2}{2}\big\rfloor< r$. 

  In this case, $N_r[v]\supseteq V(C_\gamma)$. Therefore, 
 $|N_r[v]\cap M|\leq \frac{1}{3} [\beta_{1}+1+\alpha_{1}+1+(r-x_1)+r-(\alpha_{1}+\alpha_2-x_1+1)]=\frac{1}{3} [2r+\beta_{1}-\alpha_2+1]$. We know that $\beta_{1}\leq \alpha_{1}+\beta_{2}+3\alpha_2$. Now $\beta_{1}\leq \alpha_{1}+\beta_{2}+3\alpha_2$ $\implies$ $2\beta_{1}-2\alpha_2\leq \beta_{1}+\beta_{2}+\alpha_{1}+\alpha_2$ $\implies$ $\beta_{1}-\alpha_2\leq \frac{\beta_{1}+\beta_{2}+\alpha_{1}+\alpha_2}{2}$  $\implies$ $\beta_{1}-\alpha_2+1\leq \frac{\beta_{1}+\beta_{2}+\alpha_{1}+\alpha_2}{2}+1$ $\implies$ $\beta_{1}-\alpha_2+1\leq \big\lfloor\frac{\beta_{1}+\beta_{2}+\alpha_{1}+\alpha_2}{2}\big\rfloor+1\leq r$ 
 $\implies$ $\beta_{1}-\alpha_2 +1\leq r$ $\implies$ $|N_r[v]\cap M|\leq \frac{1}{3} [2r+\beta_{1}-\alpha_2+1]\leq r$.

\vspace{0.22cm}


 Similarly, using the relation $\alpha_1\leq 3\beta_2+\alpha_2+\beta_1$, we can show that, when $v=c_{x_2}$ for some $x_2$ where $m\leq x_2\leq \gamma-1$, we can show that  $v\in \{c_i:m\leq i\leq \gamma-1\}$, then $|N_r[v]\cap M|\leq r$ for all  $r\geq 1$. Therefore, $|N_r[v]\cap M|\leq r$ for all $v\in V(C_\gamma)$ and $r\geq 1$.

 Suppose $v\in V(P_{\alpha+1})$. Then any path that joins $v$ with a vertex in $V(C_\gamma)\cup V(Q_{\beta+1})\cup V(R_{\delta+1})$ passes through $a_0$, otherwise $G$ cannot be a cactus. By Observation \ref{obs:isometric_graph} we know that $H_\gamma(c_0,\alpha,c_t,\delta,c_m,\beta)$ is an isometric subgraph of $G$.  Therefore $|N_r[v]\cap M|\leq r$ for all  $r\geq 1$. Similarly we can show that, if $v$ is in  $V(Q_{\beta+1})$ or $V(R_{\delta+1})$, then $|N_r[v]\cap M|\leq r$ for all  $r\geq 1$. Therefore $M$ is a multipacking of $H$. So, $M$ is a multipacking of $G$ by Lemma \ref{lem:multipacking_subgraph} and $|M|\geq  \big\lfloor\frac{\alpha+\alpha_1+1}{3}\big\rfloor+\big\lfloor\frac{\beta+\beta_1+1}{3}\big\rfloor -2$.   
\end{proof}

Substitute $\alpha_2=(m-1)-\alpha_1 $ and $\beta_2=(\gamma-1)-(m+\beta_1) $ in Lemma \ref{lem:multipacking_subgraph1}. This yields the following.

\begin{lemma} 
\label{lem:gamma/2}
    Let $G$ be a cactus and $H_\gamma(c_0,\alpha,c_t,\delta,c_m,\beta)$  be a subgraph of $G$. Let  $\alpha_1$ and $\beta_1$ be non-negative integers such that $\alpha_1\leq m-1$ and $\beta_1\leq (\gamma-1)-m$. If $\alpha_1\leq \big\lfloor\frac{\gamma}{2}\big\rfloor-1$ and $\beta_1\leq \big\lfloor\frac{\gamma}{2}\big\rfloor-1$, then 
 $M_\gamma(c_0,\alpha,\alpha_1,c_m,\beta,\beta_1)$  is a multipacking of $G$ of size at least
 $ \big\lfloor\frac{\alpha+\alpha_1+1}{3}\big\rfloor+\big\lfloor\frac{\beta+\beta_1+1}{3}\big\rfloor -2$.
 
\end{lemma}



\begin{figure}[ht]
    \centering
    \includegraphics[width=\textwidth]{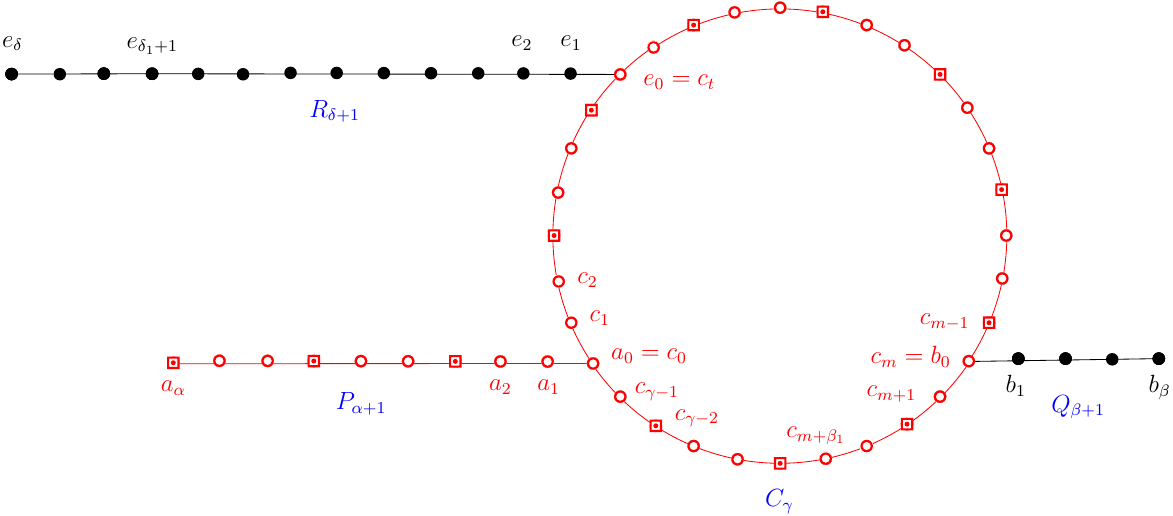}
 \caption{The circles and squares  represent the subgraph $P_{\alpha+1}\cup C_\gamma$  and the squares represent the  set $M'_\gamma(c_0,\alpha)$.}
    \label{fig:multipacking_choice_2}
\end{figure}

\smallskip
\noindent \textbf{Finding a multipacking of  $\bm{H_\gamma(c_0,\alpha,c_t,\delta,c_m,\beta)} $ according to Choice 2 :}\label{subec:multipacking_technique_2}
\medskip

Here we find a path and a cycle of  $H$ so that the set of every third vertex (with some exceptions) 
on these subgraphs provide a multipacking of $H$. Consider the path $P_{\alpha+1}$ and the cycle $C_\gamma$. Now we choose every third vertex (with some exceptions) from these subgraphs  to construct a set, called $M'_\gamma(c_0,\alpha)$ (See Fig. \ref{fig:multipacking_choice_2}).  Formally, we define \begin{align*}
M'_\gamma(c_0,\alpha)
   ={}&
   \{\,a_i : 0 \le i \le \alpha,\; i \equiv 0 \pmod{3}\,\} \\
   &\cup \{\,c_i : 0 \le i \le \gamma-1,\; i \equiv 0 \pmod{3}\,\} \setminus \{c_0\}.
\end{align*}


Similarly, for  the path $Q_{\beta+1}$ and the cycle $C_\gamma$, we define the set \begin{align*}
M'_\gamma(c_m,\beta)
   ={}&
   \{\,b_i : 0 \le i \le \beta,\; i \equiv 0 \pmod{3}\,\} \\
   &\cup \{\,c_i : 0 \le i \le \gamma-1,\; i \equiv m \pmod{3}\,\} \setminus \{c_m\}.
\end{align*}


These sets will be a multipacking of $G$ as stated in the following lemma.

\begin{lemma} 
\label{lem:multipacking_delta=0}
     Let $G$ be a cactus and $H_\gamma(c_0,\alpha,c_t,\delta,c_m,\beta)$  be a subgraph of $G$. Then $M'_\gamma(c_0,\alpha)$ is a multipacking of $G$ of size at least $\big\lfloor\frac{\gamma}{3}\big\rfloor+\big\lfloor\frac{\alpha}{3}\big\rfloor-1$ and $M'_\gamma(c_m,\beta)$ is a multipacking of $G$ of size at least $\big\lfloor\frac{\gamma}{3}\big\rfloor+\big\lfloor\frac{\beta}{3}\big\rfloor-1$.
\end{lemma}

\begin{proof}  Let $H=H_\gamma(c_0,\alpha,c_t,\delta,c_m,\beta)$,  $M'=M'_\gamma(c_0,\alpha)$.  From the definition of $M'$,   $|N_r[v]\cap M'|\leq r$ for all $v\in V(H)$ and $r\geq 1$. Hence $M'$ is a multipacking of $H$ size at least 
  $\big\lfloor\frac{\gamma}{3}\big\rfloor+\big\lfloor\frac{\alpha}{3}\big\rfloor-1$. Therefore, $M'$ is a multipacking of $G$ by Lemma \ref{lem:multipacking_subgraph}. For a  similar reason, $M'_\gamma(c_m,\beta)$ is a multipacking of $G$ of size at least $\big\lfloor\frac{\gamma}{3}\big\rfloor+\big\lfloor\frac{\beta}{3}\big\rfloor-1$. 
\end{proof}

\begin{figure}[htp]
    \centering
    \includegraphics[width=\textwidth]{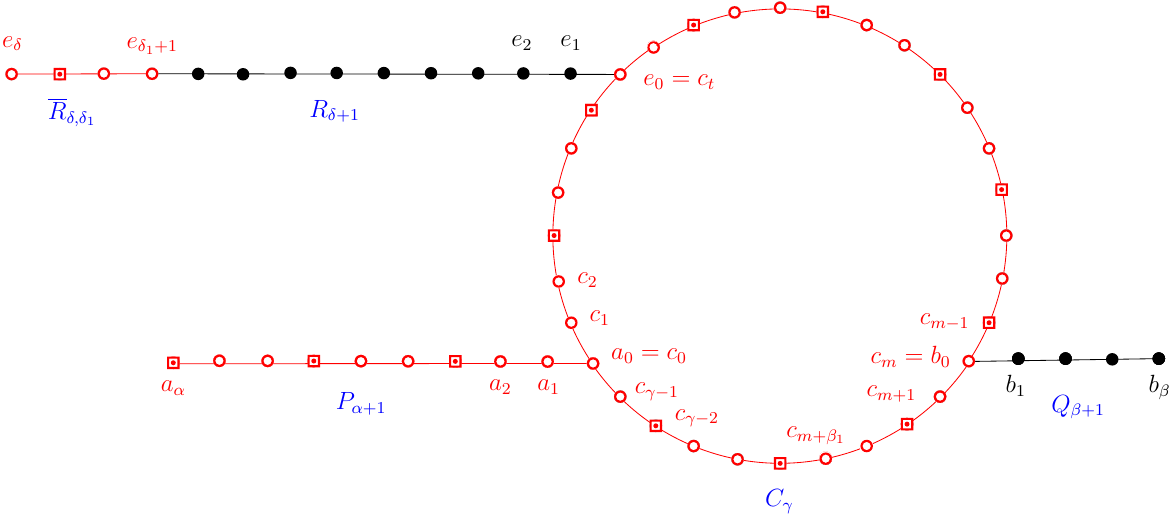}
 \caption{The circles and squares  represent the subgraph $P_{\alpha+1}\cup \overline{R}_{\delta,\delta_1}\cup C_\gamma$   and the squares represent the  set $M'_\gamma(c_0,\alpha,c_t,\delta,\delta_1)$ in this figure.}
    \label{fig:m2_11}
\end{figure}

\medskip
\noindent\textbf{Finding a multipacking in  $\bm{H_\gamma(c_0,\alpha,c_t,\delta,c_m,\beta)} $ according to Choice 3 :}\label{subec:multipacking_technique_3}
\smallskip

Next, we discuss a more general choice, similar to the last one. In the previous choice, we did not use the path $R_{\delta+1}$ to construct a multipacking. Here we are going to use $R_{\delta+1}$ also when it exists. Let   $\overline{R}_{\delta,\delta_1}=(e_{\delta_1+1},e_{\delta_1+2},\dots,e_{\delta})$, which is a path and a part of the path $R_{\delta+1}$. 
Now we choose every third vertex (with some exceptions) from the subgraphs $P_{\alpha+1}$, $\overline{R}_{\delta,\delta_1}$ and $C_\gamma$   to construct a set, called $M'_\gamma(c_0,\alpha,c_t,\delta,\delta_1)$(see Fig. \ref{fig:m2_11}). Formally, we define  \begin{align*}
M'_\gamma(c_0,\alpha,c_t,\delta,\delta_1)
   ={}&
   \{\,a_i : 0 \le i \le \alpha,\; i \equiv 0 \pmod{3}\,\} \\
   &\cup \{\,e_i : \delta_1+2 \le i \le \delta,\; i \equiv \delta_1+1 \pmod{3}\,\} \\
   &\cup \{\,c_i : 0 \le i \le \gamma-1,\; i \equiv 0 \pmod{3}\,\} \setminus \{c_0\}.
\end{align*}


Similarly, considering the subgraphs $Q_{\beta+1}$, $\overline{R}_{\delta,\delta_1}$ and $C_\gamma$, we define  \begin{align*}
M'_\gamma(c_m,\beta,c_t,\delta,\delta_1)
   ={}&
   \{\,b_i : 0 \le i \le \beta,\; i \equiv 0 \pmod{3}\,\} \\
   &\cup \{\,e_i : \delta_1+2 \le i \le \delta,\; i \equiv \delta_1+1 \pmod{3}\,\} \\
   &\cup \{\,c_i : 0 \le i \le \gamma-1,\; i \equiv m \pmod{3}\,\} \setminus \{c_m\}.
\end{align*}


These sets will be a multipacking of $G$ under the conditions stated in the following lemma.

\begin{lemma}  
\label{lem:delta1}
     Let $G$ be a cactus and $H_\gamma(c_0,\alpha,c_t,\delta,c_m,\beta)$  be a subgraph of $G$.  If $\delta_1=\big\lfloor\frac{\gamma}{2}\big\rfloor-d(c_0,c_t)$ and  $ \delta\geq \delta_1$, then $M'_\gamma(c_0,\alpha,c_t,\delta,\delta_1)$ is a multipacking of $G$ of size at least $\big\lfloor\frac{\gamma}{3}\big\rfloor+\big\lfloor\frac{\alpha}{3}\big\rfloor+ \big\lfloor\frac{\delta-\delta_1}{3}\big\rfloor-1$. Moreover, if $\delta_2=\big\lfloor\frac{\gamma}{2}\big\rfloor-d(c_m,c_t)$ and $ \delta\geq \delta_2$, then $M'_\gamma(c_m,\beta,c_t,\delta,\delta_2)$ is a multipacking of $G$ of size at least $\big\lfloor\frac{\gamma}{3}\big\rfloor+\big\lfloor\frac{\beta}{3}\big\rfloor+ \big\lfloor\frac{\delta-\delta_2}{3}\big\rfloor-1$.
\end{lemma}

\begin{proof} Let $H=H_\gamma(c_0,\alpha,c_t,\delta,c_m,\beta)$ and  $M'=M'_\gamma(c_0,\alpha,c_t,\delta,\delta_1)$.

Here $V(C_\gamma)=\{c_i:0\leq i\leq \gamma-1\}$. Suppose  $v\in V(C_\gamma)$.
Let $v=c_{x_1}$ for some $x_1$, where $0\leq x_1\leq \gamma-1$. Let  $d_1=d(c_0,v)$ and $d_2=d(c_t,v)$. Since $c_0,c_t\in V(C_\gamma)$,  $d(c_0,c_t)\leq \diam(C_\gamma)= \big\lfloor\frac{\gamma}{2}\big\rfloor$. Therefore $\big\lfloor\frac{\gamma}{2}\big\rfloor- d(c_0,c_t)\geq 0$, which implies that $ \delta_1\geq 0$. 

\vspace{0.22cm}

 \noindent \textit{\textbf{Case 1: }}
    $1\leq r \leq d_2+\delta_1$.


     In this case, $N_r[v]\cap V(P_{\alpha+1}\cup C_{\gamma})\neq \emptyset$, but $N_r[v] \cap V(\overline{R}_{\delta,\delta_1})=\emptyset$. Therefore,   $|N_r[v]\cap M'|\leq r$.

\vspace{0.22cm}
   
 \noindent \textit{\textbf{Case 2: }}  $ d_2+\delta_1<r< \big\lfloor\frac{\gamma}{2}\big\rfloor$. 


   In this case, $N_r[v]\cap V(C_{\gamma})\neq \emptyset$ and $N_r[v]\cap V(\overline{R}_{\delta,\delta_1}) \neq \emptyset$. The set $N_r[v]$ contains at most $2r+1$ vertices from the set $V(C_{\gamma})$  and at most $r-d_2-\delta_1$ vertices from the set $V(\overline{R}_{\delta,\delta_1})$. Suppose $r<d_1$. Note that $\delta_1+4=\min\{i:e_i\in V(\overline{R}_{\delta,\delta_1}) \}$. Therefore, $|N_r[v]\cap M'|\leq \big\lfloor\frac{1}{3}[2r+1+r-d_2-\delta_1]\big\rfloor\leq r$. If $r\geq d_1$, then the set $N_r[v]$ contains at most $r-d_1$ vertices from the set $V(P_{\alpha+1})$.
   
   Therefore,  \begin{align*}
\bigl|N_r[v]\cap M'\bigr|
&\le \frac{1}{3}\!\left[\,2r+1 + r-d_1 + r-d_2 - \delta_1 \right] \\[6pt]
&= \frac{1}{3}\!\left[\,2r+1 + r-d(c_0,v) + r-d(c_t,v)
      - \Bigl(\big\lfloor\tfrac{\gamma}{2}\big\rfloor - d(c_0,c_t)\Bigr) \right] \\[6pt]
&= \frac{1}{3}\!\left[\,3r 
      - \Bigl(\big\lfloor\tfrac{\gamma}{2}\big\rfloor - r - 1\Bigr)
      - \bigl(d(c_0,v) + d(c_t,v) - d(c_0,c_t)\bigr) \right] \\[6pt]
&\le r .
\end{align*}


 \vspace{0.22cm}

  \noindent \textit{\textbf{Case 3: }}  $  \big\lfloor\frac{\gamma}{2}\big\rfloor \leq r$.

   In this case, $N_r[v]\supseteq V(C_{\gamma})$, $N_r[v]\cap V(P_{\alpha+1})\neq \emptyset$ and $N_r[v]\cap V(\overline{R}_{\delta,\delta_1}) \neq \emptyset$. The set $N_r[v]$ contains all the vertices of the set $V(C_{\gamma})$. Moreover, it contains at most $r-d_1$ vertices from the set $V(P_{\alpha+1})$ and at most $r-d_2-\delta_1$ vertices from the set $V(\overline{R}_{\delta,\delta_1})$. Therefore, 
   
   \begin{align*}
\bigl|N_r[v]\cap M'\bigr|
&\le \frac{1}{3}\!\left[\,\gamma + r - d_1 + r - d_2 - \delta_1 \right] \\[6pt]
&= \frac{1}{3}\!\left[\,\gamma + r - d(c_0,v)
      + r - d(c_t,v)
      - \Bigl(\big\lfloor\tfrac{\gamma}{2}\big\rfloor - d(c_0,c_t)\Bigr) \right] \\[6pt]
&= \frac{1}{3}\!\left[\,2r
      + \bigl(\gamma - \big\lfloor\tfrac{\gamma}{2}\big\rfloor\bigr)
      - \bigl(d(c_0,v) + d(c_t,v) - d(c_0,c_t)\bigr) \right] \\[6pt]
&\le \frac{1}{3}\!\left[\,2r + \big\lceil\tfrac{\gamma}{2}\big\rceil \right] \\[6pt]
&\le r .
\end{align*}

\vspace{0.22cm}

   Therefore,  for $v\in V(C_\gamma)$ and any natural number $r$, $|N_r[v]\cap M'|\leq r$.


 
 


 Suppose $v\in V(P_{\alpha+1})$. Then any path that joins $v$ with a vertex in $V(C_\gamma)\cup V(Q_{\beta+1})\cup V(R_{\delta+1})$ passes through $a_0$, otherwise $G$ cannot be a cactus. By Observation \ref{obs:isometric_graph} we know that $H_\gamma(c_0,\alpha,c_t,\delta,c_m,\beta)$ is an isometric subgraph of $G$.  Therefore $|N_r[v]\cap M'|\leq r$ for all  $r\geq 1$. Similarly we can show that, if $v\in V(Q_{\beta+1})$ or $v \in V(R_{\delta+1})$, then $|N_r[v]\cap M'|\leq r$ for all  $r\geq 1$. Therefore $M'$ is a multipacking of $H$. So, $M'$ is a multipacking of $G$ by Lemma \ref{lem:multipacking_subgraph}. Moreover, $M'_\gamma(c_0,\alpha,c_t,\delta,\delta_1)$ is a multipacking of $G$ of size at least $\big\lfloor\frac{\gamma}{3}\big\rfloor+\big\lfloor\frac{\alpha}{3}\big\rfloor+ \big\lfloor\frac{\delta-\delta_1}{3}\big\rfloor-1$.
 
 Similarly, we can show that, if $\delta_2=\big\lfloor\frac{\gamma}{2}\big\rfloor-d(c_m,c_t)$ and $ \delta\geq \delta_2$, then $M'_\gamma(c_m,\beta,c_t,$ $\delta,\delta_2)$ is a multipacking of $G$ of size at least $\big\lfloor\frac{\gamma}{3}\big\rfloor+\big\lfloor\frac{\beta}{3}\big\rfloor+ \big\lfloor\frac{\delta-\delta_2}{3}\big\rfloor-1$.  
\end{proof}

Now we are ready to prove Lemma  \ref{lem:multipacking_radius_relation}. Here is a small observation before we start proving the lemma. We use this observation in the proof.

\begin{observation}
\label{obs:r/3} If $r$ is a positive integer, then 
    $\big\lfloor\frac{r}{3}\big\rfloor+\big\lfloor\frac{r-1}{3}\big\rfloor \geq \big\lfloor\frac{2r-1}{3}\big\rfloor-1$,   $\big\lfloor\frac{r}{3}\big\rfloor\geq \frac{r}{3}-\frac{2}{3}$,  and $\big\lfloor\frac{r}{2}\big\rfloor+\big\lceil\frac{r}{2}\big\rceil=r$.



\end{observation}

\begin{lemma} 
    \label{lem:multipacking_radius_relation}
  Let $G$ be a cactus with radius $\rad(G)$, then  $$\MP(G)\geq \frac{2}{3}\rad(G)-\frac{11}{3}.$$
\end{lemma}

\begin{proofsketch}     First, we present a very brief sketch of the proof for Lemma \ref{lem:multipacking_radius_relation}. Let $\rad(G)=r$ and $c$ be a central vertex of $G$. If $r=0$ or $1$, then $\MP(G)=1$. Therefore, for $r\leq 1$, we have $\MP(G)\geq \frac{2}{3}\rad(G)$. Now assume $r\geq 2$. Since $G$ has radius $r$, there is an isometric path $P$ in $G$ whose one endpoint  is $c$ and $l(P)=r$. Let $Q$ be a  longest isometric path in $G$ whose one endpoint  is $c$ and  $V(P)\cap V(Q)=\{c\}$. Say $l(Q)=r'$. Then $ r-1 \leq r'\leq  r$  by Lemma \ref{lem:disjoint_radial_path}.
    Let $P=(v_0,v_1, \dots , v_{r})$ and  $Q=(w_0,w_1, \dots , w_{r'})$, where $v_0=w_0=c$. From Observation \ref{obs:joining_path_atmost_1}, we know that $|X_{P,Q}|\leq 1$. 

     If  $|X_{P,Q}|=0$, we take the path $P\cup Q$ which is an isometric path of length $r+r'$ and choose every third vertex to the path to construct a multipacking of size at least $\MP(G)\geq \big\lceil\frac{2}{3}\rad(G)\big\rceil$ by Lemma \ref{lem:isometric_path}.

     Suppose $|X_{P,Q}|=1$. We have already discussed in this section that we can find a subgraph $H_\gamma(c_0,\alpha,c_t,\delta,$ $c_m,\beta)$ in $G$ and we use Lemma \ref{lem:gamma/2}, Lemma \ref{lem:multipacking_delta=0} and Lemma \ref{lem:delta1} under several cases to construct a multipacking of size at least $ \frac{2}{3}\rad(G)-\frac{11}{3}$.   
\end{proofsketch}

\begin{proof}[Proof of Lemma \ref{lem:multipacking_radius_relation}] Let $\rad(G)=r$ and $c$ be a central vertex of $G$. If $r=0$ or $1$, then $\MP(G)=1$. Therefore, for $r\leq 1$, we have $\MP(G)\geq \frac{2}{3}\rad(G)$. Now assume $r\geq 2$. Since $G$ has radius $r$, there is an isometric path $P$ in $G$ whose one endpoint  is $c$ and $l(P)=r$. Let $Q$ be a  largest isometric path in $G$ whose one endpoint  is $c$ and  $V(P)\cap V(Q)=\{c\}$. If $l(Q)=r'$, then $ r-1 \leq r'\leq  r$  by Lemma \ref{lem:disjoint_radial_path}.
    Let $P=(v_0,v_1, \dots , v_{r})$ and  $Q=(w_0,w_1, \dots , w_{r'})$ where $v_0=w_0=c$. From Observation \ref{obs:joining_path_atmost_1}, we know that $|X_{P,Q}|\leq 1$.   

\vspace{0.2cm}
\noindent
\textbf{Claim \ref{lem:multipacking_radius_relation}.1. } If  $|X_{P,Q}|=0$, then $\MP(G)\geq \big\lceil\frac{2}{3}\rad(G)\big\rceil$.

\begin{claimproof} Since  $|X_{P,Q}|=0$,  any path in $G$ that joins a vertex of $P$ with a vertex of $Q$ always passes through $c$. Therefore, $P\cup Q$ is an isometric path of length $r+r'$ and $|V(P\cup Q)|=r+r'+1$. By Lemma \ref{lem:isometric_path}, there is a multipacking in $G$ of size $\big\lceil\frac{r+r'+1}{3}\big\rceil$. Since $r'\geq r-1$, $\big\lceil\frac{r+r'+1}{3}\big\rceil \geq \big\lceil\frac{2r}{3}\big\rceil$.  Therefore, $\MP(G)\geq \big\lceil\frac{2r}{3}\big\rceil$.    
\end{claimproof} 

\medskip

Suppose $|X_{P,Q}|=1$. Let $X_{P,Q}=\{(v_i,w_j)\}$. Then $|X_{P,Q}(v_i,w_j)|=1$ by Observation \ref{obs:joining_path_atmost_1}. Let  $F_1\in X_{P,Q}(v_i,w_j)$ and $F_2=(v_i,v_{i-1},v_{i-2},\dots,v_1,v_0,$ $w_1,w_2,\dots,$ $w_{j-1},w_j)$. Therefore, $F_1$ and $F_2$ form a cycle $F_1\cup F_2$ of length  $ l(F_1)+l(F_2)$. Note that $l(F_1)\geq 1$, since $v_i\neq w_j$. Let $\gamma= l(F_1)+l(F_2)$ and  $C_\gamma=(c_0,c_1,c_2,\dots,c_{\gamma-2},$ $c_{\gamma-1},c_0) $ be a cycle of length $\gamma$.  Therefore,  $C_\gamma$ is  isomorphic to $ F_1\cup F_2$. For simplicity, we assume that $C_\gamma = F_1\cup F_2$. So, we can  assume that  $F_2=(c_0,c_1,c_2,\dots,$ $c_{m})$ and $F_1=(c_m,c_{m+1},\dots,c_{\gamma-1}, c_0)$.  
 Since $P$ and $Q$ are isometric paths,  $P'=(v_i,v_{i+1},\dots,v_{r})$ and $Q'=(w_j,w_{j+1},$ $\dots,w_{r'})$ are also isometric paths in $G$.  Therefore $C_\gamma\cup P'\cup Q'$ can be represented as 
 $H_\gamma(c_0,\alpha,c_t,0,c_m,\beta)$ or $ H_\gamma(c_0,\alpha,c_m,\beta)$, where $\alpha =l(P')$ and $\beta=l(Q')$. So, $C_\gamma\cup P'\cup Q'$ is an isometric subgraph of $G$ by Observation \ref{obs:isometric_graph}. Let $H=C_\gamma\cup P'\cup Q'$. For simplicity, we can assume that $P'=P_{\alpha+1}=(a_0,a_1,\dots,a_{\alpha+1})$ and $Q'=Q_{\beta+1}=(b_0,b_1,\dots,b_{\alpha+1})$. (See. Fig. \ref{fig:bfs3} and Fig. \ref{fig:bfs1}) 
 
 \vspace{0.2cm}
\noindent
\textbf{Claim \ref{lem:multipacking_radius_relation}.2. }  If  $|X_{P,Q}|=1$ and $l(F_1)\geq  l(F_2)$, then $\MP(G)\geq \big\lceil\frac{2}{3}\rad(G)\big\rceil$.

\begin{claimproof}
Since $G$ is a cactus and $l(F_1)\geq  l(F_2)$,  
$P\cup Q$ is an isometric path in $G$ by Observation \ref{obs:isometric_graph_path_F_1}.  Note that  $l(P\cup Q)=r+r'$ and $|V(P\cup Q)|=r+r'+1$. By Lemma \ref{lem:isometric_path},  there is a multipacking in $G$ of size $\big\lceil\frac{r+r'+1}{3}\big\rceil$. Now $r'\geq r-1$ $\implies$ $\big\lceil\frac{r+r'+1}{3}\big\rceil \geq \big\lceil\frac{2r}{3}\big\rceil$.  Therefore, $\MP(G)\geq \big\lceil\frac{2r}{3}\big\rceil$.  
\end{claimproof}

 \medskip
 
 Now assume $l(F_1)<  l(F_2)$. Therefore $F_1\cup P'\cup Q'$ is an isometric path by Observation \ref{obs:isometric_graph_path_F_1}. We know $l(F_2)=m$. Let $l(F_1)=x$   and $g=m+\big\lfloor\frac{x}{2}\big\rfloor$. Therefore $c_g$ is a central vertex of the path $F_1$. Let $S_r=\{u\in V(G):d(c_g,u)=r\}$.  Here  $S_r\neq \emptyset$, since $\rad(G)=r$. We know that $c$ is a central vertex of $G$ and  $c \in V(C_\gamma)$; more precisely, $c\in V(F_2)$. Therefore $c=c_k$ for some $k\in \{0,1,2,\dots,m\}$. Let $F_2^1=(c_0,c_1,\dots,c_k)$ and $F_2^2=(c_k,c_{k+1},\dots,c_m)$. Therefore $F_2=F_2^1\cup F_2^2$. Let  $l(F_2^1)=y$ and $l(F_2^2)=z$. Therefore, $F_1\cup F_2^1\cup F_2^2=C_\gamma$ and $x+y+z=\gamma$. Note that  $d(c_g,c_m)=d_H(c_g,c_m)=\big\lfloor\frac{x}{2}\big\rfloor$ and $d(c_g,c_0)=d_H(c_g,c_0)=x-\big\lfloor\frac{x}{2}\big\rfloor=\big\lceil\frac{x}{2}\big\rceil$.  We know that $l(P')=\alpha $ and $l(Q')=\beta$. Therefore $\alpha +y=l(P')+l(F^1_2)=l(P)=r$ and $\beta+z=l(Q')+l(F^2_2)=l(Q)=r'$. We use all these notations in the rest of the proof.

Now we want to show that,  if   $l(F_1)<  l(F_2)$, then $x$, $y$ and $z$ are upper bounded by $\big\lfloor\frac{\gamma}{2}\big\rfloor$. Now   
 $l(F_1)<  l(F_2)$ implies that $ x<y+z$, which implies $ x<\frac{x+y+z}{2}$, hence $ x<\frac{\gamma}{2}$ i.e., $ x\leq \big\lfloor\frac{\gamma}{2}\big\rfloor$.  Since $P$ is an isometric path, $F_2^1$ is a shortest path joining $c_0$ and $c$. Note that, the path  $ F_2^2\cup F_1$  also  joins $c_0$ and $c$. Therefore, $l(F_2^1)\leq l(F_2^2\cup F_1)$, hence $ l(F_2^1)\leq l(F_2^2)+l(F_1)$ and so $ y\leq z+x. $ But then $ y\leq \frac{x+y+z}{2}$ i.e., $ y\leq \big\lfloor\frac{\gamma}{2}\big\rfloor$. Similarly, we can show that  $z\leq \big\lfloor\frac{\gamma}{2}\big\rfloor$, since $F_2^2$ is a shortest path joining $c_0$ and $c$. 
 Therefore, if $l(F_1)<  l(F_2)$ , $\max\{x,y,z\}\leq \big\lfloor\frac{\gamma}{2}\big\rfloor$. 

\vspace{0.2cm}
\noindent
\textbf{Claim \ref{lem:multipacking_radius_relation}.3. }  If  $|X_{P,Q}|=1$, $l(F_1)<  l(F_2)$ and $S_r\cap P'\neq \emptyset $, then $\MP(G)\geq \big\lfloor\frac{2}{3}\rad(G)-\frac{1}{3}\big\rfloor-3$.

\begin{claimproof}
   Let $u\in S_r\cap P'$. Let $\alpha_1=x-1$ and  $\beta_1=z-1$.  Since $F_1\cup P' \cup Q' $ is an isometric path of $G$,   $\alpha+\alpha_1+1=x+\alpha =  d(c_m,v_r)\geq d(c_g,u)=r $. Here $\beta+\beta_1+1=z+\beta=r'\geq r-1$. We have $\alpha_1=x-1\leq \big\lfloor\frac{\gamma}{2}\big\rfloor-1$ and $\beta_1=z-1\leq \big\lfloor\frac{\gamma}{2}\big\rfloor-1$, since $\max\{x,y,z\}\leq \big\lfloor\frac{\gamma}{2}\big\rfloor$. We have shown that $H$ can be represented as 
 $ H_\gamma(c_0,\alpha,c_t,0,c_m,\beta) $. Therefore, there is a multipacking of $G$ of size at least 
  $ \big\lfloor\frac{\alpha+\alpha_1+1}{3}\big\rfloor+\big\lfloor\frac{\beta+\beta_1+1}{3}\big\rfloor -2$ by Lemma \ref{lem:gamma/2}.  
 Here $  \big\lfloor\frac{\alpha+\alpha_1+1}{3}\big\rfloor+\big\lfloor\frac{\beta+\beta_1+1}{3}\big\rfloor -2\geq \big\lfloor\frac{r}{3}\big\rfloor+\big\lfloor\frac{r-1}{3}\big\rfloor -2\geq \big\lfloor\frac{2r-1}{3}\big\rfloor-3$ by Observation \ref{obs:r/3}. 
\end{claimproof}

\vspace{0.2cm}
\noindent
\textbf{Claim \ref{lem:multipacking_radius_relation}.4. }  If  $|X_{P,Q}|=1$, $l(F_1)<  l(F_2)$ and $S_r\cap Q'\neq \emptyset $, then $\MP(G)\geq 2\big\lfloor\frac{\rad(G)}{3}\big\rfloor-2$.

\begin{claimproof}
Let $u\in S_r\cap Q'$. Let $\alpha_1=y-1$ and   $\beta_1=x-1$.  Since $F_1\cup P' \cup Q' $ is an isometric path of $G$, $\beta+\beta_1+1=x+\beta=d(c_0,w_{r'})\geq d(c_g,u)=r $. Moreover, $\alpha+\alpha_1+1=y+\alpha =  r$.  We have $\alpha_1=y-1\leq \big\lfloor\frac{\gamma}{2}\big\rfloor-1$ and $\beta_1=x-1\leq \big\lfloor\frac{\gamma}{2}\big\rfloor-1$, since $\max\{x,y,z\}\leq \big\lfloor\frac{\gamma}{2}\big\rfloor$. We have shown that $H$ can be represented as 
 $ H_\gamma(c_0,\alpha,c_t,0,c_m,\beta) $. Therefore,  there is a multipacking of $G$ of size at least 
  $ \big\lfloor\frac{\alpha+\alpha_1+1}{3}\big\rfloor+\big\lfloor\frac{\beta+\beta_1+1}{3}\big\rfloor -2$ by Lemma \ref{lem:gamma/2}. 
 Here $  \big\lfloor\frac{\alpha+\alpha_1+1}{3}\big\rfloor+\big\lfloor\frac{\beta+\beta_1+1}{3}\big\rfloor -2\geq \big\lfloor\frac{r}{3}\big\rfloor+\big\lfloor\frac{r}{3}\big\rfloor -2\geq 2\big\lfloor\frac{r}{3}\big\rfloor-2$.  
\end{claimproof}

\vspace{0.2cm}
\noindent
\textbf{Claim \ref{lem:multipacking_radius_relation}.5. } If  $|X_{P,Q}|=1$, $l(F_1)<  l(F_2)$ and $S_r\cap C_\gamma \neq \emptyset $, then $\MP(G)\geq \big\lceil\frac{2}{3}\rad(G)\big\rceil$.

\begin{claimproof}
    We know that $ H_\gamma(c_0,0,c_t,0,c_m,0)= C_\gamma$.  Therefore $C_\gamma$ is an isometric subgraph of $G$ by Observation \ref{obs:isometric_graph}. Now $S_r\cap C_\gamma \neq \emptyset \implies r\leq \big\lceil\frac{\gamma}{2}\big\rceil\implies 2r\leq \gamma$. Let $M=\{c_i:0\leq i\leq \gamma-1,i\equiv 0 \text{ (mod $3$)}\}$. Note that, $M$ is a multipacking of $C_\gamma$. Therefore, $M$ is a mutipacking of $G$ by Lemma \ref{lem:multipacking_subgraph}.  Here $|M|=\big\lceil\frac{\gamma}{3}\big\rceil\geq \big\lceil\frac{2r}{3}\big\rceil$, since $2r\leq \gamma$.   
\end{claimproof}

\vspace{0.2cm}
\noindent
\textbf{Claim \ref{lem:multipacking_radius_relation}.6. } If  $|X_{P,Q}|=1$, $l(F_1)<  l(F_2)$ and $S_r\cap V(H) =\emptyset $, then $\MP(G)\geq \frac{2}{3}\rad(G)-\frac{11}{3}$.

\begin{claimproof} Let $u\in S_r$. Therefore, $u\notin V(H)$. We know that $g=m+\big\lfloor\frac{x}{2}\big\rfloor$. Let $R$ be a shortest path joining $c_g$ and $u$. Therefore, $R$ is an isometric path of $G$. Let $R=(u_0,u_1,\dots,u_r)$ where $c_g=u_0$ and $u=u_r$. Suppose $h=\max\{i:u_i\in V(H)\}$.   Let $R'=(u_{h+1},u_{h+2},\dots,u_r)$. 

First, we show that  $u_h\notin V(F_1)$. Suppose $u_h\in V(F_1)$, so $u_h\in \{c_m,c_{m+1},$ $\dots,c_{\gamma-1},c_0\}$. Therefore, $u_h=c_{g'}$ for some $g'\in \{m,m+1,\dots,\gamma-1,0\}$. First consider $u_h\in \{c_{g+1},c_{g+2},\dots,$ $c_{\gamma-1},c_0\}$. We know that $c$ is a central vertex of $G$ and $c=c_k$ for some $k\in \{0,1,2,\dots,m\}$. Suppose $(c_k,c_{k-1},\dots,c_0,$ $c_{\gamma-1},\dots,c_{g'})$ is a shortest path joining $c_k(=c)$ and $c_{g'}(=u_h)$.  We have shown that $H$ is an isometric subgraph of $G$. Therefore, the distance between two vertices in $H$ is the distance in $G$. 
Since $S_r\cap V(H)= \emptyset$, we have $d(c_g,v_r)<r$. Therefore,  $d(c_g,v_r)<r=d(c_g,u)$. Now $ d(c_g,v_r)<d(c_g,u) \implies d(c_g,u_h)+d(u_h,c_0)+d(c_0,v_r)<d(c_g,u_h)+d(u_h,u)\implies d(u_h,c_0)+d(c_0,v_r)<d(u_h,u)\implies d(c_0,v_r)<d(u_h,u)$.   Since $c$ is a central vertex of $G$ and $S_r\cap V(H)= \emptyset$, we have $r\geq d(c,u)=d(c,c_0)+d(c_0,u_h)+d(u_h,u)>d(c,c_0)+d(c_0,u_h)+d(c_0,v_r)=d(c,v_r)+d(c_0,u_h)=r+d(c_0,u_h)$, which implies $ d(c_0,u_h)<0$. This is a contradiction. Now assume $(c_k,c_{k+1},\dots,c_m,c_{m+1},\dots,c_{g'})$ is a shortest path joining $c$ and $u_h$. Since $c$ is a central vertex of $G$ and $S_r\cap V(H)= \emptyset$, we have $r\geq d(c,u)=d(c,c_g)+d(c_g,u)=d(c,c_g)+r$, which implies $ d(c,c_g)=0$, hence $ c=c_g$. Therefore $r=d(c,v_{r})=d(c_g,v_r)$, which implies $ v_r\in S_r\cap V(H)$, which is a contradiction. 
Therefore, $u_h\notin \{c_{g+1},c_{g+2},\dots,c_{\gamma-1},c_0\}$. Similarly we can show that $u_h\notin \{c_{m},c_{m+1},\dots,c_{g-1},c_g\}$. So, $u_h\notin V(F_1)$.

If $u_h\in V(P')$, then $d(c_g,u)=r>d(c_g,v_r)$, which implies $ d(c_g,u_h)+d(u,u_h)>d(c_g,u_h)+d(u_h,v_r)$, hence $d(u,u_h)>d(u_h,v_r)$, and so $ d(u,u_h)+d(c,u_h)>d(u_h,v_r)+d(c,u_h)$, which implies $d(c,u)>d(c,v_r)$, therefore $ d(c,u)>r$, which is a contradiction, since $c$ is a central vertex of the graph $G$ having radius $r$. Therefore, $u_h\notin V(P')$. Similarly, we can show that $u_h\notin V(Q')$.



Therefore $u_h\in V(C_\gamma)\setminus V(F_1)= \{c_1,c_2,\dots,c_{m-1}\}$.  Since $G$ is a cactus, every edge belongs to at most one cycle. This implies that $u_i\in V(C_\gamma)$ for all $0\leq i\leq h$.  Let $u_h=c_t$. 

Suppose $x\geq \alpha $. Then $x+y+z+\beta\geq \alpha +y+z+\beta\geq r+r'\geq 2r-1$.  By Lemma \ref{lem:multipacking_delta=0},  $M'_\gamma(c_m,\beta)$ is a multipacking of $G$ of size at least $\big\lfloor\frac{\gamma}{3}\big\rfloor+\big\lfloor\frac{\beta}{3}\big\rfloor-1$. Therefore  $\big\lfloor\frac{\gamma}{3}\big\rfloor+\big\lfloor\frac{\beta}{3}\big\rfloor-1\geq \frac{\gamma}{3}-\frac{2}{3}+\frac{\beta}{3}-\frac{2}{3}-1=\frac{x+y+z+\beta}{3}-\frac{7}{3}\geq \frac{2r-1}{3}-\frac{7}{3}$. 

Suppose $x\geq \beta$. Then $x+y+z+\alpha \geq \alpha +y+z+\beta\geq r+r'\geq 2r-1$.  By Lemma \ref{lem:multipacking_delta=0},  $M'_\gamma(c_0,\alpha )$ is a multipacking of $G$ of size at least $\big\lfloor\frac{\gamma}{3}\big\rfloor+\big\lfloor\frac{\alpha }{3}\big\rfloor-1$ (See Fig. \ref{fig:multipacking_choice_2}). Therefore  $\big\lfloor\frac{\gamma}{3}\big\rfloor+\big\lfloor\frac{\alpha }{3}\big\rfloor-1\geq \frac{\gamma}{3}-\frac{2}{3}+\frac{\alpha }{3}-\frac{2}{3}-1=\frac{x+y+z+\beta}{3}-\frac{7}{3}\geq \frac{2r-1}{3}-\frac{7}{3}$. 

Now assume $x<\min\{\alpha ,\beta\}$.

Note that $S_r\cap V(H)=\emptyset$ implies that $ S_r\cap V(C_\gamma)=\emptyset$. Therefore, there is no vertex on the cycle $C_\gamma$ which is at  distance $r$ from the vertex $c_g$. Therefore, $r>\big\lfloor\frac{\gamma}{2}\big\rfloor$.

Now we split the remainder of the proof  into two cases.

\vspace{0.22cm} 

 \noindent \textit{\textbf{Case 1: }} $ \big\lfloor\frac{\gamma}{2}\big\rfloor<r\leq \big\lfloor\frac{\gamma}{2}\big\rfloor+\big\lfloor\frac{x}{2}\big\rfloor$.

 Consider the set $M'_\gamma(c_0,\alpha )$. This a multipacking of $G$ of size $\big\lfloor\frac{\gamma}{3}\big\rfloor+\big\lfloor\frac{\alpha }{3}\big\rfloor-1$ by Lemma \ref{lem:multipacking_delta=0}. Now

 \begin{align*}
\big\lfloor\tfrac{\gamma}{3}\big\rfloor
  + \big\lfloor\tfrac{\alpha}{3}\big\rfloor - 1
&\;\ge\;
\big\lfloor\tfrac{\gamma}{3}\big\rfloor
  + \big\lfloor\tfrac{x}{3}\big\rfloor - 1 \\[4pt]
&\;\ge\;
\tfrac{\gamma}{3} - \tfrac{2}{3}
  + \tfrac{x}{3} - \tfrac{2}{3} - 1 \\[4pt]
&\;\ge\;
\tfrac{2}{3}\!\left(\tfrac{\gamma}{2}+\tfrac{x}{2}\right)
  - \tfrac{7}{3} \\[4pt]
&\;\ge\;
\tfrac{2}{3}r - \tfrac{7}{3}.
\end{align*}


\vspace{0.22cm} 

 \noindent \textit{\textbf{Case 2: }} $ \big\lfloor\frac{\gamma}{2}\big\rfloor+\big\lfloor\frac{x}{2}\big\rfloor<r$.

Now consider  
 $\delta=r-d(c_g,c_t)$, $\delta_1=\big\lfloor\frac{\gamma}{2}\big\rfloor-d(c_0,c_t)$ and $\delta_2=\big\lfloor\frac{\gamma}{2}\big\rfloor-d(c_m,c_t)$.  Since $c_0,c_t\in V(C_\gamma)$,  $d(c_0,c_t)\leq \diam(C_\gamma)= \big\lfloor\frac{\gamma}{2}\big\rfloor$. Therefore $\big\lfloor\frac{\gamma}{2}\big\rfloor- d(c_0,c_t)\geq 0\implies \delta_1\geq 0$. Similarly, we can show that $\delta_2\geq 0$. Now $\delta-\delta_1=r-d(c_g,c_t)-\big\lfloor\frac{\gamma}{2}\big\rfloor+d(c_0,c_t)\geq r-d(c_g,c_0)-d(c_0,c_t)-\big\lfloor\frac{\gamma}{2}\big\rfloor+d(c_0,c_t)\geq r-\big\lceil\frac{x}{2}\big\rceil-\big\lfloor\frac{\gamma}{2}\big\rfloor\geq 0$. Therefore $\delta \geq \delta_1$ and $\delta - \delta_1\geq r-\big\lceil\frac{x}{2}\big\rceil-\big\lfloor\frac{\gamma}{2}\big\rfloor\geq 0$. Similarly, we can show that $\delta\geq \delta_2$ and $\delta - \delta_2\geq r-\big\lceil\frac{x}{2}\big\rceil-\big\lfloor\frac{\gamma}{2}\big\rfloor\geq 0$.  Now we are ready to use Lemma \ref{lem:delta1} to find multipacking in $G$.


\newpage

First assume that, $z\geq y$. 
By Lemma \ref{lem:delta1},   $M'_\gamma(c_0,\alpha ,c_t,\delta,\delta_1)$ is a multipacking of $G$ of size at least $\big\lfloor\frac{\gamma}{3}\big\rfloor+\big\lfloor\frac{\alpha }{3}\big\rfloor+ \big\lfloor\frac{\delta-\delta_1}{3}\big\rfloor-1$.  

Now, \begin{align*}
\Big\lfloor\frac{\gamma}{3}\Big\rfloor
 + \Big\lfloor\frac{\alpha}{3}\Big\rfloor
 + \Big\lfloor\frac{\delta-\delta_1}{3}\Big\rfloor - 1
&\ge 
   \frac{\gamma}{3} - \frac{2}{3}
 + \frac{\alpha}{3} - \frac{2}{3}
 + \frac{\delta-\delta_1}{3} - \frac{2}{3} - 1
\\[6pt]
&=
\frac{\gamma}{3}
 + \frac{\alpha}{3}
 + \frac{1}{3}\!\left(r - \Big\lceil\tfrac{x}{2}\Big\rceil - \Big\lfloor\tfrac{\gamma}{2}\Big\rfloor\right)
 - 3
\\[6pt]
&\ge 
\frac{\gamma}{3}
 + \frac{\alpha}{3}
 + \frac{1}{3}\!\left(r - \frac{x}{2} - 1 - \frac{\gamma}{2}\right)
 - 3
\\[6pt]
&=
\frac{1}{3}\!\left( r + \frac{\gamma}{2} - \frac{x}{2} + \alpha \right)
 - \frac{10}{3}
\\[6pt]
&\ge
\frac{1}{3}\!\left( r + \frac{\gamma}{2} - \frac{x}{2} + r - y \right)
 - \frac{10}{3}
\\[6pt]
&\ge
\frac{1}{3}\!\left( 2r + \frac{x+y+z}{2} - \frac{x}{2} - y \right)
 - \frac{10}{3}
\\[6pt]
&\ge
\frac{1}{3}\!\left( 2r + \frac{z - y}{2} \right)
 - \frac{10}{3}
\\[6pt]
&\ge
\frac{2}{3}r - \frac{10}{3}.
\end{align*}


\newpage

Suppose $z<y$. By Lemma \ref{lem:delta1},   $M'_\gamma(c_m,\beta,c_t,\delta,\delta_2)$ is a multipacking of $G$ of size at least $\big\lfloor\frac{\gamma}{3}\big\rfloor+\big\lfloor\frac{\beta}{3}\big\rfloor+ \big\lfloor\frac{\delta-\delta_2}{3}\big\rfloor-1$.  

Now, \begin{align*}
\Big\lfloor\frac{\gamma}{3}\Big\rfloor
 + \Big\lfloor\frac{\beta}{3}\Big\rfloor
 + \Big\lfloor\frac{\delta-\delta_2}{3}\Big\rfloor - 1
&\ge
   \frac{\gamma}{3} - \frac{2}{3}
 + \frac{\beta}{3} - \frac{2}{3}
 + \frac{\delta-\delta_2}{3} - \frac{2}{3}
 - 1
\\[6pt]
&=
\frac{\gamma}{3}
 + \frac{\beta}{3}
 + \frac{1}{3}\!\left( r - \Big\lceil\tfrac{x}{2}\Big\rceil - \Big\lfloor\tfrac{\gamma}{2}\Big\rfloor \right)
 - 3
\\[6pt]
&\ge
\frac{\gamma}{3}
 + \frac{\beta}{3}
 + \frac{1}{3}\!\left( r - \frac{x}{2} - 1 - \frac{\gamma}{2} \right)
 - 3
\\[6pt]
&=
\frac{1}{3}\!\left( r + \frac{\gamma}{2} - \frac{x}{2} + \beta \right)
 - \frac{10}{3}
\\[6pt]
&\ge
\frac{1}{3}\!\left( r + \frac{\gamma}{2} - \frac{x}{2} + r' - z \right)
 - \frac{10}{3}
\\[6pt]
&\ge
\frac{1}{3}\!\left( r + r - 1 + \frac{x+y+z}{2} - \frac{x}{2} - z \right)
 - \frac{10}{3}
\\[6pt]
&\ge
\frac{1}{3}\!\left( 2r + \frac{y - z}{2} \right)
 - \frac{11}{3}
\\[6pt]
&\ge
\frac{2}{3}r - \frac{11}{3}.
\end{align*}


\end{claimproof}

In each case, $G$ has a multipacking of size at least $ \frac{2}{3}\rad(G)-\frac{11}{3}$. Therefore, $\MP(G)\geq \frac{2}{3}\rad(G)-\frac{11}{3}$.  
\end{proof}

\begin{theorem}[\cite{teshima2012broadcasts,erwin2001cost}]  \label{thm:mpGleqgammabG}  If $G$ is a connected graph of order at least 2  having radius $ \rad(G) $, multipacking number  $\MP(G) $, broadcast domination number $ \gamma_{b}(G) $ and domination number $\gamma(G)$, then $$ \MP(G)\leq \gamma_{b}(G) \leq \min\{\gamma (G),\rad(G)\}.$$
\end{theorem}

\begin{proof}[Proof of Theorem \ref{thm:multipacking_broadcast_relation}]
    We have $\gamma_b(G)\leq \rad(G)$ from Theorem \ref{thm:mpGleqgammabG}. From Lemma 
\ref{lem:multipacking_radius_relation}, we get $\MP(G)\geq \frac{2}{3}\rad(G)-\frac{11}{3}$. Therefore, $ \frac{2}{3}\cdot\gamma_b(G)-\frac{11}{3}\leq \MP(G)\implies \gamma_b(G)\leq \frac{3}{2}\MP(G)+\frac{11}{2}$.
\end{proof}

\section{An approximation algorithm to find multipacking in cactus graphs}\label{sec:approximation_algorithm_to_find_Multipacking}

We provide the following algorithm to approximate multipacking of a cactus graph. This algorithm is a direct implementation of the proof of Theorem \ref{thm:multipacking_broadcast_relation}.

\smallskip
\noindent\textbf{Approximation algorithm: } Since $G$ is a cactus graph,  we can find a central vertex $c$ and  an isometric path $P$ of length $r$ whose one endpoint is $c$ in $O(n)$-time, where $n$ is the number of vertices of $G$. After that, we can find another isometric path $Q$ having length $r-1$ or $r$ whose one endpoint is $c$ and $V(P)\cap V(Q)=\{c\}$.  Then the flow of the proof of 
Lemma  \ref{lem:multipacking_radius_relation} provides an $O(n)$-time algorithm to construct a multipacking of size at least $\frac{2}{3}\rad(G)-\frac{11}{3}$. Thus we can find  a multipacking of $G$ of size at least $\frac{2}{3}\rad(G)-\frac{11}{3}$ in $O(n)$-time.  Moreover,  $\MP(G)\geq \frac{2}{3}\rad(G)-\frac{11}{3}\geq \frac{2}{3}\MP(G)-\frac{11}{3}$ by Theorem \ref{thm:mpGleqgammabG}. This implies the existence of a $(\frac{3}{2}+o(1))$-factor approximation algorithm for
the multipacking problem on cactus graph. Therefore, we have the following.

\MultipackingAlgorithm*

\section{Unboundedness of the gap between broadcast domination and multipacking number of cactus and AT-free graphs}\label{sec:Unboundedness_of_the_gap_between_Broadcast_domination_and_Multipacking}

In this section, we prove that the difference between the broadcast domination number and the multipacking number of the cactus and AT-free graphs can be arbitrarily large.  
 Moreover, we make a connection with the \emph{fractional} versions of the two concepts dominating broadcast and multipacking for cactus and AT-free graphs.

To prove that the difference $ \gamma_{b}(G) -  \MP(G) $ can be arbitrarily large, we construct the graph $G_k$ as follows. Let $A_i=(a_i,b_i,c_i,d_i,e_i,a_i)$ be a 5-cycle  for each  $i=1,2,\dots,3k$. We form $G_k $ by joining $b_i$ to $e_{i+1}$ for each $i=1,2,\dots,3k-1$ (See Fig. \ref{fig:pentagon}). This  graph is a  cactus graph (and AT-free graph). We show that $\MP(G_k)=3k$ and $\gamma_b(G_k)=4k$.

\begin{figure}[ht]
    \centering

\includegraphics[width=\textwidth]{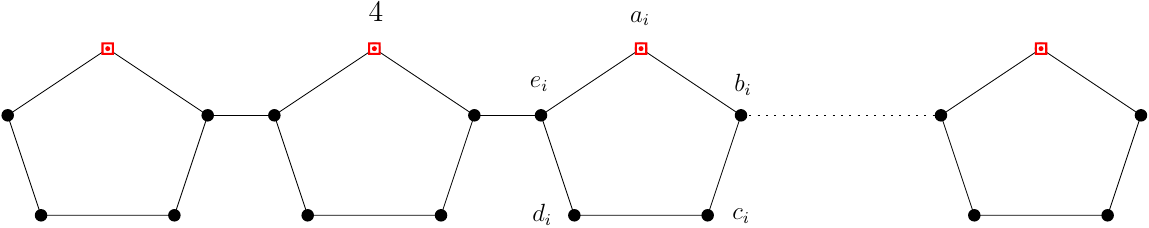}
    \caption{The graph $G_k$  with $\gamma_b(G_k)=4k$ and $\MP(G_k)=3k$. The set $\{a_i:1\leq i\leq 3k\}$ is a maximum multipacking of $G_k$.} 
    \label{fig:pentagon}
\end{figure}

\begin{lemma} \label{lem:mpGk=3k}  For each positive integer $k$, $\MP(G_{k})=3k$.
\end{lemma}

\begin{proof} The path $P=(e_1,a_1,b_1,e_2,a_2,b_2,\dots,e_{3k},a_{3k},b_{3k})$ is a diametral path of $G_k$ (Fig.\ref{fig:pentagon}). This implies $P$ is an isometric path of $G_k$ having the length  $l(P)=3\cdot 3k-1$. By Lemma \ref{lem:isometric_path}, the set consisting of every third vertex on this path is a multipacking of size $\big\lceil\frac{3\cdot 3k-1}{3}\big\rceil=3k$. Therefore, $\MP(G_k)\geq 3k$. Note that,  $\diam(A_i)=2$ for each $i$. Therefore, any multipacking of $G_k$ can contain at most one vertex of $A_i$ for each $i$. So, $\MP(G_k)\leq 3k$. Hence $\MP(G_k)= 3k$.
\end{proof}




\noindent\textbf{Fractional multipacking: }
R. C. Brewster and  L. Duchesne \cite{brewster2013broadcast} introduced fractional multipacking in 2013 (also see \cite{teshima2014multipackings}).  Suppose $G$ is a graph with $V(G)=\{v_1,v_2,v_3,\dots,v_n\}$ and $w:V(G)\rightarrow [0,\infty)$  is a function. So, $w(v)$ is a weight on a vertex $v\in V(G)$.  For $S\subseteq V(G)$, let $w(S)=\sum_{u\in S}w(u)$.  We say   $w$ is a \textit{fractional multipacking} of $G$,  if $w( N_r[v])\leq r$ for each vertex $ v \in V(G) $ and for every integer $ r \geq 1 $. The \textit{fractional multipacking number} of $ G $ is the  value $\displaystyle \max_w w(V(G)) $ where $w$ is any fractional multipacking and it
	is denoted by $ \MP_f(G) $. A \textit{maximum fractional multipacking} is a fractional multipacking $w$  of a graph $ G $ such that	$ w(V(G))=\MP_f(G)$. If $w$ is a fractional multipacking, we define   a vector $y$ with the entries $y_j=w(v_j)$.  So,  \begin{center}
	    $\MP_f(G)=\max \{y.\mathbf{1} :  yA\leq c, y_{j}\geq 0\}.$
	\end{center} This is a  linear program which is the dual of the linear program   $\min \{c.x :  Ax\geq \mathbf{1}, x_{i,k}\geq 0\}$. Let,  
	    $\gamma_{b,f}(G)=\min \{c.x : Ax\geq \mathbf{1}, x_{i,k}\geq 0\}.$
	
 Using the strong duality theorem for linear programming, we can say that \begin{center}
     $\MP(G)\leq \MP_f(G)= \gamma_{b,f}(G)\leq \gamma_{b}(G).$
 \end{center}


\begin{figure}[ht]
    \centering
    \includegraphics[width=\textwidth]{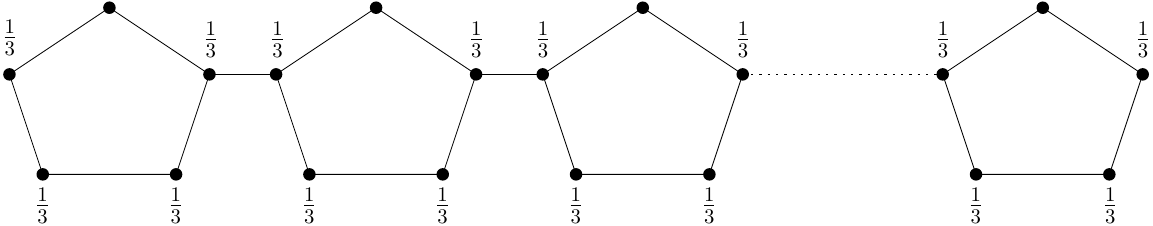}
    \caption{The graph $G_k$ with $\MP_f(G_k)=4k$.}
    \label{fig:pentagon_frac}
\end{figure}

\begin{lemma}  \label{lem:mpfGk=gammabGk=4k}
If $k$ is a positive integer, then $\MP_f(G_{k})=\gamma_b(G_{k})= 4k$.
\end{lemma}

\begin{proof}  
We define a function $w: V(G_k)\rightarrow [0,\infty)$  where $w(b_{i})=w(c_{i})=w(d_{i})=w(e_{i})=\frac{1}{3}$  for each $i=1,2,3,\dots,3k$ (Fig. \ref{fig:pentagon_frac}).  So, $w(G_k)=4k$. 
 We want to show that $w$ is a fractional multipacking of $G_{k}$. We have to prove that $w(N_r[v])\leq r$ for each vertex $ v \in V(G_{k}) $ and for every integer $ r \geq 1 $. We prove this statement using induction on $r$.	It can be checked that $w(N_r[v])\leq r$ for each vertex $ v \in V(G_{k}) $ and for each  $ r \in \{1,2,3,4\} $. Now assume that the statement is true for $r=s$, we want to prove that it is true for $r=s+4$. Observe that, $w(N_{s+4}[v]\setminus N_{s}[v])\leq 4$, $\forall v\in V(G_{k})$. Therefore,  $w(N_{s+4}[v])\leq w(N_{s}[v])+4\leq s+4$. Hence the statement is true, and $w$ is a fractional multipacking of $G_{k}$. Therefore, $\MP_f(G_{k})\geq 4k$.

Define a broadcast ${f}$ on $G_k$ as ${f}(v)=
    \begin{cases}
        4 & \text{if } v=a_i  \text{ and } i\equiv 2 \text{ (mod $3$)}  \\
        0 & \text{otherwise. }
    \end{cases}$.\\
Here  ${f}$ is an efficient dominating broadcast and $\sum_{v\in V(G_k)}{f}(v)=4k$. So, $\gamma_b(G_k)\leq 4k$, for all $ k\in \mathbb{N}$. By the strong duality theorem, $4k\leq \MP_f(G_k)= \gamma_{b,f}(G_k)\leq \gamma_{b}(G_k)\leq 4k$. Therefore, $\MP_f(G_{k})=\gamma_b(G_{k})= 4k$.  
\end{proof}

Lemma \ref{lem:mpGk=3k}  and  Lemma \ref{lem:mpfGk=gammabGk=4k} imply the following results.

\multipackingbroadcastgapecactus*


\gammabDIFFGmpGcactus*





\begin{corollary} \label{cor:mpfG-mpG}  The difference $\MP_f(G)-\MP(G)$ can be arbitrarily large for cactus graphs (and for AT-free graphs).
\end{corollary}



\section{$1/2$ - hyperbolic graphs}\label{sec:1/2-Hyperbolic graphs}

In this section, we show that the family of graphs $G_k$ (Fig. \ref{fig:pentagon}) is $\frac{1}{2}$-hyperbolic using a characterization of $\frac{1}{2}$-hyperbolic graphs. Using this fact, we  show that the difference $ \gamma_{b}(G) -  \MP(G) $ can be arbitrarily large for $\frac{1}{2}$-hyperbolic graphs. 
If the length of a shortest path $P$ between two vertices $x$ and $y$ of a cycle $C$ of $G$ is  smaller than the distance between $x$ and $y$ measured along $C$, then $P$ is called a \textit{bridge} of $C$. In $G$, a cycle $C$ is called \textit{well-bridged}  if for any vertex $x \in C$ there exists a bridge from $x$ to some vertex of $C$ or if the two neighbors of $x$ in $C$ are adjacent.

\begin{figure}[ht]
    \centering
    \includegraphics[height=6cm]{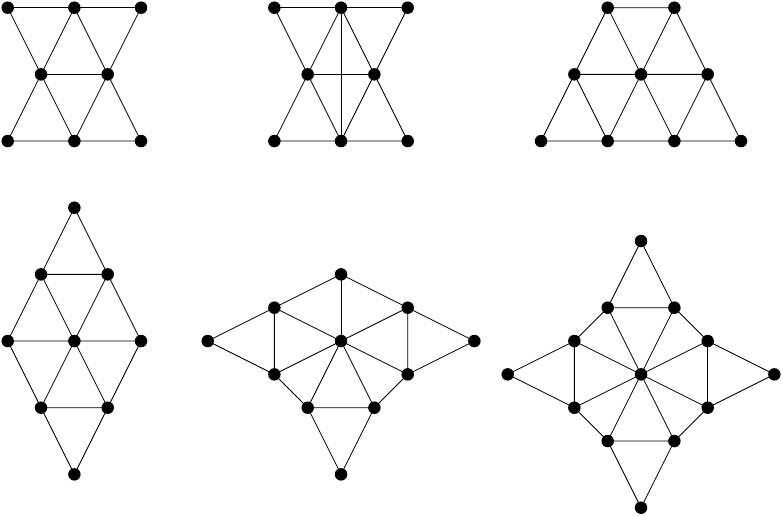}
    \caption{Forbidden isometric subgraphs for $\frac{1}{2}$-hyperbolic graphs} 
    \label{fig:hyperbolic_forbidden_graphs}
\end{figure}

\begin{theorem}[\cite{bandelt20031}] \label{thm:1/2-Hyperbolic_graphs_classification}
    A graph $G$ is $\frac{1}{2}$-hyperbolic if and only if all cycles $C_n$, $n \neq 5$, of $G$ are well-bridged and none of the graphs in Figure \ref{fig:hyperbolic_forbidden_graphs} occur as isometric subgraphs of $G$.
\end{theorem}

Theorem \ref{thm:1/2-Hyperbolic_graphs_classification} yields the following.

\begin{lemma}\label{lem:pentagon_half_hyperbolic}
    The family of graphs $G_k$ (Fig. \ref{fig:pentagon}) is $\frac{1}{2}$-hyperbolic.
    
\end{lemma}

Lemma \ref{lem:mpGk=3k}, \ref{lem:mpfGk=gammabGk=4k} and \ref{lem:pentagon_half_hyperbolic} imply the following theorem.

\hypermultipackingbroadcastgape*

\hypergammabDIFFGmpG*


From Theorem \ref{thm:1/2-Hyperbolic_graphs} and Lemma \ref{lem:mpfGk=gammabGk=4k}, we have the following.

\begin{corollary} \label{cor:hyper_mpfG-mpG}  The difference $\MP_f(G)-\MP(G)$ can be arbitrarily large for $\frac{1}{2}$-hyperbolic graphs.
\end{corollary}

\section{Conclusion}\label{sec:conclusion_cactus}

We have shown that the  bound $\gamma_b(G)\leq 2\MP(G)+3$ for general graphs $G$ can be improved to $\gamma_b(G)\leq \frac{3}{2}\MP(G)+\frac{11}{2}$ for cactus graphs. Additionally, we have given a $(\frac{3}{2} + o(1))$-factor approximation algorithm for the multipacking problem on cactus graphs. A natural direction for future research is to investigate whether a polynomial-time algorithm exists for finding a maximum multipacking on cactus graph.

It remains an interesting open problem to determine the best possible value of the expression $\lim_{\MP(G)\to \infty}$ $\sup\{\gamma_{b}(G)/\MP(G)\}$ for cactus graphs. Extending this investigation to other important graph classes could also yield valuable insights into the relationship between broadcast domination and multipacking numbers.

\chapter{Complexity of $r$-Multipacking}\label{chapter:k_multi}\hypertarget{chapter:introhref}{}
\minitoc 

    In this chapter, we prove that, for $r\geq 2$, the \textsc{$r$-Multipacking} problem is \textsc{NP-complete} even for planar bipartite graphs with bounded degree. Furthermore, we have shown that the problem is \textsc{NP-complete} for bounded diameter chordal graphs and bounded diameter bipartite graphs.

    \section{Chapter overview}
    In Section \ref{sec:preliminaries_r_multipacking}, we recall some definitions and notations. In Section \ref{sec:r_multipacking_npc_w}, we prove that, for $r\geq 2$, the \textsc{$r$-Multipacking} problem is \textsc{NP-complete} even for planar bipartite graphs with bounded degree. In Section \ref{sec:r_multipacking_hardness_on_subclasses}, we discuss the hardness of the \textsc{$r$-Multipacking} problem on some sublcasses.
We conclude in Section \ref{sec:conclusion_r_multipacking} by presenting several important questions and future research directions.

\section{Preliminaries}\label{sec:preliminaries_r_multipacking}

 For a graph $G=(V,E)$, $d_G(u,v)$ is the length of a shortest path joining two vertices $u$ and $v$ in  $G$, we simply write $d(u,v)$ when there is no confusion.  $\diam(G):=\max\{d(u,v):u,v\in V(G)\}$. Diameter is a path of $G$  of the length  $\diam(G)$. 
 A ball of radius $r$ around $u$, $N_r[u]:=\{v\in V:d(u,v)\leq r\}$ where $u\in V$. The \textit{eccentricity} $e(w)$  of a vertex $w$ is $\min \{r:N_r[w]=V\}$. The \textit{radius} of the graph $G$ is $\min\{e(w):w\in V\}$, denoted by $\rad(G)$.  The \textit{center} $C(G)$ of the graph $G$  is the set of all vertices of minimum eccentricity, i.e., $C(G):=\{v\in V:e(v)=\rad(G)\}$. Each vertex in the set $C(G)$ is called a \textit{central vertex} of the graph $G$.  If $P$ is a path in $G$, then we say $V(P)$ is the vertex set of the path $P$, $E(P)$ is the edge set of the path $P$, and  $l(P)$ is the length of the path $P$, i.e., $l(P)=|E(P)|$. 

 An \textit{r-multipacking} is a set $ M \subseteq V  $ in a
	graph $ G = (V, E) $ such that   $|N_s[v]\cap M|\leq s$ for each vertex $ v \in V $ and for every integer $s$, $ 1\leq s \leq r $. The \textit{r-multipacking number} of $ G $ is the maximum cardinality of an $r$-multipacking of $ G $ and it is denoted by $ \MP_r(G) $. The \textsc{$r$-Multipacking} problem is as follows:

\medskip
\noindent
\fbox{%
  \begin{minipage}{\dimexpr\linewidth-2\fboxsep-2\fboxrule}
  \textsc{ $r$-Multipacking} problem
  \begin{itemize}[leftmargin=1.5em]
    \item \textbf{Input:} An undirected graph $G = (V, E)$, an integer $k \in \mathbb{N}$.
    \item \textbf{Question:} Does there exist an $r$-multipacking $M \subseteq V$ of $G$ of size at least $k$?
  \end{itemize}
  \end{minipage}%
}
\medskip

\section{Complexity of the \textsc{$r$-Multipacking} problem}\label{sec:r_multipacking_npc_w}

It is known that the \textsc{Independent Set} problem is \textsc{NP-complete} for  planar graphs of maximum degree $3$~\cite{garey1974some}.  We reduce this problem to our problem to show that for $r\geq 2$, the \textsc{$r$-Multipacking} problem is \textsc{NP-complete} even for planar bipartite graphs with maximum degree $\max \{4,r\}$.

\medskip
\noindent
\fbox{%
  \begin{minipage}{\dimexpr\linewidth-2\fboxsep-2\fboxrule}
  \textsc{ Independent Set} problem
  \begin{itemize}[leftmargin=1.5em]
    \item \textbf{Input:} An undirected graph $G = (V, E)$ and an integer $k \in \mathbb{N}$.
    \item \textbf{Question:} Does there exist an \emph{independent set} $S \subseteq V$ of size at least $k$; that is, a set of at least $k$ vertices such that no two vertices in $S$ are adjacent in $G$?
  \end{itemize}
  \end{minipage}%
}
\medskip

Let $(G,k)$ be an instance of the \textsc{Independent Set} problem, where $G$ is a planar graph  with the vertex set $V=\{v_1,v_2,\dots,v_n\}$ having maximum degree $3$. We construct a graph $G'$ in the following way (Fig. \ref{fig:NPC_planar_r_multi} gives an illustration). 



\begin{enumerate}
\item[(i)] For every vertex $v_i\in V(G)$, we introduce a path $P_i=u_i u_i^1 u_i^2 \dots u_i^{r-1}$ of length $r-1$ in $G'$. 

\item[(ii)] If $v_i$ and $v_j$ are different and adjacent in $G$, we join $u_i$ and $u_j$ by a $2$-length path $u_iu_{i,j}u_j$  in $G'$ for every $1\leq i<j\leq n$. 

\item[(iii)] For each vertex $u_{i,j}$ in $G'$, we introduce a vertex $u_{i,j}^0$ and we join $u_{i,j}$ and $u_{i,j}^0$ by an edge. Further, for each vertex $u_{i,j}^0$ in $G'$, we introduce $r-1$ paths $P_{i,j}^t=u_{i,j}^0u_{i,j}^{t,1}u_{i,j}^{t,2}\dots u_{i,j}^{t,r-1}$ for $1\leq t\leq r-1$. Note that, each path $P_{i,j}^t$ has length $r-1$.
\end{enumerate}

\begin{figure}[ht]
    \centering
   \includegraphics[width=\textwidth]{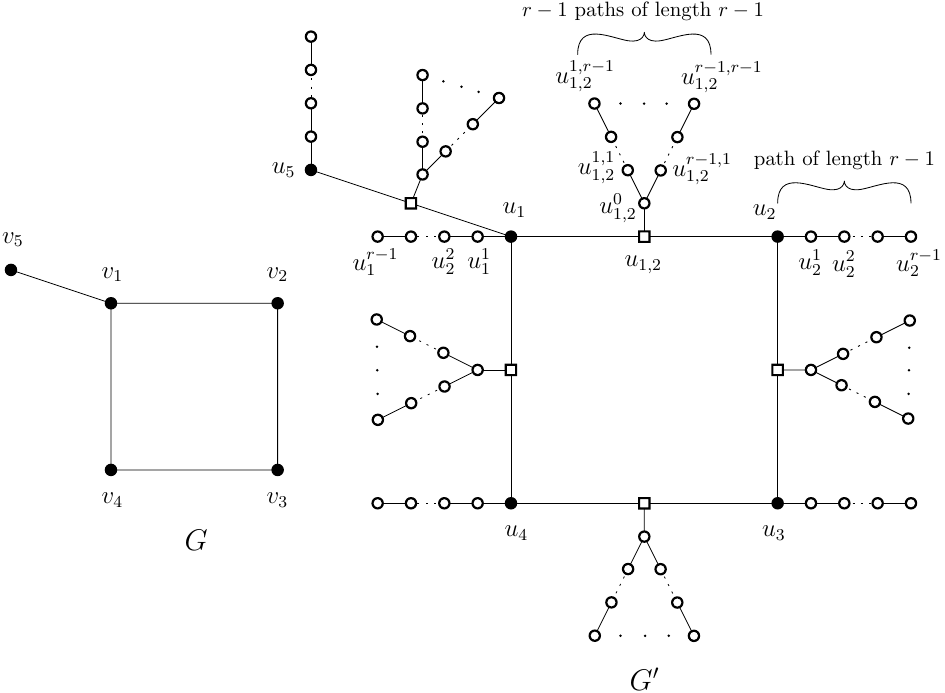}
    \caption{An illustration of the construction of the graph $G'$ used in the proof of Theorem~\ref{thm:NPC_planar_r_multi}, demonstrating the hardness of the \textsc{$r$-Multipacking} problem.}
    \label{fig:NPC_planar_r_multi}
\end{figure}

According to the construction, $G'$ is a planar graphs with maximum degree $\max \{4,r\}$, since $G$ is planar graph of maximum degree $3$. Moreover, every edge of $G$ is subdivided in $G'$, and hence every cycle in $G'$ has even length. Therefore, $G'$ is a bipartite graph.

Let $U_{i,j}=\{u_{i,j}\}\cup \bigcup_{t=1}^{r-1} V(P_{i,j}^t)$ and $S_{i,j}=V(P_i)\cup V(P_j)\cup U_{i,j}$.

Now consider the Algorithm \ref{alg:update_M} that we use in the proofs of Theorem \ref{thm:NPC_planar_r_multi}, Theorem \ref{thm:NPC_chordal_r_multi}, and Theorem \ref{thm:NPC_bipartite_r_multi}.

\begin{algorithm}[H]
\caption{\textsc{Reassign}$(G', M)$: Algorithm for the proofs of Theorems~\ref{thm:NPC_planar_r_multi}, \ref{thm:NPC_chordal_r_multi}, and \ref{thm:NPC_bipartite_r_multi}.}
\label{alg:update_M}

Initialize $M' \gets M$\;
\For{each vertex $u_{i,j} \in V(G')$ where $i < j$}{
    Compute $p \gets |S_{i,j} \cap M'|$\;
    \tcp*[l]{\textcolor{blue}{Note that $|S_{i,j} \cap M| \leq |N_r[u_{i,j}] \cap M| \leq r$.}}
    
    \eIf{$p = r$}{
        \eIf{$V(P_i) \cap M' \neq \emptyset$}{
            $M' \gets M' \setminus S_{i,j}$\;
            $M' \gets M' \cup \{u_i^{r-1}\} \cup \{u_{i,j}^{1,r-1}, u_{i,j}^{2,r-1}, \dots, u_{i,j}^{r-1,r-1}\}$\;
        }{
            \eIf{$V(P_j) \cap M' \neq \emptyset$}{
                $M' \gets M' \setminus S_{i,j}$\;
                $M' \gets M' \cup \{u_j^{r-1}\} \cup \{u_{i,j}^{1,r-1}, u_{i,j}^{2,r-1}, \dots, u_{i,j}^{r-1,r-1}\}$\;
            }{
                \tcp*[l]{\textcolor{blue}{This case cannot arise; otherwise, $|U_{i,j} \cap M| = r$ or $|N_{r-1}[u^0_{i,j}] \cap M| = r$.}}
                $M' \gets M'$\;
            }
        }
    }{
        \tcp*[l]{\textcolor{blue}{$p \not> r$ since $M$ is an $r$-multipacking and $|N_r[u_{i,j}] \cap M| \leq r$. Therefore, in this case, $p< r$.}}
        $M' \gets M' \setminus S_{i,j}$\;
        $M' \gets M' \cup \{u_{i,j}^{1,r-1}, u_{i,j}^{2,r-1}, \dots, u_{i,j}^{p,r-1}\}$\;
    }
}
$M'' \gets M'$\;
\Return $M''$\;

\end{algorithm}

    

\begin{lemma}\label{lem:reaasign}
    If $M$ is an $r$-multipacking of $G'$, then the output of $\textsc{Reassign}(G',M)$ from Algorithm \ref{alg:update_M} is an  $r$-multipacking of $G'$ where the size of the output $|\textsc{Reassign}(G',M)|=|M|$. Moreover, $\textsc{Reassign}(G',M)$ runs in polynomial time over $n$, where $n=|V(G)|$.
\end{lemma}



\begin{proof}
    Suppose $M$ is an $r$-multipacking of $G'$.  Let $M''=\textsc{Reassign}(G',M)$, from Algorithm \ref{alg:update_M}.  We want to show that $M''$ is an $r$-multipacking of $G'$.

Note that, $d(w_1,w_2)\geq 2(r-1)$ for each pair of vertices $w_1,w_2\in M''$. Therefore, $|N_s[v] \cap M''|\leq 1\leq s$ for any $s$ where $1\leq s\leq r-2$ and any vertex $v\in V(G')$. Moreover, $d(w_1,w_2)= 2(r-1)$ if and only if $w_1,w_2\in U_{i,j}$ for some $i,j$. Then $u^0_{i,j}$ is the only vertex for which $w_1,w_2\in N_{r-1}[u^0_{i,j}]$. In that case,  $|N_{r-1}[u^0_{i,j}] \cap M''|\leq r-1$ according to the Algorithm \ref{alg:update_M}. Therefore, we have $|N_s[v] \cap M''|\leq s$ for any $s$ where $1\leq s\leq r-1$ and any vertex $v\in V(G')$. So, the only thing we have to prove is that $|N_r[v] \cap M''|\leq r$ for any vertex $v\in V(G')$. Suppose there is a vertex $v$ such that $|N_r[v] \cap M''|\geq r+1$. This is true only when $v=u_{i,j}$ for some $i,j$ and $\{u_i^{r-1},u_j^{r-1}\} \cup \{u_{i,j}^{1,r-1}, u_{i,j}^{2,r-1}, \dots, u_{i,j}^{r-1,r-1}\}\subseteq M''$. Then $|S_{i,j} \cap M''|=r+1$. But in every iteration of the Algorithm \ref{alg:update_M}, the value of $|S_{i,j} \cap M'|$ either remains same or decreases, since $V(P_k)\cap M=\emptyset \implies V(P_k)\cap M'=\emptyset$ for any $k$. Therefore, $|S_{i,j} \cap M|\geq r+1$. This implies that $|N_r[u_{i,j}] \cap M|\geq r+1$. This is a contradiction, since $M$ is an $r$-multipacking. Therefore, $|N_r[v] \cap M''|\leq r$ for any vertex $v\in V(G')$. Hence $M''$ is an $r$-multipacking.

Note that, in every iteration of the Algorithm \ref{alg:update_M}, the value of $| M'|$  remains same. Therefore, $|\textsc{Reassign}(G',M)|=|M|$.

The Algorithm \ref{alg:update_M} has only one loop that runs on the number of edges of $G$.  Therefore, $\textsc{Reassign}(G',M)$ runs in polynomial time over $n$.
\end{proof}



\begin{lemma} \label{lem:NPC_planar_r_multi}
$G$ has an independent set of size at least $k$ if and only if $G'$ has an $r$-multipacking of size at least $k+m(r-1)$, where $|E(G)|=m$.
\end{lemma}

\begin{proof}
\textsf{(If part)} Suppose $G$ has an independent set $S = \{ v_1, v_2, \dots, v_k \}$. 
We want to show that $M_S =\{ u_1^{r-1}, u_2^{r-1}, \dots, u_k^{r-1} \}\cup \{u_{i,j}^{t,r-1}:v_iv_j\in E(G),i<j, 1\leq t\leq r-1\}$ is an $r$-multipacking in $G'$. Note that the distance between any two vertices of $M_S$ is at least $2(r-1)$. Therefore, $|N_s[v]\cap M_S|\leq 1\leq s$ for each vertex $ v \in V(G') $ and for every integer $ 1\leq s \leq r-2 $. If $M_S$ is not an $r$-multipacking,  there exists some vertex $v \in V(G')$ such that 
$| N_{s}[v] \cap M_S | > s$ for $s=r-1$ or $s=r$. Let $U_{i,j}=\{u_{i,j}\}\cup \bigcup_{t=1}^{r-1} V(P_{i,j}^t)$ and $S_{i,j}=V(P_i)\cup V(P_j)\cup U_{i,j}$.

\vspace{0.22cm}
 \noindent \textit{\textbf{Case 1: }} $v\in V(P_i)$ for some $i$.
 
Then $d(v,w)\geq r+1$ for each $w\in M_S\setminus\{u_i^{r-1}\}$. In that case, $| N_{s}[v] \cap M_S | \leq 1\leq s$ for $s=r-1$ and $s=r$. Therefore, $v\notin V(P_i)$ for any $i$. 

\vspace{0.22cm}
 \noindent \textit{\textbf{Case 2: }} $v\in V(P_{i,j}^t)$ for some $i,j,t$. 

Then $d(v,w)\geq r+1$ for each $w\in M_S\setminus\{u_{i,j}^{t,r-1}:1\leq t\leq r-1\}$. In that case, $| N_{s}[v] \cap M_S | \leq r-1\leq s$ for $s=r-1$ and $s=r$. Therefore, $v\notin V(P_{i,j}^t)$ for any $i,j,t$. 

\vspace{0.22cm}
 \noindent \textit{\textbf{Case 3: }} $v=u_{i,j}$ for some $i$ and $j$.

 Note that, $d(v,w)\geq r+1$ for each $w\in M_S\setminus(\{u_{i,j}^{t,r-1}:1\leq t\leq r-1\}\cup \{u_i^{r-1},u_j^{r-1}\})$. Therefore, $(\{u_{i,j}^{t,r-1}:1\leq t\leq r-1\}\cup \{u_i^{r-1},u_j^{r-1}\})\subseteq N_{r}[v] \cap M_S$, since we assumed that $| N_{r}[v] \cap M_S | \geq r+1$. In that case, $u_i^{r-1},u_j^{r-1}\in M_S$. Therefore, $v_i$ and $v_j$ are adjacent in $G$, which is a contradiction since $S$ is an independent set of $G$. Hence, $v\neq u_{i,j}$ for any $i$ and $j$.

 \noindent Hence, $M_S$ is an $r$-multipacking in $G'$. Note that, $|M_S|=k+m(r-1)$.

\medskip
 \textsf{(Only-if part)} Suppose $G'$ has an $r$-multipacking $M$ of size $k+m(r-1)$. Now we construct an $r$-multipacking $M''$ of the same size using a polynomial time Algorithm \ref{alg:update_M} (See Lemma \ref{lem:reaasign} for the correctness of the Algorithm \ref{alg:update_M}). Let $S=\{v_i\in V(G):u_i^{r-1}\in M''\}$. Now we prove that $S$ is an independent set of size at least $k$. Suppose there exist two vertices $v_i$ and $v_j$ (where $i<j$) in $S$  such that $v_i$ and $v_j$ are adjacent in $G$. Therefore,  $u_i^{r-1}$ and $u_j^{r-1}$ are in $M''$. From Algorithm \ref{alg:update_M}, we can say that $V(P_i) \cap M' \neq \emptyset$. In that case, the algorithm says,  $\{u_{i,j}^{1,r-1}, u_{i,j}^{2,r-1}, \dots, u_{i,j}^{r-1,r-1}\}\subseteq M''$. Therefore, $\{u_i^{r-1},u_j^{r-1}\}\cup\{u_{i,j}^{1,r-1}, u_{i,j}^{2,r-1}, \dots, u_{i,j}^{r-1,r-1}\}\subseteq N_r[u_{i,j}]\cap M''$. Then $|N_r[u_{i,j}]\cap M''|\geq r+1$. This is a contradiction, since $M''$ is an $r$-multipacking by Lemma \ref{lem:reaasign}. Therefore, $S$ is an independent set in $G$. 

Now we show that $|S|\geq k$. Note that, $|U_{i,j}\cap M''|\leq r-1$ for each $u_{i,j}\in V(G')$ where $i<j$. Therefore,

 $\sum_{u_{i,j}\in V(G')} |U_{i,j}\cap M''|\leq m(r-1),$\\
$\implies |\{u_i^{r-1}\in M'':v_i\in S\}| +\sum_{u_{i,j}\in V(G')} |U_{i,j}\cap M''|\leq |\{u_i^{r-1}\in M'':v_i\in S\}| + m(r-1),$\\
$\implies | M''|\leq |\{u_i^{r-1}\in M'':v_i\in S\}| + m(r-1),$\\
$\implies | M|\leq |\{u_i^{r-1}\in M'':v_i\in S\}| + m(r-1),$\\
$\implies k+m(r-1)\leq |\{u_i^{r-1}\in M'':v_i\in S\}| + m(r-1),$\\
$\implies k\leq |\{u_i^{r-1}\in M'':v_i\in S\}|,$\\
$\implies k\leq |S|$.
\end{proof}

Lemma \ref{lem:NPC_planar_r_multi} yields the following.

\NPCplanarkmulti*

\section{Hardness results on some subclasses} \label{sec:r_multipacking_hardness_on_subclasses}

In this section, we present hardness results on some important subclasses.

\NPCchordalkmulti*

\begin{proof}

It is known that the \textsc{Independent Set} problem is \textsc{NP-complete}~\cite{garey1979computers}.  We reduce this problem to our problem to show that \textsc{$r$-Multipacking} problem is \textsc{NP-complete}  for chordal graphs with maximum radius $r+1$.

Let $(G,k)$ be an instance of the \textsc{Independent Set} problem, where $G$ is a graph  with the vertex set $V=\{v_1,v_2,\dots,v_n\}$. We construct a graph $G'$ in the following way (Fig. \ref{fig:NPC_chordal_r_multi} gives an illustration).

\begin{enumerate}
\item[(i)] For every vertex $v_i\in V(G)$, we introduce a path $P_i=u_i u_i^1 u_i^2 \dots u_i^{r-1}$ of length $r-1$ in $G'$. 

\item[(ii)] If $v_i$ and $v_j$ are different and adjacent in $G$, we join $u_i$ and $u_j$ by a $2$-length path $u_iu_{i,j}u_j$  in $G'$ for every $1\leq i<j\leq n$. 

\item[(iii)] For each vertex $u_{i,j}$ in $G'$, we introduce a vertex $u_{i,j}^0$ and we join $u_{i,j}$ and $u_{i,j}^0$ by an edge. Further, for each vertex $u_{i,j}^0$ in $G'$, we introduce $r-1$ paths $P_{i,j}^t=u_{i,j}^0u_{i,j}^{t,1}u_{i,j}^{t,2}\dots u_{i,j}^{t,r-1}$ for $1\leq t\leq r-1$. Note that, each path $P_{i,j}^t$ has length $r-1$.

\item[(iv)] The set of vertices $\mathcal{T}=\{u_{i,j}\in V(G'):v_iv_j\in E(G)\}$ forms a clique in $G'$.
\end{enumerate}

\begin{figure}[ht]
    \centering
\includegraphics[width=\textwidth]{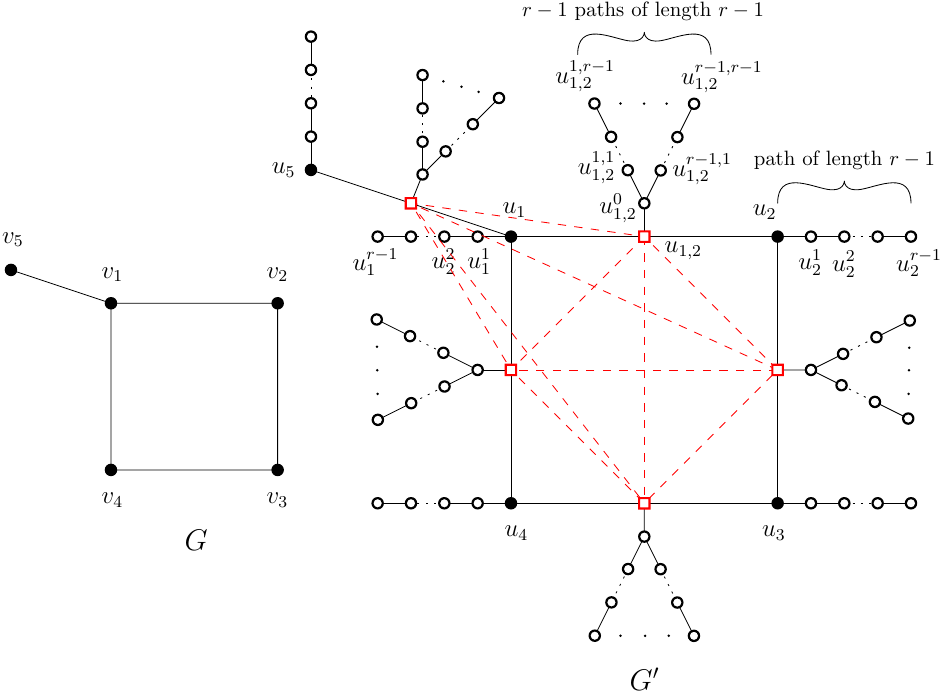}
    \caption{An illustration of the construction of the graph $G'$ used in the proof of Theorem~\ref{thm:NPC_chordal_r_multi}, demonstrating the hardness of the \textsc{$r$-Multipacking} problem.}
    \label{fig:NPC_chordal_r_multi}
\end{figure}

Note that, if $G'$ has a cycle $C=c_0c_1c_2\dots c_{t-1}c_0$ of length at least $4$, then there exists $i$ such that $c_i,c_{i+2 \! \pmod{t}}\in \mathcal{T}$ because no two vertices of the set $\{ u_1, u_2, \dots, u_n \}$ are adjacent in $G'$. Since $\mathcal{T}$ forms a clique, so $c_i$ and $c_{i+2 \! \pmod{t}}$ are endpoints of a chord. Therefore, the graph $G'$ is a chordal graph. Moreover, the eccentricity of any vertex of $\mathcal{T}$ is at most $r+1$. Therefore, $G'$ has radius at most $r+1$. 

\vspace{0.3cm}
\noindent
\textbf{Claim \ref{thm:NPC_chordal_r_multi}.1.} $G$ has an independent set of size at least $k$ if and only if $G'$ has an $r$-multipacking of size at least $k+m(r-1)$.
\begin{claimproof} The proof is completely same as the proof of Lemma \ref{lem:NPC_planar_r_multi}, since the vertices that are at least $r+1$ distance apart from the verties $u_p^{r-1}$ or $u_{i,j}^{t,r-1}$ for each $p,i,j,t$ in the proof of Lemma \ref{lem:NPC_planar_r_multi} remain same in this proof also. Moreover, the Algorithm \ref{alg:update_M} or the function $\textsc{Reassign}(G',M)$ is applicable for the graph $G'$ that we constructed here.
\end{claimproof}

    Hence, for $r\geq 2$, \textsc{$r$-Multipacking} problem is \textsc{NP-complete}  for chordal graphs with maximum radius $r+1$.
\end{proof}

Next, we show that,  for $r\geq 2$, the \textsc{$r$-Multipacking} problem is \textsc{NP-complete} for bipartite graphs with bounded radius.

\NPCbipartitekmulti*

\begin{proof} It is known that the \textsc{Independent Set} problem is \textsc{NP-complete}~\cite{garey1979computers}.  We reduce this problem to our problem to show that \textsc{$r$-Multipacking} problem is \textsc{NP-complete}  for chordal graphs with maximum radius $r+1$.

Let $(G,k)$ be an instance of the \textsc{Independent Set} problem, where $G$ is a graph  with the vertex set $V=\{v_1,v_2,\dots,v_n\}$. We construct a graph $G'$ in the following way (Fig. \ref{fig:NPC_bipartite_r_multi} gives an illustration).

\begin{enumerate}
\item[(i)] For every vertex $v_i\in V(G)$, we introduce a path $P_i=u_i u_i^1 u_i^2 \dots u_i^{r-1}$ of length $r-1$ in $G'$. 

\item[(ii)] If $v_i$ and $v_j$ are different and adjacent in $G$, we join $u_i$ and $u_j$ by a $2$-length path $u_iu_{i,j}u_j$  in $G'$ for every $1\leq i<j\leq n$. 

\item[(iii)] For each vertex $u_{i,j}$ in $G'$, we introduce a vertex $u_{i,j}^0$ and we join $u_{i,j}$ and $u_{i,j}^0$ by an edge. Further, for each vertex $u_{i,j}^0$ in $G'$, we introduce $r-1$ paths $P_{i,j}^t=u_{i,j}^0u_{i,j}^{t,1}u_{i,j}^{t,2}\dots u_{i,j}^{t,r-1}$ for $1\leq t\leq r-1$. Note that, each path $P_{i,j}^t$ has length $r-1$.

\item[(iv)] We introduce a vertex $u$ in $G'$ where $u$ adjacent to each vertex of the set  $\mathcal{T}=\{u_{i,j}\in V(G'):v_iv_j\in E(G)\}$.
\end{enumerate}

\begin{figure}[ht]
    \centering
\includegraphics[width=\textwidth]{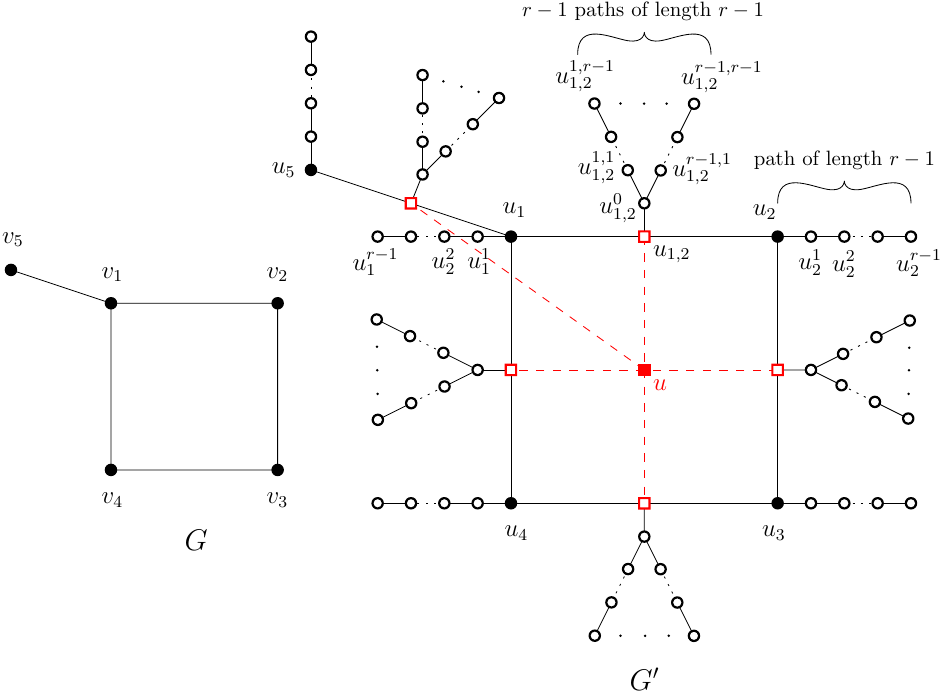}
    \caption{An illustration of the construction of the graph $G'$ used in the proof of Theorem~\ref{thm:NPC_bipartite_r_multi}, demonstrating the hardness of the \textsc{$r$-Multipacking} problem.}
    \label{fig:NPC_bipartite_r_multi}
\end{figure}

Note that, $G'$ is a bipartite graph where the partite sets are $B_1=\{u\}\cup \{u_i:1\leq i\leq n\}\cup \{u_i^j:1\leq i\leq n \text{ and } j\text{ is even} \}\cup \{u^0_{i,j}:u^0_{i,j}\in V(G')\}\cup \{u^{t,p}_{i,j}:u^{t,p}_{i,j}\in V(G'),1\leq t\leq r-1  \text{ and } p\text{ is even}\}$ and $B_2=\mathcal{T} \cup \{u_i^j:1\leq i\leq n \text{ and } j\text{ is odd} \}\cup \{u^{t,p}_{i,j}:u^{t,p}_{i,j}\in V(G'),1\leq t\leq r-1  \text{ and } p\text{ is odd}\}$. Moreover, the eccentricity of  $u$ is at most $r+1$. Therefore, $G'$ has radius at most $r+1$.  

\vspace{0.3cm}
\noindent
\textbf{Claim \ref{thm:NPC_bipartite_r_multi}.1.} $G$ has an independent set of size at least $k$ if and only if $G'$ has an $r$-multipacking of size at least $k+m(r-1)+1$.
\begin{claimproof}  The proof is similar as the proof of Lemma \ref{lem:NPC_planar_r_multi}, since the vertices that are at least $r+1$ distance apart from the verties $u_p^{r-1}$ or $u_{i,j}^{t,r-1}$ for each $p,i,j,t$ in the proof of Lemma \ref{lem:NPC_planar_r_multi} remain same in this proof also. The distance between $u$ and $u_p^{r-1}$ or $u_{i,j}^{t,r-1}$ for each $p,i,j,t$ is $r+1$. Therefore, in the \textsf{"If part"} of this proof, we add $u$ in the $r$-multipacking set $M_S$ (described in the proof of Lemma \ref{lem:NPC_planar_r_multi}), since it has no conflict with the other vertices of $M_S$.  Moreover, the Algorithm \ref{alg:update_M} or the function $\textsc{Reassign}(G',M)$ is applicable for the graph $G'$ that we constructed here. Therefore, after running the Algorithm \ref{alg:update_M}, in the \textsf{"Only-if part"} of this proof, we add $u$ in the $r$-multipacking set $M''$ (described in the proof of Lemma \ref{lem:NPC_planar_r_multi}), since it has no conflict with the other vertices of $M''$. Therefore, $M_S$ and $M''$ both have size $k+m(r-1)+1$.
\end{claimproof}

    Hence, for $r\geq 2$, \textsc{$r$-Multipacking} problem is \textsc{NP-complete}  for bipartite graphs with maximum radius $r+1$.
\end{proof}

\section{Lower Bound under the ETH}

It is known that, unless \textsc{ETH} fails, there is no $2^{o(n+m)}$-time algorithm for the \textsc{Independent Set} problem, where $n$ is the number of vertices and $m$ is the number of edges of the graph~\cite{cygan2015parameterized}. We use this result and the constructions of Theorem \ref{thm:NPC_chordal_r_multi} and \ref{thm:NPC_bipartite_r_multi} to show the following.

\rmultiETH*

\begin{proof}
Given a graph $G$ with $n$ vertices and $m$ edges, we construct a chordal graph $G'$ of maximum radius $r+1$ as described in the proof of Theorem~\ref{thm:NPC_chordal_r_multi}. Since $r$ is fixed for the \textsc{$r$-Multipacking} problem, the construction ensures that $G'$ has $O(n+m)$ vertices and $O(n+m)$ edges. This reduction implies that the existence of a $2^{o(n+m)}$-time algorithm for the \textsc{$r$-Multipacking} problem would yield a $2^{o(n+m)}$-time algorithm for the \textsc{Independent Set} problem. However, such an algorithm for \textsc{Independent Set} contradicts the \textsc{ETH}, as no $2^{o(n+m)}$-time algorithm exists for this problem~\cite{cygan2015parameterized}.

Similarly, we can show that the same result holds for bipartite graphs of maximum radius $r+1$ by using the construction described in the proof of Theorem~\ref{thm:NPC_bipartite_r_multi}.
\end{proof}

\section{Conclusion}
\label{sec:conclusion_r_multipacking}

In this work, we have explored the computational complexity of the \textsc{$r$-Multipacking} problem and established its \textsc{NP}-completeness. Moreover, we have extended these hardness results to several well-studied graph classes, including planar bipartite graphs with bounded degree, chordal graphs with bounded diameter, and bipartite graphs with bounded diameter.

Our findings naturally raise several open questions that deserve further investigation:
\begin{enumerate}
    \item Where does the \textsc{$r$-Multipacking} problem lie within the $W$-hierarchy? Is there an \textsc{FPT} algorithm when parameterized by the solution size?
    
    \item Does an \textsc{FPT} algorithm exist for general graphs when parameterized by treewidth?
    
    
    \item What is the approximation hardness of the \textsc{$r$-Multipacking} problem?
\end{enumerate}

Addressing these questions would advance our understanding of the computational complexity landscape of the \textsc{$r$-Multipacking} problem and could lead to new algorithmic insights.

\chapter{Multipacking in Geometry} \label{chapter:geo_multi}\hypertarget{chapter:introhref}{}

 In this chapter, we study the  multipacking problems for geometric point sets with respect to their Euclidean distances. For $r=1 \text{ and } 2$, we study the problem of computing a maximum $r$-multipacking of the point sets in $\mathbb{R}^2$. We show that, for the point sets in $\mathbb{R}^2$, a maximum $1$-multipacking can be computed in polynomial time but computing a maximum $2$-multipacking is \textsc{NP-hard}. Further, we provide approximation and parameterized solutions to the $2$-multipacking problem in $\mathbb{R}^2$.

\section{Chapter overview}
In Section \ref{sec:Preliminaries_geo_multi}, we recall some definitions and notations. In Section \ref{sec:Multipacking in geometry}, we provide an algorithm to find a maximum $r$-multipacking of a point set on a line and we prove some bounds on the multipacking number on a line. In Section \ref{sec:MP}, we prove that   $\MMP(2)=6$ using some geometric observations.  In Section \ref{sec:1,1/2 Multipacking in a 2D plane},  we have shown that a maximum $1$-multipacking in $\mathbb{R}^2$ can be computed in polynomial time. In Section \ref{sec:2,1/2 Multipacking in a 2D plane}, we study the hardness of the $2$-multipacking problem in $\mathbb{R}^2$. We prove that the problem is \textsc{NP-hard}. Moreover, we provide approximation and parameterized solutions to this problem. We conclude this chapter in Section \ref{sec:conclusion_geo_multi}.

\section{Preliminaries}\label{sec:Preliminaries_geo_multi}
Given a point set $P \subseteq \mathbb{R}^2$ of size $n$, the \emph{neighborhood of size $r$} of a point $v \in P$ is the subset of $P$ that consists of the $r$ nearest points of $v$ and $v$ itself. This subset is denoted by $N_r[v]$. A natural and realistic assumption here is the distances between all pairs of points are distinct. Thus the $r$-th neighbor of every point is unique. This assumption is easy to achieve through some form of controlled perturbation (see \cite{mehlhorn2006reliable}).
For a natural number $r \leq n-1$, an \textit{$r$-multipacking} is a subset $ M $ of $ P  $, such that for each point $ v \in P $ and for every integer $ 1\leq s \leq r  $, $|N_s[v]\cap M|\leq (s+1)/2$, i.e. $M$ contains at most half of the number of elements of $N_s[v]$.
 The \emph{$r$-multipacking number} of $ P $ is the maximum cardinality of an $r$-multipacking of $ P $ and is denoted by $ \MP_r(P) $. An $(n-1)$-multipacking $M$ of $P$ is referred to as a \emph{multipacking} of $P$. The \emph{multipacking number} of $P$, denoted by $\MP(P)$, refers to the cardinality of a maximum multipacking of $P$. Computing maximum multipacking  involves selecting the maximum number of points from the given point set satisfying the constraints of multipacking.

 The function $\MMP_{r}(t)$ is the smallest number such that any point set $P \subset \mathbb{R}^2$ of size $ \MMP_{r}(t)$ admits an $r$-multipacking of size at least $t$. We denote $\MMP_{n-1}(t)$ as $\MMP(t)$, i.e. $\MMP(t)$ is the smallest number such that any point set $P \subset \mathbb{R}^2$ of size $ \MMP(t)$ admits a \emph{multipacking} of size at least $t$. It follows from the definition of multipacking that $\MMP(1)=1$.

\section{Maximum multipacking in $\mathbb{R}^1$}\label{sec:Multipacking in geometry}

    


\subsection{Bound of maximum multipacking in $\mathbb{R}^1$}\label{subsec:Bounds of Multipacking in R1}
We start by bounding the multipacking number of a point set in the following lemma.

\begin{lemma}
\label{lem:mp 1D lower bound n/3 upper bound n/2}
    Let $P $ be a point set on $\mathbb{R}^1$. Then $\lfloor\frac{n}{3}\rfloor\leq \MP(P)\leq \lfloor\frac{n}{2}\rfloor$. 
\end{lemma}

\begin{proof}
    In order to show the lower bound, we show that for any point set on $\mathbb{R}^1$ of size $n$ we can always find a multipacking of size $\lfloor\frac{n}{3}\rfloor$. Let $P=\{p_1, p_2, \dots, p_n\}$ be $n$ points on $\mathbb{R}^1$ sorted with respect to the ascending order of their $x$-coordinates. Construct a set $M$ as such, $M=\{p_1, p_4, p_7, \dots, p_{3k+1} \}$ (where $k \in \mathbb{N}$ and $3k+1 \leq n < 3(k+1) + 1$) is a multipacking of $P$. Note that $M$ is a multipacking since for any $1 \leq s \leq n-1$ and $p_i\in P$, $N_s[p_i]$ contains at most $\lfloor\frac{s+1}{2}\rfloor$ points from $M$. This shows the lower bound of the lemma.  The upper bound follows from the definition of multipacking. 
\end{proof}




\noindent \textbf{Tight bound example for the upper bound:} For the upper bound of Lemma \ref{lem:mp 1D lower bound n/3 upper bound n/2}, we consider the point set $P = \{ p_1, \ldots , p_n \}$ where, $$p_i=
    \begin{cases}
        0 & \text{if $i=0$ }\\
        \frac{4}{3}(2^{i-1}-1)-\frac{i-1}{2} & \text{if $i$ is odd} \\
        \frac{4}{3}(2^{i}-1)-\frac{i}{2} - 2^{i-1} +1 & \text{if $i$ is even, $i>0$}
    \end{cases}$$

Note that  $p_i$'s in $P$ are in ascending order. From definition of this example, $p_{2k+2}-p_{2k+1}>p_{2k+1}-p_{1}$ and $p_{2k+2}-p_{2k+1}>p_{2k+3}-p_{2k+2}$ for each $k$. From this, we conclude that the set $M = \{ p_1, p_4, p_6, p_8, \ldots \}$ is a maximum multipacking and $|M| = \lfloor\frac{n}{2}\rfloor$.


\vspace{0.25 cm}
\noindent \textbf{Tight bound example for the lower bound: }
For the lower bound of Lemma \ref{lem:mp 1D lower bound n/3 upper bound n/2}, we consider the point set $P = \{ p_1, \ldots , p_n \}$ where $p_i = 2^i$. 

Note that $p_{i-1}-p_1 < p_i - p_{i-1}$ holds for each $p_i$ in $P$. From this, we conclude that the set $M = \{p_1, p_4, p_7, \ldots \}$ is a maximum  multipacking with $|M| = \lfloor\frac{n}{3}\rfloor$.


Combining these arguments together with Lemma \ref{lem:mp 1D lower bound n/3 upper bound n/2} we have the following theorem.

\thmmponeDtightlowerupperbound*


\subsection{Algorithm to find a maximum multipacking in $\mathbb{R}^1$}\label{subsec:AlgoR1}

We devise algorithm \ref{algo:2_mp} that checks possible violations for every member of $P$ and selects the eligible candidates as multipacking accordingly. It takes a point set $P \subset \mathbb{R}^1$ sorted in ascending order with respect to their $x$-coordinates  as input and outputs a maximum multipacking $M$ where $M \subseteq P$.

\medskip
\RestyleAlgo{ruled}    
\begin{algorithm}[H]
    \SetAlgoLined
    \KwIn{A point set $P$ in $\mathbb{R}$ and $M \subseteq P$ and the value of $r$.}
    \KwOut{\textit{if} $M$ is a valid multipacking on $P$ \textit{then} return $1$,\ \textit{else} return $0$}
    \For{$p_i$}{
      \For{$s=1 \rightarrow r$}{
      \eIf{$|N_s[p_i] \cap M| \leq (s+1)/2$}{continue;}{return $0$;}
      }
       
    }
    return $1$\;

 \caption{Check $r$-multipacking - $CMP_{r}(P,M)$}
 \label{algo:chk_2_mp}
\end{algorithm}
\medskip
Algorithm \ref{algo:chk_2_mp} ($CMP_{r}(P,M)$) is designed to check the eligibility of a member of $P$ as a point in multipacking and executed as a subroutine in each iteration of \cref{algo:2_mp}. $CMP_{r}(P,M)$(Algorithm \ref{algo:chk_2_mp}) takes a point set $M \subseteq P$ as input and outputs $1$ if $M$ is a multipacking and $0$ if not. 
Given a point set $P = \{ p_1, p_2, \hdots , p_n \}$, at the $i^{\text{th}}$ iteration, $p_i$ is included in $M$ and $CMP_{r}(P,M)$ is called as a subroutine. Depending on the eligibility check done by this subroutine $M$ is formed. The algorithm terminates when $P$ is exhausted and returns $M$ as output. 

\medskip
\RestyleAlgo{ruled}    
\begin{algorithm}[H]
    \SetAlgoLined
    \KwIn{A point set $P$ in $\mathbb{R}$ and the value of $r$.}
    \KwOut{A maximum multipacking $M$ of $P$.}
    $M=\phi$\; $i=1$\;
    \For{$p_i$}{
        $M = M \cup \{p_i\}$\;
       \If{$CMP_{r}(P,M) ==0$}{$M = M \setminus \{p_i\}$\;}
       $i++$\;
    }
    return $M$\;
 \caption{Maximum $r$-multipacking}
 \label{algo:2_mp}
\end{algorithm}



\rkthmonedcorrectness*

\begin{proof}
 We prove the result by showing the correctness of \cref{algo:2_mp} that returns a maximum $r$-multipacking of the given point set $P$ as output in $O(n^2r)$ time. For simplicity, we just say multipacking in place of $r$-multipacking in this proof. Let $P=\{a_1,a_2,\dots,a_n\}$ where $a_{i}<a_{i+1}$ for each $i$. Let $M$ be a solution of \cref{algo:2_mp}. According to the algorithm, we can say that $M$ is a multipacking. Let us assume that $M$ is not a maximum multipacking. Let $M'$ denote a maximum multipacking of the point set $P$. To compare $M$ and $M'$, we consider the members of $M$ and $M'$ sorted in ascending order with respect to their $x$-coordinate. Let $M=\{a_{g_1},a_{g_2},\dots,a_{g_m}\}$ and $M'=\{a_{h_1},a_{h_2},\dots,a_{h_{m'}}\}$ where $a_{g_i}<a_{g_{i+1}}$ and $a_{h_j}<a_{h_{j+1}}$ for each $i$ and $j$. By assumption, $|M|<|M'|$ or $m<m'$. Let $t=\min\{i:a_{g_i}\neq a_{h_{i}}, 1\leq i\leq m\}$. If $a_{g_t}> a_{h_{t}}$, then $\{a_{h_1},a_{h_2},\dots,a_{h_{t}}\}$ is not a multipacking of $P$ since \cref{algo:2_mp} allocates multipacking for every possible point from left to right of $P$. Therefore, $a_{g_t}< a_{h_{t}}$. Let $M''=\{a_{g_1},a_{g_2},\dots,a_{g_t},a_{h_{t+1}},a_{h_{t+2}},\dots,a_{h_{m'}}\}$  (See Fig. \ref{algo:2_mp}). 

 \begin{figure}[htbp]
    \centering
    \includegraphics[width=0.92\textwidth]{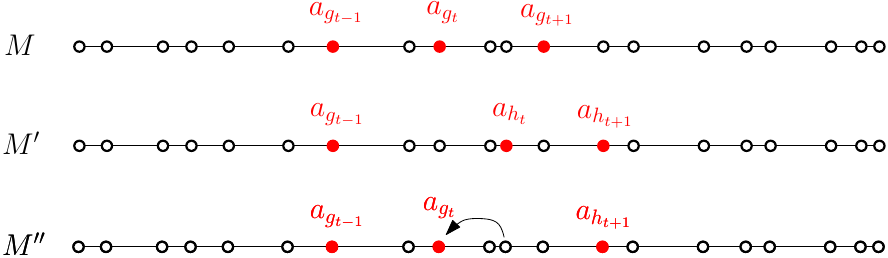}
    \caption{Correctness proof of algorithm \ref{algo:2_mp}.}
    \label{fig:MMM}
\end{figure}

 \vspace{0.2cm}
\noindent
\textbf{Claim \ref{thm:rk1D}.1. } $M''$ is a multipacking of $P$.

\begin{claimproof} We want to prove that $M''$ is a multipacking using the fact saying that $M$ and $M'$ are multipackings. Observe that, $N_s[a_i]$ is a list of consecutive points of $P$ for each $ a_i \in P $. Now fix a point $a_i\in P$. First consider $a_{g_t},a_{h_t}\in N_s[a_i]$ for some $s$. Then $|N_s[a_i]\cap M''|=|N_s[a_i]\cap M'|\leq (s+1)/2$. Suppose $a_{g_t}\in N_s[a_i]$ but $a_{h_t}\notin N_s[a_i]$. Then $N_s[a_i]\cap M''=N_s[a_i]\cap M$, this implies $|N_s[a_i]\cap M''|=|N_s[a_i]\cap M|\leq (s+1)/2$. If $a_{g_t}\notin N_s[a_i]$ but $a_{h_t}\in N_s[a_i]$, then $|N_s[a_i]\cap M''|<|N_s[a_i]\cap M'|\leq (s+1)/2$. Now we are left with only one case which says that $a_{g_t},a_{h_t}\notin N_s[a_i]$. In that case $N_s[a_i]\cap M''=N_s[a_i]\cap M'$, this implies $|N_s[a_i]\cap M''|=|N_s[a_i]\cap M'|\leq (s+1)/2$. Therefore, $|N_s[a_i]\cap M''|\leq (s+1)/2$ for every point $a_i\in P$ and for every integer $ 1\leq s \leq r$.   
\end{claimproof}

Since $|M''|=m'$, therefore $M''$ is a maximum multipacking of $P$. Similarly we can show that $\{a_{g_1},a_{g_2},\dots,$ $a_{g_t},a_{g_{t+1}},a_{h_{t+2}},\dots,a_{h_{m'}}\}$ is also a maximum multipacking using the fact: $M$ and $M''$ are multipackings. We can keep on doing the same method to arrive at the maximum multipacking $M_1=\{a_{g_1},a_{g_2},\dots,a_{g_{m-1}},a_{g_{m}},$ $a_{h_{m+1}},\dots,a_{h_{m'}}\}$. But \cref{algo:2_mp} ensures that one cannot add more points to the set $\{a_{g_1},a_{g_2},\dots,a_{g_{m-1}},a_{g_{m}}\}$ to construct a larger multipacking. Therefore, $|M_1|=|M|$, this implies $m=m'$ which is a contradiction. Therefore, \cref{algo:2_mp} returns a maximum $r$-multipacking.

    Now we analyze the running time of the \cref{algo:2_mp}. It is clear that the algorithm calls the subroutine $CMP_{r}(P,M)$ (\cref{algo:chk_2_mp}) for $n$ times. Further, the subroutine $CMP_{r}(P,M)$ (\cref{algo:chk_2_mp}) takes $O(nr)$ time to check the correctness of being multipacking for a set. Thus \cref{algo:2_mp} takes $O(n^2r)$ time to return a solution.  
 \end{proof}

It is worth mentioning that, computing the maximum multipacking set of a point set in $\mathbb{R}^2$ relates with the notion of independent set (which we show later in this chapter) which becomes \textsc{NP-hard} in general. Thus it is important to study a bound for the multipacking number in $\mathbb{R}^2$. Next, we address the bound on the size of $P$ to achieve a multipacking of a desired size irrespective of the distribution of a set of points in $\mathbb{R}^2$.

\section{Minimum number of points to achieve a desired multipacking size in $\mathbb{R}^2$}
\label{sec:MP}

In this section, we study the $\MMP$ function. Recall the definition of the same. $\MMP_{r}(t)$ is the smallest number such that any point set $P \subset \mathbb{R}^2$ with $|P|= \MMP_{r}(t)$ admits an $r$-multipacking of size at least $t$. As stated earlier $\MMP_{n-1}(t)$ is denoted by $\MMP(t)$ and $\MMP(1)=1$. Now our goal is to find the value of $\MMP(2)$, which we state in the following theorem.




\begin{figure}[htbp]
\centering%
\begin{minipage}[t]{0.50\textwidth}
	\centering
	\includegraphics[width = .50\textwidth]{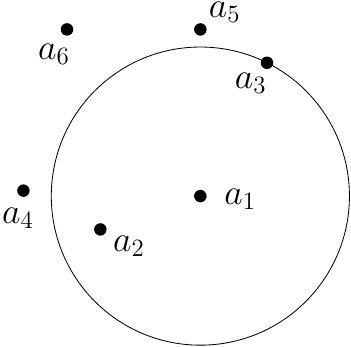}
	\caption{Any $6$ points always have $\MP(P) \geq 2$.}
	\label{fig:C_2[a_1]}
\end{minipage}%
\hfill%
\begin{minipage}[t]{0.5\textwidth}
	\centering
	\includegraphics[width = .65\textwidth]{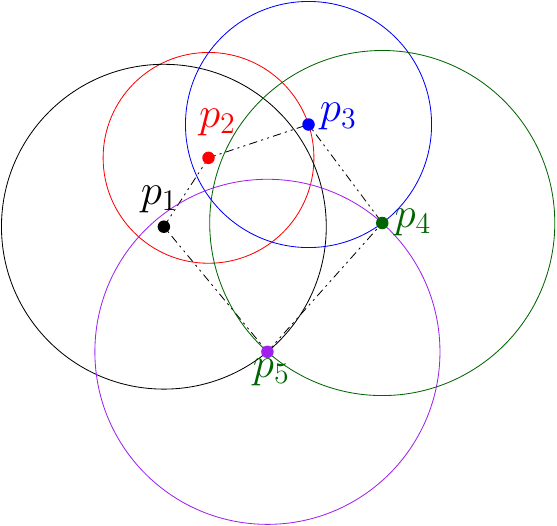}
	\caption{$5$ points with $\MP(P) =1$.}
	\label{fig:mp2}
\end{minipage}%
\end{figure}

\MPtwo*

\begin{proof}
    We prove this using contradiction. Suppose $\MMP(2)>6$. That means there is a point set $P=\{a_1,a_2,\dots,a_6\}\subset \mathbb{R}^2$ such that $\MP(P)=1$. This implies, for each $i,j(i\neq j)$ there exists $k$ such that $\{a_i, a_j\}\subset N_2[a_k]$, otherwise $\{a_i, a_j\}$ can form a multipacking of size $2$ which is not allowed. If $a_i$ and $a_j$ are the $1^{\text{st}}$ and $2^{\text{nd}}$ neighbor of $a_k$ respectively, we denote $C_2[a_k]$ to be the circle with center at $a_k$ passing through $a_j$. Without loss of generality we can assume that $C_2[a_1]$ is the circle with the largest radius among all the circles $C_2[a_k]$, $k\in \{1,2,\dots,6\}$ and moreover $a_2$ and $a_3$ are the $1^{\text{st}}$ and $2^{\text{nd}}$ neighbor of $a_1$ (See Fig. \ref{fig:C_2[a_1]}). Observe that, $a_1\notin N_2[a_k]$ for each $k\in \{4,5,6\}$, since $C_2[a_1]$ is the circle with largest radius. Therefore, each pair among the three pairs $\{a_1, a_4\}$, $\{a_1, a_5\}$ and $\{a_1, a_6\}$ is a subset of either $N_2[a_2]$ or $N_2[a_3]$. 
    By Pigeonhole Principle, there are at least two pairs among the three pairs $\{a_1, a_4\}$, $\{a_1, a_5\}$ and $\{a_1, a_6\}$ which are subsets of either $N_2[a_2]$ or $N_2[a_3]$, but both $N_2[a_2]$ and $N_2[a_3]$ have size $3$ which is a contradiction of $\MP(P)=1$.
    Therefore,  $\MMP(2)\leq 6$. 
    
    Furthermore, we have point sets of sizes $3,4$ and $5$ (See Fig. \ref{fig:mp2}) that have multipacking number $1$. We depict the point set $P$ in Fig. \ref{fig:mp2} forming a convex pentagon. Every point $p_i \in P$ has $p_{i-1}$ and $p_{i+1 (mod~ 5)}$ as its two nearest neighbors. As a result no two points can be present in the multipacking set of $P$. The correctness of the figure can be followed from the fact that each circle (marked with different colours) containing a point $p_i \in P$ as center and the second nearest neighbor of $p_i$ on its boundary, has only one other point (namely the first nearest neighbor of $p_i$) in the interior. This example shows $\MMP(2) > 5$ implying, $\MMP(2)= 6$.
    \end{proof}

\section{$1$-multipacking in $\mathbb{R}^2$}
\label{sec:1,1/2 Multipacking in a 2D plane} 

We note that Algorithm \ref{algo:2_mp} in Subsection \ref{subsec:AlgoR1}  cannot be adapted for a point set $P \subseteq \mathbb{R}^2$. Since we  cannot  ensure the sorted order of the points with respect to their $x$ or $y$ coordinate, the maximum cardinality of the output  cannot  be guaranteed. Thus it is interesting to address $r$-multipacking while $r < n-1$. 


The NNG (Nearest-Neighbor-Graph) of $P$ is a forest when considered as an undirected graph with $2$-cycles at the leaves. Then, a maximum $1$-multipacking for $P$ can be computed in polynomial time by computing the maximum independent set of NNG of $P$ \cite{eppstein1997nearest}.

We write this observation as the following lemma.

\begin{lemma}
\label{lem:1,1/2 MP 1}
    Let $P$ be a point set in $\mathbb{R}^2$ and the $NNG$ of $P$ be $G$. The maximum independent sets of $G$ are the maximum $1$-multipackings of $P$.
\end{lemma}

\section{$2$-multipacking in $\mathbb{R}^2$}
\label{sec:2,1/2 Multipacking in a 2D plane}
 
Next we study the $2$-multipacking problem  in $\mathbb{R}^2$. We are interested in studying the hardness of finding a maximum $2$-multipacking of a given set of points $P$ in the plane. 

 We construct a graph $G_P=(V,E)$ by taking the point set $P$ as its vertex set, that is, $V=P$. For every vertex $v \in V$, we introduce three edges $vu_1,vu_2,$ and $u_1u_2$  where $u_1,u_2 \in V$ are the first and second neighbors of $v$ respectively. The following lemma shows the relation between the independent set of $G_P$ and a $2$-multipacking of $P$.

\begin{lemma}
    \label{lem:graph}
    A set $M$ is a $2$-multipacking of $P$ if and only if the vertices in the graph $G_P$ corresponding to the members of $M$ form an independent set in $G_P$.
\end{lemma}
\begin{proof}
    It follows from the definition of $2$-multipacking that for every tuple $(v,u_1,$ $u_2)$, as described above, at most one can be chosen as a member of $M$. Since each edge of $G_P$ is incident on a pair of vertices appearing in a tuple, there is no edge between two members of $M$ appearing in $G_P$. Thus the vertices corresponding to the members of $M$ form an independent set of $G_P$. The converse follows from a similar argument.
\end{proof}

Thus, finding a maximum independent set of $G_P$ is equivalent to finding a maximum $2$-multipacking of $P$. In the rest of this section, we show that finding a maximum independent set of $G_P$ is NP-hard, and we devise approximation and parameterized solutions to this problem.

\subsection{NP-hardness}\label{subsec:np-hardness}

We show that finding a maximum independent set in the graph $G_P$ is NP-hard by reducing the well-known NP-complete problem \emph{Planar Rectilinear Monotone (PRM) 3-SAT}~\cite{de2012optimal} to this problem. This variant of the boolean formula contains three literals in each clause and all of them are either positive or negative. Moreover, a planar graph can be constructed for the formula such that a set of axis parallel rectangles represent the variables and the clauses with the following properties. The rectangles representing the variables lie on the $x$-axis. Whereas rectangles representing the clauses that contain the positive (negative) literals lie above (below) the $x$-axis. Additionally, the rectangles for the clauses can be connected with vertical lines to the rectangles for the variables that appear in the clauses.

Given a PRM 3-SAT formula $\phi$ with $t(>1)$ variables and $m(>1)$ clauses, we construct a planar point set $P_\phi$ such that $\phi$ is satisfiable if and only if $G_{P_\phi}$ has a maximum independent set of a given size. First, we construct a point set for the \emph{variable gadget} for each variable that has only two subsets as the maximum independent set in $G_{P_\phi}$. Then we construct another point set as clause gadgets for each clause in $\phi$ such that the clause gadget can incorporate only one vertex in the independent set only if at least one variable gadget has the desired set chosen as the maximum independent set. The clause and variable gadgets are connected via \emph{connection gadgets} which are two point sets referred to as \emph{translation} and \emph{rotation gadgets}. We describe the construction of the point sets formally below. 


\begin{figure}
    \centering
    \includegraphics[page=1,width=\textwidth]{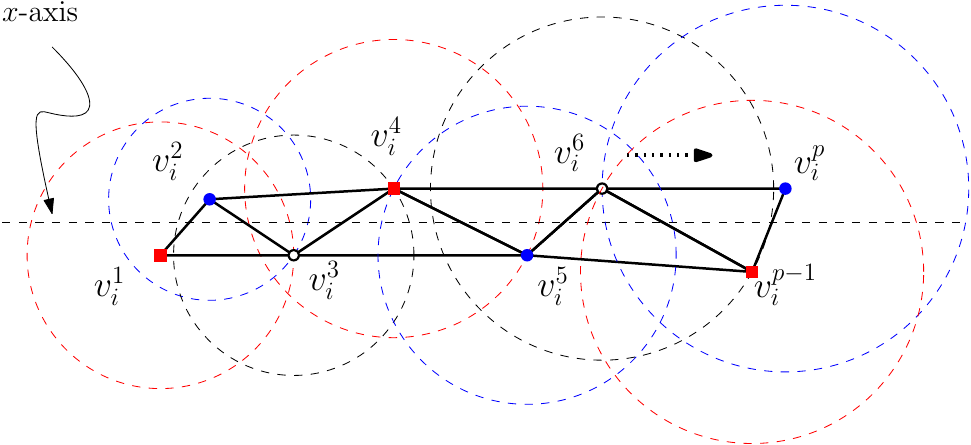}
    \caption{A variable gadget in $G_{P_\phi}$.}
    \label{fig:variable}
\end{figure}

\begin{figure}
\centering
    \includegraphics[page=2,width=\textwidth]{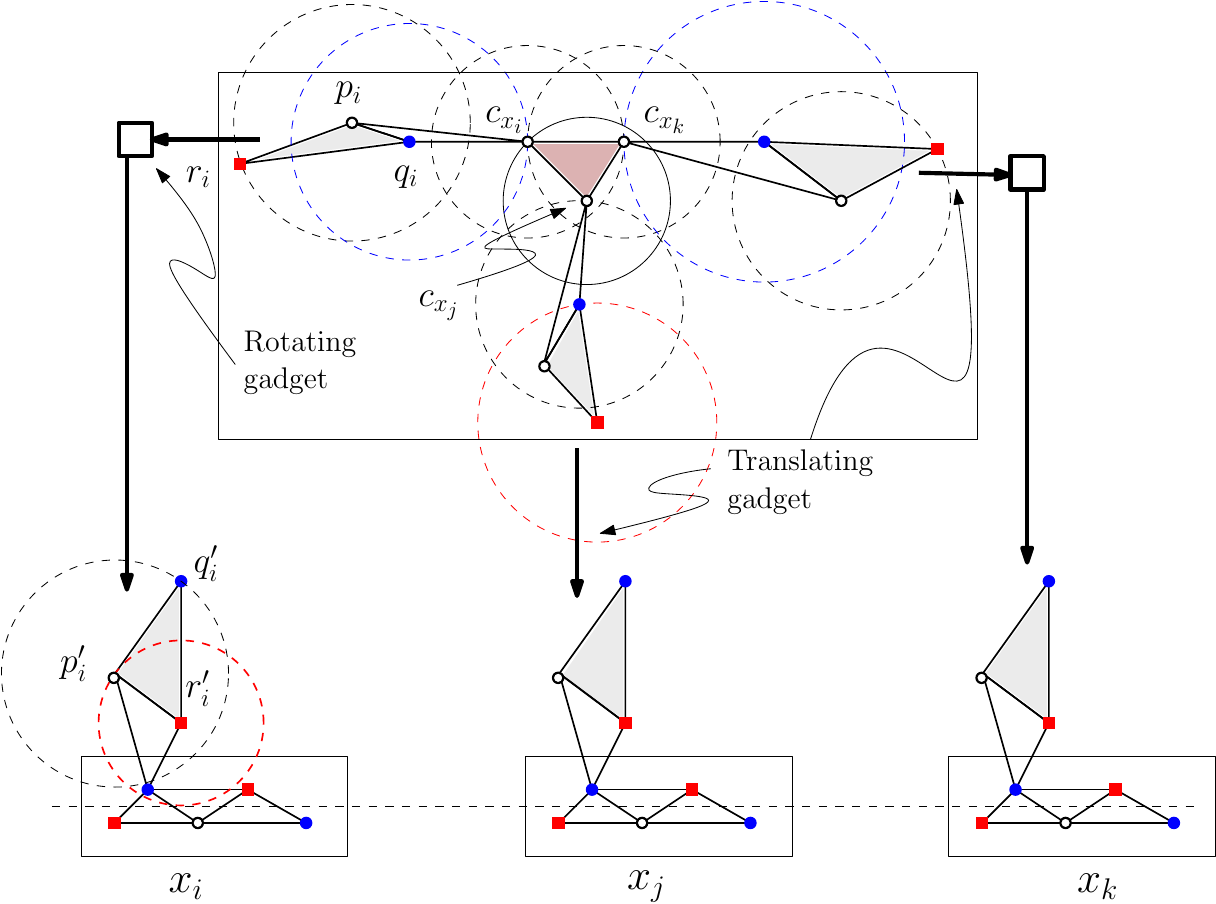}
    \caption{A clause gadget and its connection with the variable gadgets.}
    \label{fig:clause}
\end{figure}

\vspace{0.3cm}
{\bf Variable gadget.} We construct a point set for each variable $x_i$ in $\phi$ consisting of $p=6(m-1)+2$ points.  
A variable gadget is depicted in Figure~\ref{fig:variable} with $p=8$. Each point is considered as the vertex in the graph $G_{P_\phi}$. The leftmost point is $v_i^1$ having $v_i^2$ and $v_i^3$ as its first and second neighbor respectively. Thus $v_i^1$ introduces three edges $v_i^1v_i^2$, $v_i^1v_i^3$, and $v_i^2v_i^3$ in $G_{P_\phi}$. The next point from the right is $v_i^2$ which has $v_i^1$ and $v_i^3$ has its first and second neighbor. Since the edges corresponding to the adjacency of $v_i^2$ are already drawn, it does not introduce any new edges. Then the third point $v_i^3$ introduces two new edges namely $v_i^3v_i^4$ and $v_i^2v_i^4$. Continuing in a similar fashion $v_i^{p-1}$ introduces two edges $v_i^{p-1}v_i^{p}$ and $v_i^6v_i^p$. Finally, $v_i^{p}$ does not introduce any new edges having $v_i^{p-1}$ and $v_i^6$ as its first and second neighbors respectively. The second neighborhoods of the points are indicated with colored (referring to the online version in case of gray-scale documents) circles in the figure. The leftmost red circle contains the first and second neighbor of $v_i^1$ in Figure~\ref{fig:variable}. 
The following observation holds from the structure of the variable gadget in $G_{P_\phi}$. 

\begin{observation}
    \label{obs:variable}
    The maximum independent set in a variable gadget with $6(m-1)+2$ vertices, as described above, is of size $2m-1$. There are exactly two options for selecting these two sets, either the sets consisting of the vertices $V_i^T=\{v_i^1,v_i^4,\dots,v_i^{p-1}\}$ or the sets consisting of the vertices $V_i^F=\{v_i^2,v_i^5,\dots,v_i^{p}\}$. Moreover, we can place the variable gadgets in such a way that the vertices $v_i^2,v_i^4,\dots,v_i^{p-1}$ lie above the $x$-axis and the other vertices lie below the $x$-axis.
\end{observation}

We interpret the assignments of the variable $x_i$ to be true (false resp.) if  $V_i^T$ ($V_i^F$) is chosen as the maximum independent set in the variable gadget corresponding to $x_i$. Now we describe the structure of a clause gadget. 

\vspace{0.3cm}
{\bf Clause gadget.} A clause gadget consists of a point set such that the corresponding graph contains $7$ triangles. We describe the construction for a positive clause $C=(x_i \lor x_j \lor x_k)$ and the construction is analogous for the other clauses. The \emph{central} triangle of the gadget in $G_{P_\phi}$ consists of three vertices $c_{x_i},c_{x_j}$, and $c_{x_k}$ as shown (with pink shade) in Figure~\ref{fig:clause}. The triangle is introduced by placing $c_{x_j}$ and $c_{x_k}$ as the first and second neighbor of $c_{x_i}$ respectively. 
Then we draw three other vertices for each variable appearing in the clause. We describe the position of them for $x_i$ and the others are analogous. $p_i,q_i$ and $r_i$ are the three vertices placed \emph{near} $c_{x_i}$ such that $q_i$ has $c_{x_i}$ as its first neighbor and $p_i$ as its second neighbor. This introduces $\triangle q_ip_ic_{x_i}$ in $G_{P_\phi}$. Then $p_i$ has $q_i$ and $r_i$ as its first and second neighbor respectively introducing $\triangle p_iq_ir_i$ in $G_{P_\phi}$. We adjust the neighborhood of the other vertices are adjusted in such a way that no other triangle is introduced in $G_{P_\phi}$. For example, the first and second neighbors of $r_i$ are $q_i$ and $p_i$ respectively and the first and second neighbors of $c_{x_j}$ are $c_{x_i}$ and $c_{x_k}$ respectively. Considering the position of the points in the clause gadget following observation holds for the maximum independent set in the induced graph with the points in the clause gadget as the vertices. 

\begin{observation}
    \label{obs:clause}
    The maximum independent set in a clause gadget as described above has size $4$ if and only if at least one of the three vertices $r_i,r_j$, and $r_k$ is chosen in the independent set.
\end{observation}

 We consider an ordering of the clauses from left to right in the planar embedding of $\phi$. Let $C$ be the $\xi$-th clause that is connected with the variable $x_i$ from left to right. We \emph{translate} each of the triangles $\triangle p_iq_ir_i$, $\triangle p_jq_jr_j$ and $\triangle p_kq_kr_k$ \emph{near} to the corresponding variable gadgets to \emph{connect} them. We place three points $p'_i,q'_i$,and $r'_i$ near the point $v_i^{2+6(\xi-1)}$ such that $r'_i$ has $v_i^{2+6(\xi-1)}$ and $p'_i$ as its first and second neighbour. We call the vertex $v_i^{2+6(\xi-1)}$ as the connecting vertex of $C$ and $x_j$. Then we construct a point set such that the maximum independent set induced by these points as vertices in $G_{P_\phi}$ has the following property. $p_i$($q_i$ or $r_i$ resp.) can be chosen in the maximum independent set if and only if $p'_i$ ($q'_i$ or $r'_i$) is also chosen in the maximum independent set. We call this property as the \emph{connection} of the clause to the variable gadgets. Once the connections to every clause and its variables are complete we have the following observation.

\begin{observation}
    \label{obs:assign}
    The vertex $c_{x_i}$ can be chosen in an independent set in $G_{P_\phi}$ only if the maximum independent set of the variable gadget corresponding to $x_i$ is chosen to be $V_i^T$.
\end{observation}

The remaining work is to achieve the connection from $C$ to $x_i,x_j$, and $x_k$. We use two point sets (gadgets) to perform this move. The first one is a \emph{translating gadget} and the second one is a \emph{rotating gadget} depicted in Figure~\ref{fig:rt}. The following observation can be verified from the structure of the gadgets.

\begin{observation}
    \label{obs:translate}
    Both the gadgets contain $6$ points inside them and the size of the maximum independent set is $2$.
\end{observation}

 The translating gadget starts with a copy of a triangle and \emph{propagates} the choice of the vertex chosen in an independent set in the underlying graph. The rotating gadget does the same thing but the final triangle in this gadget is a copy of the initial triangle in the gadget but rotated at a right angle. Thus we can use mirror-reflections of the rotating gadget depicted in the figure to achieve the rotating angle to be $180^{\circ}$ and $270^{\circ}$. Together we call them the connecting gadgets.  

\begin{figure}[htp]
    \centering
    \includegraphics[page=3,width=.85\textwidth]{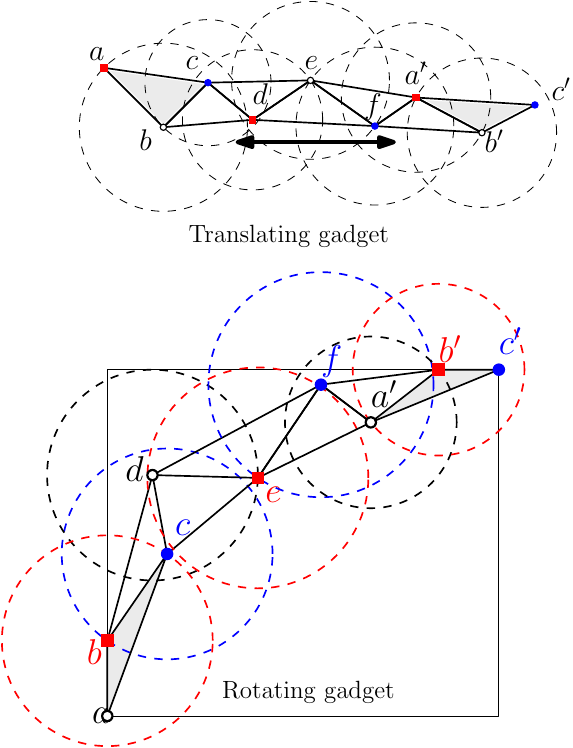}
    \caption{Translation and rotation of a triangle in $G_{P_\phi}$.}
    \label{fig:rt}
\end{figure}

 \begin{figure}[t]
    \centering
    \includegraphics[page=4,width=\textwidth]{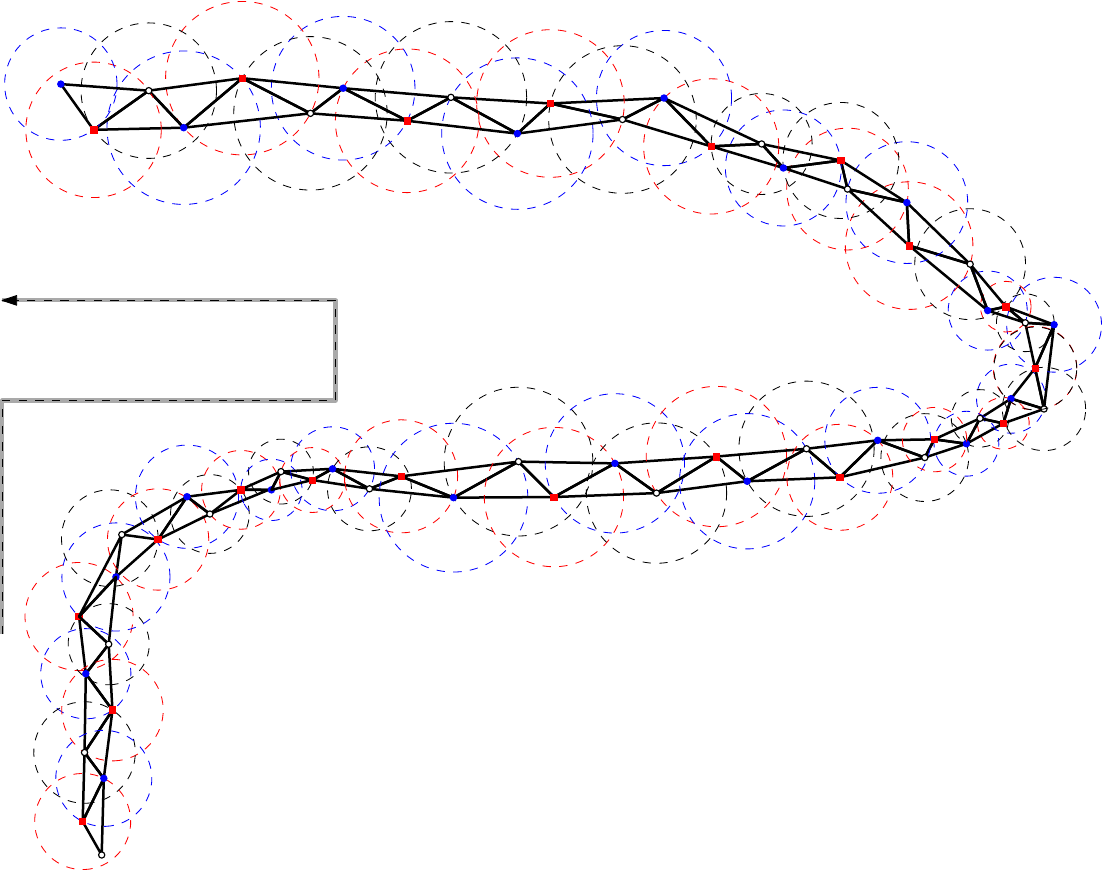}
    \caption{An example of embedding of the connecting gadgets.}
    \label{fig:emb}
\end{figure}


Now the following observation is evident from the structure of the clause gadget and its connection with the variables.

\begin{observation}
    \label{obs:number}
    Let $C=(x_i \lor x_j \lor x_k)$ one clause that is connected with its variables in $G_{P_\phi}$ with $\nu$ connecting gadgets then the maximum independent set for the point set has size $3(2m+1)+(\nu+4)$ if and only if at least one of the three variables is assigned to be true.  
\end{observation}

Now we adjust our drawing of the whole point set for $G_{P_\phi}$ such that the size of the independent set is a fixed number depending on $\phi$.  We place each clause gadget $C=(x_i \lor x_j \lor x_k)$ in such a way that the connecting gadgets for $x_j$ do not contain a rotating gadget for the connection of $c_{x_j}$ to its connecting vertex in $x_j$. We consider the fact that the nearest clause gadget from the $x$-axis can be connected with its variables by two rotating gadgets for the first and third variables and one translating gadget for the middle variable appearing in the clause.  
Then we compute the maximum number of connecting gadgets required to connect a clause with the variables.
Let $l_v$ and $w_v$ denote the length and width of the rectangles representing the variables in the planar drawing of $\phi$. We place our variable gadgets inside the rectangles. So $l_v=O(m)$ as we have placed $p=6(m-1)+2$ points in each variable gadget. Since there are $t$ variables and every clause rectangle can be connected with the variable rectangles using vertical lines, then the maximum length of a rectangle representing a clause could be $tl_v+\varepsilon$ for some small constant $\varepsilon$. On the other hand, if $w_c$ and $h_c^0$ denote the width of the clause-rectangle and the height (distance from $x$-axis) of the nearest clause-rectangle from the $x$-axis, then the distance of the farthest clause rectangle from the $x$-axis is $mw_c+h_c^0$. 
Then the maximum number of connecting gadgets required for connecting a clause gadget to its variables is at most  $3mw_c+2tl_v$.
Thus the clause gadget together with the translating and rotating gadget consists of $6(3mw_c+2tl_v)+12$ points. And thus contributing a $2(3mw_c+2tl_v)+4$ size maximum independent set. We draw our gadget in such a way that $w_c$ contains $6$ and $l_v$ contains $p$ points. Thus, every clause contains $6(18m+2tp)+12$ points together with its connecting gadgets. 

To achieve this drawing, We expand the planar embedding of $\phi$ by a scalar $\alpha = (10tm)^2$. The size of the gadgets remains the same but the space between them increases such that there is a distance of at least $\alpha$ between any two connecting gadgets to accommodate the total number of points. The \emph{convolution} to accommodate the extra points for the connection is depicted in Figure~\ref{fig:convolution}. A lower-level diagram for the embedding of such convolution with translating and rotating gadgets is depicted in Figure~\ref{fig:emb}. 

\begin{figure}[ht]
\centering
    \includegraphics[width=\textwidth,page=5]{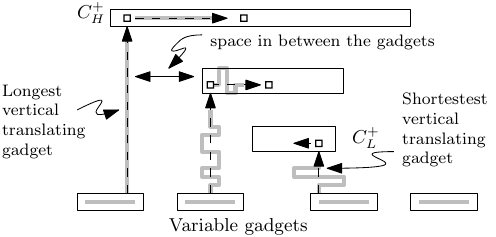}
    \caption{Convolution of the connecting gadgets in $G_{P_\phi}$.}
    \label{fig:convolution}
%
\end{figure}
Then we claim the following lemma.

\begin{lemma}
    \label{lem:NP}
    Given a PRM-3SAT formula $\phi$ we can construct a planar point set $P_\phi$ such that $G_{P_\phi}$ has a maximum independent set of size $t(m+1)+m(2(18m+2tp)+4)$ if and only if $\phi$ is satisfiable.
\end{lemma}

\begin{proof}
    If $\phi$ is satisfiable then every clause contains at least one variable that has been assigned as true. We choose $m+1$ vertices as independent set in each variable gadgets according to the truth assignment. Now for every clause $C=(x_i \lor x_j \lor x_k)$, we choose one vertex from $c_{x_i},c_{x_j}$, and $c_{x_k}$ whichever achieves a true value in the assignment. Then we choose $2(18m+2tp)$ vertices from the translation gadgets into the independent set. Since the joining vertex of  the variable gadget of a true variable will not be present in the independent set, it will not conflict with the vertex chosen from $\triangle c_{x_i}c_{x_j}c_{x_k}$. Thus we can choose $t(m+1)$ vertices from the variable gadgets and $2(18m+2tp)+4m$ vertices from the clause gadgets in the independent set. 

    On the other hand, the translation gadgets have three choices for choosing an independent set of size $2(mw_c+3tl_v)$ and the clause gadgets can contribute one vertex each towards the independent set only if at least one vertex can be chosen from the central triangle $\triangle c_{x_i}c_{x_j}c_{x_k}$. Thus we must choose the remaining $t(m+1)$ vertices in the independent set from the $t$ variable gadgets. Thus we can ensure if all clause gadgets together with their translation gadgets contributes $2(18m+2tp)+4$ vertices in the independent set then the corresponding assignment satisfies $\phi$.  
\end{proof}




Since it is possible to check the correctness of a solution for an independent set we can conclude the following theorem. 


\NPC*




\subsection{An approximation and a parameterised algorithm}\label{subsec:approximation-algorithm}

In this subsection, we proved that the maximum degree of $G_P$ is bounded.

\begin{lemma}\label{lem:degree_bounded_by_17}
    The maximum degree of the graph $G_P$ is at most $17$.
\end{lemma}

\begin{proof}
Let $p$ be any point of $P$. Consider the graph $G_P=(V,E)$ that we defined in the beginning of the Section \ref{sec:2,1/2 Multipacking in a 2D plane}. Here $V=P$.  We want to show that the degree of $p$ in $G_P$ is bounded by $17$.  Recall that, $N_r[v]$ is the subset of $P$ that includes the $r$ nearest points of $v$ and the point $v$ itself. We denote the $r$-th neighbour of $v$ by $n_r(v)$.  Here we define some set of points of $P$ that will help us to prove this lemma.
Let $S_1=\{v\in P:n_1(v)=p \text{ or }n_2(v)=p\}$,  $S_2=\{n_1(p),n_2(p)\}$ and $S_3=\{v\in P\setminus S_1\cup S_2:p,v\in N_2[w] \text{ for some }w \in P\}$. Note that, the degree of $p$ in $G_P$ is $|S_1\cup S_2\cup S_3|$. Therefore, we have to show that $|S_1\cup S_2\cup S_3|\leq 17$.

\begin{figure}[ht]
    \centering
    \includegraphics[width=0.7\textwidth]{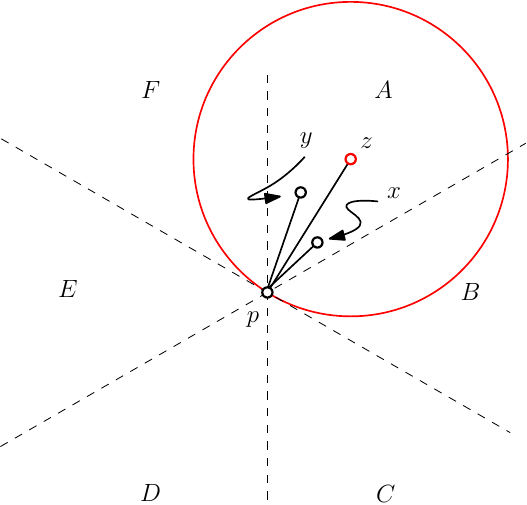}
    \caption{The maximum degree of $G_P$ is at most 17.}
    \label{fig:degree17}    
\end{figure}



\vspace{0.3cm}
\noindent
\textbf{Claim \ref{lem:degree_bounded_by_17}.1. }  $|S_1\cup S_2|\leq 12$. 
\begin{claimproof} We partition the plane by three concurrent lines meeting at $p$ that generate $6$ wedges of equal angles (See Fig.\ref{fig:degree17}). Let $A,B,C,D,E$ and $F$ be the $6$ wedges. Our goal is to prove that no wedge can contain more than two points from the set $S_1\cup S_2$. We prove this by contradiction. Without loss of generality assume that the wedge $A$ contains at least $3$ points from the set $S_1\cup S_2$. Let $S_A$ be the set of all points which belong to both $S_1\cup S_2$ and $A$. Let $z\in S_A$ be the farthest point from $p$ among all the points from $S_A$. Therefore $z\notin S_2$. This implies $p\in N_2[z]$. Let $x,y\in S_A\setminus\{z\}$. Now consider the triangle $\triangle xpz$. This is not an equilateral triangle since $px<pz$. Moreover, the angle $\angle xpz\leq 60^{\circ}$. This implies either $\angle pxz>60^{\circ}$ or $\angle pzx>60^{\circ}$.  Therefore $xz$ is not the largest side of $\triangle xpz$. Now $px<pz$ implies $pz$ is the largest side of $\triangle xpz$. Therefore $x\in N_2[z]$. Similarly, we can show that $y\in N_2[z]$. Therefore, $x,y,z,p\in N_2[z]$. This is a contradiction.   
\end{claimproof}

\vspace{0.3cm}
\noindent
\textbf{Claim \ref{lem:degree_bounded_by_17}.2. }  $|S_3|\leq 5$. 
\begin{claimproof} Suppose $|S_3|\geq 6$. Define a function $f:S_3\rightarrow P$ where we say $f(v)=w$ if either $v$ is the first neighbor of $w$ and $p$ is the second neighbor of $w$ or $v$ is the second neighbor of $w$ and $p$ is the first neighbor of $w$. Note that $f$ is an injective function. Therefore, the size of the image set, $|f(S_3)|\geq 6$. Moreover, $f(S_3)\subset S_1\cup S_2$. Therefore, $f(S_3)\cap S_3=\phi$. Now we remove all the points of $S_3$ from $P$. Let $P'=P\setminus S_3$. In $P'$, the first neighbor of each point of $f(S_3)$ is $p$. Therefore, there are at least $6$ points in $P'$ that have $p$ as their first neighbor. Then there exist two points $v_1,v_2\in f(S_3)$ such that $\angle v_1pv_2\leq 60^{\circ}$. We have $pv_1\neq pv_2$, since each vertex has only one $r$-th neighbor, for any $r$. Without loss of generality assume that $pv_1> pv_2$. This implies $pv_1$ is the largest side in $\triangle v_1pv_2$. Therefore, $p$ cannot be the first neighbor of $v_1$. This is a contradiction. 
\end{claimproof}

From the above two claims, we can say that $|S_1\cup S_2\cup S_3|\leq 17$.  
\end{proof}

Using this we provide approximation and parameterized solutions to the $2$-multipacking problem of a given set of points $P$ in the plane. We use existing techniques for computing the maximum independent set for bounded degree graphs.

The problem of finding the Maximum Independent Set in a bounded degree graph is denoted by \textit{MAX IS-$\Delta$} where the maximum node degree of the graph is bounded above by $\Delta$. The following result states the approximation bound of the \textit{MAX IS-$\Delta$} problem.

\begin{theorem}[Berman and  Fujito \cite{berman1999approximation}]\label{thm:App17}
    \textit{MAX IS-$\Delta$} can be approximated in polynomial time within a ratio arbitrarily close to $(\Delta + 3)/5$ for all $\Delta \geq 2$.
\end{theorem}

Theorem \ref{thm:App17} and Lemma \ref{lem:degree_bounded_by_17} yield the following.

\begin{restatable}{theorem}{thmapprox}\label{thm:approx}
 A maximum $2$-multipacking of a given set of points $P$ in $\mathbb{R}^2$ can be approximated in polynomial time within a ratio arbitrarily close to $4$.
\end{restatable}


In 2019, Bonnet et al. have shown the following:

\begin{theorem}[Bonnet et al.\cite{bonnet2019maximum}]\label{thm:FPT_for_bounded_degree_graph}
    For graphs of $n$ vertices with degree bounded by $\Delta$, we can compute an independent set of size $k$ in  time $(\Delta+1)^kn^{O(1)}$, if it exists.
\end{theorem}

Lemma \ref{lem:degree_bounded_by_17} and Theorem \ref{thm:FPT_for_bounded_degree_graph} yield the following result.

\begin{restatable}{theorem}{thmFPT}\label{thm:FPT}
 A $2$-multipacking of size $k$ of a given set of $n$ points in $\mathbb{R}^2$ can be computed in time $18^kn^{O(1)}$, if it exists.
\end{restatable}

\section{Conclusion}
\label{sec:conclusion_geo_multi}
In this chapter, we study some multipacking problems in a geometric setting under the Euclidean distance metric.  An intriguing open topic in this area is addressing the difficulty of computing a maximum multipacking (or, $(n-1)$-multipacking) in general. We know neither the algorithmic difficulties nor the limitations as of yet for the multipacking (or, $(n-1)$-multipacking) problem.

We have determined the value of $\MMP(2)$. A natural direction for future research is determining the exact values of $\MMP(t)$ for $t \geq 3$.

Additionally, the definition of $r$-multipacking in geometry, inspired by obnoxious facility location and graph-theoretic multipacking, can be generalized by introducing a fractional parameter $k \in [0,1]$. Specifically, the condition on a set $M$ could be generalized to $|N_s[v] \cap M| \leq k(s+1)$ for each $v \in P$. Varying $k$ may lead to different bounds on the multipacking number and introduce new computational challenges.

Further extensions of this work could explore multipackings in higher-dimensional Euclidean spaces $\mathbb{R}^d$, where geometric constraints and packing densities may yield richer structural and algorithmic results.





\chapter{Dominating broadcast in $d$-dimensional space} \label{chapter:geo_broad}\hypertarget{chapter:introhref}{}

Here we study  the \textsc{Minimum dominating broadcast} (MDB) problem for geometric point set in $\mathbb{R}^d$ where the pairwise distance between the points are measured in Euclidean metric. We present an $O(n\log n)$-time algorithm for solving the MDB problem on a point set in $\mathbb{R}^d$. We provide bounds of the broadcast domination  number using kissing number. Further, we prove tight upper and lower bounds of the broadcast domination  number of a point set in $\mathbb{R}^2$.

\section{Chapter overview}
In Section \ref{sec:Preliminaries_geo_broad}, we recall some definitions and notations. In Section \ref{sec:Broadcast domination in a 2D plane}, we present some observations on the minimum dominating broadcast for a point set in $\mathbb{R}^d$. In Section~\ref{sec:algo}, there is an $O(n\log n)$  time algorithm to compute the same for $n$ points in $\mathbb{R}^d$. Then we discuss some upper and lower bounds on the broadcast domination number of a point set in Section~\ref{sec:bound}. Finally, we conclude this chapter in Section \ref{sec:conclusion_geo_broad}.

\section{Preliminaries}\label{sec:Preliminaries_geo_broad}
We consider a point $v \in P$. We denote the subset of $P$ that has only the $r$ nearest points of $v$ and the point $v$ itself by $N_r[v]$. Formally, $N_r[v] = \{ p \in P | \text{ $p$ is the $i$-th neighbor of $v$}, 1 \leq i \leq r \} \cup \{v\}$. Hence, $|N_r[v]|=r+1$ for $1\leq r\leq n-1$.

A function 	$ f : P \rightarrow \{0, 1, 2, \dots , n-1\} $ is called a \emph{broadcast} on $P$ and $v\in P$ is a \textit{tower} of $P$ if $f(v)>0$.



\begin{definition}[Dominating Broadcast]
Let $f$ be a broadcast on $P$. For each point $ u \in P  $, if  there exists a point $ v $ in $ P $ (possibly, $ u = v $) such that $ f (v) > 0 $ and $ u\in N_{f(v)}[v] $, then $ f $ is called a \emph{dominating broadcast}  or \emph{broadcast domination} on $ P$.
\end{definition} 

 The \textit{cost} of the broadcast $f$ is the quantity $ \sigma(f)$, which is the sum of the weights of the broadcasts over all the points in $ P $. Therefore, $\sigma(f)=  \sum_{v\in P}f(v)$.

\begin{definition}[Broadcast Domination Number]  
 The minimum cost of a dominating broadcast in $P$ (taken over all dominating broadcasts)  is the \textit{broadcast domination number} of $P$, denoted by $ \gamma_{b}(P) $.  So, $ \gamma_{b}(P) = \min_{f\in D(P)} \sigma(f)= \min_{f\in D(P)} \sum_{v\in P}f(v)$, where $D(P)$ is the set of all dominating broadcasts on $P$.
\end{definition}

\begin{definition}[Minimum Dominating Broadcast]
A \emph{minimum dominating broadcast} (or \emph{optimal broadcast}, or \emph{minimum broadcast}) on $ P$ is a dominating broadcast on $ P$ with cost equal to $ \gamma_{b}(P)$.
\end{definition}

\medskip
\noindent
\fbox{%
  \begin{minipage}{\dimexpr\linewidth-2\fboxsep-2\fboxrule}
  \textsc{ Minimum Dominating Broadcast} problem
  \begin{itemize}[leftmargin=1.5em]
    \item \textbf{Input:} A point set $P$ in $\mathbb{R}^d$.
    \item \textbf{Output:} A minimum dominating broadcast $f$ on $P$.
  \end{itemize}
  \end{minipage}%
}
\medskip




The \textit{kissing number}, denoted by $\tau_d$, is the maximum number of non-overlapping (i.e., having disjoint interiors)  unit balls in the Euclidean $d$-dimensional space $\mathbb{R}^d$ that can touch a given unit ball. For a survey of kissing numbers, see \cite{boyvalenkov2015survey}. 

The values of $\tau_d$ for small dimensions are known: $\tau_1=2$, $\tau_2=6$, $\tau_3=12$ \cite{anstreicher2004thirteen,boroczky2003newton,maehara2001isoperimetric,schutte1952problem}, $\tau_4=24$ \cite{musin2008kissing} etc. However, the exact value of $\tau_d$ for general $d$ remains unknown. Nevertheless, bounds for $\tau_d$ are available.

\section{Observations on the dominating broadcast for point set in $\mathbb{R}^d$}
\label{sec:Broadcast domination in a 2D plane}

In this section, we prove a lemma and a theorem. Using those we give an algorithm to find a minimum broadcast for a point set.

\begin{lemma}\label{lem:f(w)=0}
Let $f$ be a minimum broadcast on $P$. If $f(u)>1$ for some $u \in P$, then $f(w)=0$, for each $w \in N_{f(u)}[u]\setminus\{u,u_1\}$, where $u_1$ is the nearest neighbor of $u$.    
\end{lemma}

\begin{proof}
 Let $f(u) = k > 1$, and $N_k[u] = \{u, u_1, u_2, \hdots , u_k\}$, where $u_i$ is the $i$-th neighbor of $u$. Suppose $f(u_i) \geq 1$ for some $u_i \in N_{k}[u]\setminus\{u,u_1\}$. Let $ X = \{v \in P : v \in N_{k}[u]\setminus\{u,u_1\}$ and $f(v)>0 \}$ and
 $Y = \{v \in P : v \in N_{k}[u]\setminus\{u,u_1\}$ and $f(v)=0 \}$.
 Therefore, $|X| \geq 1$ and $|X| + |Y| = k-1$. Now, define a new broadcast $f'$ on $P$, where  $f'(v)=1$ if  $v\in Y \cup \{u\}$, otherwise $f'(v)=f(v)$. Note that, $f'(v)$ is a dominating broadcast. Observe that, $ \sum_{v \in Y \cup \{u\}} f(v) = f(u) + \sum_{v \in Y} f(v) = k + 0 = k$. Also, $\sum_{v \in Y \cup \{u\}} f'(v) = f'(u) + \sum_{v \in Y} f'(v) = 1 +  \sum_{v \in Y} 1 = 1 + |Y|$. Therefore, $\sum_{v \in Y \cup \{u\}} f(v) =  k > k-1 = |x| + |y| \geq 1 + |Y| = \sum_{v \in Y \cup \{u\}} f'(v)$. Thus, $\sigma(f) > \sigma(f')$. This is a contradiction, since $f$ is a minimum broadcast.  
\end{proof}


\thmbroadcastzeroone*

\begin{proof}
   Let, $V_f = \{ v \in P : f(v) > 1 \}$. Let $f$  be a minimum broadcast on $P$ whose $|V_f|$ is minimum. To prove the theorem, it is enough to show that, $V_f = \phi$.

 Suppose $V_f \neq \phi$. Let $u\in V_f$ and $N_{f(u)}[u] = \{ u, u_1, \hdots , u_{f(u)}\}$ where $u_i$ is the $i$-th neighbor of $u$. By lemma \ref{lem:f(w)=0}, we have, $f(u_i) = 0$, for each $2\leq i \leq f(u)$. Define a new dominating broadcast $f'$ on $P$, where $f'(v)=1$ if $v\in N_{f(u)}[u]\setminus\{u_1\}$, otherwise $f'(v)=f(v)$. Here, $f'$ is a dominating broadcast and $ \sum_{v \in P} f(v) = \sum_{v \in P} f'(v)$. Therefore, $f'$ is a minimum broadcast. Note that, $V_{f'} = V_f \setminus \{u\}$. Therefore, $|V_{f'}| < |V_f|$. This is a contradiction, since $V_f$ was minimum. Therefore, $V_f = \phi$. Hence proved.
\end{proof}


From Theorem \ref{thm:f(u)in0,1},
we say $f$ is a $(0,1)$-\textit{minimum broadcast} on $P$ if $f$ is a minimum broadcast  where $f(u)\in \{0,1\}$ for each point $ u \in P  $.

\section{Algorithm to find a minimum broadcast}
\label{sec:algo}
We say, $f$ is a \textit{(0,1)-dominating broadcast} on $P$, if $f$ is a dominating broadcast and $f(p) \in \{0,1\}$ for each point $p \in P$.
The algorithm for finding minimum broadcast, based on theorem \ref{thm:f(u)in0,1}, takes a point set $P$ in $\mathbb{R}^d$ as input and outputs a $(0,1)$-minimum broadcast $f$. Consequently we have, $\sum_{p \in P} f(p) = \gamma_b(P)$. So we construct a \emph{Nearest-Neighbor-Graph (NNG)} \cite{eppstein1997nearest} of the point set $P$, and denote it by $G(V,E)$ where $V=P$.  If $p,q\in V$, then there is a directed edge from $p$ to $q$ whenever $q$ is a nearest neighbor of $p$ ($q$ is a point whose distance from $p$ is minimum among all the given points other than $p$ itself). Each vertex of the graph will have exactly one outgoing edge. The only cycles in $G$ are $2$-cycles. For $|V| \geq 2$, each weakly connected component\footnote{In a directed graph, a \textit{weakly connected component} is a maximal subgraph in which every pair of vertices is connected by an undirected path—that is, if the directions of the edges are ignored, the subgraph is connected.} $C$ of $G$ contains exactly one $2$-cycle. This pair of vertices is called the bi-root of $C$. For more information on \emph{Nearest-Neighbor-Graph (NNG)} see \cite{eppstein1997nearest}. 

Now we prove that the \emph{minimum edge cover (MEC)} of this graph will correspond to a minimum broadcast. A \emph{minimum edge cover} of a graph is the minimum cardinality subset of edges of the graph such that each vertex is adjacent to at least one of the edges in the set. 

\begin{lemma}
    \label{lem:MEC} Let $P$ be a point set on $\mathbb{R}^d$ and $G$ be the NNG of $P$. Let $ \rho(G)$ be the minimum edge cover number of $G$. Then $ \gamma_b(P)=\rho(G)$.

\end{lemma}

\begin{proof} 
Let $f$ be a $(0,1)$-minimum broadcast and  $T = \{ v\in P : f(v)=1 \}$. 
Note that, each tower in $T$ dominates exactly one point which is the nearest neighbor from that tower. Let $v$ be a tower that broadcasts $u$ through the edge $\overrightarrow{vu}$ of $G$. Let $E$ be the edge set of $G$. We denote $e_v=\overrightarrow{vu}$. Let $E_1=\{e_v\in E: v\in T\}$. Therefore, $|E_1|=|T|=\gamma_b(P)$. Now we want to show that $E_1$ is an edge cover of $G$. For each vertex $w$  in $G$, it is enough to show that there is an adjacent edge of $w$ which is in $E_1$. If $w\in T$, then the edge $e_w\in E_1$. Now consider the case when $w\notin T$. Then there is a tower $x\in T$ which broadcasts $w$. Thus $e_x=\overrightarrow{xw}\in E_1$. Therefore, $E_1$ is an edge cover of $G$. This implies $\rho(G)\leq |E_1|$, therefore $ \rho(G)\leq \gamma_b(P)$.  

Suppose $E_2$ is a minimum edge cover of $G$. Define a broadcast $f'$ on $P$ where \begin{equation}
f'(v)=
    \begin{cases}
        1 & \text{if } \overrightarrow{vu}\in E_2\\
        0 & \text{otherwise. } 
    \end{cases}
\end{equation}
Note that, $f'(v)$ is a dominating broadcast. Moreover, $\sum_{v\in P}f'(v)=|E_2|=\rho(G)$, since  each vertex of $G$ has exactly one outgoing edge.  Therefore,  $ \rho(G)\geq \gamma_b(P)$.
\end{proof}

\RestyleAlgo{ruled}    
\begin{algorithm}[htbp]
    \SetAlgoLined
    \KwIn{A point set $P$ in $\mathbb{R}^d$.}
    \KwOut{A minimum broadcast on $P$.}
    Construct the Nearest neighbor graph $G$ of $P$\;
    Compute the minimum edge cover $M$ of the graph $G$\;
    \For{every edge $e \in M$}{
        \eIf{$e$ is outgoing from a vertex $v$}{
            $f(v)=1$\;
        }{
        $f(v)=0$\;
            }
    }
 \caption{Minimum broadcast}
 \label{algo:2_factor}
\end{algorithm}

 Thus we find the minimum edge cover of $G$ and return the set of the towers in $P$, which gives an $(0,1)$-minimum broadcast and hence a minimum broadcast.

 We take each edge $vu$ from MEC and construct a tower of strength $1$ on the vertex $v$ if the edge goes from $v$ to $u$. We do this for every edge in MEC.   For a more formal reading of the algorithm, see Appendix[\ref{algo:2_factor}].

   For a point set $P$ in $\mathbb{R}^2$, the NNG of $P$ can be computed in $O(n\log n)$ time~\cite{callahan1993optimal}. To compute the NNG of a point set $P$ in $\mathbb{R}^d$, there is an $O(n \log n)$-time algorithm, but with a constant depending exponentially on the dimension \cite{clarkson1983fast,vaidya1989n}. Note that, the NNG, $G$ of a point set $P$ in $\mathbb{R}^d$ has only one directed cycle for each connected component of $G$, also the cycle has length $2$ (See~\cite{eppstein1997nearest}).  Therefore, if we think of those directed cycles as edges, then $G$ can be seen as a forest (collection of trees). It is known that maximum matching in trees can be computed in $O(n)$ time~\cite{savage1980maximum}. Therefore, maximum matching in $G$ can be computed in $O(n)$ time. The computation of MEC (Minimum Edge Cover) in a graph typically begins with finding a maximum matching and then extending it by greedily adding edges to cover the remaining uncovered vertices. Therefore, MEC in $G$ can be computed in $O(n)$ time. This yields the following.

\thmbroadcastalgortihmcomplexity*

%

\section{Tight bounds on the minimum broadcast domination number}
\label{sec:bound}

In this section we discuss the bounds on domination number of a point set.
\thmgammablowerboundford*

\begin{proof} Let $T = \{ v\in P : f(v)=1 \}$ be the set of towers and $N = \{v \in P : f(v)=0 \}$ be the set of non-tower points.  From Theorem \ref{thm:f(u)in0,1}, there is a $(0,1)$-dominating broadcast $f$.  Therefore $\frac{n}{2} \leq \gamma_b(P) .$ For the tight bound example see fig. \ref{fig:lbd}.

\begin{figure}[htbp]
\centering
\includegraphics[width=\linewidth]{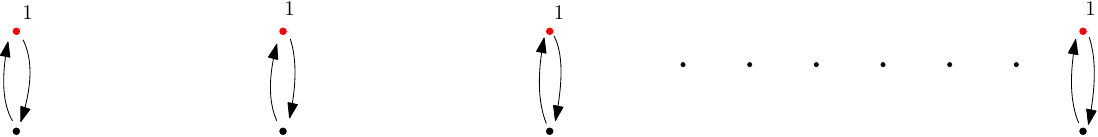}
\caption{Example showing the tight lower bound for minimum broadcast in $\mathbb{R}^d$.}
\label{fig:lbd}
\end{figure}

Now we show the upper bound of $\gamma_b(P)$. Let $f$ be a $(0,1)$-minimum broadcast on $P$. Let $G=(P, E)$ be the \emph{Nearest-Neighbor-Graph (NNG)} of the point set $P$.  Also $|T| = \gamma_b(P)$. Let $N = \{u_1,u_2,\dots,u_k\}$ and let $S_i=\{v\in T: \text{ $v$ broadcasts $u_i$} \}\cup \{u_i\}$ for each $i\in [k]$. Note that, the only non-tower point in $S_i$ is $u_i$.

\vspace{0.2cm}
\noindent
\textbf{Claim 1.} $S_i \cap S_j = \phi$ for $i \neq j$.

\begin{claimproof} Suppose $x \in S_i \cap S_j$. If $x\in N$, then $i=j$. This is not possible. Therefore,  $x\notin N$. If $x\in T$, then $x$ broadcasts $u_i$ and $u_j$. But $x$ broadcasts only one point. Therefore, $i=j$, which is a contradiction. Hence the claim is true.  \end{claimproof} 

 Let $t_i$ be the number of towers in $S_i$ and $\Delta(G)$ be the maximum indegree of the graph $G$. Therefore, $t_i\leq\Delta(G)$ for each $i\in [k]$. Since $G$ is an NNG of a point set in $\mathbb{R}^d$, we have $\Delta(G) \leq \tau_d$, the kissing number on $d$-dimension.  Therefore $t_i\leq \tau_d$ for each $i\in [k]$. Thus, $\frac{1}{\tau_d} \leq \frac{1}{t_i} \implies \frac{\tau_d+1}{\tau_d} \leq \frac{t_i +1}{t_i}  \implies \sum_{i =1}^k \frac{t_i}{t_i +1}(t_i + 1) \leq \sum_{i =1}^k \frac{\tau_d}{\tau_d+1}|S_i|$$ $$ \implies \gamma_b (P) \leq \frac{\tau_d}{\tau_d+1}n.$ 
 \end{proof}

Recall that, we have the assumption: the $r$-th neighbor of every point is unique, for each $r$. When the point set $P$ is in $\mathbb{R}^2$, the kissing number $\tau_2=6$. In that case there exist $6$ points in $P$ whose first neighbors are not unique, which violates our assumption. Therefore, the NNG of $P$ cannot have maximum indegree $6$. So, the NNG $G$ has maximum indegree at most $5$. Thus we have the following.

  \thmgammablowerbound*

For the tight upper bound example in $\mathbb{R}^2$ we refer to fig. \ref{fig:upb}.

\begin{figure}[htbp]
\centering
\includegraphics[width=\linewidth]{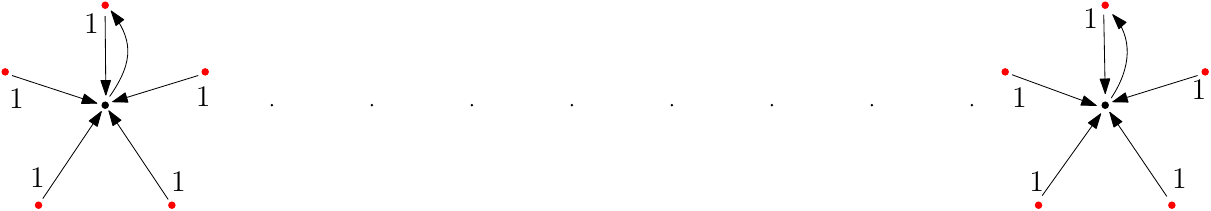}
\caption{Example showing tight upper bound for broadcast domination number in $\mathbb{R}^2$.}
\label{fig:upb}
\end{figure}

\section{Conclusion}
\label{sec:conclusion_geo_broad} We have addressed the problem of computing the minimum dominating broadcast for a point set in $\mathbb{R}^d$. Our polynomial time algorithm relies on computing an edge cover of the  Nearest-Neighbor-Graph (NNG) of the point set. However, the bounds of the broadcast domination number in higher dimensions require exact values of the \emph{kissing numbers} in higher dimensions~\cite{li2025generalized}, which itself is a long standing open problem.

\chapter{Conclusion} \label{chapter:conclude}\hypertarget{chapter:introhref}{}

We study the hardness of the \textsc{Multipacking} problem and we prove that the \textsc{Multipacking} problem is \textsc{NP-complete} for undirected graphs. Moreover, it is \textsc{W[2]-hard} for undirected graphs when parameterized by the solution size. Furthermore, we have shown that the problem is \textsc{NP-complete} and \textsc{W[2]-hard} (when parameterized by the solution size) even for various subclasses.   We study the relationship between $\MP(G)$ and $\gamma_b(G)$ for chordal, cactus, and $\delta$-hyperbolic graphs. We provide approximation algorithms for the \textsc{Multipacking} problem for chordal, cactus, and $\delta$-hyperbolic graphs. Further, we studied computational hardness and algorithms for multipacking and broadcast domination problems in geometric domain. In this chapter, we give some future directions and some preliminary work related to them.

\section{Complexity of \textsc{Multipacking}}

In Chapter~\ref{chapter:NPcomplete}, we have analyzed the computational complexity of the \textsc{Multipacking} problem. Specifically, we proved that the problem is \textsc{NP-complete} and established its \textsc{W[2]}-hardness when parameterized by the size of the solution. Additionally, we extended these hardness results to various well-studied graph classes, such as chordal $\cap$ $\frac{1}{2}$-hyperbolic graphs, bipartite graphs, claw-free graphs, regular graphs, and \textsc{CONV} graphs.

Our findings open up several avenues for future research, including the following questions:

\begin{itemize}
    \item Is there a fixed-parameter tractable (FPT) algorithm for general graphs with respect to treewidth?
    \item Is it possible to design a sub-exponential time algorithm for this problem?
    \item What is the complexity status of the \textsc{Multipacking} problem on planar graphs? In particular, does an FPT algorithm exist for planar graphs when parameterized by the solution size?
    \item While a $(2 + o(1))$-approximation algorithm is already known~\cite{beaudou2019broadcast}, can this be improved to a polynomial-time approximation scheme (PTAS)?
\end{itemize}

Progress on these questions would significantly enhance our understanding of the problem's algorithmic and structural properties and could lead to the development of new algorithmic techniques.

\section{Multipacking on a bounded hyperbolic graph-class: chordal graphs}

In Chapter~\ref{chapter:chordal}, we have shown that for any connected chordal graph $G$, the bound $\gamma_{b}(G) \leq \big\lceil \frac{3}{2} \MP(G) \big\rceil$ holds, whereas $\gamma_b(G) \leq 2\MP(G) + 3$ is the known bound for arbitrary graphs. 

While it is known that for strongly chordal graphs, $\gamma_{b}(G) = \MP(G)$, our results demonstrate that this equality does not extend to all connected chordal graphs. In fact, we have shown that the gap $\gamma_{b}(G) - \MP(G)$ can be made arbitrarily large. To support this, we constructed infinitely many connected chordal graphs $G$ for which $\gamma_{b}(G) / \MP(G) = 10/9$ and where $\MP(G)$ is arbitrarily large. Further, we improve this by constructing infinitely many connected chordal graphs $G$ for which $\gamma_{b}(G) / \MP(G) = 4/3$ and where $\MP(G)$ is arbitrarily large. Hence, we determined that, for connected chordal graphs, the expression \[
\lim_{\MP(G)\to \infty} \sup \left\{ \frac{\gamma_{b}(G)}{\MP(G)} \right\}\in [4/3,3/2].
\] A natural direction for future research is to determine the exact value of the expression for connected chordal graphs, whereas the exact value of the expression is $2$ for general connected graphs \cite{rajendraprasad2025multipacking}. Exploring this ratio for other graph classes, such as the one depicted in Figure~\ref{fig:diagram}, also presents an interesting line of investigation.


\section{Multipacking on an unbounded hyperbolic graph-class: cactus graphs}

In Chapter~\ref{chapter:cactus}, we established that for any cactus graph $G$, $\gamma_b(G) \leq \frac{3}{2}\MP(G) + \frac{11}{2}$. In addition, we presented a $(\frac{3}{2} + o(1))$-approximation algorithm for computing multipacking in cactus graphs.

A natural direction for future work is to explore whether an exact polynomial-time algorithm can be designed for finding a maximum multipacking in cactus graphs.


\section{Complexity of $r$-Multipacking}

In Chapter \ref{chapter:k_multi}, we established the \textsc{NP}-completeness of the \textsc{$r$-Multipacking} problem. Additionally, we extended these hardness results to several well-known graph classes, including planar bipartite graphs with bounded degree, chordal graphs with bounded diameter, and bipartite graphs with bounded diameter.

These results naturally lead to several open problems that merit further investigation:
\begin{enumerate}
    \item What is the position of the \textsc{$r$-Multipacking} problem within the $W$-hierarchy? In particular, does there exist an \textsc{FPT} algorithm when parameterized by the solution size?
    
    \item Does an \textsc{FPT} algorithm exist for general graphs when parameterized by treewidth?
    
    
    \item What is the approximation hardness of the \textsc{$r$-Multipacking} problem?
\end{enumerate}

Resolving these questions would contribute to a deeper understanding of the algorithmic and complexity-theoretic aspects of the \textsc{$r$-Multipacking} problem.

\section{Multipacking in geometric domain}

In Chapter~\ref{chapter:geo_multi}, we have explored multipacking problems within a geometric framework, specifically under the Euclidean distance metric. One compelling and unresolved question in this domain is the complexity of computing a maximum multipacking (or, $(n-1)$-multipacking). At present, little is known about the algorithmic challenges or theoretical limits of this extreme case.

As part of our contributions, we have determined the value of $\MMP(2)$. A natural next step is to investigate the exact values of $\MMP(t)$ for higher values of $t$, especially for $t \geq 3$, which may reveal deeper structural insights.

We also propose a generalized variant of geometric multipacking inspired by connections to obnoxious facility location problems and classical graph-theoretic multipacking. In this generalized version, a fractional parameter $k \in [0,1]$ controls the packing condition, requiring that for every point $v \in P$, the set $M$ satisfies $|N_s[v] \cap M| \leq k(s+1)$. Exploring how the multipacking number behaves under varying values of $k$ could lead to new bounds and computational problems.

Further extensions of this work could explore multipackings in higher-dimensional Euclidean spaces $\mathbb{R}^d$. Studying multipacking in higher dimensions may uncover richer geometric structures and pose new algorithmic challenges related to distance, packing density, and dimensional constraints.

\section{Dominating broadcast in $d$-dimensional space}

In Chapter~\ref{chapter:geo_broad}, we investigated the problem of determining the minimum dominating broadcast for a set of points in $\mathbb{R}^d$. Our approach leverages a polynomial-time algorithm based on computing an edge cover of the Nearest-Neighbor Graph (NNG) constructed from the point set. 

Obtaining tight bounds on the broadcast domination number in higher dimensions poses additional challenges. These bounds are closely tied to the precise values of the \emph{kissing numbers} in higher dimensions~\cite{li2025generalized}, a classical problem in discrete geometry that remains unresolved in general.

\clearpage
\addcontentsline{toc}{chapter}{References}
\newpage
\vspace{10cm}

\Large
\noindent \textbf{List of Publications from the Content of this Thesis} 

\vspace{0.5cm}

\normalsize
In this thesis, the results from the publications or preprints (the list is given below) are included as follows:
\begin{itemize}
    \item Paper [7] in Chapter~\ref{chapter:NPcomplete}
    \item Papers [1] and [4] in Chapter~\ref{chapter:chordal}
    \item Papers [3] and [6] in Chapter~\ref{chapter:cactus}
    \item Paper [8] in Chapter~\ref{chapter:k_multi}
    \item Papers [2] and [5] in Chapter~\ref{chapter:geo_multi}
    \item Paper [9] in Chapter~\ref{chapter:geo_broad}
\end{itemize}

\vspace{0.6cm}

\normalsize
\noindent \textbf{Conference Publications}
\begin{itemize}

\bibitem{das2023caldam}
Sandip Das, Florent Foucaud, Sk Samim Islam, and Joydeep Mukherjee. 
"\textbf{Relation between broadcast domination and multipacking numbers on chordal graphs}". 
In \textit{Conference on Algorithms and Discrete Applied Mathematics (CALDAM)}, 
pages 297--308. Springer, year 2023.


\url{https://doi.org/10.1007/978-3-031-25211-2_23}

\bibitem{das2025caldam} 
Arun Kumar Das, Sandip Das, Sk Samim Islam, Ritam Manna Mitra, and Bodhayan Roy. 
"\textbf{Multipacking in the Euclidean metric space}". 
In \textit{Conference on Algorithms and Discrete Applied Mathematics (CALDAM)}, 
pages 109--120. Springer, year 2025.


\url{https://doi.org/10.1007/978-3-031-83438-7_10}

\bibitem{das2025walcom} 
Sandip Das and Sk Samim Islam. 
"\textbf{Multipacking and broadcast domination on cactus graphs and its impact on hyperbolic graphs}". 
In \textit{International Conference and Workshops on Algorithms and Computation (WALCOM)}, 
pages 111--126. Springer, year 2025.


\url{https://doi.org/10.1007/978-981-96-2845-2_8}

\end{itemize}

\vspace{0.2cm}

\noindent \textbf{Journal Submissions}
\begin{itemize}

\bibitem{das2023CALDAMarxiv} 
Sandip Das, Florent Foucaud, Sk Samim Islam, and Joydeep Mukherjee. 
"\textbf{Relation between broadcast domination and multipacking numbers on chordal and other hyperbolic graphs}". Submitted in \emph{Discrete Applied Mathematics (DAM)}. Manuscript no: DA16107. 
\textit{arXiv preprint}, year 2023. 

\href{https://arxiv.org/abs/2312.10485}{arXiv:2312.10485}


    \bibitem{das2024CALDAMarxiv} 
Arun Kumar Das, Sandip Das, Sk Samim Islam, Ritam Manna Mitra, and Bodhayan Roy. 
"\textbf{Multipacking in Euclidean metric space}". Submitted in \emph{Computational Geometry: Theory and Applications (CGTA)}. Manuscript no: CGTA-D-25-00035. 
\textit{arXiv preprint}, year 2024.

\href{https://arxiv.org/abs/2411.12351}{arXiv:2411.12351}

\bibitem{das2023WALCOMarxiv}
Sandip Das and Sk Samim Islam. 
"\textbf{Multipacking and broadcast domination on cactus graphs and its impact on hyperbolic graphs}". Submitted in \emph{Theoretical Computer Science (TCS)}. Manuscript no: TCS-D-25-00486. 
\textit{arXiv preprint}, year 2023.

\href{https://arxiv.org/abs/2308.04882}{arXiv:2308.04882}

\end{itemize}

\vspace{0.2cm}

\noindent \textbf{Manuscripts to Submit}
\begin{itemize}
 \bibitem{das2025hardnessMultipacking} 
Sandip Das, Sk Samim Islam, and Daniel Lokshtanov. 
"\textbf{On the complexity of Multipacking}". 
 \textit{preprint}, year 2025.

 \bibitem{das2025rMultipacking} 
Arun Kumar Das, Sandip Das, Sk Samim Islam, Ritam Manna Mitra, and Bodhayan Roy. 
"\textbf{On the complexity of $r$-Multipacking}". 
 \textit{preprint}, year 2025.

 \bibitem{das2025rbroadcast} 
Arun Kumar Das, Sandip Das, Sk Samim Islam, Ritam Manna Mitra, and Bodhayan Roy. 
"\textbf{Dominating broadcast in $d$-dimensional space}". 
 \textit{preprint}, year 2025.

\end{itemize}

\newpage

\nocite{*}
\bibliographystyle{plainurl}
\bibliography{references}
\end{document}